\documentclass[a4paper,12pt]{article}

\usepackage{docmute}
\usepackage[margin=1.0in]{geometry}
\usepackage{hyperref}
\usepackage[sort,numbers]{natbib}
\usepackage{multirow}
\usepackage{esvect}
\usepackage{tabularx}
\usepackage{amsmath,amsfonts,amssymb,xspace}
\usepackage{amsthm}
\usepackage{tikz}
\usetikzlibrary{fit,shapes}
\usepackage{etoolbox}
\usepackage{datetime}
\usepackage{amssymb}
\def\matrixgap{\hspace{5pt}}
\def\beginmatrix{\left(\!\!\begin{array}{c@{\matrixgap}c@{\matrixgap}c@{\matrixgap}c}}
\def\endmatrix{\end{array}\!\!\right)}
\def\vs{\vphantom{\frac 2 {\sqrt{5}}}}

\newenvironment{pic}[1][]
{\begin{aligned}\begin{tikzpicture}[#1]}
{\end{tikzpicture}\end{aligned}}
\newcommand{\edges}[1][]%
{%\end{scope}\end{pgfonlayer}\begin{pgfonlayer}{foreground}\begin{scope}[#1]
}
\theoremstyle{plain}
\newtheorem{theorem}{Theorem}
\newtheorem{lemma}[theorem]{Lemma}
\newtheorem{proposition}[theorem]{Proposition}
\newtheorem{corollary}[theorem]{Corollary}

\theoremstyle{definition} 
\newtheorem{definition}[theorem]{Definition} 
\newtheorem{example}[theorem]{Example} 
 
\newtheorem*{conj*}{Conjecture}

\newtheorem{remark}[theorem]{Remark}
\newcommand{\tr}{\text{Tr}}
\newcommand{\ob}{\text{Ob}}
\makeatletter
\def\calign@preamble{%
   &\hfil\strut@
    \setboxz@h{\@lign$\m@th\displaystyle{##}$}%
    \ifmeasuring@\savefieldlength@\fi
    \set@field
    \hfil
    \tabskip\alignsep@
}
\let\cmeasure@\measure@
\patchcmd\cmeasure@{\divide\@tempcntb\tw@}{}{}{}
\patchcmd\cmeasure@{\divide\@tempcntb\tw@}{}{}{}
\patchcmd\cmeasure@{\ifodd\maxfields@
  \global\advance\maxfields@\@ne
  \fi}{}{}{}    
  
\makeatother

\renewcommand\matrix[1]
{\hspace{-1.5pt}\left( \hspace{-2pt} \raisebox{-0pt}{$\begin{array}{@{\extracolsep{-5.5pt}\,}cccc}#1\end{array}$}  \hspace{0pt} \right)\hspace{-1.5pt}}
\newcommand\fmatrix[1]
{\hspace{-1.5pt}\left( \hspace{-2pt} \raisebox{-0pt}{\footnotesize$\begin{array}{@{\extracolsep{-5.5pt}\,}cccc}#1\end{array}$}  \hspace{0pt} \right)\hspace{-1.5pt}}
\newcommand\tinymatrix[1]
{\left( \hspace{-2pt} \renewcommand\thickspace{\kern2pt} \scriptstyle\begin{smallmatrix} #1 \end{smallmatrix} \hspace{-2pt} \right)}
\newcommand\foplus{\!\oplus\!}

\newcommand\ignore[1]{}
\newcommand\superequals[1]{\ensuremath{\stackrel {\makebox[0pt]{\ensuremath{\scriptstyle #1}}}{=}}}
\newcommand\equalref[1]{\superequals{\eqref{#1}}}
\newcommand\equalreftwo[2]{\ensuremath{\stackrel{\makebox[0pt]{\ensuremath{\stackrel{\scriptstyle\eqref{#1}}{\scriptstyle\eqref{#2}}}}}{=}}}

\newcommand{\name}[1]{\ensuremath{\ulcorner #1 \urcorner}}
\newcommand{\coname}[1]{\ensuremath{\llcorner #1 \lrcorner}}
\newcommand{\op}{\ensuremath{\mathrm{op}}}
\newcommand{\tuple}[2]{\ensuremath{\left(\begin{smallmatrix} #1\\#2 \end{smallmatrix}\right)}}%{\ensuremath{\protect\langle #1,#2 \protect\rangle}}
\DeclareMathOperator{\Tr}{Tr}
\DeclareMathOperator{\dom}{dom}
\DeclareMathOperator{\cod}{cod}
\DeclareMathOperator{\rank}{rank}
\newcommand\cat[1]{\ensuremath{\mathbf{#1}}}
\renewcommand\dag{\ensuremath{\dagger}}
\newcommand\lecture[2]{\section{#1}}
\newcommand\id[1][]{\ensuremath{\mathrm{id}_{#1}}}
\def\totimes{{\textstyle \bigotimes}}
\newcommand\sxleftarrow[1]{\xleftarrow{\smash{#1}}}
\newcommand\sxrightarrow[1]{\xrightarrow{\smash{#1}}}
\newcommand\xto[1]{\xrightarrow{#1}}
\newcommand\sxto[1]{\sxrightarrow{#1}}
\renewcommand{\-}[0]{\nobreakdash-\hspace{0pt}}
\DeclareMathOperator{\Ob}{Ob}
\newcommand\C{\ensuremath{\mathbb{C}}}
\newcommand{\inprod}[2]{\ensuremath{\langle #1\hspace{0.5pt}|\hspace{0.5pt}#2 \rangle}}
\newcommand{\defined}{\ensuremath{\mathop{\downarrow}}}
\newcommand\Hom{\ensuremath{\textrm{Hom}}}
\newcommand\pdag{{\phantom{\dagger}}}
\DeclareMathOperator{\CP}{CP}
\DeclareMathOperator{\CPq}{CP_q}
\DeclareMathOperator{\CPc}{CP_c}
\DeclareMathOperator{\CPstar}{CP^*}
\newcommand{\proj}{\ensuremath{p}}
\newcommand{\inj}{\ensuremath{i}}
\DeclareMathOperator{\Prob}{Prob}
\newcommand{\ie}{\textit{i.e.}\xspace}
\newcommand{\eg}{\textit{e.g.}\xspace}
\newcommand\ud{\ensuremath{\mathrm{d}}}
\usepackage{color}
\newcommand{\todo}[1]{{\color{red}#1}}
\newcommand{\labpos}[1]{label={[label distance=0.05cm]330:#1}}

\def\arraystretch{1.0}
\newcommand\grid[1]{\ensuremath{\def\arraystretch{1.4}\begin{array}{|c|c|c|c|c|c|c|c|c|}\hline#1\\\hline\end{array}}}
\newcommand\diag{\mathrm{diag}}
\newcommand\inv{{-1}}
\newcommand\I{\ensuremath{\mathbb I}}
\newcommand\M{\ensuremath{\mathcal M}}
\newcommand\F{\ensuremath{\mathcal F}}
\newcommand\super[2]{\stackrel{\makebox[0pt]{\smash{\tiny #1}}}{#2}}
\newcommand\T{\ensuremath{\mathrm T}}
\newcommand\bra[1]{\langle #1|}
\newcommand\ket[1]{{|} #1 \rangle}
\newcommand\braket[2]{\langle #1 | #2 \rangle}
\newcommand\ketbra[2]{|#1 \rangle \hspace{-1pt} \langle #2 |}
\newcommand\range[1]{\ensuremath{ \{ #1 ,...,d-1 \} }}
\newcommand\dfrob{\ensuremath{\dagger\text{-SCFA}}}
\newcommand\tinylab[1]{label={[label distance=-0.05cm]330:\tiny{$#1$}}}
\newcommand\projs{\ensuremath{\scalebox{1.5}{$\boldsymbol{*}$}}}

\newcounter{jamiecomment}
\newcommand\JVcomm[1]{\ensuremath{{}^{\color{red}\thejamiecomment}}\marginpar{\color{red}\tiny\raggedright \thejamiecomment: #1}\stepcounter{jamiecomment}}

\newcommand\UEBM{\text{partitioned UEB}}
\tikzset{proofdiagram/.style={scale=0.65}}
\tikzset{projdot/.style={scale=0.2}}

\usepackage{amsbsy}
\usepackage[normalem]{ulem}
\allowdisplaybreaks[1]
\usepackage{docmute}
\usepackage{hyperref}
\usepackage{enumitem,multirow}
\usepackage{graphicx}
\usepackage{amsmath,amsfonts,amssymb,mathtools,xspace}
\usepackage{amsthm}
\usepackage{tikz}
\usepackage{array}
\usetikzlibrary{decorations.pathreplacing}
\usetikzlibrary{decorations.markings}
\usetikzlibrary{calc}
\usepackage{etoolbox}
\usepackage{datetime}
\def\minus{\text{-}}

\colorlet{lgray}{gray!20}
\tikzstyle{dc}   = [circle, minimum width=8pt, draw, inner sep=0pt, path picture={\draw (path picture bounding box.south east) -- (path picture bounding box.north west) (path picture bounding box.south west) -- (path picture bounding box.north east);}]
\tikzstyle{dp}   = [circle, minimum width=8pt, draw, inner sep=0pt, path picture={\draw (path picture bounding box.west) -- (path picture bounding box.north) (path picture bounding box.south west) -- (path picture bounding box.north) (path picture bounding box.south west) -- (path picture bounding box.north east) (path picture bounding box.north east) -- (path picture bounding box.south) (path picture bounding box.south) -- (path picture bounding box.east);}]
\tikzstyle{ls} = [dc,scale=0.65]
\tikzstyle{ls'} = [dp,scale=1.3]
\tikzstyle{dpt}  = [circle, minimum width=8pt, draw, inner sep=0pt, path picture={\draw (path picture bounding box.south) -- (path picture bounding box.north) (path picture bounding box.west) -- (path picture bounding box.east);}]
\tikzstyle{agg} = [dpt,scale=0.65]
\renewenvironment{pic}[1][]
{\begin{aligned}\begin{tikzpicture}[#1]}
{\end{tikzpicture}\end{aligned}}

\makeatletter
\def\calign@preamble{%
   &\hfil\strut@
    \setboxz@h{\@lign$\m@th\displaystyle{##}$}%
    \ifmeasuring@\savefieldlength@\fi
    \set@field
    \hfil
    \tabskip\alignsep@
}
\let\cmeasure@\measure@
\patchcmd\cmeasure@{\divide\@tempcntb\tw@}{}{}{}
\patchcmd\cmeasure@{\divide\@tempcntb\tw@}{}{}{}
\patchcmd\cmeasure@{\ifodd\maxfields@
  \global\advance\maxfields@\@ne
  \fi}{}{}{}    

\makeatother


\pgfdeclarelayer{foreground}
\pgfdeclarelayer{background}
\pgfsetlayers{main,foreground,background}

\makeatletter
\pgfkeys{%
  /tikz/on layer/.code={
    \pgfonlayer{#1}\begingroup
    \aftergroup\endpgfonlayer
    \aftergroup\endgroup
  },
  /tikz/node on layer/.code={
    \gdef\node@@on@layer{%
      \setbox\tikz@tempbox=\hbox\bgroup\pgfonlayer{#1}\unhbox\tikz@tempbox\endpgfonlayer\egroup}
    \aftergroup\node@on@layer
  },
  /tikz/end node on layer/.code={
    \endpgfonlayer\endgroup\endgroup
  }
}
\def\node@on@layer{\aftergroup\node@@on@layer}
\makeatother\def\thickness{0.7pt}
\tikzstyle{oldmorphism}=[minimum width=30pt, minimum height=16pt, draw, font=\small, inner sep=0pt, fill=white, line width=\thickness]
\tikzstyle{cross}=[preaction={draw=white, -, line width=10pt}]
\tikzstyle{braid}=[double=black, line width=3*\thickness, double distance=\thickness, white]
\tikzstyle{string}=[line width=\thickness]
\tikzstyle{scalar}=[circle, inner sep=0pt, minimum width=15pt, draw, line width=\thickness]
\tikzstyle{dot}=[circle, draw=black, fill=black!25, inner sep=.5ex, line width=\thickness, node on layer=foreground]
\tikzstyle{blackdot}=[circle, draw=black, fill=black, inner sep=.5ex, line width=\thickness, node on layer=foreground]
\tikzstyle{whitedot}=[circle, draw=black, fill=white, inner sep=.5ex, line width=\thickness, node on layer=foreground]
\tikzstyle{pinkdot}=[circle, draw=black, fill=red!20, inner sep=.5ex, line width=\thickness, node on layer=foreground]
\tikzstyle{reddot}=[circle, draw=black, fill=red, inner sep=.5ex, line width=\thickness, node on layer=foreground]
\tikzstyle{bluedot}=[circle, draw=black, fill=blue, inner sep=.5ex, line width=\thickness, node on layer=foreground]
\tikzstyle{yellowdot}=[circle, draw=black, fill=yellow, inner sep=.5ex, line width=\thickness, node on layer=foreground]
\tikzstyle{greendot}=[circle, draw=black, fill=green, inner sep=.5ex, line width=\thickness, node on layer=foreground]
\tikzstyle{lsdot}=[circle, draw=black, fill=white, inner sep=.5ex, line width=\thickness, node on layer=foreground, path picture={\draw (path picture bounding box.south east) -- (path picture bounding box.north west) (path picture bounding box.south west) -- (path picture bounding box.north east);}]
\tikzstyle{lssdot}=[circle, draw=black, fill=white, inner sep=.5ex, line width=\thickness, node on layer=foreground, path picture={\draw (path picture bounding box.south) -- (path picture bounding box.north) (path picture bounding box.west) -- (path picture bounding box.east);}]
\tikzstyle{mixedmorphism}=[morphism, minimum width=30pt, minimum height=16pt, draw, font=\small, inner sep=0pt, fill=white, line width=\thickness,rounded corners=1ex]
\tikzstyle{thick}=[line width=\thickness]
\tikzstyle{tiny}=[font=\tiny]

\tikzset{arrow/.style={decoration={
    markings,
    mark=at position #1 with \arrow{thickarrow}},
    postaction=decorate}
}
\tikzset{reverse arrow/.style={decoration={
    markings,
    mark=at position #1 with \arrow{reversethickarrow}},
    postaction=decorate}
}

\newif\ifblack\pgfkeys{/tikz/black/.is if=black}
\newif\ifwedge\pgfkeys{/tikz/wedge/.is if=wedge}
\newif\ifvflip\pgfkeys{/tikz/vflip/.is if=vflip}
\newif\ifhflip\pgfkeys{/tikz/hflip/.is if=hflip}
\newif\ifhvflip\pgfkeys{/tikz/hvflip/.is if=hvflip}
\newif\ifconnectnw\pgfkeys{/tikz/connect nw/.is if=connectnw}
\newif\ifconnectne\pgfkeys{/tikz/connect ne/.is if=connectne}
\newif\ifconnectsw\pgfkeys{/tikz/connect sw/.is if=connectsw}
\newif\ifconnectse\pgfkeys{/tikz/connect se/.is if=connectse}
\newif\ifconnectn\pgfkeys{/tikz/connect n/.is if=connectn}
\newif\ifconnects\pgfkeys{/tikz/connect s/.is if=connects}
\newif\ifconnectnwf\pgfkeys{/tikz/connect nw >/.is if=connectnwf}
\newif\ifconnectnef\pgfkeys{/tikz/connect ne >/.is if=connectnef}
\newif\ifconnectswf\pgfkeys{/tikz/connect sw >/.is if=connectswf}
\newif\ifconnectsef\pgfkeys{/tikz/connect se >/.is if=connectsef}
\newif\ifconnectnf\pgfkeys{/tikz/connect n >/.is if=connectnf}
\newif\ifconnectsf\pgfkeys{/tikz/connect s >/.is if=connectsf}
\newif\ifconnectnwr\pgfkeys{/tikz/connect nw </.is if=connectnwr}
\newif\ifconnectner\pgfkeys{/tikz/connect ne </.is if=connectner}
\newif\ifconnectswr\pgfkeys{/tikz/connect sw </.is if=connectswr}
\newif\ifconnectser\pgfkeys{/tikz/connect se </.is if=connectser}
\newif\ifconnectnr\pgfkeys{/tikz/connect n </.is if=connectnr}
\newif\ifconnectsr\pgfkeys{/tikz/connect s </.is if=connectsr}
\tikzset{keylengthnw/.initial=\connectheight}
\tikzset{keylengthn/.initial =\connectheight}
\tikzset{keylengthne/.initial=\connectheight}
\tikzset{keylengthsw/.initial=\connectheight}
\tikzset{keylengths/.initial =\connectheight}
\tikzset{keylengthse/.initial=\connectheight}
\tikzset{connect nw length/.style={connect nw=true, keylengthnw={#1}}}
\tikzset{connect n length/.style ={connect n =true, keylengthn ={#1}}}
\tikzset{connect ne length/.style={connect ne=true, keylengthne={#1}}}
\tikzset{connect sw length/.style={connect sw=true, keylengthsw={#1}}}
\tikzset{connect s length/.style ={connect s =true, keylengths ={#1}}}
\tikzset{connect se length/.style={connect se=true, keylengthse={#1}}}
\tikzset{connect nw < length/.style={connect nw <=true, keylengthnw={#1}}}
\tikzset{connect n < length/.style ={connect n <=true,  keylengthn ={#1}}}
\tikzset{connect ne < length/.style={connect ne <=true, keylengthne={#1}}}
\tikzset{connect sw < length/.style={connect sw <=true, keylengthnw={#1}}}
\tikzset{connect s < length/.style ={connect s <=true,  keylengths ={#1}}}
\tikzset{connect se < length/.style={connect se <=true, keylengthse={#1}}}
\tikzset{connect nw > length/.style={connect nw >=true, keylengthnw={#1}}}
\tikzset{connect n > length/.style ={connect n >=true,  keylengthn ={#1}}}
\tikzset{connect ne > length/.style={connect ne >=true, keylengthne={#1}}}
\tikzset{connect sw > length/.style={connect sw >=true, keylengthsw={#1}}}
\tikzset{connect s > length/.style ={connect s >=true,  keylengths ={#1}}}
\tikzset{connect se > length/.style={connect se >=true, keylengthse={#1}}}

\newlength\morphismheight
\newlength\minimummorphismwidth
\newlength\stateheight
\newlength\minimumstatewidth
\newlength\connectheight
\tikzset{width/.initial=\minimummorphismwidth}

\makeatletter
\pgfarrowsdeclare{thickarrow}{thickarrow}
{
  \pgfutil@tempdima=-0.84pt%
  \advance\pgfutil@tempdima by-1.3\pgflinewidth%
  \pgfutil@tempdimb=-1.7pt%
  \advance\pgfutil@tempdimb by.625\pgflinewidth%
  \pgfarrowsleftextend{+\pgfutil@tempdima}
  \pgfarrowsrightextend{+\pgfutil@tempdimb}
}
{
  \pgfmathparse{\pgfgetarrowoptions{thickarrow}}%
  \pgfsetlinewidth{1.25 pt}
  \pgfutil@tempdima=0.28pt%
  \advance\pgfutil@tempdima by.3\pgflinewidth%
  \pgfsetlinewidth{0.8\pgflinewidth}
  \pgfsetdash{}{+0pt}
  \pgfsetroundcap
  \pgfsetroundjoin
  \pgfpathmoveto{\pgfqpoint{-3\pgfutil@tempdima}{4\pgfutil@tempdima}}
  \pgfpathcurveto
  {\pgfqpoint{-2.75\pgfutil@tempdima}{2.5\pgfutil@tempdima}}
  {\pgfqpoint{0pt}{0.25\pgfutil@tempdima}}
  {\pgfqpoint{0.75\pgfutil@tempdima}{0pt}}
  \pgfpathcurveto
  {\pgfqpoint{0pt}{-0.25\pgfutil@tempdima}}
  {\pgfqpoint{-2.75\pgfutil@tempdima}{-2.5\pgfutil@tempdima}}
  {\pgfqpoint{-3\pgfutil@tempdima}{-4\pgfutil@tempdima}}
  \pgfusepathqstroke
}
\pgfarrowsdeclare{reversethickarrow}{reversethickarrow}
{
  \pgfutil@tempdima=-0.84pt%
  \advance\pgfutil@tempdima by-1.3\pgflinewidth%
  \pgfutil@tempdimb=0.2pt%
  \advance\pgfutil@tempdimb by.625\pgflinewidth%
  \pgfarrowsleftextend{+\pgfutil@tempdima}
  \pgfarrowsrightextend{+\pgfutil@tempdimb}
}
{
  \pgftransformxscale{-1}
  \pgfmathparse{\pgfgetarrowoptions{thickarrow}}%
  \ifpgfmathunitsdeclared%
    \pgfmathparse{\pgfmathresult pt}%
  \else%  
    \pgfmathparse{\pgfmathresult*\pgflinewidth}%
  \fi%
  \let\thickness=\pgfmathresult
  \pgfsetlinewidth{1.25 pt}
  \pgfutil@tempdima=0.28pt%
  \advance\pgfutil@tempdima by.3\pgflinewidth%
  \pgfsetlinewidth{0.8\pgflinewidth}
  \pgfsetdash{}{+0pt}
  \pgfsetroundcap
  \pgfsetroundjoin
  \pgfpathmoveto{\pgfqpoint{-3\pgfutil@tempdima}{4\pgfutil@tempdima}}
  \pgfpathcurveto
  {\pgfqpoint{-2.75\pgfutil@tempdima}{2.5\pgfutil@tempdima}}
  {\pgfqpoint{0pt}{0.25\pgfutil@tempdima}}
  {\pgfqpoint{0.75\pgfutil@tempdima}{0pt}}
  \pgfpathcurveto
  {\pgfqpoint{0pt}{-0.25\pgfutil@tempdima}}
  {\pgfqpoint{-2.75\pgfutil@tempdima}{-2.5\pgfutil@tempdima}}
  {\pgfqpoint{-3\pgfutil@tempdima}{-4\pgfutil@tempdima}}
  \pgfusepathqstroke
}
\makeatother

\makeatletter
\pgfdeclareshape{ground}
{
    \savedanchor\centerpoint
    {
        \pgf@x=0pt
        \pgf@y=0pt
    }
    \anchor{center}{\centerpoint}
    \anchorborder{\centerpoint}
    \anchor{north}
    {
        \pgf@x=0pt
        \pgf@y=0.5*0.33*\stateheight
    }
    \anchor{south}
    {
        \pgf@x=0pt
        \pgf@y=0pt
    }
    \saveddimen\overallwidth
    {
        \pgfkeysgetvalue{/pgf/minimum width}{\minwidth}
        \pgf@x=\minimumstatewidth
        \ifdim\pgf@x<\minwidth
            \pgf@x=\minwidth
        \fi
    }
    \backgroundpath
    {
        \begin{pgfonlayer}{foreground}
        \pgfsetstrokecolor{black}
        \pgfsetcolor{gray}
        \pgfsetlinewidth{1.25pt}
        \ifhflip
            \pgftransformyscale{-1}
        \fi
        \pgftransformscale{0.5}
        \pgfpathmoveto{\pgfpoint{-0.5*\overallwidth}{0}}
        \pgfpathlineto{\pgfpoint{0.5*\overallwidth}{0}}
        \pgfpathmoveto{\pgfpoint{-0.33*\overallwidth}{0.33*\stateheight}}
        \pgfpathlineto{\pgfpoint{0.33*\overallwidth}{0.33*\stateheight}}
        \pgfpathmoveto{\pgfpoint{-0.16*\overallwidth}{0.66*\stateheight}}
        \pgfpathlineto{\pgfpoint{0.16*\overallwidth}{0.66*\stateheight}}
        \pgfpathmoveto{\pgfpoint{-0.02*\overallwidth}{\stateheight}}
        \pgfpathlineto{\pgfpoint{0.02*\overallwidth}{\stateheight}}
        \end{pgfonlayer}
    }
}
\tikzset{forward arrow style/.style={every to/.style, decoration={
    markings,
    mark=at position 0.5 with \arrow{thickarrow}},
    postaction=decorate}}
\tikzset{reverse arrow style/.style={every to/.style, decoration={
    markings,
    mark=at position 0.5 with \arrow{reversethickarrow}},
    postaction=decorate}}
\pgfdeclareshape{morphism}
{
    \savedanchor\centerpoint
    {
        \pgf@x=0pt
        \pgf@y=0pt
    }
    \anchor{center}{\centerpoint}
    \anchorborder{\centerpoint}
    \saveddimen\savedlengthnw
    {
        \pgfkeysgetvalue{/tikz/keylengthnw}{\len}
        \pgf@x=\len
    }
    \saveddimen\savedlengthn
    {
        \pgfkeysgetvalue{/tikz/keylengthn}{\len}
        \pgf@x=\len
    }
    \saveddimen\savedlengthne
    {
        \pgfkeysgetvalue{/tikz/keylengthne}{\len}
        \pgf@x=\len
    }
    \saveddimen\savedlengthsw
    {
        \pgfkeysgetvalue{/tikz/keylengthsw}{\len}
        \pgf@x=\len
    }
    \saveddimen\savedlengths
    {
        \pgfkeysgetvalue{/tikz/keylengths}{\len}
        \pgf@x=\len
    }
    \saveddimen\savedlengthse
    {
        \pgfkeysgetvalue{/tikz/keylengthse}{\len}
        \pgf@x=\len
    }
    \saveddimen\overallwidth
    {
        \pgfkeysgetvalue{/tikz/width}{\minwidth}
        \pgf@x=\wd\pgfnodeparttextbox
        \ifdim\pgf@x<\minwidth
            \pgf@x=\minwidth
        \fi
    }
    \savedanchor{\upperrightcorner}
    {
        \pgf@y=.5\ht\pgfnodeparttextbox
        \advance\pgf@y by -.5\dp\pgfnodeparttextbox
        \pgf@x=.5\wd\pgfnodeparttextbox
    }
    \anchor{north}
    {
        \pgf@x=0pt
        \pgf@y=0.5\morphismheight
    }
    \anchor{north east}
    {
        \pgf@x=\overallwidth
        \multiply \pgf@x by 2
        \divide \pgf@x by 5
        \pgf@y=0.5\morphismheight
    }
    \anchor{east}
    {
        \pgf@x=\overallwidth
        \divide \pgf@x by 2
        \advance \pgf@x by 5pt
        \pgf@y=0pt
    }
    \anchor{west}
    {
        \pgf@x=-\overallwidth
        \divide \pgf@x by 2
        \advance \pgf@x by -5pt
        \pgf@y=0pt
    }
    \anchor{north west}
    {
        \pgf@x=-\overallwidth
        \multiply \pgf@x by 2
        \divide \pgf@x by 5
        \pgf@y=0.5\morphismheight
    }
    \anchor{connect nw}
    {
        \pgf@x=-\overallwidth
        \multiply \pgf@x by 2
        \divide \pgf@x by 5
        \pgf@y=0.5\morphismheight
        \advance\pgf@y by \savedlengthnw
    }
    \anchor{connect ne}
    {
        \pgf@x=\overallwidth
        \multiply \pgf@x by 2
        \divide \pgf@x by 5
        \pgf@y=0.5\morphismheight
        \advance\pgf@y by \savedlengthne
    }
    \anchor{connect sw}
    {
        \pgf@x=-\overallwidth
        \multiply \pgf@x by 2
        \divide \pgf@x by 5
        \pgf@y=-0.5\morphismheight
        \advance\pgf@y by -\savedlengthsw
    }
    \anchor{connect se}
    {
        \pgf@x=\overallwidth
        \multiply \pgf@x by 2
        \divide \pgf@x by 5
        \pgf@y=-0.5\morphismheight
        \advance\pgf@y by -\savedlengthse
    }
    \anchor{connect n}
    {
        \pgf@x=0pt
        \pgf@y=0.5\morphismheight
        \advance\pgf@y by \savedlengthn
    }
    \anchor{connect s}
    {
        \pgf@x=0pt
        \pgf@y=-0.5\morphismheight
        \advance\pgf@y by -\savedlengths
    }
    \anchor{south east}
    {
        \pgf@x=\overallwidth
        \multiply \pgf@x by 2
        \divide \pgf@x by 5
        \pgf@y=-0.5\morphismheight
    }
    \anchor{south west}
    {
        \pgf@x=-\overallwidth
        \multiply \pgf@x by 2
        \divide \pgf@x by 5
        \pgf@y=-0.5\morphismheight
    }
    \anchor{south}
    {
        \pgf@x=0pt
        \pgf@y=-0.5\morphismheight
    }
    \anchor{text}
    {
        \upperrightcorner
        \pgf@x=-\pgf@x
        \pgf@y=-\pgf@y
    }
    \backgroundpath
    {
        \pgfsetstrokecolor{black}
        \pgfsetlinewidth{\thickness}
        \begin{scope}
                \ifhflip
                    \pgftransformyscale{-1}
                \fi
                \ifvflip
                    \pgftransformxscale{-1}
                \fi
                \ifhvflip
                    \pgftransformxscale{-1}
                    \pgftransformyscale{-1}
                \fi
                \pgfpathmoveto{\pgfpoint
                    {-0.5*\overallwidth-5pt}
                    {0.5*\morphismheight}}
                \pgfpathlineto{\pgfpoint
                    {0.5*\overallwidth+5pt}
                    {0.5*\morphismheight}}
                \ifwedge
                    \pgfpathlineto{\pgfpoint
                        {0.5*\overallwidth + 15pt}
                        {-0.5*\morphismheight}}
                \else
                    \pgfpathlineto{\pgfpoint
                        {0.5*\overallwidth + 5pt}
                        {-0.5*\morphismheight}}
                \fi
                \pgfpathlineto{\pgfpoint
                    {-0.5*\overallwidth-5pt}
                    {-0.5*\morphismheight}}
                \pgfpathclose
                \pgfusepath{stroke}
        \end{scope}
        \ifconnectnw
            \pgfpathmoveto{\pgfpoint
                {-0.4*\overallwidth}
                {0.5*\morphismheight}}
            \pgfpathlineto{\pgfpoint
                {-0.4*\overallwidth}
                {0.5*\morphismheight+\savedlengthnw}}
            \pgfusepath{stroke}
        \fi
        \ifconnectne
            \pgfpathmoveto{\pgfpoint
                {0.4*\overallwidth}
                {0.5*\morphismheight}}
            \pgfpathlineto{\pgfpoint
                {0.4*\overallwidth}
                {0.5*\morphismheight+\savedlengthne}}
            \pgfusepath{stroke}
        \fi
        \ifconnectsw
            \pgfpathmoveto{\pgfpoint
                {-0.4*\overallwidth}
                {-0.5*\morphismheight}}
            \pgfpathlineto{\pgfpoint
                {-0.4*\overallwidth}
                {-0.5*\morphismheight-\savedlengthsw}}
            \pgfusepath{stroke}
        \fi
        \ifconnectse
            \pgfpathmoveto{\pgfpoint
                {0.4*\overallwidth}
                {-0.5*\morphismheight}}
            \pgfpathlineto{\pgfpoint
                {0.4*\overallwidth}
                {-0.5*\morphismheight-\savedlengthse}}
            \pgfusepath{stroke}
        \fi
        \ifconnectn
            \pgfpathmoveto{\pgfpoint
                {0pt}
                {0.5*\morphismheight}}
            \pgfpathlineto{\pgfpoint
                {0pt}
                {0.5*\morphismheight+\savedlengthn}}
            \pgfusepath{stroke}
        \fi
        \ifconnects
            \pgfpathmoveto{\pgfpoint
                {0pt}
                {-0.5*\morphismheight}}
            \pgfpathlineto{\pgfpoint
                {0pt}
                {-0.5*\morphismheight-\savedlengths}}
            \pgfusepath{stroke}
        \fi
        \ifconnectnwf
            \draw [forward arrow style] (-0.4*\overallwidth,0.5*\morphismheight)
                to (-0.4*\overallwidth,0.5*\morphismheight+\savedlengthnw);
        \fi
        \ifconnectnef
            \draw [forward arrow style] (0.4*\overallwidth,0.5*\morphismheight)
                to (0.4*\overallwidth,0.5*\morphismheight+\savedlengthne);
        \fi
        \ifconnectswf
            \draw [forward arrow style] (-0.4*\overallwidth,-0.5*\morphismheight-\savedlengthsw)
                to (-0.4*\overallwidth,-0.5*\morphismheight);
        \fi
        \ifconnectsef
            \draw [forward arrow style] (0.4*\overallwidth,-0.5*\morphismheight-\savedlengthse)
                to (0.4*\overallwidth,-0.5*\morphismheight);
        \fi
        \ifconnectnf
            \draw [forward arrow style] (0,0.5*\morphismheight)
                to (0,0.5*\morphismheight+\savedlengthn);
        \fi
        \ifconnectsf
            \draw [forward arrow style] (0,-0.5*\morphismheight-\savedlengths)
                to (0,-0.5*\morphismheight);
        \fi
        \ifconnectnwr
            \draw [reverse arrow style] (-0.4*\overallwidth,0.5*\morphismheight)
                to (-0.4*\overallwidth,0.5*\morphismheight+\savedlengthnw);
        \fi
        \ifconnectner
            \draw [reverse arrow style] (0.4*\overallwidth,0.5*\morphismheight)
                to (0.4*\overallwidth,0.5*\morphismheight+\savedlengthne);
        \fi
        \ifconnectswr
            \draw [reverse arrow style] (-0.4*\overallwidth,-0.5*\morphismheight-\savedlengthsw)
                to (-0.4*\overallwidth,-0.5*\morphismheight);
        \fi
        \ifconnectser
            \draw [reverse arrow style] (0.4*\overallwidth,-0.5*\morphismheight-\savedlengthse)
                to (0.4*\overallwidth,-0.5*\morphismheight);
        \fi
        \ifconnectnr
            \draw [reverse arrow style] (0,0.5*\morphismheight)
                to (0,0.5*\morphismheight+\savedlengthn);
        \fi
        \ifconnectsr
            \draw [reverse arrow style] (0,-0.5*\morphismheight-\savedlengths)
                to (0,-0.5*\morphismheight);
        \fi
    }
}
\pgfdeclareshape{swish right}
{
    \savedanchor\centerpoint
    {
        \pgf@x=0pt
        \pgf@y=0pt
    }
    \anchor{center}{\centerpoint}
    \anchorborder{\centerpoint}
    \anchor{north}
    {
        \pgf@x=\minimummorphismwidth
        \divide\pgf@x by 5
        \pgf@y=\morphismheight
        \divide\pgf@y by 2
        \advance\pgf@y by \connectheight
    }
    \anchor{south}
    {
        \pgf@x=-\minimummorphismwidth
        \divide\pgf@x by 5
        \pgf@y=-\morphismheight
        \divide\pgf@y by 2
        \advance\pgf@y by -\connectheight
    }
    \backgroundpath
    {
        \pgfsetstrokecolor{black}
        \pgfsetlinewidth{\thickness}
        \pgfpathmoveto{\pgfpoint
            {-0.2*\minimummorphismwidth}
            {-0.5*\morphismheight-\connectheight}}
        \pgfpathcurveto
            {\pgfpoint{-0.2*\minimummorphismwidth}{0pt}}
            {\pgfpoint{0.2*\minimummorphismwidth}{0pt}}
            {\pgfpoint
                {0.2*\minimummorphismwidth}
                {0.5*\morphismheight+\connectheight}}
        \pgfusepath{stroke}
    }
}
\pgfdeclareshape{swish left}
{
    \savedanchor\centerpoint
    {
        \pgf@x=0pt
        \pgf@y=0pt
    }
    \anchor{center}{\centerpoint}
    \anchorborder{\centerpoint}
    \anchor{north}
    {
        \pgf@x=-\minimummorphismwidth
        \divide\pgf@x by 5
        \pgf@y=\morphismheight
        \divide\pgf@y by 2
        \advance\pgf@y by \connectheight
    }
    \anchor{south}
    {
        \pgf@x=\minimummorphismwidth
        \divide\pgf@x by 5
        \pgf@y=-\morphismheight
        \divide\pgf@y by 2
        \advance\pgf@y by -\connectheight
    }
    \backgroundpath
    {
        \pgfsetstrokecolor{black}
        \pgfsetlinewidth{\thickness}
        \pgfpathmoveto{\pgfpoint
            {0.2*\minimummorphismwidth}
            {-0.5*\morphismheight-\connectheight}}
        \pgfpathcurveto
            {\pgfpoint{0.2*\minimummorphismwidth}{0pt}}
            {\pgfpoint{-0.2*\minimummorphismwidth}{0pt}}
            {\pgfpoint
                {-0.2*\minimummorphismwidth}
                {0.5*\morphismheight+\connectheight}}
        \pgfusepath{stroke}
    }
}
\pgfdeclareshape{state}
{
    \savedanchor\centerpoint
    {
        \pgf@x=0pt
        \pgf@y=0pt
    }
    \anchor{center}{\centerpoint}
    \anchorborder{\centerpoint}
    \saveddimen\overallwidth
    {
        \pgf@x=3\wd\pgfnodeparttextbox
        \ifdim\pgf@x<\minimumstatewidth
            \pgf@x=\minimumstatewidth
        \fi
    }
    \savedanchor{\upperrightcorner}
    {
        \pgf@x=.5\wd\pgfnodeparttextbox
        \pgf@y=.5\ht\pgfnodeparttextbox
        \advance\pgf@y by -.5\dp\pgfnodeparttextbox
    }
    \anchor{A}
    {
        \pgf@x=-\overallwidth
        \divide\pgf@x by 4
        \pgf@y=0pt
    }
    \anchor{B}
    {
        \pgf@x=\overallwidth
        \divide\pgf@x by 4
        \pgf@y=0pt
    }
    \anchor{text}
    {
        \upperrightcorner
        \pgf@x=-\pgf@x
        \ifhflip
            \pgf@y=-\pgf@y
            \advance\pgf@y by 0.4\stateheight
        \else
            \pgf@y=-\pgf@y
            \advance\pgf@y by -0.4\stateheight
        \fi
    }
    \backgroundpath
    {
        \pgfsetstrokecolor{black}
        \pgfsetlinewidth{\thickness}
        \pgfpathmoveto{\pgfpoint{-0.5*\overallwidth}{0}}
        \pgfpathlineto{\pgfpoint{0.5*\overallwidth}{0}}
        \ifhflip
            \pgfpathlineto{\pgfpoint{0}{\stateheight}}
        \else
            \pgfpathlineto{\pgfpoint{0}{-\stateheight}}
        \fi
        \pgfpathclose
        \ifblack
            \pgfsetfillcolor{black}
            \pgfusepath{fill,stroke}   
        \else
             \pgfsetfillcolor{red!20}
             \pgfusepath{fill,stroke} 
        \fi     
    }
}
\makeatother

\newcommand{\tinycomult}[1][blackdot]{
\smash{\raisebox{-2pt}{\hspace{-5pt}\ensuremath{\begin{pic}[scale=0.4,string]
    \node (0) at (0,0) {};
    \node[#1, inner sep=1.5pt] (1) at (0,0.55) {};
    \node (2) at (-0.5,1) {};
    \node (3) at (0.5,1) {};
    \draw (0.center) to (1.center);
    \draw (1.center) to [out=left, in=down, out looseness=1.5] (2.center);
    \draw (1.center) to [out=right, in=down, out looseness=1.5] (3.center);
\end{pic}
}\hspace{-3pt}}}}
\newcommand{\tinydot}[1][blackdot]{
\smash{\raisebox{-2pt}{\hspace{-5pt}\ensuremath{\begin{pic}[scale=0.4,string]
    \node (0) at (0,0) {};
    \node[#1, inner sep=1.5pt] (1) at (0,0.45) {};
    \node (2) at (-0.5,1) {};
    \node (3) at (0.5,1) {};  
\end{pic}
}\hspace{-3pt}}}}

\newcommand{\tinycounit}[1][blackdot]{
\smash{\raisebox{-2pt}{\ensuremath{\hspace{-3pt}\begin{pic}[scale=0.4,string]
        \node (0) at (0,0) {};
        \node (1) at (0,1) {};
        \node[#1, inner sep=1.5pt] (d) at (0,0.55) {};
        \draw (0.center) to (d.center);
    \end{pic}
    \hspace{-1pt}}}}}
    
\newcommand{\tinymult}[1][blackdot]{
\smash{\raisebox{-2pt}{\hspace{-5pt}\ensuremath{\begin{pic}[scale=0.4,string,yscale=-1]
    \node (0) at (0,0) {};
    \node[#1, inner sep=1.5pt] (1) at (0,0.55) {};
    \node (2) at (-0.5,1) {};
    \node (3) at (0.5,1) {};
    \draw (0.center) to (1.center);
    \draw (1.center) to [out=left, in=down, out looseness=1.5] (2.center);
    \draw (1.center) to [out=right, in=down, out looseness=1.5] (3.center);
\end{pic}
}\hspace{-3pt}}}}
\newcommand{\tinymultgr}[1][greendot]{
\smash{\raisebox{-2pt}{\hspace{-5pt}\ensuremath{\begin{pic}[scale=0.4,string,yscale=-1]
    \node (0) at (0,0) {};
    \node[#1, inner sep=1.5pt] (1) at (0,0.55) {};
    \node (2) at (-0.5,1) {};
    \node (3) at (0.5,1) {};
    \draw[green] (0.center) to (1.center);
    \draw[green] (1.center) to [out=left, in=down, out looseness=1.5] (2.center);
    \draw[green] (1.center) to [out=right, in=down, out looseness=1.5] (3.center);
\end{pic}
}\hspace{-3pt}}}}
\newcommand{\tinymultcl}[1][dot]{
\smash{\raisebox{-2pt}{\hspace{-5pt}\ensuremath{\begin{pic}[scale=0.4,string,yscale=-1]
    \node (0) at (0,0) {};
    \node[#1dot, inner sep=1.5pt] (1) at (0,0.55) {};
    \node (2) at (-0.5,1) {};
    \node (3) at (0.5,1) {};
    \draw[#1] (0.center) to (1.center);
    \draw[#1] (1.center) to [out=left, in=down, out looseness=1.5] (2.center);
    \draw[#1] (1.center) to [out=right, in=down, out looseness=1.5] (3.center);
\end{pic}
}\hspace{-3pt}}}}

\newcommand{\tinyunit}[1][blackdot]{
\smash{\raisebox{-2pt}{\ensuremath{\hspace{-3pt}\begin{pic}[scale=0.4,string,yscale=-1]
        \node (0) at (0,0) {};
        \node (1) at (0,1) {};
        \node[#1, inner sep=1.5pt] (d) at (0,0.55) {};
        \draw (0.center) to (d.center);
    \end{pic}
    \hspace{-1pt}}}}}

\newcommand{\tinyhandle}[1][dot]{\raisebox{-2pt}{\ensuremath{\hspace{-3pt}\begin{pic}[scale=0.4,string]
        \node (0) at (0,0) {};
        \node[dot, inner sep=1.0pt] (1) at (0,0.3) {};
        \node[dot, inner sep=1.0pt] (2) at (0,0.7) {};
        \node (3) at (0,1) {};
        \draw (0.center) to (1.center);
        \draw (2.center) to (3.center);
        \draw[in=180, out=180, looseness=2] (1.center) to (2.center);
        \draw[in=0, out=0, looseness=2] (1.center) to (2.center);
\end{pic}\hspace{-1pt}}}}

\newcommand{\tinystate}[1][]{
\smash{\raisebox{-2pt}{\ensuremath{\hspace{-3pt}\begin{pic}[scale=0.4,string,yscale=-1]
        \node (0) at (0,0) {};
        \node (1) at (0,1) {};
        \node[#1, scale=0.25] (d) at (0,0.55) {};
        \draw (0.center) to (d.center);
    \end{pic}
    \hspace{-1pt}}}}}

\newcommand{\tinyeffect}[1][]{
\smash{\raisebox{-2pt}{\ensuremath{\hspace{-3pt}\begin{pic}[scale=-0.4,string,yscale=-1]
        \node (0) at (0,0) {};
        \node (1) at (0,1) {};
        \node[#1, scale=0.25] (d) at (0,0.55) {};
        \draw (0.center) to (d.center);
    \end{pic}
    \hspace{-1pt}}}}}   

\renewcommand\dag{\ensuremath{\dagger}}
\def\totimes{{\textstyle \bigotimes}}
\newcommand{\tinystatelab}[1][]{
\smash{\raisebox{-2pt}{\ensuremath{\hspace{-3pt}\begin{pic}[scale=0.4,string,yscale=-1]
        \node (0) at (0,0) {};
        \node (1) at (0,1) {};
        \node[black,state,label={[label distance=-0.08cm]330:\tiny{$#1$}}, scale=0.25] (d) at (0,0.55) {};
        \draw (0.center) to (d.center);
    \end{pic}
    \hspace{-1pt}}}}}

\newcommand{\tinyeffectlab}[1][]{
\smash{\raisebox{-2pt}{\ensuremath{\hspace{-3pt}\begin{pic}[scale=-0.4,string,yscale=-1]
        \node (0) at (0,0) {};
        \node (1) at (0,1) {};
        \node[black,hflip,state,label={[label distance=-0.08cm]30:\tiny{$#1$}}, scale=0.25] (d) at (0,0.55) {};
        \draw (0.center) to (d.center);
    \end{pic}
    \hspace{-1pt}}}}}

\title{Characterizing maximal families of  \\ mutually unbiased bases}
\author{
Benjamin Musto
\\
\texttt{benjamin.musto@cs.ox.ac.uk}
\\[20pt]
Department of Computer Science, University of Oxford}
\date{\today}

\def\titlerunning{Finite fields}
\def\authorrunning{B. Musto}
\begin{document}

\maketitle

\begin{abstract}
We show that maximal families of mutually unbiased bases  are characterized in all dimensions by partitioned unitary error bases, up to a choice of a family of Hadamards. Furthermore, we give a new construction of  partitioned unitary error bases, and thus  maximal families of mutually unbiased bases, from a finite field, which is simpler and more direct than previous proposals.
We introduce new tensor diagrammatic characterizations of
maximal families of mutually unbiased bases, partitioned unitary error
bases, and finite fields as
algebraic structures defined over Hilbert spaces.
\end{abstract}
\section{Introduction}

In this paper we present the following results:
\begin{itemize}
\item  an equivalence between \textit{partitioned unitary error bases} ($\UEBM$s) and maximal families of MUBs equipped with families of \textit{Hadamards};
\item a construction of
maximal families of MUBs from finite fields that is simpler than those proposed previously;
\item  a new tensor diagrammatic axiomatisation of maximal families of MUBs.
\end{itemize}
It has been shown that the largest family of $d$-dimensional mutually unbiased bases that can exist is $d+1$~\cite{Bandyopadhyay}. In light of this result we will refer to a family of $d+1$ MUBs as a \textit{maximal family of MUBs}.
 Maximal families of MUBs represent $d+1$ measurements that are, in some sense `as far apart as possible', and can perfectly distinguish any density operator on a  $d$-dimensional Hilbert space~\cite{woott}. Maximal families of MUBs are fundamental to areas such as quantum tomography~\cite{ivonovic} and  quantum key distribution~\cite{cerf} and are as such of great importance to quantum information. In spite of this much is still to be discovered about maximal families of MUBs. In general dimension, it is not even known whether maximal families of MUBs exist although the existence of maximal families of MUBs in prime power dimension has however long been \mbox{established~\cite{Bandyopadhyay,godsil,boykin}}.  The results in this paper build on previous work establishing a new diagrammatic framework to tackle a much researched area. We have utilised this approach to clarify the exact nature of the relationship between UEBs and maximal families of MUBs. This is just the beginning for this framework and we expect further progress to follow.
We briefly give definitions of the key structures which our
theorems relate.
\begin{definition}[Mutually unbiased bases~\cite{woott}]\label{def:mub}
A pair of orthonormal bases $\ket{a_i}$ and $\ket{b_i}$ with $i \in \{ 0,...,d-1 \}$ are \textit{mutually unbiased} if $|\braket{a_i}{b_j}|^2 =1/d$ for all $i,j \in \{ 0,...,d-1 \}$.  
\end{definition}
\begin{definition}[Family of mutually unbiased bases]\label{def:FMUB} A family of bases of a $d$-dimensional Hilbert space are a \textit{mutually unbiased family} if  they are pairwise mutually unbiased. 
 
\end{definition} 
\begin{definition}[Unitary error basis~\cite{knill}]\label{def:ueb}
A \textit{unitary error basis} on a $d$-dimensional Hilbert space is a family of $d^2$ unitary operators $U_{ij}$ where $i,j \in \{ 0,...,d-1 \}$ such that:
\begin{equation}\label{eq:trueb}
 \tr (U^{\dag}_{ij} \circ U_{mn})=\delta_{im}\delta_{jn}d 
\end{equation}
\end{definition}

Throughout this paper we make use of Penrose tensor diagrams closely related to those, which are used by the tensor networks community. We give a very brief introduction here in order to introduce one of the main results of the paper, a characterization of maximal families of MUBs using tensor diagrams. Refer to Section~\ref{section:graphical aa} for an in depth introduction to the tensor diagrams and other necessary background material. 

We use  wires to represent Hilbert spaces and boxes and nodes to represent linear maps. Tensor products are given by horizontal composition with composition of linear maps represented by connecting wires vertically. Direct sums are given by summation. We use the convention that our diagrams are read from bottom to top. Bra-ket notation can be translated into tensor diagrams as follows: $\ket i:=\tinystatelab[i]$ and $\bra i:=\tinyeffectlab[i]$ with the inner product given by connecting the wires.   
\begin{definition}[Basis tensors]
Given an orthonormal basis for a Hilbert space, $\ket{i},0 \leq i <n-1,$ we canonically define the following four linear maps: 
     
\begin{align*}
\begin{pic}[scale=1]
\node[pinkdot] (a) at (0,-0.25){};
\node (1) at (0.25,-1){};
\node (2) at (-0.25,-1){};
\node (3) at (0,0.5){};
\draw[string,in=right,out=90] (1.center) to (a.center);
\draw[string,in=left,out=90] (2.center) to (a.center);
\draw[string,in=90,out=270] (3.center) to (a.center);
\end{pic}
&:=\sum^{n-1}_{i=0}
\begin{pic}
\node(3)[state,scale=0.5] at (0,0){$i$};
\node(1)[state,scale=0.5,hflip] at (0.25,-0.5){$i$};
\node(2)[state,scale=0.5,hflip] at (-0.25,-0.5){$i$};
\draw[string] (1) to +(0,-0.5);
\draw[string] (2) to +(0,-0.5);
\draw[string] (3) to +(0,0.5);
\end{pic} & 
\begin{pic}[scale=1]
\node[pinkdot] (a) at (0,-0.25){};
\node (3) at (0,0.5){};
\draw[string,in=90,out=270] (3.center) to (a.center);
\end{pic}
&:=\sum^{n-1}_{i=0}
\begin{pic}
\node(3)[state,scale=0.5] at (0,0){$i$};
\draw[string] (3) to +(0,0.5);
\end{pic}
\\
\begin{pic}[yscale=1]
\node[pinkdot] (a) at (0,0){};
\node (1) at (0.25,0.75){};
\node (2) at (-0.25,0.75){};
\node (3) at (0,-0.75){};
\draw[string,in=right,out=270] (1.center) to (a.center);
\draw[string,in=left,out=270] (2.center) to (a.center);
\draw[string,in=270,out=90] (3.center) to (a.center);
\end{pic}
&:=\sum^{n-1}_{i=0}
\begin{pic}[yscale=-1]
\node(3)[state,scale=0.5,hflip] at (0,0){$i$};
\node(1)[state,scale=0.5] at (0.25,-0.5){$i$};
\node(2)[state,scale=0.5] at (-0.25,-0.5){$i$};
\draw[string] (1) to +(0,-0.5);
\draw[string] (2) to +(0,-0.5);
\draw[string] (3) to +(0,0.5);
\end{pic}&
\begin{pic}[scale=-1]
\node[pinkdot] (a) at (0,-0.25){};
\node (3) at (0,0.5){};
\draw[string,in=90,out=270] (3.center) to (a.center);
\end{pic}
&:=\sum^{n-1}_{i=0}
\begin{pic}[yscale=-1]
\node(3)[state,hflip,scale=0.5] at (0,0){$i$};
\draw[string] (3) to +(0,0.5);
\end{pic}
\end{align*}
Connected diagrams made up of these four linear maps are \textit{basis tensors}, and are uniquely determined
by the number of input and output wires. \end{definition}
We make extensive use of the basis tensor corresponding to the computational orthonormal basis which we will denote with black dots $\tinymult$. In Section~\ref{section:MMUB UEBM} we start by proving that the following diagrammatic equation completely characterizes maximal families of MUBs using the computational basis tensor, where $M$ is a linear map of type $\mathcal H \otimes \mathcal H \rightarrow \mathcal H$:
\begin{equation}\label{eq:MMUB}
\begin{aligned}
\begin{tikzpicture}[scale=0.75]
\node (m1) [morphism,wedge,hflip,scale=0.75] at (-1.25,0.5) {$M$};
\node(m2)[morphism,wedge,scale=0.75]  at (-1.25,-0.5) {$M$};
\node  at (m1.connect n) {};
\node at (m2.connect s){};
\draw (m1.south) to (m2.north)[];

\node (m3) [morphism,wedge,hflip,vflip,scale=0.75] at (1.25,0.5) {$M$};
\node(m4)[morphism,wedge,vflip,scale=0.75]  at (1.25,-0.5) {$M$};
\node  at (m3.connect n) {};
\node at (m4.connect s){};
\draw (m3.south) to (m4.north)[];

\node(b2)[blackdot] at (0.15,1.25){};
\draw[string,out=90,in=left] (m1.north east) to (b2.center);
\draw[string,out=right,in=90] (b2.center) to (m3.north west);
\draw (b2.center) to +(0,0.75);

\node(b4)[blackdot] at (0.15,-1.25){};
\draw[string,out=270,in=left] (m2.south east) to (b4.center);
\draw[string,out=right,in=270] (b4.center) to (m4.south west);
\draw (b4.center) to +(0,-0.75);

\node(b1)[blackdot] at (-0.15,1.65){};
\draw[string,out=90,in=left] (m1.north west) to (b1.center);
\draw[string,out=right,in=90] (b1.center) to (m3.north east);
\draw (b1.center) to +(0,0.35);

\node(b3)[blackdot] at (-0.15,-1.65){};
\draw[string,out=270,in=left] (m2.south west) to (b3.center);
\draw[string,out=right,in=270] (b3.center) to (m4.south east);
\draw (b3.center) to +(0,-0.35);
\end{tikzpicture}
\end{aligned}
=
\frac{1}{d}\left(
\begin{aligned}
\begin{tikzpicture}
\node(1) at (-0.25,2){};
\node(2) at (0.25,2){};
\node(3) at (-0.25,-2){};
\node(4) at (0.25,-2){};
\node(b1)[blackdot] at (-0.25,1){};
\node(b2)[blackdot] at (0.25,1){};
\node(b3)[blackdot] at (-0.25,-1){};
\node(b4)[blackdot] at (0.25,-1){};
\draw (b1.center) to (1);
\draw (b2.center) to (2);
\draw (b3.center) to (3);
\draw (b4.center) to (4);
\end{tikzpicture}
\end{aligned}
\quad -\quad
\begin{aligned}
\begin{tikzpicture}
\node(1) at (-0.25,2){};
\node(2) at (0.25,2){};
\node(3) at (-0.25,-2){};
\node(4) at (0.25,-2){};
\node(b2)[blackdot] at (0.25,1){};
\node(b4)[blackdot] at (0.25,-1){};
\draw (3) to (1);
\draw (b2.center) to (2);
\draw (b4.center) to (4);
\end{tikzpicture}
\end{aligned}\right)
\quad + \quad
\begin{aligned}
\begin{tikzpicture}
\node(1) at (-0.25,2){};
\node(2) at (0.25,2){};
\node(3) at (-0.25,-2){};
\node(4) at (0.25,-2){};
\draw (3) to (1);
\draw (4) to (2);
\end{tikzpicture}
\end{aligned}
\end{equation}
   
\textit{Unitary error bases} (UEBs) are fundamental to protocols such as teleportation and dense coding as well as finding application in quantum error correction~\cite{werner2001all,klapp,mypaper1,steane}.

Partitioned UEBs~\cite{Bandyopadhyay} are unitary error bases containing the identity operator equipped with a
partition into the identity and $d+1$ disjoint classes each containing $d-1$ commuting operators. From a partitioned UEB we can obtain a maximal family of MUBs by taking the common eigenbases of each of the commuting classes of operators~\cite{Bandyopadhyay}. We denote the map taking a partitioned UEB to its eigenbases by $\theta$.
  
Later in Section~\ref{section:MMUB UEBM}, making use of our diagrammatic maximal families of MUBs we introduce the following converse map $\phi_{H}:\text{Maximal family of MUBs} \to \text{Partitioned UEB}$ given extra data in the form of a family of Hadamards $G$ and $H$ (here $*$ is a projector defined using $\tinymult$):
\begin{equation}
\begin{aligned}
\begin{tikzpicture}[scale=0.5]
\node (f) [morphism,wedge, connect s length=1.5cm, connect se length=1.5cm, width=1cm, connect n length=1.25cm,connect sw length=1.5cm] at (0,0) {$\phi_{H}(M)$};
\node  at (f.connect s) {};
\node  at (f.connect se) {};
\node  at (f.connect n) {};
\end{tikzpicture}
\end{aligned}
:=\quad
\begin{aligned}
\begin{tikzpicture}[scale=0.75]
\node (m1) [morphism,wedge,scale=0.5] at (0,0.75) {M};
\node(m2)[morphism,wedge,hflip,scale=0.5]  at (0,-0.75) {M};
\draw (m1.south east) to (m2.north east)[string];
\draw (m1.south west) to (m2.north west)[string];
\node(b1)[blackdot] at (-0.25,0.1){};
\node(b2)[blackdot] at (0.25,0.2){};
\node(h)[wedge,morphism,scale=0.5] at (1.75,-1){$H$};
\node(b4)[blackdot] at (1.25,-1.5){};
\node(m) at (1,-0.25){};
\node(s)[whitedot,inner sep=0.1pt] at (1.75,-2){$\scalebox{1.5}{$\boldsymbol{*}$}$};
\draw[string,out=left,in=right,looseness=1] (m.center) to (b1.center);
\draw[string,out=left,in=right,looseness=2] (b4.center) to (m.center);
\draw[string,out=right,in=south west] (b4.center) to (h.south west);
\draw[string,out=90,in=right] (h.north) to (b2.center);
\node(i1) at (0,-2.5){};
\node(i2) at (1.25,-2.5){};
\node(i3) at (1.75,-2.5){};
\node(o) at (0,2){};
\draw[string] (i1.center) to (m2.south);
\draw[string] (i2.center) to (b4.center);
\draw[string] (i3.center) to (h.south);
\draw[string] (o.center) to (m1.north);
\end{tikzpicture}
\end{aligned}
\quad + \quad
\begin{aligned}
\begin{tikzpicture}[scale=0.75]
\node(b2)[blackdot] at (0.425,0.2){};
\node(h)[morphism,wedge,scale=0.5] at (1.25,-1.25){$G$};
\node(m) at (1,-0.25){};
\node(s)[state,black,hflip,scale=0.25,label={[label distance=-0.08cm]30:\tiny{0}}] at (1.75,-2){};
\draw[string,out=90,in=right] (h.north) to (b2.center);
\node(i1) at (0,-2.5){};
\node(i2) at (1.25,-2.5){};
\node(i3) at (1.75,-2.5){};
\node(o) at (0.425,2){};
\draw[string,out=90,in=left,looseness=0.5] (i1.center) to (b2.center);
\draw[string] (i2.center) to (h.south);
\draw[string] (i3.center) to (s.center);
\draw[string] (o.center) to (b2.center);
\end{tikzpicture}
\end{aligned}
\end{equation}

We show that given family of Hadamards $H$ the the composition $\theta\circ\phi_{H}$ is the identity and thus conclude that all maximal families of MUBs can be obtained in this way up to a choice of $H$ which we identify with a choice of eigenvalues. Each maximal MUB is associated with an infinite family of partitioned UEBs which are not necessarily equivalent.

In Section~\ref{section:gff} we introduce a diagrammatic axiomatisation for finite fields as algebras over Hilbert spaces. We show that, given a finite field, the following unitary operators form a partitioned UEB:\begin{equation}\label{eq:mainconst}
\begin{aligned}
\begin{tikzpicture}
\node (f) [morphism,wedge, connect s length=1.5cm width=0.5mm, connect se length=1.5cm, width=1cm, connect n length=1.25cm,connect sw length=1.5cm] at (0,0) {$U_{FF}$};
\node  at (f.connect s) {};
\node  at (f.connect se) {};
\node  at (f.connect n) {};
\end{tikzpicture}
\end{aligned}
:=\quad
\begin{aligned}
\begin{tikzpicture}[scale=0.75]
\node (chi) [morphism,wedge,scale=0.75] at (0,-0.1) {\large$\chi$};
\node(b1)[blackdot]  at (-0.25,0.5) {};
\node(add)[reddot] at (0.25,1){};
\node(x)[yellowdot] at (1.25,-0.5){};
\node(b2)[blackdot] at (0.5,-1){};
\node(*) at (0.5,-1.5){};
\node(m) at (1.25,-1){};
\node(i1) at (-1,-2.5){};
\node(i2) at (0.5,-2.5){};
\node(i3) at (1.25,-2.5){};
\node(i2a) at (0.5,-1.75){};
\node(i3a) at (1.25,-1.75){};
\node(o) at (0.25,2){};
\node(g1)[whitedot,projdot] at (1.25,-2.25){};
\node(g2)[whitedot,inner sep=0.1pt] at (1.25,-2){$\projs$};
\draw[green][string, line width=1pt](g1.center) to (g2.center);
\draw[string,in=left,out=90,looseness=0.6] (i1.center) to (b1.center);
\draw[string,in=90,out=right] (b1.center) to (chi.north);
\draw[string,in=left,out=90] (b1.center) to (add.center);
\draw[string,in=90,out=right] (add.center) to (x.center);
\draw[string,in=270,out=90] (add.center) to (o.center);
\draw[string,in=270,out=90,looseness=1] (i2.center) to (i2a.center);
\draw[string,in=270,out=90,looseness=1] (i2a.center) to (m.center);
\draw[string,in=right,out=90,looseness=1] (m.center) to (x.center);
\draw[string,in=270,out=90] (i3.center) to (g1.center);
\draw[string,in=270,out=90] (g2.center) to (i3a.center);
\draw[string,in=270,out=90] (i3a.center) to (b2.center);
\draw[string,in=270,out=left] (b2.center) to (chi.south);
\draw[string,in=left,out=right] (b2.center) to (x.center);
\end{tikzpicture}
\end{aligned}
\quad + \quad
\begin{aligned}
\begin{tikzpicture}[scale=0.75]
\node(b2)[reddot] at (0.6,0.2){};
\node(m) at (1,-0.25){};
\node(s)[reddot] at (1.75,-2){};
\draw[string,out=90,in=right,looseness=0.5] (i3.center) to (b2.center);
\node(i1) at (0,-2.5){};
\node(i2) at (1.2,-2.5){};
\node(i3) at (1.75,-2.5){};
\node(o) at (0.6,2){};
\draw[string,out=90,in=left,looseness=0.5] (i1.center) to (b2.center);
\draw[string] (i3.center) to (s.center);
\draw[string] (o.center) to (b2.center);
\end{tikzpicture}
\end{aligned}
\end{equation}
Here $\tinymult$ represents the computational basis, $\tinymult[reddot]$ and $\tinymult[yellowdot]$ are the linear extension of the addition and multiplication respectively of the finite field and $\chi$ is the Fourier transform for the additive group. We finish by giving an example of our construction in dimension $d=4$.

\paragraph{Related work.}In their 2002 paper, Bandyopadhyay et al~\cite{Bandyopadhyay} introduced partitioned UEBs and showed how to obtain maximal families of MUBs from them.
 Gogioso and Zeng gave a diagrammatic axiomatisation of complex group algebras in their 2015 paper~\cite{stefwill}. We have extended this to finite fields in order to give our diagrammatic construction of partitioned UEBs and thus, maximal families of MUBs.    

\section{Background}\label{section:graphical aa}
 We begin with the definitions of mutually unbiased bases and unitary error bases (UEBs), we will then review the necessary categorical quantum mechanics material  using Penrose tensor diagrams to present our main results diagrammatically in Section~\ref{section:MMUB UEBM} and Section~\ref{section:gff}. \paragraph{Basic defininitions.}
The following result was given a particularly simple and elegant proof by Bandyopahyay et al~\cite{Bandyopadhyay}.
\begin{theorem}
The largest family of MUBs that can exist on a $d$-dimensional Hilbert space is a family of $d+1$ MUBs. 
\end{theorem}
In Section~\ref{section:MMUB UEBM} we prove the equivalent tensor diagrammatic characterisation of maximal families of MUBs given in the introduction.

We now define  the equivalence of pairs of UEBs.

\begin{definition}[Equivalent UEBs~\cite{werner2001all}]
Given UEBs $X_{ij}$ and $Y_{ij}$ they are \textit{equivalent} if there exist
unitary operators $U$ and $V$, complex numbers with unit absolute value $c_{ij}$ and permutation $p$ such that:
\begin{equation}\label{eq:uebeq}
X_{ij}=c_{ij}UY_{p(i,j)}V
\end{equation}
\end{definition}
 Later, we will formally define partitioned UEBs in
Definition~\ref{def:pueb}, and give a tensor diagrammatic characterization of UEBs~\cite{mustothesis} in
Proposition~\ref{prop:grueb}, with an additional tensor diagrammatic axiom for $\UEBM$s given in Lemma~\ref{lem:guebp}. 
\begin{definition}[Hadamard]\label{def:had}
A \textit{Hadamard} matrix of order $d$ is a $d\times d$ matrix $H$, such that $|H_{ij}|=1$ and $H  H^{\dag}=H^\dag H=d \mathbb{I}_d$~\cite{Hadamard}. \end{definition}
Hadamards with the addition of a normalization constant, are precisely change of basis matrices between  pairs of mutually unbiased bases~\cite{mubs}. To see how mutually unbiased bases can be recovered from Definition~\ref{def:had}, consider the following. Given a Hadamard $H$, the matrix $\frac{1}{\sqrt{d}}H$ is unitary, since $H  H^{\dag}=H^\dag H=d \mathbb{I}_d$. So $\frac{1}{\sqrt{d}}H$ is a change of basis matrix between two orthonormal bases. Let $H':=\frac{1}{\sqrt{d}}H$, and represent the computational basis states by $\ket{a_i}$ and define $H'\ket{a_i}:=\ket{b_i}.$ The other Hadamard condition gives us that for all $i,j$ we have  $|H_{ij}|=1$ thus $|H_{ij}|^2=1$ and so: 
\begin{equation*}|\braket{a_j}{b_i}|^2=|\bra{a_j} H' \ket{a_i}|^2=|\bra{a_j} \frac{1}{\sqrt{d}}H\ket{a_i}|^2=\frac{1}{d}|\bra{a_j} H\ket{a_i}|^2=\frac{1}{d}|H_{ij}|^2=\frac{1}{d}\end{equation*}
This gives us back Definition~\ref{def:mub}.

After we have established the necessary tensor diagrammatic notation below we introduce a tensor diagrammatic axiomatisation of Hadamards in Lemma~\ref{lem:had} and controlled Hadamards which are indexed families of Hadamard in Definition~\ref{def:ch}.
\paragraph{Penrose tensor diagrams.}   We now introduce some results from categorical quantum mechanics and its graphical calculus for tensors~\cite{bob-book, abramskycoecke2004, surveycategoricalquantummechanics} which is different although closely related to the diagrams used in the tensor network community~\cite{TN1,TN2,TN3}. For those with knowledge of category theory we will be  working in  $\cat{FHilb}$, the category of linear maps and finite dimensional Hilbert spaces which is a $\dagger$-symmetric monoidal category. Much of the following could be interpreted in a general $\dagger$-symmetric monoidal category, however the main results of this paper use structure particular to $\cat{FHilb}$.   

In our formalism wires represent Hilbert spaces and boxes and nodes represent linear maps between Hilbert spaces. Different Hilbert spaces are represented by different coloured wires. We take the convention that diagrams are read from bottom to top. Composite linear maps are represented by vertical composition along the wires. Tensor products are represented by horizontal composition. Let $\mathcal{A,\, B,\, C\, \text{ and }D}$ be Hilbert spaces represented by black, red, blue and green wires respectively. The following diagram therefore represents the linear map $F\otimes G$ for $F:\mathcal A \rightarrow \mathcal B$ and $G:\mathcal C \rightarrow \mathcal D$. 
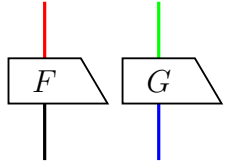
\begin{equation*}
\begin{aligned}
\begin{tikzpicture}
\begin{pgfonlayer}{background}
\node (f) [morphism,wedge] at (0,0) {$F$};
\node (g) [morphism,wedge] at (1.5,0) {$G$};
\end{pgfonlayer}
\begin{pgfonlayer}{foreground}
\draw[string,line width=1.2pt,red] (f.north) to node [auto] {} +(0,0.75);
\draw[string,line width=1.2pt] (f.south) to [out=down, in=up] node [auto, swap] {} +(0,-0.75);
\draw[string,line width=1.2pt,green] (g.north) to node [auto] {} +(0,0.75);
\draw[string,line width=1.2pt,blue] (g.south) to [out=down, in=up] node [auto, swap] {} +(0,-0.75);
\end{pgfonlayer}
\end{tikzpicture}
\end{aligned}
\end{equation*}
We will occasionally have to make use of different coloured wires but most of the time we will just use black wires to represent a single Hilbert space. We will use the convention that reflection in a horizontal axis represents adjunction or `taking the $\dag$' and reflection in a vertical axis represents complex conjugation. Since linear algebraic equations remain true under taking the adjoint of both sides all diagrammatic equations will remain true under  reflection about a horizontal axis. We will use the asymmetry of boxes representing linear maps to make it clear when we have taken an adjoint or complex conjugate of a given linear map. This works as follows:\begin{align*}
\left(
\begin{aligned}
\begin{tikzpicture}
  \begin{pgfonlayer}{background}
    \node (f) [morphism,wedge] at (0,0) {$F$};
  \end{pgfonlayer}
  \begin{pgfonlayer}{foreground}
    \draw[string] (f.north) to
%        node [auto] {$B$}
        +(0,0.75) node [above] {};
    \draw[string] (f.south) to
%        node [auto, swap] {$A$}
        +(0,-0.75) node [below] {};
  \end{pgfonlayer}
\end{tikzpicture}
\end{aligned}
\right)^{{\Huge \dag}}
\quad=\quad
\begin{aligned}
\begin{tikzpicture}
  \begin{pgfonlayer}{background}
    \node (f) [morphism,hflip,wedge] at (0,0) {$F$};
  \end{pgfonlayer}
  \begin{pgfonlayer}{foreground}
    \draw[string] (f.north) to
%        node [auto] {$B$}
        +(0,0.75) node [above] {};
    \draw[string] (f.south) to
%        node [auto, swap] {$A$}
        +(0,-0.75) node [below] {};
  \end{pgfonlayer}
\end{tikzpicture}
\end{aligned}\qquad \quad & \qquad \quad\left(
\begin{aligned}
\begin{tikzpicture}
  \begin{pgfonlayer}{background}
    \node (f) [morphism,wedge] at (0,0) {$F$};
  \end{pgfonlayer}
  \begin{pgfonlayer}{foreground}
    \draw[string] (f.north) to
%        node [auto] {$B$}
        +(0,0.75) node [above] {};
    \draw[string] (f.south) to
%        node [auto, swap] {$A$}
        +(0,-0.75) node [below] {};
  \end{pgfonlayer}
\end{tikzpicture}
\end{aligned}
\right)^{{\Huge *}}
\quad=\quad
\begin{aligned}
\begin{tikzpicture}
  \begin{pgfonlayer}{background}
    \node (f) [morphism,vflip,wedge] at (0,0) {$F$};
  \end{pgfonlayer}
  \begin{pgfonlayer}{foreground}
    \draw[string] (f.north) to
%        node [auto] {$B$}
        +(0,0.75) node [above] {};
    \draw[string] (f.south) to
%        node [auto, swap] {$A$}
        +(0,-0.75) node [below] {};
  \end{pgfonlayer}
\end{tikzpicture}
\end{aligned}
\end{align*}
We represent states as boxes with wires going out but no wires coming in, and effects as boxes with wires coming in but none going out. Boxes with no wires in or out are complex scalars.  The correspondence with bra ket notation is therefore as follows:
\begin{align*}\ket j :=
\begin{pic}
   \node[black,state,scale=0.75,label={[label distance=0.15cm]330:$j$}] (2) at (0,0) {};
   \draw[string](2.center) to +(0,1);
\end{pic}\qquad & \qquad \quad
\bra i:=
\begin{pic}[yscale=-1]
   \node[black,state,hflip,scale=0.75,label={[label distance=0.15cm]10:$i$}] (2) at (0,0) {};
   \draw[string](2.center) to +(0,1);
\end{pic} & \quad \braket{i}{j}:=
\begin{pic}[yscale=1]
   \node[black,state,scale=0.75,label={[label distance=0.15cm]330:$j$}] (2) at (0,0) {};
   \draw[string](2.center) to +(0,1);
   \node[black,state,hflip,scale=0.75,label={[label distance=0.15cm]10:$i$}] (2) at (0,1) {};
\end{pic}
\end{align*}
In order to capture the structure of maximal families of MUBs and $\UEBM$s using tensor diagrams in the next section, we use the computational basis tensor as an indexing set by making use of certain algebras on Hilbert space. We now introduce those algebras known as   \mbox{$\dagger$-special} commutative Frobenius algebras ($\dfrob$s) on Hilbert spaces that are equivalent to orthonormal bases as we will show. First we use tensor diagrams to define the properties that make up the  $\dfrob$ axioms. 
\begin{definition}[Associativity]\label{def:ass}
The linear map $\tinymult$ is \textit{associative} if the following equation holds:
\begin{equation}\label{eq:as}
\begin{aligned}\begin{tikzpicture}[scale=0.65]
          \node (0a) at (-1,0.25) {};
          \node (0b) at (-0.5,0.25) {};
          \node (0c) at (0.5,0.25) {};

          \node[blackdot] (1) at (0,1) {};
          \node[blackdot] (2) at (-0.5,1.5) {};

          \node (5a) at (0-0.5,2) {};

          \draw[string, out=90, in =180] (0a.center) to (2.center);
          \draw[string, out=0, in=90] (2.center) to (1.center);
          \draw[string, out=0, in=90] (1.center) to (0c.center);
          \draw[string, out=180, in=90] (1.center) to (0b.center);

          \draw[string] (2.center) to (5a.center);
         
          \end{tikzpicture}\end{aligned}   
  \quad   = \quad
\begin{aligned}\begin{tikzpicture}[scale=0.65]
          \node (0a) at (0.75,0.25) {};
          \node (0b) at (0.25,0.25) {};
          \node (0c) at (-0.75,0.25) {};

          \node[blackdot] (1) at (-0.25,1) {};
          \node[blackdot] (2) at (0.25,1.5) {};

          \node (5a) at (0.25,2) {};

          \draw[string, out=90, in =00] (0a.center) to (2.center);
          \draw[string, out=180, in=90] (2.center) to (1.center);
          \draw[string, out=180, in=90] (1.center) to (0c.center);
          \draw[string, out=0, in=90] (1.center) to (0b.center);

          \draw[string] (2.center) to (5a.center);
\end{tikzpicture}\end{aligned} 
\end{equation}

\end{definition}
\begin{definition}[Unitality]\label{def:uni}
The linear map $\tinymult[blackdot]$ is \textit{unital} if there exists a state  $\tinyunit[blackdot]$ such that:
\begin{equation}\label{eq:un}
\begin{aligned}\begin{tikzpicture}[scale=0.65]
          
          \node (0a) at (-1,-0.75) {};
          \node (0b) at (-0.5,-0.5) {};
          \node[blackdot] (0c) at (0,-0.5) {};
          \node[blackdot] (7) at (-0.5,0) {};       
          \node (10a) at (-0.5,0.5) {};       
          \draw[string, out=90, in =180] (0a.center) to (7.center);
          \draw[string, out=0, in=90] (7.center) to (0c.center);
          \draw[string] (7.center) to (10a.center);
         
          \end{tikzpicture}\end{aligned}   
  \quad   = \quad
\begin{aligned}\begin{tikzpicture}[scale=0.65]
           \node (10a) at (-0.5,0) {}; 
           \node (5a) at (-0.5,1.25) {}; 
           \draw[string] (5a.center) to (10a.center);
\end{tikzpicture}\end{aligned}
\quad=\quad
\begin{aligned}\begin{tikzpicture}[scale=0.65]

      \node[blackdot] (0a) at (-1,-0.5) {};
          \node (0b) at (-0.5,-0.5) {};
          \node (0c) at (0,-0.75) {};
          \node[blackdot] (7) at (-0.5,0) {};       
          \node (10a) at (-0.5,0.5) {};       
          \draw[string, out=90, in =180] (0a.center) to (7.center);
          \draw[string, out=0, in=90] (7.center) to (0c.center);
          \draw[string] (7.center) to (10a.center);

\end{tikzpicture}\end{aligned}  
  \end{equation}
The state $\tinyunit[blackdot]$ is called the \textit{unit}.
\end{definition}
\begin{definition}[Commutativity]\label{def:comm}
The linear map $\tinymult$ is \textit{commutative} if the following equation holds:
\begin{equation}\label{eq:comm}
\begin{aligned}\begin{tikzpicture}[scale=0.6]
          \node (0a) at (-1,0.25) {};
          \node (0b) at (-0.5,0.25) {};
          \node (0c) at (0,0.25) {};

          \node[blackdot] (2) at (-0.5,1.5) {};

          \node (5a) at (-0.5,2) {};

          \draw[string, out=90, in =180] (0a.center) to (2.center);
          \draw[string, out=0, in=90] (2.center) to (0c.center);

          \draw[string] (2.center) to (5a.center);
         
          \end{tikzpicture}\end{aligned}   
  \quad   = \quad
\begin{aligned}\begin{tikzpicture}[scale=0.6]
           \node (0a) at (-1,0.25) {};
          \node (0b) at (-0.5,0.25) {};
          \node (0c) at (0,0.25) {};

          \node[blackdot] (2) at (-0.5,1.5) {};

          \node (5a) at (-0.5,2) {};
          \node(m1) at (-0.1,1.1){};
          \node(m2) at (-0.9,1.1){};
       
          \draw[string, out=up, in =down] (0a.center) to (m1.center);
          \draw[string, out=up, in=right,looseness=1.4] (m1.center) to (2.center);
          \draw[string, out=down, in=90] (m2.center) to (0c.center);
          \draw[string, out=180, in=up,looseness=1.4] (2.center) to (m2.center);
          \draw[string] (2.center) to (5a.center);
\end{tikzpicture}\end{aligned}\end{equation}   \end{definition}
If we take the adjoint of both sides of equations~\eqref{eq:as},\eqref{eq:un} and~\eqref{eq:comm} we obtain the definitions of coassociativity, counitality and cocommutativity respectively.
\begin{definition}[(Co)monoid]
The linear map $\tinymult$ together with $\tinyunit$ is a \textit{monoid} if it is associative and unital. The linear map $\tinycomult$ together with $\tinycounit$ is a \textit{comonoid} if it is coassociative and counital.
\end{definition}
\begin{definition}[Comonoid homomorphism]\label{def:cohom}
A linear map F is a \textit{comonoid homomorphism} for the comonoid $\tinycomult$ if the following equations hold:
\begin{align}
\begin{pic}
   \node[morphism,wedge,scale=0.5](f) at (0,0){$F$};
   \node[blackdot](b) at (0,0.5){};
   \node(i) at (0,-0.5){};
   \node(o1) at (-0.35,1){};
   \node(o2) at (0.35,1){};
   \draw[string] (i) to (f.south);
   \draw[string] (b.center) to (f.north);
   \draw[string,out=left,in=270] (b.center) to (o1);
   \draw[string,out=right,in=270] (b.center) to (o2);
\end{pic}\quad = \quad
\begin{pic}
  \begin{pgfonlayer}{foreground}
   \node[morphism,wedge,scale=0.5](f1) at (-0.35,0.5){$F$};
   \node[morphism,wedge,scale=0.5](f2) at (0.35,0.5){$F$};
  \end{pgfonlayer}
  \begin{pgfonlayer}{background}
   \node[blackdot](b) at (0,0){};
   \node(i) at (0,-0.5){};
   \node(o1) at (-0.35,1){};
   \node(o2) at (0.35,1){};
   \draw[string] (i) to (f.south);
   \draw[string] (b.center) to (i);
   \draw[string,out=left,in=270] (b.center) to (f1.south);
   \draw[string] (f1.north) to (o1);
   \draw[string,out=right,in=270] (b.center) to (f2.south);
   \draw[string] (f2.north) to (o2);
  \end{pgfonlayer}
\end{pic}\qquad & \qquad
 \begin{pic}
   \node[morphism,wedge,scale=0.5](f) at (0,0){$F$};
   \node[blackdot](b) at (0,0.5){};
   \node(i) at (0,-0.5){};
   \node(o1) at (-0.35,1){};
   \node(o2) at (0.35,1){};
   \draw[string] (i) to (f.south);
   \draw[string] (b.center) to (f.north);
\end{pic}\quad = \quad
\begin{pic}
   \node(f) at (0,0){};
   \node[blackdot](b) at (0,0.5){};
   \node(i) at (0,-0.5){};
   \node(o1) at (-0.35,1){};
   \node(o2) at (0.35,1){};
   \draw[string] (i) to (f.center);
   \draw[string] (b.center) to (f.center);
  \end{pic}
  \end{align}
\end{definition}
\begin{definition}[Special, quasi-special]The linear maps $\tinymult$ and $\tinycomult[blackdot]$ are \textit{special} if the left hand side  equation holds and \textit{quasi-special} if the right hand side equation holds with $d$ the dimension of the Hilbert space:
\begin{align}\label{eq:sp}
\begin{aligned}\begin{tikzpicture}[scale=0.65]
          \node (0a) at (-1,0) {};
          \node (0b) at (-0.5,0) {};
          \node (0c) at (0,0) {};
          \node[blackdot] (2) at (-0.5,0.5) {};
          \node (5a) at (-0.5,1) {};
          \node[blackdot] (7) at (-0.5,-0.5) {};       
          \node (10a) at (-0.5,-1) {};       
          \draw[string, out=90, in =180] (0a.center) to (2.center);
          \draw[string, out=0, in=90] (2.center) to (0c.center);
          \draw[string] (2.center) to (5a.center);
          \draw[string, out=270, in =180] (0a.center) to (7.center);
          \draw[string, out=0, in=270] (7.center) to (0c.center);
          \draw[string] (7.center) to (10a.center);         
          \end{tikzpicture}\end{aligned}   
  \quad   = \quad
\begin{aligned}\begin{tikzpicture}[scale=0.65]
           \node (10a) at (-0.5,-1) {}; 
           \node (5a) at (-0.5,1) {}; 
           \draw[string] (5a.center) to (10a.center);
\end{tikzpicture}\end{aligned}  
\qquad & \qquad
\begin{aligned}\begin{tikzpicture}[scale=0.65]
          \node (0a) at (-1,0) {};
          \node (0b) at (-0.5,0) {};
          \node (0c) at (0,0) {};
          \node[blackdot] (2) at (-0.5,0.5) {};
          \node (5a) at (-0.5,1) {};
          \node[blackdot] (7) at (-0.5,-0.5) {};       
          \node (10a) at (-0.5,-1) {};       
          \draw[string, out=90, in =180] (0a.center) to (2.center);
          \draw[string, out=0, in=90] (2.center) to (0c.center);
          \draw[string] (2.center) to (5a.center);
          \draw[string, out=270, in =180] (0a.center) to (7.center);
          \draw[string, out=0, in=270] (7.center) to (0c.center);
          \draw[string] (7.center) to (10a.center);         
          \end{tikzpicture}\end{aligned}   
  \quad   = \quad d
\begin{aligned}\begin{tikzpicture}[scale=0.65]
           \node (10a) at (-0.5,-1) {}; 
           \node (5a) at (-0.5,1) {}; 
           \draw[string] (5a.center) to (10a.center);
\end{tikzpicture}\end{aligned}  
\end{align}
\end{definition}
\begin{definition}[$\dagger$-Frobenius law]
The linear map $\tinymult$ and its adjoint $\tinycomult$ obey the
\emph{ \mbox{$\dagger$-Frobenius} law} if the following equations hold:\begin{equation}
    \begin{aligned}\begin{tikzpicture}[yscale=0.5,xscale=0.6]
          \node (0) at (0,0) {};
          \node (0a) at (0,1) {};
          \node[blackdot] (1) at (0.5,2) {};
          \node[blackdot] (2) at (1.5,1) {};
          \node (3) at (1.5,0) {};
          \node (4) at (2,3) {};
          \node (4a) at (2,2) {};
          \node (5) at (0.5,3) {};
          \draw[string] (0) to (0a.center);
          \draw[string,out=90,in=180] (0a.center) to (1.center);
          \draw[string,out=0,in=180] (1.center) to (2.center);
          \draw[string,out=0,in=270] (2.center) to (4a.center);
          \draw[string] (4a.center) to (4);
          \draw[string] (2.center) to (3);
          \draw[string] (1.center) to (5);
      \end{tikzpicture}\end{aligned}
    \quad = \quad
    \begin{aligned}\begin{tikzpicture}[yscale=0.5,xscale=0.6]
          \node (0a) at (-0.5,0) {};
          \node (0b) at (0.5,0) {};
          \node[blackdot] (1) at (0,1) {};
          \node[blackdot] (2) at (0,2) {};
          \node (3a) at (-0.5,3) {};
          \node (3b) at (0.5,3) {};
          \draw[string,out=90,in=180] (0a) to (1.center);
          \draw[string,out=90,in=0] (0b) to (1.center);
          \draw[string] (1.center) to (2.center);
          \draw[string,out=180,in=270] (2.center) to (3a);
          \draw[string,out=0,in=270] (2.center) to (3b);
      \end{tikzpicture}\end{aligned}
    \quad = \quad
    \begin{aligned}\begin{tikzpicture}[yscale=0.5,xscale=-0.6]
          \node (0) at (0,0) {}; 
          \node (0a) at (0,1) {};
          \node[blackdot] (1) at (0.5,2) {};
          \node[blackdot] (2) at (1.5,1) {};
          \node (3) at (1.5,0) {};
          \node (4) at (2,3) {};
          \node (4a) at (2,2) {};
          \node (5) at (0.5,3) {};
          \draw[string] (0) to (0a.center);
          \draw[string,out=90,in=180] (0a.center) to (1.center);
          \draw[string,out=0,in=180] (1.center) to (2);
          \draw[string,out=0,in=270] (2.center) to (4a.center);
          \draw[string] (4a.center) to (4);
          \draw[string] (2.center) to (3);
          \draw[string] (1.center) to (5);
      \end{tikzpicture}\end{aligned}
  \end{equation} 
\end{definition}
\begin{definition}[$\dagger$-special commutative Frobenius algebra] The linear map $\tinymult$ and state  $\tinyunit[blackdot]$ together with their adjoints are  a  $\dagger$-\textit{special commutative Frobenius algebra} ($\dagger$-SCFA) if they are a commutative monoid, special and obey the $\dagger$-Frobenius law. 
\end{definition}
The following result due to Coecke, Pavlovi{\'c}, and Vicary will prove important throughout this paper.
\begin{proposition}~\cite{spiderONB} In $\cat{FHilb}$ $\dagger$-SCFAs are in one to one correspondence with orthonormal bases. All $\dfrob$s can be written in terms of the corresponding orthonormal bases as follows:
\begin{equation}~\label{classtr}
\begin{pic}[xscale=0.6,yscale=0.4]
          \node (0a) at (-1,0) {};
          \node (0b) at (-0.5,0) {};
          \node (0c) at (0,0) {};      
          \node[blackdot] (2) at (-0.5,1.5) {};
          
          \node (5a) at (-0.5,3) {};
       
          \draw[string, out=90, in =180] (0a.center) to (2.center);
          \draw[string, out=0, in=90] (2.center) to (0c.center);
     
          \draw[string] (2.center) to (5a.center);
\end{pic}\quad =\quad \sum_i
\begin{pic}
[xscale=0.6,yscale=0.4]
          \node (0a) at (-0.9,0) {};
          \node (0b) at (-0.5,0) {};
          \node (0c) at (-0.1,0) {};      
          \node[black,state,scale=0.25,label={[label distance=-0.05cm]330:\tiny{$i$}}] (2) at (-0.5,2) {};
          \node[black,state,scale=0.25,hflip,label={[label distance=-0.05cm]30:\tiny{$i$}}] (3) at (-0.9,1) {};     
          \node[black,state,scale=0.25,hflip,label={[label distance=-0.05cm]30:\tiny{$i$}}] (4) at (-0.1,1) {};
          \node (5a) at (-0.5,3) {};
       
          \draw[string] (5a.center) to (2.center);
          \draw[string] (4.center) to (0c.center);
     
          \draw[string] (3.center) to (0a.center);
\end{pic}
\end{equation}

\ignore{ \begin{equation}
    \overbrace{\underbrace{
    \begin{pic}[scale=0.5]
     \node[blackdot,scale=1] (A){};
     \draw[string,out=190,in=90] (A) to (-2,-1);
     \draw[string,out=350,in=90] (A) to (2,-1);
     \draw[string,out=190,in=90] (A) to (-1.7,-1);
     \draw[string,out=350,in=90] (A) to (1.7,-1);  
     \draw[string,out=170,in=270] (A) to (-2,1);
     \draw[string,out=10,in=270] (A) to (2,1);
     \draw[string,out=170,in=270] (A) to (-1.7,1);
     \draw[string,out=10,in=270] (A) to (1.7,1); 
     \draw[loosely dotted] (-1.65,0.9) to (1.65,0.9); 
     \draw[loosely dotted] (-1.65,-0.9) to (1.65,-0.9);          
    \end{pic}}_{m}}^{n}
    \quad=\quad
    \sum^{d-1}_{i=0}
    \overbrace{\underbrace{
    \begin{pic}[scale=0.5]
     \node[state,black,scale=0.25,label={[label distance=-0.05cm]330:\tiny{$i$}}] (A)at (-1.75,0.5){};
     \node[state,hflip,black,scale=0.25,label={[label distance=-0.05cm]30:\tiny{$i$}}] (B)at (-2,-0.5){};
     \node[state,black,scale=0.25,label={[label distance=-0.05cm]330:\tiny{$i$}}] (C)at (1.75,0.5){};
     \node[state,hflip,black,scale=0.25,label={[label distance=-0.05cm]30:\tiny{$i$}}] (D)at (2,-0.5){};
     \node[state,black,scale=0.25,label={[label distance=-0.05cm]330:\tiny{$i$}}] (A')at (-1.25,0.5){};
     \node[state,hflip,black,scale=0.25,label={[label distance=-0.05cm]30:\tiny{$i$}}] (B')at (-1.5,-0.5){};
     \node[state,black,scale=0.25,label={[label distance=-0.05cm]330:\tiny{$i$}}] (C')at (1.25,0.5){};
     \node[state,hflip,black,scale=0.25,label={[label distance=-0.05cm]30:\tiny{$i$}}] (D')at (1.5,-0.5){};
     \draw[string] (A) to (-1.75,1);
     \draw[string] (B) to (-2,-1);
     \draw[string] (C) to (1.75,1);
     \draw[string] (D) to (2,-1);  
     \draw[string] (A') to (-1.25,1);
     \draw[string] (B') to (-1.5,-1);
     \draw[string] (C') to (1.25,1);
     \draw[string] (D') to (1.5,-1); 
     \draw[loosely dotted] (-1,0.1) to (1,0.1); 
     \draw[loosely dotted] (-1.25,-0.1) to (1.25,-0.1);          
    \end{pic}}_{m}}^{n}
  \end{equation}}
  
\end{proposition}
Equation~\eqref{classtr} can be used to show that the following very useful result holds:   
\begin{corollary}~\cite{spiderONB}
For $\tinymult$ a $\dfrob$ any two connected diagrams of black dots with the same numbers of inputs and outputs are equal. For example:
\begin{equation}\label{eq:sm}
    \overbrace{\underbrace{
    \begin{pic}[scale=0.5]
     \node[blackdot,scale=1] (A) at (-0.5,0.25){};
     \node[blackdot,scale=1] (B) at (0.5,-0.25){};
     \draw[string,out=190,in=90] (B.center) to (-2,-1);
     \draw[string,out=350,in=90] (B.center) to (2,-1);
     \draw[string,out=190,in=90] (B.center) to (-1.7,-1);
     \draw[string,out=350,in=90] (B.center) to (1.7,-1);  
     \draw[string,out=170,in=270] (A.center) to (-2,1);
     \draw[string,out=10,in=270] (A.center) to (2,1);
     \draw[string,out=170,in=270] (A.center) to (-1.7,1);
     \draw[string,out=10,in=270] (A.center) to (1.7,1); 
     \draw[loosely dotted] (-1.65,0.9) to (1.65,0.9); 
     \draw[loosely dotted] (-1.65,-0.9) to (1.65,-0.9); 
     \draw[dotted] (-0.1,-0.1) to (0.15,0.1);   
     \draw[string,in=170,out=300] (A.center) to (B.center);  
     \draw[string,in=100,out=340] (A.center) to (B.center);   
    \end{pic}}_{r}}^{s}
    \quad = \quad
    \overbrace{\underbrace{
    \begin{pic}[scale=0.5]
     \node[blackdot,scale=1] (A){};
     \draw[string,out=190,in=90] (A.center) to (-2,-1);
     \draw[string,out=350,in=90] (A.center) to (2,-1);
     \draw[string,out=190,in=90] (A.center) to (-1.7,-1);
     \draw[string,out=350,in=90] (A.center) to (1.7,-1);  
     \draw[string,out=170,in=270] (A.center) to (-2,1);
     \draw[string,out=10,in=270] (A.center) to (2,1);
     \draw[string,out=170,in=270] (A.center) to (-1.7,1);
     \draw[string,out=10,in=270] (A.center) to (1.7,1); 
     \draw[loosely dotted] (-1.65,0.9) to (1.65,0.9); 
     \draw[loosely dotted] (-1.65,-0.9) to (1.65,-0.9);       
    \end{pic}}_{r}}^{s}
\end{equation}

\end{corollary} 
Given a $\dfrob$ $\tinymult$ we will refer to the corresponding basis as the \emph{black basis}, the basis states $\tinystate[black,state]$ as \emph{black states} and their adjoints $\tinyeffect[black,state,hflip]$ as \emph{black effects} (we also extend this terminology to other colours). We can also see by combining equations~\eqref{classtr} and ~\eqref{eq:sm} that, for a  $\dfrob$ $\tinymult$ any black state or effect composed with a connected diagram of black dots will be \textit{copied} in the following sense:
 \begin{equation}\label{eq:spider}
    \overbrace{\underbrace{
    \begin{pic}[scale=0.5]
     \node[blackdot,scale=1] (A){};
     \draw[string,out=190,in=90] (A) to (-2,-1);
     \draw[string,out=350,in=90] (A) to (2,-1);
     \draw[string,out=190,in=90] (A) to (-1.7,-1);
     \draw[string,out=350,in=90] (A) to (1.7,-1);  
     \draw[string,out=170,in=270] (A) to (-2,1);
     \draw[string,out=10,in=270] (A) to (2,1);
     \draw[string,out=170,in=270] (A) to (-1.7,1);
     \draw[string,out=10,in=270] (A) to (1.7,1); 
     \draw[loosely dotted] (-1.65,0.9) to (1.65,0.9); 
     \draw[loosely dotted] (-1.65,-0.9) to (1.65,-0.9); 
     \node[state,black,scale=0.25,label={[label distance=-0.05cm]330:\tiny{$k$}}] (s)at (-1,-1){}; 
     \draw[string,out=260,in=up] (A.center) to (s.center); 
            \end{pic}}_{m}}^{n}
    \quad=\quad
    \overbrace{\underbrace{
    \begin{pic}[scale=0.5]
     \node[state,black,scale=0.25,label={[label distance=-0.05cm]330:\tiny{$k$}}] (A)at (-1.75,0.5){};
     \node[state,hflip,black,scale=0.25,label={[label distance=-0.05cm]30:\tiny{$k$}}] (B)at (-2,-0.5){};
     \node[state,black,scale=0.25,label={[label distance=-0.05cm]330:\tiny{$k$}}] (C)at (1.75,0.5){};
     \node[state,hflip,black,scale=0.25,label={[label distance=-0.05cm]30:\tiny{$k$}}] (D)at (2,-0.5){};
     \node[state,black,scale=0.25,label={[label distance=-0.05cm]330:\tiny{$k$}}] (A')at (-1.25,0.5){};
     \node[state,hflip,black,scale=0.25,label={[label distance=-0.05cm]30:\tiny{$k$}}] (B')at (-1.5,-0.5){};
     \node[state,black,scale=0.25,label={[label distance=-0.05cm]330:\tiny{$k$}}] (C')at (1.25,0.5){};
     \node[state,hflip,black,scale=0.25,label={[label distance=-0.05cm]30:\tiny{$k$}}] (D')at (1.5,-0.5){};
     \draw[string] (A) to (-1.75,1);
     \draw[string] (B) to (-2,-1);
     \draw[string] (C) to (1.75,1);
     \draw[string] (D) to (2,-1);  
     \draw[string] (A') to (-1.25,1);
     \draw[string] (B') to (-1.5,-1);
     \draw[string] (C') to (1.25,1);
     \draw[string] (D') to (1.5,-1); 
     \draw[loosely dotted] (-1,0.1) to (1,0.1); 
     \draw[loosely dotted] (-1.25,-0.1) to (1.25,-0.1);          
    \end{pic}}_{m-1}}^{n}
  \end{equation} 

This allows us to use the basis states of a $\dfrob$ as an indexing set. For the rest of this paper black wires will represent the $d$ dimensional Hilbert space $\mathcal H\cong\C^d$. We will take the black $\dfrob$ $\tinymult$ to be the $\dfrob$ corresponding to the computational basis of $\mathcal H$. As a first example, we now define
 \textit{Hadamards} and \textit{controlled Hadamards} using tensor diagrams.
\begin{lemma}[Hadamard]\label{lem:had}
Given a  $\dfrob$ $\tinymult$ on a  $d$-dimensional Hilbert space $\mathcal H$, a linear map of type  $H:\mathcal H    \rightarrow \mathcal H$  is a \textit{ Hadamard} if and only if the following equations hold:
\begin{align}\label{eq:had}
\begin{aligned}\begin{tikzpicture}%[scale=0.65]
           \node (2) at (0,-1.5) {}; 
           \node (3) at (0,1.5) {};
           \node[morphism,wedge,scale=0.5] (h1) at (0,-0.5) {$H$}; 
           \node[morphism,wedge,scale=0.5,hflip] (h11) at (0,0.5) {$H$};            \draw[string] (2.center) to (h1.south);
           \draw[string] (h11.north) to (3.center);
           \draw[string] (h11.south) to (h1.north);   
\end{tikzpicture}\end{aligned} 
\quad = \quad
\begin{aligned}\begin{tikzpicture}%[scale=0.65]
           \node (2) at (0,-1.5) {}; 
           \node (3) at (0,1.5) {};
           \node[morphism,wedge,scale=0.5,hflip] (h1) at (0,-0.5) {$H$};            \node[morphism,wedge,scale=0.5] (h11) at (0,0.5) {$H$};            \draw[string] (2.center) to (h1.south);
           \draw[string] (h11.north) to (3.center);
           \draw[string] (h11.south) to (h1.north);   
\end{tikzpicture}\end{aligned} 
\quad = \quad
 d
\begin{aligned}\begin{tikzpicture}
\draw[string](0.75,0) to (0.75,3);
\end{tikzpicture}\end{aligned}
\qquad & \qquad 
\begin{aligned}\begin{tikzpicture}[scale=0.75]
           \node (2) at (0,-2) {}; 
           \node (3) at (0.35,0) {};
           \node[morphism,wedge,scale=0.5] (h1) at (0,-1) {$H$}; 
           \node[blackdot] (r1) at (0.35,-0.5) {};
           \node[blackdot] (r2) at (-0.35,-1.5) {}; 
           \node(a) at (-1,0){};
           \draw[string,out=left,in=270] (r2.center) to (a.center);
           \node(a2) at (1,-2){};
           \node(a3) at (-0.35,-2){};
           \draw[string,out=right,in=90] (r1.center) to (a2.center);
           \draw[string] (a3.center) to (r2.center) ;
           \draw[string] (r1.center) to (3.center);
           \draw[string,out=left,in=90] (r1.center) to (h1.north);           \draw[string,out=right,in=270] (r2.center) to (h1.south);
           \node (21) at (0,2) {}; 
           \node (31) at (0.35,0) {};
           \node[morphism,wedge,scale=0.5,hflip] (h11) at (0,1) {$H$}; 
           \node[blackdot] (r11) at (0.35,0.5) {};
           \node[blackdot] (r21) at (-0.35,1.5) {}; 
           \node(a1) at (-1,0){};
           \draw[string,out=left,in=90] (r21.center) to (a1.center);
           \node(a21) at (1,2){};
           \node(a31) at (-0.35,2){};
           \draw[string,out=right,in=270] (r11.center) to (a21.center);
           \draw[string] (a31.center) to (r21.center) ;
           \draw[string] (r11.center) to (31.center);
           \draw[string,out=left,in=270] (r11.center) to (h11.south);           \draw[string,out=right,in=90] (r21.center) to (h11.north);
\end{tikzpicture}\end{aligned}
\quad = \quad
\begin{aligned}\begin{tikzpicture}[xscale=-0.75,yscale=0.75]
           \node (2) at (0,-2) {}; 
           \node (3) at (0.35,0) {};
           \node[morphism,wedge,scale=0.5,hflip] (h1) at (0,-1) {$H$}; 
           \node[blackdot] (r1) at (0.35,-0.5) {};
           \node[blackdot] (r2) at (-0.35,-1.5) {}; 
           \node(a) at (-1,0){};
           \draw[string,out=left,in=270] (r2.center) to (a.center);
           \node(a2) at (1,-2){};
           \node(a3) at (-0.35,-2){};
           \draw[string,out=right,in=90] (r1.center) to (a2.center);
           \draw[string] (a3.center) to (r2.center) ;
           \draw[string] (r1.center) to (3.center);
           \draw[string,out=left,in=90] (r1.center) to (h1.north);           \draw[string,out=right,in=270] (r2.center) to (h1.south);
           \node (21) at (0,2) {}; 
           \node (31) at (0.35,0) {};
           \node[morphism,wedge,scale=0.5] (h11) at (0,1) {$H$}; 
           \node[blackdot] (r11) at (0.35,0.5) {};
           \node[blackdot] (r21) at (-0.35,1.5) {}; 
           \node(a1) at (-1,0){};
           \draw[string,out=left,in=90] (r21.center) to (a1.center);
           \node(a21) at (1,2){};
           \node(a31) at (-0.35,2){};
           \draw[string,out=right,in=270] (r11.center) to (a21.center);
           \draw[string] (a31.center) to (r21.center) ;
           \draw[string] (r11.center) to (31.center);
           \draw[string,out=left,in=270] (r11.center) to (h11.south);           \draw[string,out=right,in=90] (r21.center) to (h11.north);
\end{tikzpicture}\end{aligned}
\quad = \quad
\begin{aligned}\begin{tikzpicture}
\draw[string](0,0) to (0,3);
\draw[string](0.75,0) to (0.75,3);
\end{tikzpicture}\end{aligned}
  \end{align} \end{lemma} 
\begin{proof}
The left hand side of equation~\eqref{eq:had} is simply a tensor diagrammatic translation of \mbox{$H  H^{\dag}=H^\dag H=d \mathbb{I}_d$}. We now show the equivalence of the right hand side of equation~\eqref{eq:had} and the other condition of Definition~\ref{def:had}, that for all $i$, $j$, $|H_{ij}|=1$.  We have for all $i,j$:
\begin{align*}
|H_{ij}|&=1\\
\Leftrightarrow | \bra i H \ket j | &= 1\\
\Leftrightarrow  \bra i H \ket j \bra j H^\dag \ket i &= 1\\
\Leftrightarrow  \bra i H \ket j \bra j H^\dag \ket i &= \braket{i}{i}\braket{j}{j}\\
\end{align*}
We now translate this final equation into the graphical calculus, for all $i,j$:
\begin{equation*}
\begin{pic}[xscale=0.75]
\node[morphism,wedge,scale=0.5,hflip](hdag) at (0.5,0) {$H$};
\node[morphism,wedge,scale=0.5](h) at (-0.5,0) {$H$};
\node[state,black,scale=0.25,label={[label distance=-0.05cm]330:\tiny{$i$}}](i1) at (-0.5,-0.25){};
\node[state,hflip,black,scale=0.25,label={[label distance=-0.05cm]30:\tiny{$j$}}](j1) at (-0.5,0.25){};
\node[state,black,scale=0.25,label={[label distance=-0.05cm]330:\tiny{$j$}}](j2) at (0.5,-0.25){};
\node[state,hflip,black,scale=0.25,label={[label distance=-0.05cm]30:\tiny{$i$}}](i2) at (0.5,0.25){};
\draw[string] (i1) to (h.south);
\draw[string] (j1) to (h.north);
\draw[string] (i2) to (hdag.north);
\draw[string] (j2) to (hdag.south);
\end{pic}\quad = \quad
\begin{pic}[xscale=0.75]
\node[state,black,scale=0.25,label={[label distance=-0.05cm]330:\tiny{$i$}}](i1) at (-0.5,-0.25){};
\node[state,hflip,black,scale=0.25,label={[label distance=-0.05cm]30:\tiny{$i$}}](j1) at (-0.5,0.25){};
\node[state,black,scale=0.25,label={[label distance=-0.05cm]330:\tiny{$j$}}](j2) at (0.5,-0.25){};
\node[state,hflip,black,scale=0.25,label={[label distance=-0.05cm]30:\tiny{$j$}}](i2) at (0.5,0.25){};
\draw[string] (j1) to (i1);
\draw[string] (i2) to (j2);
\end{pic}
\end{equation*} 
Rearranging the left hand side we have for all $i,j$:
\begin{align*}
\begin{pic}[scale=0.75]
\node[morphism,wedge,scale=0.5,hflip](hdag) at (0,1) {$H$};
\node[morphism,wedge,scale=0.5](h) at (0,-1) {$H$};
\node[state,black,scale=0.25,label={[label distance=-0.05cm]330:\tiny{$i$}}](i1) at (0,-1.85){};
\node[state,hflip,black,scale=0.25,label={[label distance=-0.05cm]30:\tiny{$j$}}](j1) at (0,-0.5){};
\node[state,black,scale=0.25,label={[label distance=-0.05cm]330:\tiny{$j$}}](j2) at (0,0.5){};
\node[state,hflip,black,scale=0.25,label={[label distance=-0.05cm]30:\tiny{$i$}}](i2) at (0,1.85){};
\draw[string] (i1) to (h.south);
\draw[string] (j1) to (h.north);
\draw[string] (i2) to (hdag.north);
\draw[string] (j2) to (hdag.south);
\end{pic}\quad = \quad
\begin{pic}[scale=0.75]
\node[morphism,wedge,scale=0.5,hflip](hdag) at (0,1) {$H$};
\node[morphism,wedge,scale=0.5](h) at (0,-1) {$H$};
\node[state,black,scale=0.25,label={[label distance=-0.05cm]330:\tiny{$i$}}](i1) at (0,-1.85){};
\node[state,hflip,black,scale=0.25,label={[label distance=-0.05cm]30:\tiny{$j$}}](j1) at (0,-0.5){};
\node[state,black,scale=0.25,label={[label distance=-0.05cm]330:\tiny{$j$}}](j2) at (0,0.5){};
\node[state,hflip,black,scale=0.25,label={[label distance=-0.05cm]30:\tiny{$i$}}](i2) at (0,1.85){};
\draw[string] (i1) to (h.south);
\draw[string] (j1) to (h.north);
\draw[string] (i2) to (hdag.north);
\draw[string] (j2) to (hdag.south);
\node[state,black,scale=0.25,label={[label distance=-0.05cm]330:\tiny{$i$}}](i1) at (-0.75,-1.85){};
\node[state,hflip,black,scale=0.25,label={[label distance=-0.05cm]30:\tiny{$i$}}](j1) at (-0.75,1.85){};
\node[state,black,scale=0.25,label={[label distance=-0.05cm]330:\tiny{$j$}}](j2) at (0.5,-0.25){};
\node[state,hflip,black,scale=0.25,label={[label distance=-0.05cm]30:\tiny{$j$}}](i2) at (0.5,0.25){};
\draw[string] (j1) to (i1);
\draw[string] (i2) to (j2);
\end{pic}
\quad \super{\eqref{eq:spider}}= \quad
\begin{pic}[scale=0.75]
           \node (2) at (0,-2) {}; 
           \node (3) at (0.35,0) {};
           \node[morphism,wedge,scale=0.5] (h1) at (0,-1) {$H$}; 
           \node[blackdot] (r1) at (0.35,-0.5) {};
           \node[blackdot] (r2) at (-0.35,-1.5) {}; 
           \node(a) at (-1,0){};
           \draw[string,out=left,in=270] (r2.center) to (a.center);
           \node[state,black,scale=0.25,label={[label distance=-0.05cm]330:\tiny{$j$}}](a2) at (1,-1.85){};
           \node[state,black,scale=0.25,label={[label distance=-0.05cm]330:\tiny{$i$}}](a3) at (-0.35,-1.85){};
           \draw[string,out=right,in=90] (r1.center) to (a2.center);
           \draw[string] (a3.center) to (r2.center) ;
           \draw[string] (r1.center) to (3.center);
           \draw[string,out=left,in=90] (r1.center) to (h1.north);                  \draw[string,out=right,in=270] (r2.center) to (h1.south);
           \node (21) at (0,2) {}; 
           \node (31) at (0.35,0) {};
           \node[morphism,wedge,scale=0.5,hflip] (h11) at (0,1) {$H$}; 
           \node[blackdot] (r11) at (0.35,0.5) {};
           \node[blackdot] (r21) at (-0.35,1.5) {}; 
           \node(a1) at (-1,0){};
           \draw[string,out=left,in=90] (r21.center) to (a1.center);
           \node[state,hflip,black,scale=0.25,label={[label distance=-0.05cm]30:\tiny{$j$}}](a21) at (1,1.85){};
           \node[state,hflip,black,scale=0.25,label={[label distance=-0.05cm]30:\tiny{$i$}}](a31) at (-0.35,1.85){};
           \draw[string,out=right,in=270] (r11.center) to (a21.center);
           \draw[string] (a31.center) to (r21.center) ;
           \draw[string] (r11.center) to (31.center);
           \draw[string,out=left,in=270] (r11.center) to (h11.south);           \draw[string,out=right,in=90] (r21.center) to (h11.north);
\end{pic}
\end{align*} 
Thus we have that for all $i,j$:
\begin{equation*}
\begin{pic}[scale=0.75]
           \node (2) at (0,-2) {}; 
           \node (3) at (0.35,0) {};
           \node[morphism,wedge,scale=0.5] (h1) at (0,-1) {$H$}; 
           \node[blackdot] (r1) at (0.35,-0.5) {};
           \node[blackdot] (r2) at (-0.35,-1.5) {}; 
           \node(a) at (-1,0){};
           \draw[string,out=left,in=270] (r2.center) to (a.center);
           \node[state,black,scale=0.25,label={[label distance=-0.05cm]330:\tiny{$j$}}](a2) at (1,-1.85){};
           \node[state,black,scale=0.25,label={[label distance=-0.05cm]330:\tiny{$i$}}](a3) at (-0.35,-1.85){};
           \draw[string,out=right,in=90] (r1.center) to (a2.center);
           \draw[string] (a3.center) to (r2.center) ;
           \draw[string] (r1.center) to (3.center);
           \draw[string,out=left,in=90] (r1.center) to (h1.north);                  \draw[string,out=right,in=270] (r2.center) to (h1.south);
           \node (21) at (0,2) {}; 
           \node (31) at (0.35,0) {};
           \node[morphism,wedge,scale=0.5,hflip] (h11) at (0,1) {$H$}; 
           \node[blackdot] (r11) at (0.35,0.5) {};
           \node[blackdot] (r21) at (-0.35,1.5) {}; 
           \node(a1) at (-1,0){};
           \draw[string,out=left,in=90] (r21.center) to (a1.center);
           \node[state,hflip,black,scale=0.25,label={[label distance=-0.05cm]30:\tiny{$j$}}](a21) at (1,1.85){};
           \node[state,hflip,black,scale=0.25,label={[label distance=-0.05cm]30:\tiny{$i$}}](a31) at (-0.35,1.85){};
           \draw[string,out=right,in=270] (r11.center) to (a21.center);
           \draw[string] (a31.center) to (r21.center) ;
           \draw[string] (r11.center) to (31.center);
           \draw[string,out=left,in=270] (r11.center) to (h11.south);           \draw[string,out=right,in=90] (r21.center) to (h11.north);
\end{pic}
\quad = \quad
\begin{pic}[scale=0.75]
\node[state,black,scale=0.25,label={[label distance=-0.05cm]330:\tiny{$i$}}](i1) at (-0.5,-1.85){};
\node[state,hflip,black,scale=0.25,label={[label distance=-0.05cm]30:\tiny{$i$}}](j1) at (-0.5,1.85){};
\node[state,black,scale=0.25,label={[label distance=-0.05cm]330:\tiny{$j$}}](j2) at (0.5,-1.85){};
\node[state,hflip,black,scale=0.25,label={[label distance=-0.05cm]30:\tiny{$j$}}](i2) at (0.5,1.85){};
\draw[string] (j1) to (i1);
\draw[string] (i2) to (j2);
\end{pic}
\qquad \Leftrightarrow \qquad
\begin{aligned}\begin{tikzpicture}[xscale=0.75,yscale=0.75]
           \node (2) at (0,-2) {}; 
           \node (3) at (0.35,0) {};
           \node[morphism,wedge,scale=0.5] (h1) at (0,-1) {$H$}; 
           \node[blackdot] (r1) at (0.35,-0.5) {};
           \node[blackdot] (r2) at (-0.35,-1.5) {}; 
           \node(a) at (-1,0){};
           \draw[string,out=left,in=270] (r2.center) to (a.center);
           \node(a2) at (1,-2){};
           \node(a3) at (-0.35,-2){};
           \draw[string,out=right,in=90] (r1.center) to (a2.center);
           \draw[string] (a3.center) to (r2.center) ;
           \draw[string] (r1.center) to (3.center);
           \draw[string,out=left,in=90] (r1.center) to (h1.north);           \draw[string,out=right,in=270] (r2.center) to (h1.south);

           \node (21) at (0,2) {}; 
           \node (31) at (0.35,0) {};
           \node[morphism,wedge,scale=0.5,hflip] (h11) at (0,1) {$H$}; 

           \node[blackdot] (r11) at (0.35,0.5) {};
           \node[blackdot] (r21) at (-0.35,1.5) {}; 
           \node(a1) at (-1,0){};
           \draw[string,out=left,in=90] (r21.center) to (a1.center);
           \node(a21) at (1,2){};
           \node(a31) at (-0.35,2){};
           \draw[string,out=right,in=270] (r11.center) to (a21.center);
           \draw[string] (a31.center) to (r21.center) ;
           \draw[string] (r11.center) to (31.center);
           \draw[string,out=left,in=270] (r11.center) to (h11.south);           \draw[string,out=right,in=90] (r21.center) to (h11.north);
\end{tikzpicture}\end{aligned}
\quad = \quad
\begin{aligned}\begin{tikzpicture}
\draw[string](0,0) to (0,3);
\draw[string](0.75,0) to (0.75,3);
\end{tikzpicture}\end{aligned}
\end{equation*}
Since $\bra i H \ket j \bra j H^\dag \ket i = \braket{i}{i}\braket{j}{j}\Leftrightarrow\bra j H^\dag \ket i \bra i H \ket j = \braket{i}{i}\braket{j}{j}$ the other part of the right hand side of equation~\eqref{eq:had} follows similarly.
\end{proof}
We now introduce a mathematical object which captures the idea of an indexed family of Hadamards. We introduce another Hilbert space which we will represent with a red wire, equipped with a $\dfrob$. The states copyable by this $\dfrob$ will index the Hadamards in the family.   
\begin{definition}[Controlled Hadamard]\label{def:ch}
Given a  $\dfrob$ $\tinymult$ on a  $d$-dimensional Hilbert space $\mathcal H$ and another $\dfrob$, $\tinymultcl[red]$ on a, possibly different, Hilbert space $\mathcal G$, a linear map  $H:\mathcal G \otimes \mathcal H   \rightarrow \mathcal G$  is a \textit{controlled Hadamard} if the following equations hold.
\begin{align}\label{eq:ch}
\begin{aligned}\begin{tikzpicture}[xscale=0.65]
           \node (2) at (0,-1.5) {}; 
           \node (3) at (0,1.5) {};
           \node[morphism,wedge,scale=0.5] (h1) at (0,-0.65) {$H$}; 
           \node[reddot] (b1) at (-0.75,-1) {};
           \draw[red][string,in=right,out=225] (h1.south west) to (b1.center);
           \draw[red][string] (-0.75,-1.5) to (b1.center);
           \node[morphism,wedge,hflip,scale=0.5] (h11) at (0,0.65) {$H$};            \node[reddot] (b11) at (-0.75,1) {};
           \draw[red][string,in=right,out=135] (h11.north west) to (b11.center);
           \draw[red][string] (-0.75,1.5) to (b11.center);                          \draw[red][string,out=left,in=left] (b11.center) to (b1.center);
           \draw[string] (2.center) to (h1.south);
           \draw[string] (h1.north) to (h11.south);
           \draw[string] (h11.north) to (3.center);
\end{tikzpicture}\end{aligned} 
 = 
\begin{aligned}\begin{tikzpicture}[xscale=0.65]
           \node (2) at (0,-1.5) {}; 
           \node (3) at (0,1.5) {};
           \node[morphism,wedge,hflip,scale=0.5] (h1) at (0,-0.65) {$H$};            \node[reddot] (b1) at (-0.35,-0.2) {};
           \draw[red][string,in=right,out=90] (h1.north west) to (b1.center);
           \draw[red][string,out=90,in=left] (-0.75,-1.5) to (b1.center);
           \node[morphism,wedge,scale=0.5] (h11) at (0,0.65) {$H$}; 
           \node[reddot] (b11) at (-0.35,0.2) {};
           \draw[red][string,in=right,out=270] (h11.south west) to (b11.center);
           \draw[red][string,out=270,in=left] (-0.75,1.5) to (b11.center);            \draw[red][string] (b11.center) to (b1.center);
           \draw[string] (2.center) to (h1.south);
           \draw[string] (h1.north) to (h11.south);
           \draw[string] (h11.north) to (3.center);
\end{tikzpicture}\end{aligned} 
=  d
\begin{aligned}\begin{tikzpicture}[xscale=0.65]
\draw[red][string](0,0) to (0,3);
\draw[string](0.75,0) to (0.75,3);
\end{tikzpicture}\end{aligned}
\qquad & \qquad 
\begin{aligned}\begin{tikzpicture}[scale=0.75]
           \node (2) at (0,-2) {}; 
           \node (3) at (0.35,0) {};
           \node[morphism,wedge,scale=0.5] (h1) at (0,-0.85) {$H$}; 
           \node[reddot] (b1) at (-1.25,-1.5) {};
           \node[blackdot] (r1) at (0.35,-0.35) {};
           \node[blackdot] (r2) at (-0.35,-1.5) {}; 
           \node(a) at (-1,0){};
           \draw[string,out=left,in=270] (r2.center) to (a.center);
           \node(a2) at (1,-2){};
           \node(a3) at (-0.35,-2){};
           \draw[string,out=right,in=90] (r1.center) to (a2.center);
           \draw[string] (a3.center) to (r2.center) ;
           \draw[string] (r1.center) to (3.center);
           \draw[string,out=left,in=90] (r1.center) to (h1.north);                  \draw[string,out=right,in=240] (r2.center) to (h1.south);
           \draw[red][string,in=right,out=225] (h1.south west) to (b1.center);
           \draw[red][string] (-1.25,-2) to (b1.center);                
           \node (21) at (0,2) {}; 
           \node (31) at (0.35,0) {};
           \node[morphism,wedge,hflip,scale=0.5] (h11) at (0,0.85) {$H$};            \node[reddot] (b11) at (-1.25,1.5) {};
           \node[blackdot] (r11) at (0.35,0.35) {};
           \node[blackdot] (r21) at (-0.35,1.5) {}; 
           \node(a1) at (-1,0){};
           \draw[string,out=left,in=90] (r21.center) to (a1.center);
           \node(a21) at (1,2){};
           \node(a31) at (-0.35,2){};
           \draw[string,out=right,in=270] (r11.center) to (a21.center);
           \draw[string] (a31.center) to (r21.center) ;
           \draw[string] (r11.center) to (31.center);
           \draw[string,out=left,in=270] (r11.center) to (h11.south);               \draw[string,out=right,in=90] (r21.center) to (h11.north);
           \draw[red][string,in=right,out=135] (h11.north west) to (b11.center);
           \draw[red][string,out=left,in=left] (b1.center) to (b11.center);
           \draw[red][string] (-1.25,2) to (b11.center);    
\end{tikzpicture}\end{aligned}
 = 
\begin{aligned}\begin{tikzpicture}[scale=0.75]
           \node (2) at (0,0) {}; 
           \node (3) at (0.35,2) {};
           \node[morphism,wedge,scale=0.5] (h1) at (0,1.15) {$H$}; 
           \node[reddot] (b1) at (-1.25,0.5) {};
           \node[blackdot] (r1) at (0.35,1.65) {};
           \node[blackdot] (r2) at (-0.35,0.5) {}; 
           \node(a) at (-1,2){};
           \draw[string,out=left,in=270] (r2.center) to (a.center);
           \node(a2) at (1,0){};
           \node(a3) at (-0.35,0){};
           \draw[string,out=right,in=90] (r1.center) to (a2.center);
           \draw[string] (a3.center) to (r2.center) ;
           \draw[string] (r1.center) to (3.center);
           \draw[string,out=left,in=90] (r1.center) to (h1.north);                  \draw[string,out=right,in=270] (r2.center) to (h1.south);
           \draw[red][string,in=right,out=225] (h1.south west) to (b1.center);
           \draw[red][string] (-1.25,0) to (b1.center);
           \draw[red][string,out=left,in=270] (b1.center) to (-2,2);
           \node (21) at (0,0) {}; 
           \node (31) at (0.35,-2) {};
           \node[morphism,wedge,scale=0.5,hflip] (h11) at (0,-1.15) {$H$};
           \node[reddot] (b11) at (-1.25,-0.5) {};
           \node[blackdot] (r11) at (0.35,-1.65) {};
           \node[blackdot] (r21) at (-0.35,-0.5) {}; 
           \node(a1) at (-1,-2){};
           \draw[string,out=left,in=90] (r21.center) to (a1.center);
           \node(a21) at (1,0){};
           \node(a31) at (-0.35,0){};
           \draw[string,out=right,in=90] (r11.center) to (a21.center);
           \draw[string] (a31.center) to (r21.center) ;
           \draw[string] (r11.center) to (31.center);
           \draw[string,out=left,in=270] (r11.center) to (h11.south);               \draw[string,out=right,in=90] (r21.center) to (h11.north);
           \draw[red][string,in=right,out=135] (h11.north west) to (b11.center);
           \draw[red][string] (-1.25,0) to (b11.center);
           \draw[red][string,out=left,in=90] (b11.center) to (-2,-2);    \end{tikzpicture}\end{aligned}
\quad = \quad
\begin{aligned}\begin{tikzpicture}[xscale=0.75]
\draw[string](0,0) to (0,3);
\draw[string](0.75,0) to (0.75,3);
\draw[red][string](-0.75,0) to (-0.75,3);
\end{tikzpicture}\end{aligned}
  \end{align}  
\end{definition}
Controlled Hadamards are indexed families of Hadamards in the following sense.   \begin{lemma} Given a controlled Hadamard $H$ and  some red state $i$, define $H_i$ as follows:
\begin{equation}
\begin{pic}
\node[morphism,wedge,scale=0.5](h){$H_i$};
\draw[string] (h.south) to +(0,-0.5);
\draw[string] (h.north) to +(0,0.35);
\end{pic}:= 
\begin{pic}
\node[morphism,wedge,scale=0.5](h){$H$};
\node(a) at (h.south west){};
\node[state,red,scale=0.25,label={[label distance=-0.08cm]210:\tiny{$i$}}](r) at (-0.15,-0.3){};
\draw[string] (h.south) to +(0,-0.5);
\draw[string] (h.north) to +(0,0.35);
\draw[red][string] (h.south west) to (r);
\end{pic}
\end{equation}
For all red states $i$, $H_i$ as defined above is a Hadamard.
\end{lemma}
\begin{proof}If we compose equations~\eqref{eq:ch} with $i$ we have:
\begin{align*}
\begin{aligned}\begin{tikzpicture}[xscale=0.65]
           \node[state,red,hflip,scale=0.25,label={[label distance=-0.05cm]30:\tiny{$i$}}](rs1) at (-0.75,1.25){};
           \node[state,red,scale=0.25,label={[label distance=-0.05cm]330:\tiny{$i$}}](rs2) at (-0.75,-1.25){};
           \node (2) at (0,-1.5) {}; 
           \node (3) at (0,1.5) {};
           \node[morphism,wedge,scale=0.5] (h1) at (0,-0.65) {$H$}; 
           \node[reddot] (b1) at (-0.75,-1) {};
           \draw[red][string,in=right,out=225] (h1.south west) to (b1.center);
           \draw[red][string] (rs2) to (b1.center);
           \node[morphism,wedge,hflip,scale=0.5] (h11) at (0,0.65) {$H$};            \node[reddot] (b11) at (-0.75,1) {};
           \draw[red][string,in=right,out=135] (h11.north west) to (b11.center);
           \draw[red][string] (rs1) to (b11.center);                          \draw[red][string,out=left,in=left] (b11.center) to (b1.center);
           \draw[string] (2.center) to (h1.south);
           \draw[string] (h1.north) to (h11.south);
           \draw[string] (h11.north) to (3.center);
\end{tikzpicture}\end{aligned} 
 &= 
\begin{aligned}\begin{tikzpicture}[xscale=0.65]
           \node[state,red,hflip,scale=0.25,label={[label distance=-0.05cm]30:\tiny{$i$}}](rs1) at (-0.75,1.25){};
           \node[state,red,scale=0.25,label={[label distance=-0.05cm]330:\tiny{$i$}}](rs2) at (-0.75,-1.25){};                 
           \node (2) at (0,-1.5) {}; 
           \node (3) at (0,1.5) {};
           \node[morphism,wedge,hflip,scale=0.5] (h1) at (0,-0.65) {$H$}; 
           \node[reddot] (b1) at (-0.35,-0.2) {};
           \draw[red][string,in=right,out=90] (h1.north west) to (b1.center);
           \draw[red][string,out=90,in=left] (rs2) to (b1.center);            \node[morphism,wedge,scale=0.5] (h11) at (0,0.65) {$H$};
           \node[reddot] (b11) at (-0.35,0.2) {};
           \draw[red][string,in=right,out=270] (h11.south west) to (b11.center);
           \draw[red][string,out=270,in=left] (rs1) to (b11.center);            \draw[red][string] (b11.center) to (b1.center);
           \draw[string] (2.center) to (h1.south);
           \draw[string] (h1.north) to (h11.south);
           \draw[string] (h11.north) to (3.center);
\end{tikzpicture}\end{aligned} 
=  d
\begin{aligned}\begin{tikzpicture}[xscale=0.65]
           \node[state,red,hflip,scale=0.25,label={[label distance=-0.05cm]30:\tiny{$i$}}](rs1) at (0,2.75){};
           \node[state,red,scale=0.25,label={[label distance=-0.05cm]330:\tiny{$i$}}](rs2) at (0,0.25){};
           \draw[red][string](0,0.25) to (0,2.75);
           \draw[string](0.75,0) to (0.75,3);
\end{tikzpicture}\end{aligned}
\qquad  \qquad 
\begin{aligned}\begin{tikzpicture}[scale=0.75]
           \node[state,red,hflip,scale=0.25,label={[label distance=-0.05cm]30:\tiny{$i$}}](rs1) at (-1.25,1.75){};
           \node[state,red,scale=0.25,label={[label distance=-0.05cm]330:\tiny{$i$}}](rs2) at (-1.25,-1.75){};
           \node (2) at (0,-2) {}; 
           \node (3) at (0.35,0) {};
           \node[morphism,wedge,scale=0.5] (h1) at (0,-0.85) {$H$}; 
           \node[reddot] (b1) at (-1.25,-1.5) {};
           \node[blackdot] (r1) at (0.35,-0.35) {};
           \node[blackdot] (r2) at (-0.35,-1.5) {}; 
           \node(a) at (-1,0){};
           \draw[string,out=left,in=270] (r2.center) to (a.center);
           \node(a2) at (1,-2){};
           \node(a3) at (-0.35,-2){};
           \draw[string,out=right,in=90] (r1.center) to (a2.center);
           \draw[string] (a3.center) to (r2.center) ;
           \draw[string] (r1.center) to (3.center);
           \draw[string,out=left,in=90] (r1.center) to (h1.north);                  \draw[string,out=right,in=240] (r2.center) to (h1.south);
           \draw[red][string,in=right,out=225] (h1.south west) to (b1.center);
           \draw[red][string] (rs2) to (b1.center);                
           \node (21) at (0,2) {}; 
           \node (31) at (0.35,0) {};
           \node[morphism,wedge,hflip,scale=0.5] (h11) at (0,0.85) {$H$};
           \node[reddot] (b11) at (-1.25,1.5) {};
           \node[blackdot] (r11) at (0.35,0.35) {};
           \node[blackdot] (r21) at (-0.35,1.5) {}; 
           \node(a1) at (-1,0){};
           \draw[string,out=left,in=90] (r21.center) to (a1.center);
           \node(a21) at (1,2){};
           \node(a31) at (-0.35,2){};
           \draw[string,out=right,in=270] (r11.center) to (a21.center);
           \draw[string] (a31.center) to (r21.center) ;
           \draw[string] (r11.center) to (31.center);
           \draw[string,out=left,in=270] (r11.center) to (h11.south);               \draw[string,out=right,in=90] (r21.center) to (h11.north);
           \draw[red][string,in=right,out=135] (h11.north west) to (b11.center);
           \draw[red][string,out=left,in=left] (b1.center) to (b11.center);
           \draw[red][string] (rs1) to (b11.center);    
\end{tikzpicture}\end{aligned}
 = 
\begin{aligned}\begin{tikzpicture}[scale=0.75]
           \node[state,red,hflip,scale=0.25,label={[label distance=-0.05cm]30:\tiny{$i$}}](rs1) at (-2,1.75){};
           \node[state,red,scale=0.25,label={[label distance=-0.05cm]330:\tiny{$i$}}](rs2) at (-2,-1.75){};
           \node (2) at (0,0) {}; 
           \node (3) at (0.35,2) {};
           \node[morphism,wedge,scale=0.5] (h1) at (0,1.15) {$H$}; 
           \node[reddot] (b1) at (-1.25,0.5) {};
           \node[blackdot] (r1) at (0.35,1.65) {};
           \node[blackdot] (r2) at (-0.35,0.5) {}; 
           \node(a) at (-1,2){};
           \draw[string,out=left,in=270] (r2.center) to (a.center);
           \node(a2) at (1,0){};
           \node(a3) at (-0.35,0){};
           \draw[string,out=right,in=90] (r1.center) to (a2.center);
           \draw[string] (a3.center) to (r2.center) ;
           \draw[string] (r1.center) to (3.center);
           \draw[string,out=left,in=90] (r1.center) to (h1.north);                  \draw[string,out=right,in=270] (r2.center) to (h1.south);
           \draw[red][string,in=right,out=225] (h1.south west) to (b1.center);
           \draw[red][string] (-1.25,0) to (b1.center);
           \draw[red][string,out=left,in=270] (b1.center) to (-2,1.75);
           \node (21) at (0,0) {}; 
           \node (31) at (0.35,-2) {};
           \node[morphism,wedge,scale=0.5,hflip] (h11) at (0,-1.15) {$H$};
           \node[reddot] (b11) at (-1.25,-0.5) {};
           \node[blackdot] (r11) at (0.35,-1.65) {};
           \node[blackdot] (r21) at (-0.35,-0.5) {}; 
           \node(a1) at (-1,-2){};
           \draw[string,out=left,in=90] (r21.center) to (a1.center);
           \node(a21) at (1,0){};
           \node(a31) at (-0.35,0){};
           \draw[string,out=right,in=90] (r11.center) to (a21.center);
           \draw[string] (a31.center) to (r21.center) ;
           \draw[string] (r11.center) to (31.center);
           \draw[string,out=left,in=270] (r11.center) to (h11.south);                  \draw[string,out=right,in=90] (r21.center) to (h11.north);
           \draw[red][string,in=right,out=135] (h11.north west) to (b11.center);
           \draw[red][string] (-1.25,0) to (b11.center);
           \draw[red][string,out=left,in=90] (b11.center) to (-2,-1.75);    \end{tikzpicture}\end{aligned}
\quad = \quad
\begin{aligned}\begin{tikzpicture}[xscale=0.75]
\draw[string](0,0) to (0,3);
\draw[string](0.75,0) to (0.75,3);
           \node[state,red,hflip,scale=0.25,label={[label distance=-0.05cm]30:\tiny{$i$}}](rs1) at (-0.75,2.75){};
           \node[state,red,scale=0.25,label={[label distance=-0.05cm]330:\tiny{$i$}}](rs2) at (-0.75,0.25){};
           \draw[red][string](-0.75,0.25) to (-0.75,2.75);
\end{tikzpicture}\end{aligned}\\ \Leftrightarrow
\begin{aligned}\begin{tikzpicture}[xscale=0.65]
           \node (2) at (0,-1.5) {}; 
           \node (3) at (0,1.5) {};
           \node[morphism,wedge,hflip,scale=0.5] (h11) at (0,0.65) {$H$};
           \node[morphism,wedge,scale=0.5] (h1) at (0,-0.65) {$H$};
           \node[state,red,hflip,scale=0.25,label={[label distance=-0.05cm]150:\tiny{$i$}}](rs1) at (-0.28,0.95){};
           \node[state,red,scale=0.25,label={[label distance=-0.05cm]210:\tiny{$i$}}](rs2) at (-0.28,-0.95){};
           \node[state,red,hflip,scale=0.25,label={[label distance=-0.05cm]30:\tiny{$i$}}](rs3) at (-1.3,0.25){};
           \node[state,red,scale=0.25,label={[label distance=-0.05cm]330:\tiny{$i$}}](rs4) at (-1.3,-0.25){}; 
           \draw[red][string,in=90,out=270] (h1.south west) to (rs2);
           \draw[red][string] (h11.north west) to (rs1);
           \draw[red][string] (rs3) to (rs4);
            \draw[string] (2.center) to (h1.south);
           \draw[string] (h1.north) to (h11.south);
           \draw[string] (h11.north) to (3.center);
\end{tikzpicture}\end{aligned} 
 &= 
\begin{aligned}\begin{tikzpicture}[xscale=0.65]
\node[state,red,hflip,scale=0.25,label={[label distance=-0.05cm]150:\tiny{$i$}}](rs1) at (-0.28,-0.35){};
           \node[state,red,scale=0.25,label={[label distance=-0.05cm]210:\tiny{$i$}}](rs2) at (-0.28,0.35){};
           \node[state,red,hflip,scale=0.25,label={[label distance=-0.05cm]30:\tiny{$i$}}](rs3) at (-1.3,0.25){};
           \node[state,red,scale=0.25,label={[label distance=-0.05cm]330:\tiny{$i$}}](rs4) at (-1.3,-0.25){};
          \draw[red][string] (rs3) to (rs4);
           \draw[red][string] (h11.south west) to (rs2);
           \draw[red][string] (h1.north west) to (rs1);
           \node (2) at (0,-1.5) {}; 
           \node (3) at (0,1.5) {};
           \node[morphism,wedge,hflip,scale=0.5] (h1) at (0,-0.65) {$H$};
          \node[morphism,wedge,scale=0.5] (h11) at (0,0.65) {$H$};
           \draw[string] (2.center) to (h1.south);
           \draw[string] (h1.north) to (h11.south);
           \draw[string] (h11.north) to (3.center);
\end{tikzpicture}\end{aligned} 
=  d
\begin{aligned}\begin{tikzpicture}[xscale=0.65]
           \node[state,red,hflip,scale=0.25,label={[label distance=-0.05cm]30:\tiny{$i$}}](rs3) at (0,1.75){};
           \node[state,red,scale=0.25,label={[label distance=-0.05cm]330:\tiny{$i$}}](rs4) at (0,1.25){};
          \draw[red][string] (rs3) to (rs4);
\draw[string](0.75,0) to (0.75,3);
\end{tikzpicture}\end{aligned}
\qquad  \qquad
\begin{aligned}\begin{tikzpicture}[scale=0.75]
           \node[state,red,hflip,scale=0.25,label={[label distance=-0.05cm]150:\tiny{$i$}}](rs1) at (-0.28,0.95){};
           \node[state,red,scale=0.25,label={[label distance=-0.05cm]210:\tiny{$i$}}](rs2) at (-0.28,-0.95){};
           \node[state,red,hflip,scale=0.25,label={[label distance=-0.05cm]30:\tiny{$i$}}](rs3) at (-1.3,0.25){};
           \node[state,red,scale=0.25,label={[label distance=-0.05cm]330:\tiny{$i$}}](rs4) at (-1.3,-0.25){}; 
           \node (2) at (0,-2) {}; 
           \node (3) at (0.35,0) {};
           \node[morphism,wedge,scale=0.5] (h1) at (0,-0.65) {$H$}; 
           \node[blackdot,scale=0.9] (r1) at (0.35,-0.2) {};
           \node[blackdot] (r2) at (-0.35,-1.5) {}; 
           \node(a) at (-1,0){};
           \draw[string,out=left,in=270] (r2.center) to (a.center);
           \node(a2) at (1,-2){};
           \node(a3) at (-0.35,-2){};
           \draw[string,out=right,in=90] (r1.center) to (a2.center);
           \draw[string] (a3.center) to (r2.center) ;
           \draw[string] (r1.center) to (3.center);
           \draw[string,out=left,in=90] (r1.center) to (h1.north);                  \draw[string,out=right,in=240] (r2.center) to (h1.south);
           \node (21) at (0,2) {}; 
           \node (31) at (0.35,0) {};
           \node[morphism,wedge,hflip,scale=0.5] (h11) at (0,0.65) {$H$};
           \node[blackdot,scale=0.9] (r11) at (0.35,0.2) {};
           \node[blackdot] (r21) at (-0.35,1.5) {}; 
           \node(a1) at (-1,0){};
           \draw[string,out=left,in=90] (r21.center) to (a1.center);
           \node(a21) at (1,2){};
           \node(a31) at (-0.35,2){};
           \draw[string,out=right,in=270] (r11.center) to (a21.center);
           \draw[string] (a31.center) to (r21.center) ;
           \draw[string] (r11.center) to (31.center);
           \draw[string,out=left,in=270] (r11.center) to (h11.south);           \draw[string,out=right,in=90] (r21.center) to (h11.north);
                  \draw[red][string,in=90,out=270] (h1.south west) to (rs2);
           \draw[red][string] (h11.north west) to (rs1);
           \draw[red][string] (rs3) to (rs4);
\end{tikzpicture}\end{aligned}
 = 
\begin{aligned}\begin{tikzpicture}[scale=0.75]
           \node (2) at (0,0) {}; 
           \node (3) at (0.35,2) {};
           \node[morphism,wedge,scale=0.5] (h1) at (0,1.25) {$H$}; 
           \node[blackdot] (r1) at (0.35,1.65) {};
           \node[blackdot] (r2) at (-0.35,0.5) {}; 
           \node(a) at (-1,2){};
           \draw[string,out=left,in=270] (r2.center) to (a.center);
           \node(a2) at (1,0){};
           \node(a3) at (-0.35,0){};
           \draw[string,out=right,in=90] (r1.center) to (a2.center);
           \draw[string] (a3.center) to (r2.center) ;
           \draw[string] (r1.center) to (3.center);
           \draw[string,out=left,in=90] (r1.center) to (h1.north);                  \draw[string,out=right,in=270] (r2.center) to (h1.south);
           \node (21) at (0,0) {}; 
           \node (31) at (0.35,-2) {};
           \node[morphism,wedge,scale=0.5,hflip] (h11) at (0,-1.25) {$H$};
           \node[blackdot] (r11) at (0.35,-1.65) {};
           \node[blackdot] (r21) at (-0.35,-0.5) {}; 
           \node(a1) at (-1,-2){};
           \draw[string,out=left,in=90] (r21.center) to (a1.center);
           \node(a21) at (1,0){};
           \node(a31) at (-0.35,0){};
           \draw[string,out=right,in=90] (r11.center) to (a21.center);
           \draw[string] (a31.center) to (r21.center) ;
           \draw[string] (r11.center) to (31.center);
           \draw[string,out=left,in=270] (r11.center) to (h11.south);                  \draw[string,out=right,in=90] (r21.center) to (h11.north);
           \node[state,red,hflip,scale=0.25,label={[label distance=0cm]180:\tiny{$i$}}](rs1) at (-0.28,-0.9){};
           \node[state,red,scale=0.25,label={[label distance=0cm]180:\tiny{$i$}}](rs2) at (-0.28,0.9){};
           \node[state,red,hflip,scale=0.25,label={[label distance=-0.05cm]30:\tiny{$i$}}](rs3) at (-1.3,0.25){};
           \node[state,red,scale=0.25,label={[label distance=-0.05cm]330:\tiny{$i$}}](rs4) at (-1.3,-0.25){};
          \draw[red][string] (rs3) to (rs4);
           \draw[red][string] (h1.south west) to (rs2);
           \draw[red][string] (h11.north west) to (rs1);   \end{tikzpicture}\end{aligned}
\quad = \quad
\begin{aligned}\begin{tikzpicture}[xscale=0.75]
           \node[state,red,hflip,scale=0.25,label={[label distance=-0.05cm]30:\tiny{$i$}}](rs3) at (-0.75,1.75){};
           \node[state,red,scale=0.25,label={[label distance=-0.05cm]330:\tiny{$i$}}](rs4) at (-0.75,1.25){};
          \draw[red][string] (rs3) to (rs4);
\draw[string](0,0) to (0,3);
\draw[string](0.75,0) to (0.75,3);
\end{tikzpicture}\end{aligned}\\ \Leftrightarrow \quad
\begin{aligned}\begin{tikzpicture}%[scale=0.65]
           \node (2) at (0,-1.5) {}; 
           \node (3) at (0,1.5) {};
           \node[morphism,wedge,scale=0.5] (h1) at (0,-0.5) {$H_i$}; 
           \node[morphism,wedge,scale=0.5,hflip] (h11) at (0,0.5) {$H_i$};            \draw[string] (2.center) to (h1.south);
           \draw[string] (h11.north) to (3.center);
           \draw[string] (h11.south) to (h1.north);   
\end{tikzpicture}\end{aligned} 
\quad &= \quad
\begin{aligned}\begin{tikzpicture}%[scale=0.65]
           \node (2) at (0,-1.5) {}; 
           \node (3) at (0,1.5) {};
           \node[morphism,wedge,scale=0.5,hflip] (h1) at (0,-0.5) {$H_i$};            \node[morphism,wedge,scale=0.5] (h11) at (0,0.5) {$H_i$};            \draw[string] (2.center) to (h1.south);
           \draw[string] (h11.north) to (3.center);
           \draw[string] (h11.south) to (h1.north);   
\end{tikzpicture}\end{aligned} 
\quad = \quad
 d
\begin{aligned}\begin{tikzpicture}
\draw[string](0.75,0) to (0.75,3);
\end{tikzpicture}\end{aligned}
\qquad \qquad 
\begin{aligned}\begin{tikzpicture}[scale=0.75]
           \node (2) at (0,-2) {}; 
           \node (3) at (0.35,0) {};
           \node[morphism,wedge,scale=0.5] (h1) at (0,-1) {$H_i$}; 
           \node[blackdot] (r1) at (0.35,-0.5) {};
           \node[blackdot] (r2) at (-0.35,-1.5) {}; 
           \node(a) at (-1,0){};
           \draw[string,out=left,in=270] (r2.center) to (a.center);
           \node(a2) at (1,-2){};
           \node(a3) at (-0.35,-2){};
           \draw[string,out=right,in=90] (r1.center) to (a2.center);
           \draw[string] (a3.center) to (r2.center) ;
           \draw[string] (r1.center) to (3.center);
           \draw[string,out=left,in=90] (r1.center) to (h1.north);           \draw[string,out=right,in=270] (r2.center) to (h1.south);
           \node (21) at (0,2) {}; 
           \node (31) at (0.35,0) {};
           \node[morphism,wedge,scale=0.5,hflip] (h11) at (0,1) {$H_i$}; 
           \node[blackdot] (r11) at (0.35,0.5) {};
           \node[blackdot] (r21) at (-0.35,1.5) {}; 
           \node(a1) at (-1,0){};
           \draw[string,out=left,in=90] (r21.center) to (a1.center);
           \node(a21) at (1,2){};
           \node(a31) at (-0.35,2){};
           \draw[string,out=right,in=270] (r11.center) to (a21.center);
           \draw[string] (a31.center) to (r21.center) ;
           \draw[string] (r11.center) to (31.center);
           \draw[string,out=left,in=270] (r11.center) to (h11.south);           \draw[string,out=right,in=90] (r21.center) to (h11.north);
\end{tikzpicture}\end{aligned}
\quad = \quad
\begin{aligned}\begin{tikzpicture}[xscale=-0.75,yscale=0.75]
           \node (2) at (0,-2) {}; 
           \node (3) at (0.35,0) {};
           \node[morphism,wedge,scale=0.5,hflip] (h1) at (0,-1) {$H_i$}; 
           \node[blackdot] (r1) at (0.35,-0.5) {};
           \node[blackdot] (r2) at (-0.35,-1.5) {}; 
           \node(a) at (-1,0){};
           \draw[string,out=left,in=270] (r2.center) to (a.center);
           \node(a2) at (1,-2){};
           \node(a3) at (-0.35,-2){};
           \draw[string,out=right,in=90] (r1.center) to (a2.center);
           \draw[string] (a3.center) to (r2.center) ;
           \draw[string] (r1.center) to (3.center);
           \draw[string,out=left,in=90] (r1.center) to (h1.north);           \draw[string,out=right,in=270] (r2.center) to (h1.south);
           \node (21) at (0,2) {}; 
           \node (31) at (0.35,0) {};
           \node[morphism,wedge,scale=0.5] (h11) at (0,1) {$H_i$}; 
           \node[blackdot] (r11) at (0.35,0.5) {};
           \node[blackdot] (r21) at (-0.35,1.5) {}; 
           \node(a1) at (-1,0){};
           \draw[string,out=left,in=90] (r21.center) to (a1.center);
           \node(a21) at (1,2){};
           \node(a31) at (-0.35,2){};
           \draw[string,out=right,in=270] (r11.center) to (a21.center);
           \draw[string] (a31.center) to (r21.center) ;
           \draw[string] (r11.center) to (31.center);
           \draw[string,out=left,in=270] (r11.center) to (h11.south);           \draw[string,out=right,in=90] (r21.center) to (h11.north);
\end{tikzpicture}\end{aligned}
\quad = \quad
\begin{aligned}\begin{tikzpicture}
\draw[string](0,0) to (0,3);
\draw[string](0.75,0) to (0.75,3);
\end{tikzpicture}\end{aligned}
\end{align*}  
So by Lemma~\ref{lem:had}, for all $i$, $H_i$ is a Hadamard. 
\end{proof}
Thus given a controlled Hadamard the  number of Hadamards in the family is  equal to the dimension of the red Hilbert space which in practice, for our purposes is often the same as the black Hilbert space. Considering the above lemma and the discussion below Definition~\ref{def:had} we have the following corollary.
\begin{corollary}\label{cor:chadmub}
Given a controlled Hadamard $H$ with black $\dfrob$ $\tinymult$ and red \mbox{$\dfrob$} $\tinymult[reddot]$ define the following bases $\mathcal B^i:=\ket{b^i_j}$ ; for each red state $i$, and black state $j$:
\begin{equation}
\ket{b^i_j}:=\frac{1}{\sqrt{d}}
\begin{pic}
\node (f) [morphism,wedge,  connect se length=0.5cm, width=1cm, connect n] at (0,0) {$H$};
\node(rs)[state,scale=0.375,red,label={[label distance=0.05cm]330:$i$}]  at (f.connect sw) {};
\node[state,scale=0.375,black,label={[label distance=0.05cm]330:$j$}]  at (f.connect se) {};
\draw[red][string](f.south west) to (rs);
\node  at (f.connect n) {};
\end{pic} \end{equation}
Then each basis $\mathcal B^i$ is mutually unbiased to the black basis. 
\end{corollary} 
\begin{definition}[Permutation]\label{def:perm}
A \textit{permutation} with respect to a $\dfrob$ is a comonoid homomorphism (of the comonoid part of the $\dfrob$) which is  unitary.  
\end{definition}
\begin{remark}
In \cat{FHilb} Definition~\ref{def:perm} gives the usual notion of a permutation matrix where the $\dfrob$ is a choice of basis.
\end{remark}
We will later require a tensor diagrammatic characterisation for a permutation $P$ of type $\mathcal H \otimes \mathcal H \rightarrow \mathcal H \otimes \mathcal H$ with respect to the tensor product of the standard black $\dfrob$, $\tinymult$ with itself. The condition that $P$ must be unitary becomes:
\begin{equation}\label{eq:puni}
\begin{pic}[scale=0.75]
\node(i2) at (-0.32,-1.8){};
\node(i3) at (0.32,-1.8){};
\node(p)[morphism,wedge] at (0,-0.6){$P$};
\draw[string] (i2.center) to (p.south west);
\draw[string] (i3.center) to (p.south east);
\node(i2a) at (-0.32,1.8){};
\node(i3a) at (0.32,1.8){};
\node(pa)[morphism,wedge,hflip] at (0,0.6){$P$};
\draw[string] (i2a.center) to (pa.north west);
\draw[string] (i3a.center) to (pa.north east);
\draw[string] (p.north east) to (pa.south east);
\draw[string] (p.north west) to (pa.south west);
\end{pic}
\quad = \quad
\begin{pic}[scale=0.75]
\node(i2) at (-0.32,-1.8){};
\node(i3) at (0.32,-1.8){};
\node(p)[morphism,wedge,hflip] at (0,-0.6){$P$};
\draw[string] (i2.center) to (p.south west);
\draw[string] (i3.center) to (p.south east);
\node(i2a) at (-0.32,1.8){};
\node(i3a) at (0.32,1.8){};
\node(pa)[morphism,wedge] at (0,0.6){$P$};
\draw[string] (i2a.center) to (pa.north west);
\draw[string] (i3a.center) to (pa.north east);
\draw[string] (p.north east) to (pa.south east);
\draw[string] (p.north west) to (pa.south west);
\end{pic}
\quad = \quad
\begin{pic}[scale=0.75]
\node(i1) at (-0.32,-1.8){};
\node(i2) at (0.32,-1.8){};
\node(o1) at (-0.32,1.8){};
\node(o2) at (0.32,1.8){};
\draw[string](i1.center) to (o1.center);
\draw[string](i2.center) to (o2.center);
\end{pic}
\end{equation}
Referring to Definition~\ref{def:cohom} we require the following equations:

\begin{align}
\label{eq:pfunc}
\begin{pic}[scale=0.75]
\node(o2) at (-0.36,3.15){};
\node(o3) at (0.36,3.15){};
\node(o4) at (1,3.15){};
\node(i2) at (-0.32,0){};
\node(i3) at (0.32,0){};
\node(o1) at (-1,3.15){};
\node(p)[morphism,wedge] at (0,1.2){$P$};
\node(b1)[blackdot] at (-0.32,2){};
\node(b2)[blackdot] at (0.32,2){};
\draw[string] (i2.center) to (p.south west);
\draw[string] (i3.center) to (p.south east);
\draw[string] (b2.center) to (p.north east);
\draw[string] (b1.center) to (p.north west);
\draw[string,out=left,in=270] (b1.center) to (o1.center);
\draw[string,out=right,in=270] (b1.center) to (o3.center);
\draw[string,out=left,in=270] (b2.center) to (o2.center);
\draw[string,out=right,in=270] (b2.center) to (o4.center);
\end{pic}\quad = \quad
\begin{pic}[scale=0.75]
\node(o1) at (-0.32,2.4){};
\node(m) at (0.32,2.4){};
\node(i2) at (-0.32,0){};
\node(i3) at (0.32,0){};
\node(p)[morphism,wedge] at (0,1.2){$P$};
\draw[string] (m.center) to (p.north east);
\draw[string] (o1.center) to (p.north west);
\node(o1a) at (1.68,2.4){};
\node(ma) at (2.32,2.4){};
\node(i2a) at (1.68,0){};
\node(i3a) at (2.32,0){};
\node(pa)[morphism,wedge] at (2,1.2){$P$};
\node(b1)[blackdot] at (0.68,-0.25){};
\node(b2)[blackdot] at (1.32,-0.25){};
\draw[string] (ma.center) to (pa.north east);
\draw[string] (o1a.center) to (pa.north west);
\draw[string,out=left,in=270] (b1.center) to (p.south west);
\draw[string,out=left,in=270,looseness=0.6] (b2.center) to (p.south east);
\draw[string,out=right,in=270,looseness=0.6] (b1.center) to (pa.south west);
\draw[string,out=right,in=270] (b2.center) to (pa.south east);
\node(a1) at (0.68,-0.75){};
\node(a2) at (1.32,-0.75){};
\draw[string] (a1.center) to (b1.center);
\draw[string] (a2.center) to (b2.center);
\end{pic}\qquad & \qquad
\begin{pic}[scale=0.75]
\node(o2) at (-0.36,3.15){};
\node(o3) at (0.36,3.15){};
\node(o4) at (1,3.15){};
\node(i2) at (-0.32,0){};
\node(i3) at (0.32,0){};
\node(o1) at (-1,3.15){};
\node(p)[morphism,wedge] at (0,1.2){$P$};
\node(b1)[blackdot] at (-0.32,2){};
\node(b2)[blackdot] at (0.32,2){};
\draw[string] (i2.center) to (p.south west);
\draw[string] (i3.center) to (p.south east);
\draw[string] (b2.center) to (p.north east);
\draw[string] (b1.center) to (p.north west);
\draw[string,out=left,in=270] (b1.center) to (o1.center);
\draw[string,out=right,in=270] (b1.center) to (o3.center);
\draw[string,out=left,in=270] (b2.center) to (o2.center);
\draw[string,out=right,in=270] (b2.center) to (o4.center);
\end{pic}\quad = \quad
\begin{pic}[scale=0.75]
\node(o2) at (-0.36,3.15){};
\node(o3) at (0.36,3.15){};
\node(o4) at (1,3.15){};
\node(i2) at (-0.32,0){};
\node(i3) at (0.32,0){};
\node(o1) at (-1,3.15){};
\node(p)[morphism,wedge] at (0,1.2){$P$};
\node(b1)[blackdot] at (-0.32,2){};
\node(b2)[blackdot] at (0.32,2){};
\draw[string] (i2.center) to (p.south west);
\draw[string] (i3.center) to (p.south east);
\draw[string] (b2.center) to (p.north east);
\draw[string] (b1.center) to (p.north west);
\draw[string,out=left,in=270] (b1.center) to (o1.center);
\draw[string,out=right,in=270] (b1.center) to (o3.center);
\draw[string,out=left,in=270] (b2.center) to (o2.center);
\draw[string,out=right,in=270] (b2.center) to (o4.center);
\end{pic}
\end{align}
This ensures that given black basis states $\ket i$ and $\ket j$ we have  $P(\ket i \otimes \ket j)= \ket m \otimes \ket n$ for some $n,m$ also black basis states.

A well known result which is not difficult to prove and will be useful is that isometric operators on finite dimensional Hilbert spaces are always unitary (\cite{hilbtbook},page 130). Since we will mainly be working with finite dimensional Hilbert spaces we will make use of this to shorten proofs of unitarity. 
  \section{Maximal families of MUBs and \UEBM s}\label{section:MMUB UEBM}
We now move on to the discussion of maximal MUBs. For this section we will use only black wires and  all wires will represent the Hilbert space $\mathcal H \cong\mathbb C^d$ as usual. We consider the black basis states of the $\dagger$-SCFA $\tinymult$ as the computational basis denoted by $\tinystate[black,state,label={[label distance=-0.08cm]330:\tiny{$i$}}], i \in \range{0}$ and use them as an indexing set in the way described in the previous section.
\paragraph{tensor diagrammatic characterisations.}
We now give a tensor diagrammatic characterisation of a maximal MUB, which we will show to be equivalent to Definition~\ref{def:mub}.
We characterise a maximal family of MUBs\ as a linear map $M$ of type  $\mathcal H \otimes \mathcal H \rightarrow \mathcal H$ together with the computational basis $\dfrob$ $\tinymult$. Let the $d+1$ bases of a maximal MUB be denoted $\mathcal B^i, i \in \{*,0,1,...d-1\}$, where the $k$th basis state of $\mathcal B^x$ is denoted $\ket{b_x^k}$. 

We take  the states of the basis $\mathcal B^*$ to be those copyable by $\tinymult$, so $\tinystate[black,state,label={[label distance=-0.08cm]330:\tiny{$i$}}]:=\ket{b_i^*}$. The linear map $M$ encodes the $d^2$ basis states of  the remaining $d$ bases in the following way.
Given black basis states $\tinystate[black,state,label={[label distance=-0.08cm]330:\tiny{$k$}}]$ and $\tinystate[black,state,label={[label distance=-0.08cm]330:\tiny{$x$}}]$:
\begin{equation}
\ket{b_x^k}=
\begin{aligned}
\begin{tikzpicture}
\node (f) [morphism,wedge, connect sw, connect se length=0.5cm, width=1cm, connect n] at (0,0) {$M$};
\node[state,scale=0.375,black,label={[label distance=0.05cm]330:$k$}]  at (f.connect sw) {};
\node[state,scale=0.375,black,label={[label distance=0.05cm]330:$x$}]  at (f.connect se) {};
\node  at (f.connect n) {};
\end{tikzpicture}
\end{aligned}
\end{equation}
\begin{theorem}[Tensor diagrammatic maximal MUBs]
Given a $\dfrob$ $\tinymult$ on a $d$-dimensional Hilbert space $\mathcal H$, a linear map $M$ of type $\mathcal H \otimes \mathcal H \rightarrow \mathcal H$ is a maximal family of MUBs iff $\sqrt{d}M$ is a controlled Hadamard and the following equation holds. 
\begin{equation}\label{eq:MMUB}
\begin{aligned}
\begin{tikzpicture}
\node (m1) [morphism,wedge,hflip] at (-1.25,0.5) {$M$};
\node(m2)[morphism,wedge]  at (-1.25,-0.5) {$M$};
\node  at (m1.connect n) {};
\node at (m2.connect s){};
\draw (m1.south) to (m2.north)[];

\node (m3) [morphism,wedge,hflip,vflip] at (1.25,0.5) {$M$};
\node(m4)[morphism,wedge,vflip]  at (1.25,-0.5) {$M$};
\node  at (m3.connect n) {};
\node at (m4.connect s){};
\draw (m3.south) to (m4.north)[];

\node(b2)[blackdot] at (0.15,1.25){};
\draw[string,out=90,in=left] (m1.north east) to (b2.center);
\draw[string,out=right,in=90] (b2.center) to (m3.north west);
\draw (b2.center) to +(0,0.75);

\node(b4)[blackdot] at (0.15,-1.25){};
\draw[string,out=270,in=left] (m2.south east) to (b4.center);
\draw[string,out=right,in=270] (b4.center) to (m4.south west);
\draw (b4.center) to +(0,-0.75);

\node(b1)[blackdot] at (-0.15,1.65){};
\draw[string,out=90,in=left] (m1.north west) to (b1.center);
\draw[string,out=right,in=90] (b1.center) to (m3.north east);
\draw (b1.center) to +(0,0.35);

\node(b3)[blackdot] at (-0.15,-1.65){};
\draw[string,out=270,in=left] (m2.south west) to (b3.center);
\draw[string,out=right,in=270] (b3.center) to (m4.south east);
\draw (b3.center) to +(0,-0.35);
\end{tikzpicture}
\end{aligned}
=
\frac{1}{d}\left(
\begin{aligned}
\begin{tikzpicture}
\node(1) at (-0.25,2){};
\node(2) at (0.25,2){};
\node(3) at (-0.25,-2){};
\node(4) at (0.25,-2){};
\node(b1)[blackdot] at (-0.25,1){};
\node(b2)[blackdot] at (0.25,1){};
\node(b3)[blackdot] at (-0.25,-1){};
\node(b4)[blackdot] at (0.25,-1){};
\draw (b1.center) to (1);
\draw (b2.center) to (2);
\draw (b3.center) to (3);
\draw (b4.center) to (4);
\end{tikzpicture}
\end{aligned}
\quad -\quad
\begin{aligned}
\begin{tikzpicture}
\node(1) at (-0.25,2){};
\node(2) at (0.25,2){};
\node(3) at (-0.25,-2){};
\node(4) at (0.25,-2){};
\node(b2)[blackdot] at (0.25,1){};
\node(b4)[blackdot] at (0.25,-1){};
\draw (3) to (1);
\draw (b2.center) to (2);
\draw (b4.center) to (4);
\end{tikzpicture}
\end{aligned}\right)
\quad + \quad
\begin{aligned}
\begin{tikzpicture}
\node(1) at (-0.25,2){};
\node(2) at (0.25,2){};
\node(3) at (-0.25,-2){};
\node(4) at (0.25,-2){};
\draw (3) to (1);
\draw (4) to (2);
\end{tikzpicture}
\end{aligned}
\end{equation}
\end{theorem}
\begin{proof}Consider composition by arbitrary black basis states  $\tinystate[black,state,label={[label distance=-0.08cm]330:\tiny{$i$}}]\tinystate[black,state,label={[label distance=-0.08cm]330:\tiny{$j$}}]$ and  effects $\tinyeffect[black,state,hflip,label={[label distance=-0.15cm]330:\tiny{$m$}}]\tinyeffect[black,state,hflip,label={[label distance=-0.15cm]330:\tiny{$n$}}]$ on both sides of the equation~\eqref{eq:MMUB}. 
\begin{align*}
\begin{aligned}
\begin{tikzpicture}[scale=0.5]
\node (m1) [morphism,wedge,hflip,scale=0.5] at (-1.25,0.5) {$M$};
\node(m2)[morphism,wedge,scale=0.5]  at (-1.25,-0.5) {$M$};
\node  at (m1.connect n) {};
\node at (m2.connect s){};
\draw (m1.south) to (m2.north)[];
\node (m3) [morphism,wedge,hflip,vflip,scale=0.5] at (1.25,0.5) {$M$};
\node(m4)[morphism,wedge,vflip,scale=0.5]  at (1.25,-0.5) {$M$};
\node  at (m3.connect n) {};
\node at (m4.connect s){};
\draw (m3.south) to (m4.north)[];
\node(b2)[blackdot,scale=0.75] at (0.25,1.25){};
\draw[string,out=90,in=left] (m1.north east) to (b2.center);
\draw[string,out=right,in=90] (b2.center) to (m3.north west);
\draw (b2.center) to +(0,0.75);
\node(b4)[blackdot,scale=0.75] at (0.25,-1.25){};
\draw[string,out=270,in=left] (m2.south east) to (b4.center);
\draw[string,out=right,in=270] (b4.center) to (m4.south west);
\draw (b4.center) to +(0,-0.75);
\node(b1)[blackdot,scale=0.75] at (-0.25,1.65){};
\draw[string,out=90,in=left] (m1.north west) to (b1.center);
\draw[string,out=right,in=90] (b1.center) to (m3.north east);
\draw (b1.center) to +(0,0.35);
\node(b3)[blackdot,scale=0.75] at (-0.25,-1.65){};
\draw[string,out=270,in=left] (m2.south west) to (b3.center);
\draw[string,out=right,in=270] (b3.center) to (m4.south east);
\draw (b3.center) to +(0,-0.35);
\node(i)[state,black,scale=0.25,label={[label distance=-0.08cm]330:\tiny{$i$}}] at (-0.25,-2){};
\node(j)[state,black,scale=0.25,label={[label distance=-0.08cm]330:\tiny{$j$}}] at (0.25,-2){};
\node(m)[state,black,hflip,scale=0.25,label={[label distance=-0.07cm]30:\tiny{$m$}}] at (-0.25,2){};
\node(n)[state,black,hflip,scale=0.25,label={[label distance=-0.07cm]30:\tiny{$n$}}] at (0.25,2){};
\end{tikzpicture}
\end{aligned}
&=
\frac{1}{d}\left(
\begin{aligned}
\begin{tikzpicture}[scale=0.5]
\node(1) at (-0.25,2){};
\node(2) at (0.25,2){};
\node(3) at (-0.25,-2){};
\node(4) at (0.25,-2){};
\node(b1)[blackdot,scale=0.75] at (-0.25,1){};
\node(b2)[blackdot,scale=0.75] at (0.25,1){};
\node(b3)[blackdot,scale=0.75] at (-0.25,-1){};
\node(b4)[blackdot,scale=0.75] at (0.25,-1){};
\draw (b1.center) to (1.center);
\draw (b2.center) to (2.center);
\draw (b3.center) to (3.center);
\draw (b4.center) to (4.center);
\node(i)[state,black,scale=0.25,label={[label distance=-0.08cm]330:\tiny{$i$}}] at (-0.25,-2){};
\node(j)[state,black,scale=0.25,label={[label distance=-0.08cm]330:\tiny{$j$}}] at (0.25,-2){};
\node(m)[state,black,hflip,scale=0.25,label={[label distance=-0.07cm]30:\tiny{$m$}}] at (-0.25,2){};
\node(n)[state,black,hflip,scale=0.25,label={[label distance=-0.07cm]30:\tiny{$n$}}] at (0.25,2){};
\end{tikzpicture}
\end{aligned}
\quad -\quad
\begin{aligned}
\begin{tikzpicture}[scale=0.5]
\node(1) at (-0.25,2){};
\node(2) at (0.25,2){};
\node(3) at (-0.25,-2){};
\node(4) at (0.25,-2){};
\node(b2)[blackdot,scale=0.75] at (0.25,1){};
\node(b4)[blackdot,scale=0.75] at (0.25,-1){};
\draw (3.center) to (1.center);
\draw (b2.center) to (2.center);
\draw (b4.center) to (4.center);
\node(i)[state,black,scale=0.25,label={[label distance=-0.08cm]330:\tiny{$i$}}] at (-0.25,-2){};
\node(j)[state,black,scale=0.25,label={[label distance=-0.08cm]330:\tiny{$j$}}] at (0.25,-2){};
\node(m)[state,black,hflip,scale=0.25,label={[label distance=-0.07cm]30:\tiny{$m$}}] at (-0.25,2){};
\node(n)[state,black,hflip,scale=0.25,label={[label distance=-0.07cm]30:\tiny{$n$}}] at (0.25,2){};
\end{tikzpicture}
\end{aligned}\right)
\quad + \quad
\begin{aligned}
\begin{tikzpicture}[scale=0.5]
\node(1) at (-0.25,2){};
\node(2) at (0.25,2){};
\node(3) at (-0.25,-2){};
\node(4) at (0.25,-2){};
\draw (3.center) to (1.center);
\draw (4.center) to (2.center);
\node(i)[state,black,scale=0.25,label={[label distance=-0.08cm]330:\tiny{$i$}}] at (-0.25,-2){};
\node(j)[state,black,scale=0.25,label={[label distance=-0.08cm]330:\tiny{$j$}}] at (0.25,-2){};
\node(m)[state,black,hflip,scale=0.25,label={[label distance=-0.07cm]30:\tiny{$m$}}] at (-0.25,2){};
\node(n)[state,black,hflip,scale=0.25,label={[label distance=-0.07cm]30:\tiny{$n$}}] at (0.25,2){};
\end{tikzpicture}
\end{aligned}\\
\Leftrightarrow \quad \left |
\begin{aligned}\begin{tikzpicture}[scale=0.5]
\node (m3) [morphism,wedge,hflip,scale=0.5] at (0,0.5) {$M$};
\node(m4)[morphism,wedge,scale=0.5]  at (0,-0.5) {$M$};
\node  at (m3.connect n) {};
\node at (m4.connect s){};
\draw (m3.south) to (m4.north)[];
\node(i)[state,black,scale=0.25,label={[label distance=-0.08cm]330:\tiny{$i$}}] at (-0.25,-2){};
\node(j)[state,black,scale=0.25,label={[label distance=-0.08cm]330:\tiny{$j$}}] at (0.25,-2){};
\node(m)[state,black,hflip,scale=0.25,label={[label distance=-0.07cm]30:\tiny{$m$}}] at (-0.25,2){};
\node(n)[state,black,hflip,scale=0.25,label={[label distance=-0.07cm]30:\tiny{$n$}}] at (0.25,2){};
\draw[string] (i.center) to (m4.south west);
\draw[string] (j.center) to (m4.south east);
\draw[string] (n.center) to (m3.north east);
\draw[string] (m.center) to (m3.north west);
\end{tikzpicture}
\end{aligned}\right |^2&=\frac{1}{d}(1-\delta_{im})+\delta_{im}\delta_{jn}
\end{align*}
Since  $i,j,m$ and $n$ were chosen arbitrarily this holds for all values of $i,j,m$ and $n$. So our tensor diagrammatic axiom is equivalent to the following; for all $i,j,m,n$:
\begin{equation}
|\braket{b_j^i}{b_n^m}|^2=\frac{1}{d}(1-\delta_{im})+\delta_{im}\delta_{jn}
\end{equation}
 For $i=m$ we have $|\braket{b_j^i}{b_n^i}|^2=\delta_{jn}$, which indicates that for all $i$, $\mathcal B^i$ is an orthonormal basis. For $i \neq m$ we have $|\braket{b_j^i}{b_n^m}|^2=1/d$, in other words $\mathcal B^i$ and $\mathcal B^m$ are mutually unbiased.  

The requirement that $\sqrt{d}M$ is a controlled Hadamard ensures that each basis is mutually unbiased to black basis by Corollary~\ref{cor:chadmub}.
\end{proof} 
We now give a tensor diagrammatic characterisation of unitary error bases which first appeared in the author's masters thesis~\cite{mustothesis} and is equivalent to Definition~\ref{def:ueb} as we show. 
\begin{proposition}[Tensor diagrammatic unitary error bases]\label{prop:grueb}
Given a $d$-dimensional Hilbert space $\mathcal H$ with a $\dfrob$ $\tinymult$, and linear map $U:\mathcal H \otimes \mathcal H \otimes \mathcal H\rightarrow \mathcal H$, define the following family of linear maps $U_{ij}|i,j \in \range{0} $:\begin{equation}
U_{ij}:=
\begin{aligned}
\begin{tikzpicture}
\node (f) [morphism,wedge, connect s, connect se length=0.5cm, width=1cm, connect n,connect sw length = 1cm] at (0,0) {$U$};
\node[state,scale=0.375,black,label={[label distance=0.05cm]330:$i$}]  at (f.connect s) {};
\node[state,scale=0.375,black,label={[label distance=0.05cm]330:$j$}]  at (f.connect se) {};
\node  at (f.connect n) {};
\end{tikzpicture}
\end{aligned}
\end{equation}\label{eq:ueb}
The linear maps $U_{ij}$ are a unitary error basis iff the following equations hold: 
\begin{align}\label{eq:ueb}\begin{aligned}
\begin{tikzpicture}[scale=0.75]
\node (u2) [morphism,wedge, connect s length=0cm, connect se length=0cm, scale=0.75,connect sw length=1cm] at (0,-0.5) {$U$};
\node (u1) [morphism,hflip,wedge, connect ne length=0cm, scale=0.75, connect n length=0cm,connect nw length=1cm] at (0,0.5) {$U$};
\node(o2)  at (0.75,1.8) {};
\node(o3)  at (1.5,1.8) {};
\node(i2)  at (0.75,-1.8) {};
\node(i3)  at (1.5,-1.8) {};
\node(b3)[blackdot, scale=0.75]  at (0.75,-1.4) {};
\node(b4)[blackdot, scale=0.75]  at (1.5,-1.4) {};
\node(b1)[blackdot, scale=0.75]  at (0.75,1.4) {};
\node(b2)[blackdot, scale=0.75]  at (1.5,1.4) {};
\draw[string] (u1.south) to (u2.north);
\draw[string,in=270,out=left] (b3.center) to (u2.connect s);
\draw[string,in=270,out=left,looseness=0.6] (b4.center) to (u2.connect se);
\draw[string] (i2.center) to (b3.center);
\draw[string] (i3.center) to (b4.center);
\draw[string,in=90,out=left] (b1.center) to (u1.connect n);
\draw[string,in=90,out=left,looseness=0.6] (b2.center) to (u1.connect ne);
\draw[string] (o2.center) to (b1.center);
\draw[string] (o3.center) to (b2.center);
\draw[string,out=right,in=right] (b1.center) to (b3.center);
\draw[string,out=right,in=right] (b2.center) to (b4.center);
\end{tikzpicture}
\end{aligned}
\quad = \quad
\begin{aligned}
\begin{tikzpicture}[scale=0.75]
\draw (-0.4,1.8) to (-0.4,-1.8);
\draw (0.75,1.8) to (0.75,-1.8);
\draw (1.5,1.8) to (1.5,-1.8);
\end{tikzpicture}
\end{aligned}
\qquad & \qquad
\begin{aligned}
\begin{tikzpicture}[scale=0.75]
\node (u2) [morphism,wedge, connect s length=1cm, connect se length=1cm, scale=0.75,connect sw length=0cm] at (0,-0.5) {$U$};
\node (u1) [morphism,hflip,wedge, connect ne length=1cm, scale=0.75, connect n length=1cm,connect nw length=0cm] at (0,0.5) {$U$};
\node(o2)  at (-0.75,1.8) {};
\node(i2)  at (-0.75,-1.8) {};
\node(b3)[blackdot, scale=0.75]  at (-0.75,-1.4) {};
\node(b1)[blackdot, scale=0.75]  at (-0.75,1.4) {};
\draw[string] (u1.south) to (u2.north);
\draw[string,in=270,out=right] (b3.center) to (u2.connect sw);
\draw[string,in=90,out=right] (b1.center) to (u1.connect nw);
\draw[string,out=left,in=left] (b1.center) to (b3.center);
\end{tikzpicture}
\end{aligned}
\quad = \quad d
\begin{aligned}
\begin{tikzpicture}[scale=0.75]
\draw (0.75,1.8) to (0.75,-1.8);
\draw (1,1.8) to (1,-1.8);
\end{tikzpicture}
\end{aligned}
\end{align}
\end{proposition}
\begin{proof}
We first show that the left hand equation of equation~\eqref{eq:ueb} is equivalent to each $U_{ij}$ being unitary. 

We compose the left hand equation with the black states $\tinystate[black,state,label={[label distance=-0.08cm]330:\tiny{$i$}}]\tinystate[black,state,label={[label distance=-0.08cm]330:\tiny{$j$}}]$ and  effects $\tinyeffect[black,state,hflip,label={[label distance=-0.15cm]330:\tiny{$m$}}]\tinyeffect[black,state,hflip,label={[label distance=-0.15cm]330:\tiny{$n$}}]$ as follows; for all $i,j,m,n$:
\begin{align*}
\begin{pic}[scale=0.75]
\node (u2) [morphism,wedge, connect s length=0cm, connect se length=0cm, scale=0.75,connect sw length=1cm] at (0,-0.5) {$U$};
\node (u1) [morphism,hflip,wedge, connect ne length=0cm, scale=0.75, connect n length=0cm,connect nw length=1cm] at (0,0.5) {$U$};
\node[state,black,scale=0.25,hflip,label={[label distance=-0.08cm]30:\tiny{$m$}}](o2)  at (0.75,1.65) {};
\node[state,black,scale=0.25,hflip,label={[label distance=-0.08cm]30:\tiny{$n$}}](o3)  at (1.5,1.65) {};
\node[state,black,scale=0.25,label={[label distance=-0.08cm]330:\tiny{$i$}}](i2)  at (0.75,-1.65) {};
\node[state,black,scale=0.25,label={[label distance=-0.08cm]330:\tiny{$j$}}](i3)  at (1.5,-1.65) {};
\node(b3)[blackdot, scale=0.75]  at (0.75,-1.4) {};
\node(b4)[blackdot, scale=0.75]  at (1.5,-1.4) {};
\node(b1)[blackdot, scale=0.75]  at (0.75,1.4) {};
\node(b2)[blackdot, scale=0.75]  at (1.5,1.4) {};
\draw[string] (u1.south) to (u2.north);
\draw[string,in=270,out=left] (b3.center) to (u2.connect s);
\draw[string,in=270,out=left,looseness=0.6] (b4.center) to (u2.connect se);
\draw[string] (i2.center) to (b3.center);
\draw[string] (i3.center) to (b4.center);
\draw[string,in=90,out=left] (b1.center) to (u1.connect n);
\draw[string,in=90,out=left,looseness=0.6] (b2.center) to (u1.connect ne);
\draw[string] (o2.center) to (b1.center);
\draw[string] (o3.center) to (b2.center);
\draw[string,out=right,in=right] (b1.center) to (b3.center);
\draw[string,out=right,in=right] (b2.center) to (b4.center);
\end{pic}=
\begin{pic}[scale=0.75]
\node[state,black,scale=0.25,label={[label distance=-0.08cm]330:\tiny{$i$}}](i) at (0.75,-1.65){};
\node[state,black,scale=0.25,label={[label distance=-0.08cm]330:\tiny{$j$}}](j) at (1.5,-1.65){};
\node[state,black,scale=0.25,hflip,label={[label distance=-0.08cm]30:\tiny{$m$}}](m) at (0.75,1.65){};
\node[state,black,scale=0.25,hflip,label={[label distance=-0.08cm]30:\tiny{$n$}}](n) at (1.5,1.65){};
\draw (-0.4,1.8) to (-0.4,-1.8);
\draw (m) to (i);
\draw (n) to (j);
\end{pic}
& \qquad \Leftrightarrow \qquad
\begin{pic}[scale=0.75]
\node (u2) [morphism,wedge, connect s length=0.25cm, connect se length=0.25cm, scale=0.75,connect sw length=1cm] at (0,-0.5) {$U$};
\node (u1) [morphism,hflip,wedge, connect ne length=0.25cm, scale=0.75, connect n length=0.25cm,connect nw length=1cm] at (0,0.5) {$U$};
\node[state,black,scale=0.25,hflip,label={[label distance=-0.08cm]30:\tiny{$m$}}](m)  at (0.75,0.125) {};
\node[state,black,scale=0.25,hflip,label={[label distance=-0.08cm]30:\tiny{$n$}}](n)  at (1.5,0.125) {};
\node[state,black,scale=0.25,label={[label distance=-0.08cm]330:\tiny{$i$}}](i)  at (0.75,-0.125) {};
\node[state,black,scale=0.25,label={[label distance=-0.08cm]330:\tiny{$j$}}](j)  at (1.5,-0.125) {};
\node(b3)[state,black,scale=0.2,label={[label distance=0cm]270:\tiny{$i$}}]  at (u2.connect s) {};
\node(b4)[state,black,scale=0.2,label={[label distance=-0.08cm]330:\tiny{$j$}}]  at (u2.connect se) {};
\node(b1)[state,black,scale=0.2,hflip,label={[label distance=0cm]90:\tiny{$m$}}]  at (u1.connect n) {};
\node(b2)[state,black,scale=0.2,hflip,label={[label distance=-0.08cm]30:\tiny{$n$}}]  at (u1.connect ne) {};
\draw[string] (u1.south) to (u2.north);
\draw (m) to (i);
\draw (n) to (j);
\end{pic}=
\begin{pic}[scale=0.75]
\node[state,black,scale=0.25,label={[label distance=-0.08cm]330:\tiny{$i$}}](i) at (0.75,-0.125){};
\node[state,black,scale=0.25,label={[label distance=-0.08cm]330:\tiny{$j$}}](j) at (1.5,-0.125){};
\node[state,black,scale=0.25,hflip,label={[label distance=-0.08cm]30:\tiny{$m$}}](m) at (0.75,0.125){};
\node[state,black,scale=0.25,hflip,label={[label distance=-0.08cm]30:\tiny{$n$}}](n) at (1.5,0.125){};
\draw (-0.4,1.8) to (-0.4,-1.8);
\draw (m) to (i);
\draw (n) to (j);
\end{pic}
& \qquad \Leftrightarrow \qquad 
\begin{pic}[scale=0.75]
\node (u2) [morphism,wedge, scale=0.75,connect sw length=1cm] at (0,-0.5) {$U_{ij}$};
\node (u1) [morphism,hflip,wedge, scale=0.75,connect nw length=1cm] at (0,0.5) {$U_{ij}$};
\draw[string] (u1.south) to (u2.north);
\end{pic}=
\begin{pic}[scale=0.75]
\draw (-0.4,1.8) to (-0.4,-1.8);
\end{pic}
\end{align*}
So it is equivalent to all $U_{ij}$ being isometric operators and thus unitary operators. We now show that the right hand side equation of~\eqref{eq:ueb} is equivalent to equation~\eqref{eq:trueb}. We again compose by black states and effects to obtain; for all $i,j,m,n$:
\begin{align*}\begin{aligned}
\begin{tikzpicture}[scale=0.75]
\node (u2) [morphism,wedge, connect s length=0.25cm, connect se length=0.25cm, scale=0.75,connect sw length=0cm] at (0,-0.5) {$U$};
\node (u1) [morphism,hflip,wedge, connect ne length=0.25cm, scale=0.75, connect n length=0.25cm,connect nw length=0cm] at (0,0.5) {$U$};
\node(o2)  at (-0.75,1.8) {};
\node(i2)  at (-0.75,-1.8) {};
\node(b3)[blackdot, scale=0.75]  at (-0.75,-1.4) {};
\node(b1)[blackdot, scale=0.75]  at (-0.75,1.4) {};
\draw[string] (u1.south) to (u2.north);
\draw[string,in=270,out=right] (b3.center) to (u2.connect sw);
\draw[string,in=90,out=right] (b1.center) to (u1.connect nw);
\draw[string,out=left,in=left] (b1.center) to (b3.center);
\node[state,black,scale=0.2,label={[label distance=0cm]270:\tiny{$i$}}](i) at (u2.connect s){};
\node[state,black,scale=0.2,label={[label distance=-0.08cm]330:\tiny{$j$}}](j) at (u2.connect se){};
\node[state,black,scale=0.2,hflip,label={[label distance=0cm]90:\tiny{$m$}}](m) at (u1.connect n){};
\node[state,black,scale=0.2,hflip,label={[label distance=-0.08cm]30:\tiny{$n$}}](n) at (u1.connect ne){};
\end{tikzpicture}
\end{aligned}
\quad = \quad d
\begin{aligned}
\begin{tikzpicture}[scale=0.75]
\node[state,black,scale=0.2,label={[label distance=0cm]270:\tiny{$i$}}](i) at (u2.connect s){};
\node[state,black,scale=0.2,label={[label distance=-0.08cm]330:\tiny{$j$}}](j) at (u2.connect se){};
\node[state,black,scale=0.2,hflip,label={[label distance=0cm]90:\tiny{$m$}}](m) at (u1.connect n){};
\node[state,black,scale=0.2,hflip,label={[label distance=-0.08cm]30:\tiny{$n$}}](n) at (u1.connect ne){};
\draw (i) to (m);
\draw (j) to (n);
\end{tikzpicture}
\end{aligned}
& \qquad
\Leftrightarrow \qquad \tr (U^{\dag}_{ij} \circ U_{mn})=\delta_{im}\delta_{jn}d
\end{align*} 
This completes the proof.   
\end{proof}
 We now introduce notation for a projector which we will require in our description of $\UEBM$s:
\begin{equation}\label{eq:proj}
\begin{aligned}
\begin{tikzpicture}
\node(1) at (0,0.75){};
\node(3) at (0,-0.75){};
\node(b1)[whitedot,scale=1,inner sep=0.25pt] at (0,0){$\projs$};
\draw (b1.center) to (1);
\draw (b1.center) to (3);
\end{tikzpicture}\end{aligned}
\quad := \quad
\begin{aligned}
\begin{tikzpicture}
\node(1) at (0,0.75){};
\node(3) at (0,-0.75){};
\node(b1) at (0,0){};
\node(b1) at (0,0){};
\draw (b1.center) to (1);
\draw (b1.center) to (3);
\end{tikzpicture}\end{aligned}
\quad - \quad
\begin{aligned}
\begin{tikzpicture}
\node(1) at (0,0.75){};
\node(3) at (0,-0.75){};
\node[state,black,scale=0.25,label={[label distance=-0.08cm]330:\tiny{0}}](b1) at (0,0.3){};
\node[state,black,scale=0.25,hflip,label={[label distance=-0.08cm]30:\tiny{0}}](b2) at (0,-0.3){};
\draw (b1.center) to (1);
\draw (b2.center) to (3);
\end{tikzpicture}\end{aligned}
\end{equation}
Note that:
\begin{equation}\label{eq:projzero}
\begin{aligned}
\begin{tikzpicture}
\node[state,black,scale=0.25,label={[label distance=-0.08cm]330:\tiny{0}}](3) at  (0,-0.6){};
\node(1) at (0,0.75){};
\node(b1)[whitedot,scale=1,inner sep=0.25pt] at (0,0){$\projs$};
\draw (b1.center) to (1);
\draw (b1.center) to (3);
\end{tikzpicture}\end{aligned}
\quad = \quad
\begin{aligned}
\begin{tikzpicture}
\node(1) at (0,0.75){};
\node[state,black,scale=0.25,label={[label distance=-0.08cm]330:\tiny{0}}](3) at  (0,-0.6){};
\node(b1) at (0,0){};
\node(b1) at (0,0){};
\draw (b1.center) to (1);
\draw (b1.center) to (3);
\end{tikzpicture}\end{aligned}
\quad - \quad
\begin{aligned}
\begin{tikzpicture}
\node(1) at (0,0.75){};
\node[state,black,scale=0.25,label={[label distance=-0.08cm]330:\tiny{0}}](3) at  (0,-0.6){};
\node[state,black,scale=0.25,label={[label distance=-0.08cm]330:\tiny{0}}](b1) at (0,0.3){};
\node[state,black,scale=0.25,hflip,label={[label distance=-0.08cm]30:\tiny{0}}](b2) at (0,-0.3){};
\draw (b1.center) to (1);
\draw (b2.center) to (3);
\end{tikzpicture}\end{aligned}\quad = \quad 0
\end{equation}
Also note that for $n\neq0$:
\begin{equation}\label{eq:projn}
\begin{aligned}
\begin{tikzpicture}
\node[state,black,scale=0.25,label={[label distance=-0.08cm]330:\tiny{$n$}}](3) at  (0,-0.6){};
\node(1) at (0,0.75){};
\node(b1)[whitedot,scale=1,inner sep=0.25pt] at (0,0){$\projs$};
\draw (b1.center) to (1);
\draw (b1.center) to (3);
\end{tikzpicture}\end{aligned}
\quad = \quad
\begin{aligned}
\begin{tikzpicture}
\node(1) at (0,0.75){};
\node[state,black,scale=0.25,label={[label distance=-0.08cm]330:\tiny{$n$}}](3) at  (0,-0.6){};
\node(b1) at (0,0){};
\node(b1) at (0,0){};
\draw (b1.center) to (1);
\draw (b1.center) to (3);
\end{tikzpicture}\end{aligned}
\quad - \quad
\begin{aligned}
\begin{tikzpicture}
\node(1) at (0,0.75){};
\node[state,black,scale=0.25,label={[label distance=-0.08cm]330:\tiny{$n$}}](3) at  (0,-0.6){};
\node[state,black,scale=0.25,label={[label distance=-0.08cm]330:\tiny{0}}](b1) at (0,0.3){};
\node[state,black,scale=0.25,hflip,label={[label distance=-0.08cm]30:\tiny{0}}](b2) at (0,-0.3){};
\draw (b1.center) to (1);
\draw (b2.center) to (3);
\end{tikzpicture}\end{aligned}\quad = \quad 
\begin{aligned}
\begin{tikzpicture}
\node(1) at (0,0.75){};
\node[state,black,scale=0.25,label={[label distance=-0.08cm]330:\tiny{$n$}}](3) at  (0,-0.6){};
\node(b1) at (0,0){};
\node(b1) at (0,0){};
\draw (b1.center) to (1);
\draw (b1.center) to (3);
\end{tikzpicture}\end{aligned}
\end{equation}
It can easily be shown that there exists a maximum of $d$ commuting unitary operators in dimension $d$~\cite{Bandyopadhyay}. We now define partitioned UEBs ($\UEBM$s). 
\begin{definition}[Partitioned unitary error basis~\cite{Bandyopadhyay}]\label{def:pueb}A \textit{partitioned unitary error basis} ($\UEBM$), is a $d$-dimensional UEB containing the identity, with a partition \mbox{$\{\id[d] \} \sqcup C_* \sqcup C_0 \sqcup...\sqcup C_{d-1}$}, such that each class $C_i, i \in \{*,0,...,d-1 \}$ contains exactly $d-1$ matrices, which together with $\I_d$ form maximal classes of $d$ commuting operators. 
\end{definition}
We now give a tensor diagrammatic characterization of partitioned UEBs. We assume that the partitioned UEB has been ordered such that $U_{00}=\I_d$, $C_*=\{U_{a0}|a \in \{ 1,...,d-1 \}$ and for $i \in \range{0}$, $C_i=\{ U_{ik}|k \in \range{1} \}$. Up to equivalence (see Definition~\ref{eq:uebeq}) any $\UEBM$ can be written in this way. We also choose a computational basis $\dfrob$ $\tinymult$ such that the class $C_*$ is diagonal with respect to it.  
\begin{lemma}\label{lem:guebp}
A unitary error basis  $U$, with $U_{00}$ equal to the identity, is  a  $\UEBM$ iff the following tensor diagrammatic equation holds.
\begin{equation}\label{eq:uebm}
\begin{aligned}
\begin{tikzpicture}[scale=0.75]
\node (f) [morphism,wedge,scale=0.75, connect n ] at (0,1) {$U$};
\node(i2)  at (-0.4125,-2.2) {};
\node(i2a)  at (0.4125,-2.2) {};
\node[whitedot,scale=0.5,inner sep=0.25pt](i3)  at (0.75,-1.7) {$\projs$};
\node[whitedot,scale=0.5,inner sep=0.25pt](i4)  at (1.25,-1.7) {$\projs$};
\node (e) [morphism,wedge, scale=0.75,connect sw length = 1.65cm] at (-0.24,-0.25) {$U$};
\node(b)[blackdot] at (0,-1.1){};
\node(c)[blackdot] at (0,-1.7){};
\node(m1) at (0.75,-0.25) {};
\node(m2) at (1,-0.25) {};
\draw[string,out=90,in=left] (i2.center) to (c.center);
\draw[string,out=90,in=right] (i2a.center) to (c.center);
\draw[string] (b.center) to (c.center);
\draw[string] (e.north) to (f.south west);
\draw[string,in=270,out=left] (b.center) to (e.south);
\draw[string,in=270,out=right,looseness=1.25] (b.center) to (m1.center);
\draw[string,in=270,out=90,looseness=1.25] (m1.center) to (f.south);
\draw[string,in=270,out=90] (i3.center) to (e.south east);
\draw[string,out=90,in=270] (i4.center) to (m2.center);
\draw[string,out=90,in=270] (m2.center) to (f.south east);
\draw[string] (i3.center) to +(0,-0.5);
\draw[string] (i4.center) to +(0,-0.5);
\end{tikzpicture}
\end{aligned}
\quad + \quad
\begin{aligned}
\begin{tikzpicture}[scale=0.75]
\node (f) [morphism,wedge,scale=0.75, connect n ] at (0,1) {$U$};
\node(i2)  at (-0.4125,-1.7) {};
\node(i2a)  at (0.4125,-1.7) {};
\node (e) [morphism,wedge, scale=0.75,connect sw length = 1.65cm,connect s length = 1.65cm] at (-0.24,-0.25) {$U$};
\node(m1) at (0.75,-0.25) {};
\node(m2) at (1,-0.25) {};
\node[state,black,scale=0.25,label={[label distance=-0.125cm]315:\tiny{0}}](b1) at (0.75,-1.4){};
\node[state,black,scale=0.25,hflip,label={[label distance=-0.125cm]45:\tiny{0}}](b2) at (0.75,-2){};
\node[state,black,scale=0.25,label={[label distance=-0.125cm]315:\tiny{0}}](b3) at (1.25,-1.4){};
\node[state,black,scale=0.25,hflip,label={[label distance=-0.125cm]45:\tiny{0}}](b4) at (1.25,-2){};
\draw[string] (e.north) to (f.south west);
\draw[string,in=270,out=90,looseness=1.25] (i2a.center) to (m1.center);
\draw[string,in=270,out=90,looseness=1.25] (m1.center) to (f.south);
\draw[string,in=270,out=90] (b1.center) to (e.south east);
\draw[string,out=90,in=270] (b3.center) to (m2.center);
\draw[string,out=90,in=270] (m2.center) to (f.south east);
\draw[string] (i2a.center) to +(0,-0.5);
\draw[string] (b2.center) to (0.75,-2.2);
\draw[string] (b4.center) to (1.25,-2.2);
\end{tikzpicture}
\end{aligned}
\quad = \quad
\begin{aligned}
\begin{tikzpicture}[scale=0.75]
\node (f) [morphism,wedge,scale=0.75, connect n ] at (0,1) {$U$};
\node(i2)  at (-0.4125,-2.2) {};
\node(i2a)  at (0.4125,-2.2) {};
\node[whitedot,scale=0.5,inner sep=0.25pt](i3)  at (0.75,-1.7) {$\projs$};
\node[whitedot,scale=0.5,inner sep=0.25pt](i4)  at (1.25,-1.7) {$\projs$};
\node (e) [morphism,wedge,scale=0.75,connect sw length = 1.65cm] at (-0.24,-0.25){$U$};
\node(b)[blackdot,scale=0.95] at (0,-1.3){};
\node(c)[blackdot,scale=0.95] at (0,-1.7){};
\node(m1) at (0.75,-0.25) {};
\node(m2) at (1,-0.25) {};
\draw[string,out=90,in=left] (i2.center) to (c.center);
\draw[string,out=90,in=right] (i2a.center) to (c.center);
\draw[string] (b.center) to (c.center);
\draw[string] (e.north) to (f.south west);
\draw[string,in=270,out=left] (b.center) to (e.south);
\draw[string,in=270,out=right,looseness=1.25] (b.center) to (m1.center);
\draw[string,in=270,out=90,looseness=1.25] (m1.center) to (f.south);
\draw[string,in=270,out=90] (i3.center) to (e.south east);
\draw[string,out=90,in=270] (i4.center) to (m2.center);
\draw[string,out=90,in=270] (m2.center) to (f.south east);
\draw[string,out=270,in=90] (i3.center) to +(0.5,-0.5);
\draw[string,out=270,in=90] (i4.center) to +(-0.5,-0.5);
\end{tikzpicture}
\end{aligned}
\quad + \quad
\begin{aligned}
\begin{tikzpicture}[scale=0.75]
\node (f) [morphism,wedge,scale=0.75, connect n ] at (0,1) {$U$};
\node(i2)  at (-0.4125,-1.7) {};
\node(i2a)  at (0.4125,-1.7) {};
\node (e) [morphism,wedge, scale=0.75,connect sw length = 1.65cm] at (-0.24,-0.25){$U$};
\node(m1) at (0.75,-0.25){} ;
\node(m2) at (1,-0.25) {};
\node[state,black,scale=0.25,label={[label distance=-0.125cm]315:\tiny{0}}](b1) at (0.75,-1.4){};
\node[state,black,scale=0.25,hflip,label={[label distance=-0.125cm]45:\tiny{0}}](b2) at (0.75,-2){};
\node[state,black,scale=0.25,label={[label distance=-0.125cm]315:\tiny{0}}](b3) at (1.25,-1.4){};
\node[state,black,scale=0.25,hflip,label={[label distance=-0.125cm]45:\tiny{0}}](b4) at (1.25,-2){};
\draw[string] (i2.center) to (e.south);
\draw[string] (e.north) to (f.south west);
\draw[string,in=270,out=90,looseness=1.25] (i2a.center) to (m1.center);
\draw[string,in=270,out=90,looseness=1.25] (m1.center) to (f.south);
\draw[string,in=270,out=90] (b1.center) to (e.south east);
\draw[string,out=90,in=270] (b3.center) to (m2.center);
\draw[string,out=90,in=270] (m2.center) to (f.south east);
\draw[string,out=270,in=90] (i2.center) to +(0.825,-0.5);
\draw[string,out=270,in=90] (i2a.center) to +(-0.825,-0.5);
\draw[string] (b2.center) to (0.75,-2.2);
\draw[string] (b4.center) to (1.25,-2.2);
\end{tikzpicture}
\end{aligned}
\end{equation}
\end{lemma}
\begin{proof}
To show the equivalence with Definition~\ref{def:pueb} we compose with black states $\tinystate[black,state,label={[label distance=-0.08cm]330:\tiny{$i$}}]\tinystate[black,state,label={[label distance=-0.08cm]330:\tiny{$j$}}]\tinystate[black,state,label={[label distance=-0.08cm]330:\tiny{$m$}}]\tinystate[black,state,label={[label distance=-0.08cm]330:\tiny{$n$}}]$ in the following way; for all $i,j,m,n$:
\begin{align*}
\begin{aligned}
\begin{tikzpicture}[scale=0.75]
\node (f) [morphism,wedge,scale=0.75, connect n ] at (0,1) {$U$};
\node(i2)[state,black,scale=0.2,label={[label distance=-0.1cm]315:\tiny{$i$}}]  at (-0.3,-2.0) {};
\node(i2a)[state,black,scale=0.2,label={[label distance=-0.1cm]315:\tiny{$m$}}]  at (0.3,-2) {};
\node(p1)[state,black,scale=0.2,label={[label distance=-0.1cm]315:\tiny{$j$}}]  at (0.75,-2) {};
\node(p2)[state,black,scale=0.2,label={[label distance=-0.1cm]315:\tiny{$n$}}]  at (1.25,-2) {};
\node (e) [morphism,wedge, scale=0.75,connect sw length = 1.65cm] at (-0.24,-0.25) {$U$};
\node[whitedot,scale=0.5,inner sep=0.25pt](i3)  at (0.75,-1.7) {$\projs$};
\node[whitedot,scale=0.5,inner sep=0.25pt](i4)  at (1.25,-1.7) {$\projs$};
\node(b)[blackdot] at (0,-1.1){};
\node(c)[blackdot] at (0,-1.7){};
\node(m1) at (0.75,-0.25) {};
\node(m2) at (1,-0.25) {};
\draw[string,out=90,in=left] (i2.center) to (c.center);
\draw[string,out=90,in=right] (i2a.center) to (c.center);
\draw[string] (b.center) to (c.center);
\draw[string] (e.north) to (f.south west);
\draw[string,in=270,out=left] (b.center) to (e.south);
\draw[string,in=270,out=right,looseness=1.25] (b.center) to (m1.center);
\draw[string,in=270,out=90,looseness=1.25] (m1.center) to (f.south);
\draw[string,in=270,out=90] (i3.center) to (e.south east);
\draw[string,out=90,in=270] (i4.center) to (m2.center);
\draw[string,out=90,in=270] (m2.center) to (f.south east);
\draw[string] (i3.center) to (p1);
\draw[string] (i4.center) to (p2);
\end{tikzpicture}
\end{aligned}
\quad + \quad
\begin{aligned}
\begin{tikzpicture}[scale=0.75]
\node (f) [morphism,wedge,scale=0.75, connect n ] at (0,1) {$U$};
\node(i2)[state,black,scale=0.2,label={[label distance=-0.1cm]315:\tiny{$i$}}]  at (-0.25,-2.0) {};
\node(i2a)[state,black,scale=0.2,label={[label distance=-0.1cm]315:\tiny{$m$}}]  at (0.3,-2) {};
\node(p1)[state,black,scale=0.2,label={[label distance=-0.1cm]315:\tiny{$j$}}]  at (0.75,-2) {};
\node(p2)[state,black,scale=0.2,label={[label distance=-0.1cm]315:\tiny{$n$}}]  at (1.25,-2) {};
\node (e) [morphism,wedge, scale=0.75,connect sw length = 1.65cm,connect s length = 1.45cm] at (-0.24,-0.25) {$U$};
\node(m1) at (0.75,-0.25) {};
\node(m2) at (1,-0.25) {};
\node[state,black,scale=0.2,label={[label distance=-0.125cm]315:\tiny{0}}](b1) at (0.75,-1.25){};
\node[state,black,scale=0.2,hflip,label={[label distance=-0.125cm]45:\tiny{0}}](b2) at (0.75,-1.85){};
\node[state,black,scale=0.2,label={[label distance=-0.125cm]315:\tiny{0}}](b3) at (1.25,-1.25){};
\node[state,black,scale=0.2,hflip,label={[label distance=-0.125cm]45:\tiny{0}}](b4) at (1.25,-1.85){};
\draw[string] (e.north) to (f.south west);
\draw[string,in=270,out=90,looseness=1.25] (i2a.center) to (m1.center);
\draw[string,in=270,out=90,looseness=1.25] (m1.center) to (f.south);
\draw[string,in=270,out=90] (b1.center) to (e.south east);
\draw[string,out=90,in=270] (b3.center) to (m2.center);
\draw[string,out=90,in=270] (m2.center) to (f.south east);
\draw[string] (b2.center) to (p1);
\draw[string] (b4.center) to (p2);
\end{tikzpicture}
\end{aligned}
\quad &= \quad
\begin{aligned}
\begin{tikzpicture}[scale=0.75]
\node (f) [morphism,wedge,scale=0.75, connect n ] at (0,1) {$U$};
\node(i2)[state,black,scale=0.2,label={[label distance=-0.1cm]315:\tiny{$i$}}]  at (-0.3,-2.0) {};
\node(i2a)[state,black,scale=0.2,label={[label distance=-0.1cm]315:\tiny{$m$}}]  at (0.3,-2) {};
\node(p1)[state,black,scale=0.2,label={[label distance=-0.1cm]315:\tiny{$j$}}]  at (0.75,-2) {};
\node(p2)[state,black,scale=0.2,label={[label distance=-0.1cm]315:\tiny{$n$}}]  at (1.25,-2) {};
\node[whitedot,scale=0.5,inner sep=0.25pt](i3)  at (0.75,-1.5) {$\projs$};
\node[whitedot,scale=0.5,inner sep=0.25pt](i4)  at (1.25,-1.5) {$\projs$};
\node (e) [morphism,wedge,scale=0.75,connect sw length = 1.65cm] at (-0.24,-0.25){$U$};
\node(b)[blackdot,scale=0.95] at (0,-1.3){};
\node(c)[blackdot,scale=0.95] at (0,-1.7){};
\node(m1) at (0.75,-0.25) {};
\node(m2) at (1,-0.25) {};
\draw[string,out=90,in=left] (i2.center) to (c.center);
\draw[string,out=90,in=right] (i2a.center) to (c.center);
\draw[string] (b.center) to (c.center);
\draw[string] (e.north) to (f.south west);
\draw[string,in=270,out=left] (b.center) to (e.south);
\draw[string,in=270,out=right,looseness=1.25] (b.center) to (m1.center);
\draw[string,in=270,out=90,looseness=1.25] (m1.center) to (f.south);
\draw[string,in=270,out=90] (i3.center) to (e.south east);
\draw[string,out=90,in=270] (i4.center) to (m2.center);
\draw[string,out=90,in=270] (m2.center) to (f.south east);
\draw[string,out=270,in=90] (i3.center) to +(0.5,-0.5);
\draw[string,out=270,in=90] (i4.center) to +(-0.5,-0.5);
\end{tikzpicture}
\end{aligned}
\quad + \quad
\begin{aligned}
\begin{tikzpicture}[scale=0.75]
\node (f) [morphism,wedge,scale=0.75, connect n ] at (0,1) {$U$};
\node(i2)[state,black,scale=0.2,label={[label distance=-0.1cm]315:\tiny{$i$}}]  at (-0.25,-2.0) {};
\node(i2a)[state,black,scale=0.2,label={[label distance=-0.1cm]315:\tiny{$m$}}]  at (0.3,-2) {};
\node(p1)[state,black,scale=0.2,label={[label distance=-0.1cm]315:\tiny{$j$}}]  at (0.75,-2) {};
\node(p2)[state,black,scale=0.2,label={[label distance=-0.1cm]315:\tiny{$n$}}]  at (1.25,-2) {};
\node (e) [morphism,wedge, scale=0.75,connect sw length = 1.65cm] at (-0.24,-0.25){$U$};
\node(m1) at (0.75,-0.25){} ;
\node(m2) at (1,-0.25) {};
\node[state,black,scale=0.2,label={[label distance=-0.125cm]315:\tiny{0}}](b1) at (0.75,-1.25){};
\node[state,black,scale=0.2,hflip,label={[label distance=-0.125cm]45:\tiny{0}}](b2) at (0.75,-1.85){};
\node[state,black,scale=0.2,label={[label distance=-0.125cm]315:\tiny{0}}](b3) at (1.25,-1.25){};
\node[state,black,scale=0.2,hflip,label={[label distance=-0.125cm]45:\tiny{0}}](b4) at (1.25,-1.85){};
\draw[string,out=90,in=270] (i2a.center) to (e.south);
\draw[string] (e.north) to (f.south west);
\draw[string,in=270,out=90,looseness=1.25] (i2.center) to (m1.center);
\draw[string,in=270,out=90,looseness=1.25] (m1.center) to (f.south);
\draw[string,in=270,out=90] (b1.center) to (e.south east);
\draw[string,out=90,in=270] (b3.center) to (m2.center);
\draw[string,out=90,in=270] (m2.center) to (f.south east);
\draw[string] (b2.center) to (p1);
\draw[string] (b4.center) to (p2);
\end{tikzpicture}
\end{aligned}
\\ \super{\eqref{eq:spider}}\Leftrightarrow \quad
\begin{aligned}
\begin{tikzpicture}[scale=0.75]
\node (u1) [morphism,wedge,scale=0.75, connect n,connect s length = 0.25cm] at (0,1) {$U$};
\node(i2)[state,black,scale=0.2,label={[label distance=0cm]270:\tiny{$i$}}]  at (u2.connect s) {};
\node(i3)[state,black,scale=0.2,label={[label distance=-0.13cm]240:\tiny{$i$}}]  at (u1.connect s) {};
\node(p2)[state,black,scale=0.2,label={[label distance=-0.1cm]315:\tiny{$m$}}]  at (0.35,-2) {};
\node (u2) [morphism,wedge, scale=0.75,connect sw length = 1.65cm,connect s length = 0.25cm] at (-0.24,-0.25) {$U$};
\node[state,black,scale=0.2,label={[label distance=-0.125cm]315:\tiny{$j$}}](b1) at (0.2,-1.35){};
\node[state,black,scale=0.2,label={[label distance=-0.125cm]315:\tiny{$n$}}](b3) at (0.75,-0.25){};
\node[state,black,scale=0.2,hflip,label={[label distance=-0.125cm]45:\tiny{$i$}}](b4) at (0.35,-1.85){};
\node[whitedot,scale=0.5,inner sep=0.25pt](s1)  at (0.2,-1) {$\projs$};
\node[whitedot,scale=0.5,inner sep=0.25pt](s2)  at (0.5,0.15) {$\projs$};
\draw[string] (u2.north) to (u1.south west);
\draw[string] (b4.center) to (p2);
\draw[string,out=270,in=90] (u2.south east) to (s1.center);
\draw[string,out=270,in=90] (u1.south east) to (s2.center);
\draw[string,out=270,in=90] (s1.center) to (b1.center);
\draw[string,out=270,in=90] (s2.center) to (b3.center);
\end{tikzpicture}
\end{aligned}\quad +\quad
\begin{aligned}
\begin{tikzpicture}[scale=0.75]
\node (u1) [morphism,wedge,scale=0.75, connect n,connect s length = 0.25cm,connect se length = 0.25cm] at (0,1) {$U$};
\node(i2)[state,black,scale=0.2,label={[label distance=0.05cm]270:\tiny{$i$}}]  at (u2.connect s) {};
\node(i3)[state,black,scale=0.2,label={[label distance=0.035cm]270:\tiny{$m$}}]  at (u1.connect s) {};
\node(p1)[state,black,scale=0.2,label={[label distance=-0.1cm]315:\tiny{$j$}}]  at (0.75,-2) {};
\node(p2)[state,black,scale=0.2,label={[label distance=-0.1cm]315:\tiny{$n$}}]  at (1.25,-2) {};
\node (u2) [morphism,wedge, scale=0.75,connect sw length = 1.65cm,connect s length = 0.25cm,connect se length = 0.25cm] at (-0.24,-0.25) {$U$};
\node[state,black,scale=0.2,label={[label distance=-0.125cm]315:\tiny{0}}](b1) at (u2.connect se){};
\node[state,black,scale=0.2,hflip,label={[label distance=-0.125cm]45:\tiny{0}}](b2) at (0.75,-1.85){};
\node[state,black,scale=0.2,label={[label distance=-0.125cm]315:\tiny{0}}](b3) at (u1.connect se){};
\node[state,black,scale=0.2,hflip,label={[label distance=-0.125cm]45:\tiny{0}}](b4) at (1.25,-1.85){};
\draw[string] (u2.north) to (u1.south west);
\draw[string] (b2.center) to (p1);
\draw[string] (b4.center) to (p2);
\end{tikzpicture}
\end{aligned}\quad &= \quad
\begin{aligned}
\begin{tikzpicture}[scale=0.75]
\node (u1) [morphism,wedge,scale=0.75, connect n,connect s length = 0.25cm] at (0,1) {$U$};
\node(i2)[state,black,scale=0.2,label={[label distance=0cm]270:\tiny{$i$}}]  at (u2.connect s) {};
\node(i3)[state,black,scale=0.2,label={[label distance=-0.13cm]240:\tiny{$i$}}]  at (u1.connect s) {};
\node(p2)[state,black,scale=0.2,label={[label distance=-0.1cm]315:\tiny{$m$}}]  at (0.35,-2) {};
\node (u2) [morphism,wedge, scale=0.75,connect sw length = 1.65cm,connect s length = 0.25cm] at (-0.24,-0.25) {$U$};
\node[state,black,scale=0.2,label={[label distance=-0.125cm]315:\tiny{$n$}}](b1) at (0.2,-1.35){};
\node[state,black,scale=0.2,label={[label distance=-0.125cm]315:\tiny{$j$}}](b3) at (0.75,-0.25){};
\node[state,black,scale=0.2,hflip,label={[label distance=-0.125cm]45:\tiny{$i$}}](b4) at (0.35,-1.85){};
\node[whitedot,scale=0.5,inner sep=0.25pt](s1)  at (0.2,-1) {$\projs$};
\node[whitedot,scale=0.5,inner sep=0.25pt](s2)  at (0.5,0.15) {$\projs$};
\draw[string] (u2.north) to (u1.south west);
\draw[string] (b4.center) to (p2);
\draw[string,out=270,in=90] (u2.south east) to (s1.center);
\draw[string,out=270,in=90] (u1.south east) to (s2.center);
\draw[string,out=270,in=90] (s1.center) to (b1.center);
\draw[string,out=270,in=90] (s2.center) to (b3.center);
\end{tikzpicture}
\end{aligned}\quad +\quad
\begin{aligned}
\begin{tikzpicture}[scale=0.75]
\node (u1) [morphism,wedge,scale=0.75, connect n,connect s length = 0.25cm,connect se length = 0.25cm] at (0,1) {$U$};
\node(i2)[state,black,scale=0.2,label={[label distance=0.05cm]270:\tiny{$m$}}]  at (u2.connect s) {};
\node(i3)[state,black,scale=0.2,label={[label distance=-0.13cm]240:\tiny{$i$}}]  at (u1.connect s) {};
\node(p1)[state,black,scale=0.2,label={[label distance=-0.1cm]315:\tiny{$j$}}]  at (0.75,-2) {};
\node(p2)[state,black,scale=0.2,label={[label distance=-0.1cm]315:\tiny{$n$}}]  at (1.25,-2) {};
\node (u2) [morphism,wedge, scale=0.75,connect sw length = 1.65cm,connect s length = 0.25cm,connect se length = 0.25cm] at (-0.24,-0.25) {$U$};
\node[state,black,scale=0.2,label={[label distance=-0.125cm]315:\tiny{0}}](b1) at (u2.connect se){};
\node[state,black,scale=0.2,hflip,label={[label distance=-0.125cm]45:\tiny{0}}](b2) at (0.75,-1.85){};
\node[state,black,scale=0.2,label={[label distance=-0.125cm]315:\tiny{0}}](b3) at (u1.connect se){};
\node[state,black,scale=0.2,hflip,label={[label distance=-0.125cm]45:\tiny{0}}](b4) at (1.25,-1.85){};
\draw[string] (u2.north) to (u1.south west);
\draw[string] (b2.center) to (p1);
\draw[string] (b4.center) to (p2);
\end{tikzpicture}
\end{aligned}
\end{align*}
If $j=0$ and $n \neq 0$ we obtain $0+0=0+0$, the first zero in each summand being due to equation~\eqref{eq:projzero}, the second summands are multiplied by $\braket{0}{n}=0$.  Similarly if $j \neq 0$ and $n=0$ we obtain $0+0=0+0$. So no condition is imposed by equation~\eqref{eq:uebm} unless either, case one $j=n=0$ or case two $j\neq 0$ and $n\neq0$. 
\paragraph{Case one.}For $j=n=0$, again by equation~\eqref{eq:projzero} we obtain for all $i,m$: \begin{align*}
0+\braket{0}{0}^2U_{m0}U_{i0}&=0+\braket{0}{0}^2U_{i0}U_{m0}\\
\Leftrightarrow\quad U_{m0}U_{i0}&= U_{i0}U_{m0}
\end{align*}
This shows that the class $C_*$ together with the identity, $U_{00}$ form a maximal class of commuting operators. 
\paragraph{Case two.} For $j\neq 0$ and $n\neq0$ by equation~\eqref{eq:projn} we obtain for all $i,m$:
\begin{align*}
\braket{i}{m}U_{in}U_{ij}+\braket{0}{j}\braket{0}{n}U_{m0}U_{i0}&=\braket{i}{m}U_{ij}U_{in}+\braket{0}{j}\braket{0}{n}U_{i0}U_{m0}\\
\Rightarrow\quad \delta_{im}U_{in}U_{ij}&= \delta_{im}U_{ij}U_{in}
\end{align*}
For $i\neq m$ this gives $0=0$, for $i=m$ we have that for each $i$ the $d-1$ operators $U_{ik}$ with $k\in \range{1}$ pairwise commute. Thus the classes $C_i$ with $i\in \range{0}$ together with the identity $U_{00}$ form maximal classes of commuting operators. This completes the proof.
\end{proof}
\paragraph{Main results.}   We first present the following theorem due to Bandyopadhyay et al~\cite{Bandyopadhyay}. 
\begin{theorem}\label{thm:uebp}
Given $U$, a $\UEBM$, the common eigenbases $\ket{b^i_k}$ for each class \mbox{$C_i|i \in \{0,...,d-1\}$} form a maximal family of MUBs. 
\end{theorem}  
As a notational point we introduce $\theta$ to represent the map from partitioned UEBs to maximal families of MUBs given by Theorem~\ref{thm:uebp}. In their paper Banyopadhyay et al also give a construction which takes a maximal MUB in dimension $d$ and the Fourier matrix for the cyclic group $\mathbb Z_d$ and gives a $\UEBM$. The following construction generalises theirs.   

In the following construction we will use a Hadamard $G$ and a controlled Hadamard $H$ (see Definition~\ref{def:ch}). Every Hadamard is equivalent to a Hadamard with ones along the first column and first row~\cite{mubs}. We assume that each Hadamard in the controlled family as well as $G$ are in this form. This gives us the following axioms.
\begin{align}\label{eq:had1}
\begin{aligned}\begin{tikzpicture}%[scale=0.65]
           \node[state,scale=0.375,black,label={[label distance=0.05cm]330:$0$}] (2) at (0,0.25) {}; 
           \node (3) at (0,2) {};
           \node (b1) at (-0.5,0) {};
           \node[morphism,wedge,scale=0.5] (h1) at (0,1) {$H$}; 
           \draw[string] (2.center) to (h1.south);
           \draw[string] (h1.north) to (3.center);
           \draw[string] (h1.south west) to (b1.center);          
\end{tikzpicture}\end{aligned}
 = 
\begin{aligned}\begin{tikzpicture}%[scale=0.65]
           \node[blackdot] (2) at (0,1.25) {}; 
           \node (3) at (0,2) {};
           \node (b1) at (-0.5,0) {};
           \node[blackdot] (h1) at (-0.5,0.75) {}; 
           \draw[string] (2.center) to (3.center);
           \draw[string] (h1.center) to (b1.center);          
\end{tikzpicture}\end{aligned} 
\quad & \qquad 
\begin{aligned}\begin{tikzpicture}%[scale=0.65]
           \node[hflip,state,scale=0.375,black,label={[label distance=0.05cm]330:$0$}] (2) at (0,1.75) {}; 
           \node (3) at (0,0) {};
           \node (b1) at (-0.5,0) {};
           \node[morphism,wedge,scale=0.5] (h1) at (0,1) {$H$}; 
           \draw[string] (2.center) to (h1.north);
           \draw[string] (h1.south) to (3.center);
           \draw[string] (h1.south west) to (b1.center);          
\end{tikzpicture}\end{aligned} 
=
\begin{aligned}\begin{tikzpicture}%[scale=0.65]
           \node[blackdot] (2) at (0,0.75) {}; 
           \node (3) at (0,0) {};
           \node  at (0,1.95) {};
           \node (b1) at (-0.5,0) {};
           \node[blackdot] (h1) at (-0.5,0.75) {}; 
           \draw[string] (2.center) to (3.center);
           \draw[string] (h1.center) to (b1.center);          
\end{tikzpicture}\end{aligned} 
\quad & \qquad
\begin{aligned}\begin{tikzpicture}%[scale=0.65]
           \node[state,scale=0.375,black,label={[label distance=0.05cm]330:$0$}] (2) at (0,0.25) {}; 
           \node (3) at (0,2) {};
           \node (b1) at (-0.5,0) {};
           \node[morphism,wedge,scale=0.5](h1) at (0,1) {$G$}; 
           \draw[string] (2.center) to (h1.south);
           \draw[string] (h1.north) to (3.center);         
\end{tikzpicture}\end{aligned}
 = 
\begin{aligned}\begin{tikzpicture}%[scale=0.65]
           \node[blackdot] (2) at (0,1.25) {}; 
           \node (3) at (0,2) {};
           \node (b1) at (0,-0.15) {}; 
           \draw[string] (2.center) to (3.center);          
\end{tikzpicture}\end{aligned} 
\quad & \qquad 
\begin{aligned}\begin{tikzpicture}%[scale=0.65]
           \node[hflip,state,scale=0.375,black,label={[label distance=0.05cm]330:$0$}] (2) at (0,1.75) {}; 
           \node (3) at (0,0) {};
           \node (b1) at (-0.5,0) {};
           \node[morphism,wedge,scale=0.5] (h1) at (0,1) {$G$};            \draw[string] (h1.south) to (3.center);
           \draw[string] (h1.north) to (2.center);   
\end{tikzpicture}\end{aligned} 
=
\begin{aligned}\begin{tikzpicture}%[scale=0.65]
           \node[blackdot] (2) at (0,0.75) {}; 
           \node (3) at (0,0) {};
           \node  at (0,1.95) {};
           \node (b1) at (0,0) {}; 
           \draw[string] (2.center) to (3.center);          
\end{tikzpicture}\end{aligned} 
\end{align}
We now provide the main result of this section, a converse to Theorem~\ref{thm:uebp} taking a maximal MUB and a controlled Hadamard to construct  a $\UEBM$. We will later show in Theorem~\ref{thm:compid} that if we start with a $\UEBM$ and then obtain a maximal MUB by taking the eigenbases (see Theorem~\ref{thm:uebp})
and then perform the following construction we recover the $\UEBM$ we started with.\begin{theorem}\label{thm:rev}
Given a maximal MUB on a $d$ dimensional Hilbert space, a controlled Hadamard $H$ and an additional Hadamard $G$ the following map $\phi_H$ gives a $\UEBM$:
\begin{equation}
\begin{aligned}

\end{aligned}
\end{equation}
\end{theorem}
\begin{proof} We show that $\phi_H(M)$ is a UEB. First the left hand equation of~\eqref{eq:ueb}.
\begin{align*}\begin{aligned}
%
\end{aligned}
\end{align*}
Now the right hand  equation of~\eqref{def:ueb}:

\begin{align*}
\begin{aligned}
%
\end{aligned}
\end{align*}
We now show that equation~\eqref{eq:uebm} holds:
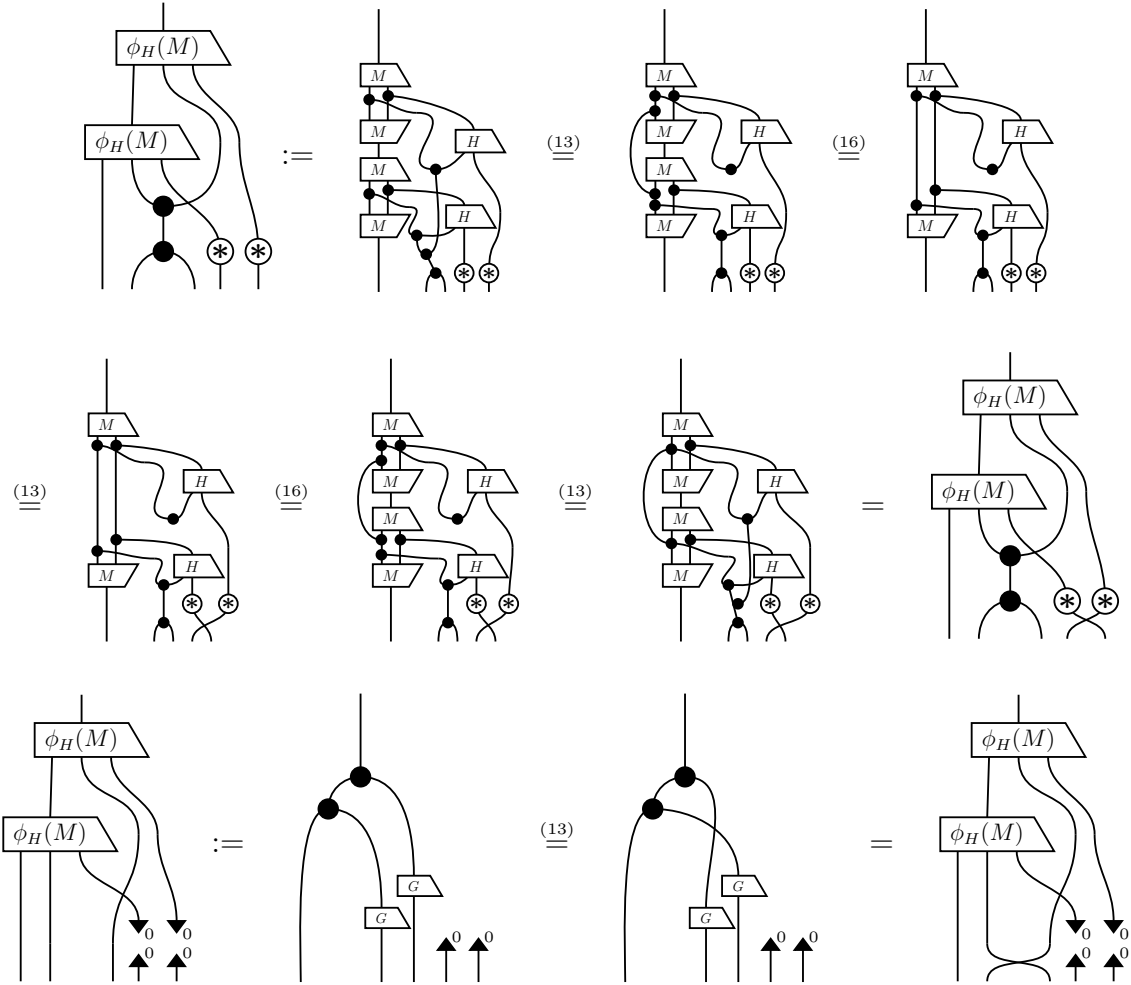
\begin{equation*}
\begin{aligned}
%
\end{aligned}
\end{equation*}
\end{proof}
Let $\theta$ be the map that takes a partitioned UEB and gives the corresponding maximal family of MUBs according to Theorem~\ref{thm:uebp}. We now investigate  the map $\theta$ and the infinite family of maps $\phi_{H}$ each taking a maximal family of MUBs and giving a partitioned UEB, given by Theorem~\ref{thm:rev} above. Given a controlled Hadamard, $H$ and a maximal family of MUBs we now consider the effect of the composition $\theta \circ \phi_{H^i}  $ on the maximal family of MUBs:
\begin{theorem}\label{thm:compid}
Given a maximal family of MUBs $M$ and a contolled Hadamard ${H^i}$: \begin{equation}\theta \circ \phi_{H}   (M)=M\end{equation} 
\end{theorem}
\begin{proof}
\begin{equation*}
\phi_{H}   (M):=\quad
\begin{aligned}
\begin{tikzpicture}
\node (m1) [morphism,wedge] at (0,0.75) {$M$};
\node(m2)[morphism,wedge,hflip]  at (0,-0.75) {$M$};
\draw (m1.south east) to (m2.north east)[string];
\draw (m1.south west) to (m2.north west)[string];
\node(b1)[blackdot] at (-0.25,0.1){};
\node(b2)[blackdot] at (0.25,0.2){};
\node(h)[morphism,wedge] at (1.75,-1){$H$};
\node(b4)[blackdot] at (1.25,-1.5){};
\node(m) at (1,-0.25){};
\node(s)[whitedot,inner sep=0.25pt] at (1.75,-2){$\boldsymbol{*}$};
\draw[string,out=left,in=right,looseness=1] (m.center) to (b1.center);
\draw[string,out=left,in=right,looseness=2] (b4.center) to (m.center);
\draw[string,out=right,in=south west] (b4.center) to (h.south west);
\draw[string,out=90,in=right] (h.north) to (b2.center);
\node(i1) at (0,-2.5){};
\node(i2) at (1.25,-2.5){};
\node(i3) at (1.75,-2.5){};
\node(o) at (0,2){};
\draw[string] (i1.center) to (m2.south);
\draw[string] (i2.center) to (b4.center);
\draw[string] (i3.center) to (h.south);
\draw[string] (o.center) to (m1.north);
\end{tikzpicture}
\end{aligned}
\quad + \quad
\begin{aligned}
\begin{tikzpicture}
\node(b2)[blackdot] at (0.425,0.2){};
\node(h)[morphism,wedge] at (1.25,-1.25){$G$};
\node(m) at (1,-0.25){};
\node(s)[state,black,hflip,scale=0.25,label={[label distance=-0.08cm]30:\tiny{0}}] at (1.75,-2){};
\draw[string,out=90,in=right] (h.north) to (b2.center);
\node(i1) at (0,-2.5){};
\node(i2) at (1.25,-2.5){};
\node(i3) at (1.75,-2.5){};
\node(o) at (0.425,2){};
\draw[string,out=90,in=left,looseness=0.5] (i1.center) to (b2.center);
\draw[string] (i2.center) to (h.south);
\draw[string] (i3.center) to (s.center);
\draw[string] (o.center) to (b2.center);
\end{tikzpicture}
\end{aligned}
\end{equation*}
By design we have a partition $\{ U_{00}\} \sqcup C_* \sqcup C_0 \sqcup ... \sqcup C_{d-1}$, where $C_*=\{U_{i0} |i \in \range{1} \}$ and for $k \in \range{0}, C_a=\{ U_{aj}|j \in \range{1} \}$. Clearly the eigenbasis of $C_*$ is the black basis. Let $\ket{b^i_k}$ be the $k$th state of the $i$th basis of $M$. We claim that $\ket{b^i_k}$ is the $k$th eigenstate of $C_i$. To see this consider the following composite linear map.
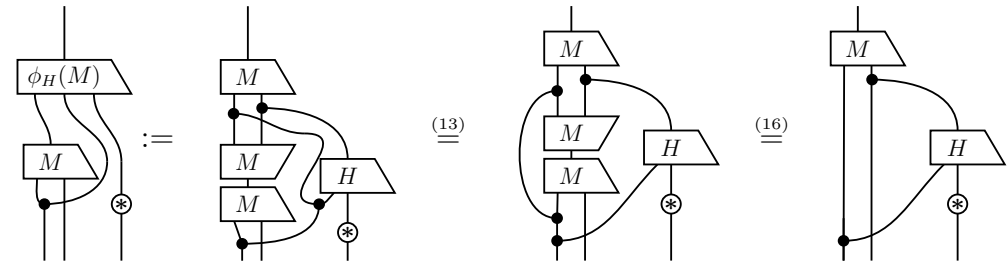
\begin{equation*}
\begin{aligned}\begin{tikzpicture}[scale=0.75]
\node[morphism,wedge,scale=0.75](m) at (0,1){$\phi_H   (M)$};
\node[morphism,wedge,scale=0.75](e) at (-0.24,-0.5){$M$};
\node(i1) at (-0.35,-2.25){};
\node(i2) at (0,-2.25){};
\node(o) at(0,2.25){};
\node[blackdot,scale=0.5](b) at (-0.35,-1.25){};
\node(g) at (0.75,-0.5){};
\node(h) at (1,-0.5){};
\node(i3) at (1,-2.25){};
\node(*)[whitedot,inner sep=0.25pt,scale=0.75] at (1,-1.25){$\boldsymbol{*}$};
\draw[string,out=90,in=270] (i3.center) to (h.center);
\draw[string,out=90,in=270] (h.center) to (m.south east);
\draw[string] (m.north) to (o.center);
\draw[string,out=90,in=270] (e.north) to (m.south west);
\draw[string,out=left,in=270] (b.center) to (e.south west);
\draw[string,out=90,in=270] (i1.center) to (b.center);
\draw[string,out=90,in=270] (i2.center) to (e.south east);
\draw[string,out=right,in=270] (b.center) to (g.center);
\draw[string,out=90,in=270] (g.center) to (m.south);
\end{tikzpicture}\end{aligned}:= \quad
\begin{aligned}
\begin{tikzpicture}[scale=0.75]
\node (m1) [morphism,wedge,scale=0.75] at (0,0.75) {$M$};
\node(m2)[morphism,wedge,hflip,scale=0.75]  at (0,-0.75) {$M$};
\node(m3)[morphism,wedge,scale=0.75]  at (0,-1.5) {$M$};
\node[blackdot,scale=0.5](b) at (-0.1,-2.2){};
\draw (m1.south east) to (m2.north east)[string];
\draw (m1.south west) to (m2.north west)[string];
\node(b1)[blackdot,scale=0.5] at (-0.25,0.1){};
\node(b2)[blackdot,scale=0.5] at (0.25,0.2){};
\node(h)[morphism,wedge,scale=0.75] at (1.75,-1){$H$};
\node(b4)[blackdot,scale=0.5] at (1.25,-1.5){};
\node(m) at (1,-0.25){};
\node(s)[whitedot,inner sep=0.25pt,scale=0.75] at (1.75,-2){$\boldsymbol{*}$};
\draw[string,out=left,in=right,looseness=1] (m.center) to (b1.center);
\draw[string,out=left,in=right,looseness=2] (b4.center) to (m.center);
\draw[string,out=right,in=south west] (b4.center) to (h.south west);
\draw[string,out=90,in=right] (h.north) to (b2.center);
\node(i1) at (-0.1,-2.5){};
\node(i2) at (1.25,-2.5){};
\node(i3) at (1.75,-2.5){};
\node(i1a) at (0.24,-2.5){};

\node(o) at (0,2){};
\draw[string] (m3.north) to (m2.south);
\draw[string,out=right,in=270] (b.center) to (b4.center);
\draw[string] (b.center) to (m3.south west);
\draw[string] (i3.center) to (h.south);
\draw[string] (o.center) to (m1.north);
\draw[string] (b.center) to (i1.center);
\draw[string] (i1a.center) to (m3.south east);
\end{tikzpicture}
\end{aligned}
\quad \super{\eqref{eq:sm}}= \quad
\begin{aligned}
\begin{tikzpicture}[scale=0.75]
\node (m1) [morphism,wedge,scale=0.75] at (0,0.75) {$M$};
\node(m2)[morphism,wedge,hflip,scale=0.75]  at (0,-0.75) {$M$};
\node(m3)[morphism,wedge,scale=0.75]  at (0,-1.5) {$M$};
\node[blackdot,scale=0.5](b) at (-0.25,-2.25){};
\draw (m1.south east) to (m2.north east)[string];
\draw (m1.south west) to (m2.north west)[string];
\node(b1)[blackdot,scale=0.5] at (-0.25,0){};
\node(b2)[blackdot,scale=0.5] at (0.25,0.2){};
\node(h)[morphism,wedge,scale=0.75] at (1.75,-1){$H$};
\node(m) at (1,-0.25){};
\node(b5)[blackdot,scale=0.5] at (-0.25,-2.65){};
\draw[string,out=right,in=south west] (b5.center) to (h.south west);
\node(s)[whitedot,inner sep=0.25pt,scale=0.75] at (1.75,-2){$\boldsymbol{*}$};
\draw[string,out=90,in=right] (h.north) to (b2.center);
\node(i1) at (-0.25,-3){};
\node(i2) at (1.25,-3){};
\node(i3) at (1.75,-3){};
\node(i1a) at (0.24,-3){};
\node(o) at (0,1.5){};
\draw[string] (m3.north) to (m2.south);
\draw[string,out=left,in=left] (b.center) to (b1.center);
\draw[string] (b.center) to (m3.south west);
\draw[string] (i3.center) to (h.south);
\draw[string] (o.center) to (m1.north);
\draw[string] (b.center) to (i1.center);
\draw[string] (i1a.center) to (m3.south east);
\end{tikzpicture}\end{aligned}
\quad \super{\eqref{eq:ch}}= \quad
\begin{aligned}
\begin{tikzpicture}[scale=0.75]
\node (m1) [morphism,wedge,scale=0.75] at (0,0.75) {$M$};

\node(b) at (-0.25,-2.25){};

\node(b1) at (-0.25,0){};
\node(b2)[blackdot,scale=0.5] at (0.25,0.2){};
\node(h)[morphism,wedge,scale=0.75] at (1.75,-1){$H$};
\node(m) at (1,-0.25){};
\node(b5)[blackdot,scale=0.5] at (-0.25,-2.65){};
\draw[string,out=right,in=south west] (b5.center) to (h.south west);
\node(s)[whitedot,inner sep=0.25pt,scale=0.75] at (1.75,-2){$\boldsymbol{*}$};
\draw[string,out=90,in=right] (h.north) to (b2.center);
\node(i1) at (-0.25,-3){};
\node(i2) at (1.25,-3){};
\node(i3) at (1.75,-3){};
\node(i1a) at (0.24,-3){};
\node(o) at (0,1.5){};
\draw[string] (i3.center) to (h.south);
\draw[string] (o.center) to (m1.north);
\draw[string] (b.center) to (i1.center);
\draw (m1.south east) to (i1a.center)[string];
\draw (m1.south west) to (i1.center)[string];
\end{tikzpicture}
\end{aligned}\end{equation*}  
If we input black states $i,j,k$ with $k \neq0$  the above equation becomes \mbox{$[\phi_H(M)]_{ik}\ket{b^i_j}=[H^i]_{jk}\ket{b^i_k}$}. Thus the bases of the original maximal family of MUBs are the eigenbases of $\phi_H(M)$  as required. 
\end{proof}

Given that the composite map $\theta \circ \phi_H$ is the identity we conclude that $\theta$ is surjective and for all controlled Hadamards $H$,  the map $\phi_H$ is injective. Given some maximal family of MUBs $M$ and controlled Hadamard ${H}$,  the proof to Theorem~\ref{thm:compid}  allows us to identify  the eigenvalues of the $\UEBM$ $\phi_{H}(M)$ with the entries of the  controlled Hadamard. This is precisely the information lost by the map $\theta$. So every maximal family of MUBs corresponds to an infinite family of $\UEBM$s for different choices of eigenvalues. These UEBs are in general inequivalent. Also every $\UEBM$ corresponds to a maximal family of MUBs with a particular choice of controlled Hadamard. This holds in any dimension, so the existence problem for maximal families of MUBs in arbitrary dimension can be phrased in terms of $\UEBM$s. Similarly for non-maximal families of MUBs we have corresponding UEBs with partial partition into maximally commuting sub-families.    
\section{Tensor diagrammatic finite fields}\label{section:gff}
We now  present a construction of $\UEBM$s from a finite field. In order to achieve this  we first introduce a Tensor diagrammatic characterisation of finite fields, so that we can interpret finite fields as algebraic structures in Hilbert space.  We begin this section by reviewing the Tensor diagrammatic properties of abelian groups in Hilbert space~\cite{stefwill}. We recall how the character theory of abelian groups can be reconstructed using tensor diagrams. This gives us a graphical representation of  complex group algebras, and the usual complex character theory~\cite{stefwill,abeliangroupgraph}. We will then build on this framework to discuss finite fields as algebraic structures defined over Hilbert spaces. 
\subsection{Abelian groups}
First we define abelian groups.
\begin{definition}[Abelian group]~\cite{lidl}
A set together with a binary operation is an \textit{abelian group} if it is closed, unital, associative, commutative and every element has an inverse.
\end{definition}
We continue with the convention that $\tinymult$ is a $\dfrob$ and use the black states as an indexing set. In this case we are indexing the elements of the abelian group, and later the elements of the finite field. 

\paragraph{Unitality, associativity and commutativity.}We introduce another Frobenius algebra which will represent the binary operator of the group. Let $\tinymult[reddot]$ be a \mbox{$\dagger$-quasi-special} commutative Frobenius algebra ($\dagger$-qSCFA). Since $\tinymult[reddot]$ is a commutative Frobenius algebra
it is unital, associative and commutative by definition.

\paragraph{Closure.}It is the interaction of red and black that gives us the structure of a group. We require that $\tinymult[reddot]$ is closed with respect to $\tinymult$ this means that we need $\tinymult[reddot]$ to take pairs of black basis states to black basis states. This is equivalent to $\tinymult[reddot]$ being a comonoid homomorphism for $\tinycomult$ (see Definition~\ref{def:cohom}). This is encapsulated by the following axiom:
\begin{definition}[Bialgebra]\label{def:bialg}A pair of unital associative algebras are a \textit{bialgebra} if:
   \begin{align}\label{eq:bialg} \begin{aligned}\begin{tikzpicture}[yscale=0.5,xscale=0.65]
          \node (0a) at (-0.5,0) {};
          \node (0b) at (0.5,0) {};
          \node[reddot] (1) at (0,1) {};
          \node[blackdot] (2) at (0,2) {};
          \node (3a) at (-0.5,3) {};
          \node (3b) at (0.5,3) {};
          \draw[string,out=90,in=180] (-0.5,0) to (1.center);
          \draw[string,out=90,in=0] (0.5,0) to (1.center);
          \draw[string] (1.center) to (0,2);
          \draw[string,out=180,in=-90] (0,2) to (-0.5,3);
          \draw[string,out=0,in=270] (2.center) to (0.5,3);
      \end{tikzpicture}\end{aligned}
    =
    \begin{aligned}\begin{tikzpicture}[yscale=0.5,xscale=0.65]
           \node[blackdot] (2) at (0,0.9) {};
           \node[blackdot] (3) at (1,0.9) {};
           \node[reddot] (4) at (0,2.1) {};
           \node[reddot] (5) at (1,2.1) {};
           \draw[string] (0,0) to (2.center);
           \draw[string] (1,0) to (3.center);
           \draw[string] (4.center) to (0,3);
           \draw[string] (5.center) to (1,3);
           \draw[string, in=180, out=180, looseness=1.2] (2.center) to (4.center);
           \draw[string, in=0, out=0, looseness=1.2] (3.center) to (5.center);
           \draw[string, in=180, out=right] (2.center) to (5.center);
           \draw[string, in=0, out=180] (3.center) to (4.center);
    \end{tikzpicture}\end{aligned}
    \qquad & \quad
 \begin{aligned}\begin{tikzpicture}[yscale=0.5,xscale=0.65]
          \node[reddot] (1) at (0,1) {};
          \node[blackdot] (2) at (0,2) {};
          \draw[string,out=90,in=180] (-0.5,0) to (1.center);
          \draw[string,out=90,in=0] (0.5,0) to (1.center);
          \draw[string] (1.center) to (2.center);
          \draw[string,out=180,in=270, white] (2.center) to (-0.5,3);
          \draw[string,out=0,in=270, white] (2.center) to (0.5,3);
      \end{tikzpicture}\end{aligned}
  =
  \begin{aligned}\begin{tikzpicture}[yscale=0.5,xscale=0.65]
        \draw [white, string] (0,0) to (0,3);
        \node[blackdot] (2) at (0,1) {};
        \node[blackdot] (3) at (1,1) {};
        \draw[string] (0,0) to (0,1) {};
        \draw[string] (1,0) to (1,1) {};
    \end{tikzpicture}\end{aligned}\quad & \quad
    \begin{aligned}\begin{tikzpicture}[yscale=-0.5,xscale=0.65]
          \node[blackdot] (1) at (0,1) {};
          \node[reddot] (2) at (0,2) {};
          \draw[string,out=90,in=180] (-0.5,0) to (0,1);
          \draw[string,out=90,in=0] (0.5,0) to (1.center);
          \draw[string] (1.center) to (2.center);
          \draw[string,out=180,in=270, white] (2) to (-0.5,3);
          \draw[string,out=0,in=270, white] (2) to (0.5,3);
      \end{tikzpicture}\end{aligned}
  =
  \begin{aligned}\begin{tikzpicture}[yscale=-0.5,xscale=0.65]
        \draw [white, string] (0,0) to (0,3);
        \node[reddot] (2) at (0,1) {};
        \node[reddot] (3) at (1,1) {};
        \draw[string] (0,0) to (2.center) {};
        \draw[string] (1,0) to (3.center) {};
    \end{tikzpicture}\end{aligned}
  \qquad & \quad
    \begin{aligned}\begin{tikzpicture}[yscale=-0.5,xscale=0.65]
        \draw[string, white] (0,-1) to (0,2);
        \node[reddot] (0) at (0,0) {};
        \node[blackdot] (1) at (0,1) {};
        \draw[string] (0.center) to (1.center);
    \end{tikzpicture}\end{aligned}
    =\qquad \end{align} 
   Note that the right hand side of the second equation is the empty diagram indicating the identity complex scalar $1$. 
\end{definition} 
Given Frobenius algebras $\tinymult$ and $\tinymult[reddot]$,
we call them Frobenius bialgebras if they obey the bialgebra laws. Note that the third bialgebra rule means that the red unit is copyable by, and thus a state of, the black basis. We will assume that the basis is ordered such that $\tinyunit[reddot]=\tinystate[black,state,label={[label distance=-0.08cm]330:\tiny{$0$}}]$.  Since we think of $\tinymult[reddot]$ as a binary operation taking black states to other black states, the morphism $\tinymult[reddot]$ should be real valued with respect to the black basis when considered as a linear map in Hilbert space. The following axiom gives this property in the general case:
\begin{definition}[$\tinydot$\hspace{-6pt} -real~\cite{rosskev}]
In a $\dagger$-symmetric monoidal category given an object $A$ with a $\dfrob$ $\tinymult$, a morphism $F:A^{\otimes n}\rightarrow A^{\otimes m}$   is $\tinydot$\hspace{-6pt} -\textit{real} if: 
\begin{equation}
\begin{pic}
     \node(f)[morphism,wedge,hflip,connect se length=0.5cm,connect sw length=0.5cm,connect ne length=0.5cm,connect nw length=0.5cm] at (0,0){$F$};
     \draw[dotted] (-0.15,0.55) to (0.25,0.55); 
     \draw[dotted] (-0.15,-0.55) to (0.25,-0.55);      
\end{pic}
\quad = \quad
\begin{pic}
     \node(f)[morphism,wedge,connect se length=0.125cm,connect sw length=0.125cm,connect ne length=0.125cm,connect nw length=0.125cm] at (0,0){$F$};
     \draw[dotted] (-0.15,0.4) to (0.25,0.4); 
     \draw[dotted] (-0.15,-0.4) to (0.25,-0.4);
     \draw[dotted] (0.3,0.7) to (0.7,0.7);
     \draw[dotted] (0.95,-0.7) to (1.35,-0.7);
     \draw[dotted] (-0.3,-0.7) to (-0.7,-0.7);
     \draw[dotted] (-0.95,0.7) to (-1.35,0.7);     
     \node(b1)[blackdot,scale=0.7] at (0.25,0.65){};
     \node(b2)[blackdot,scale=0.7] at (0.75,0.65){};
     \draw[string,out=90,in=180,looseness=0.8] (f.connect nw) to (b1.center);      \node(i1) at (0.9,-0.95){};      
     \node(i2) at (1.4,-0.95){};
     \draw[string,out=90,in=180,looseness=0.8] (f.connect ne) to (b2.center);      \draw[string,out=0,in=90] (b1.center) to (i1);      
     \draw[string,out=0,in=90] (b2.center) to (i2);     
     \node(b11)[blackdot,scale=0.7] at (-0.25,-0.65){};
     \node(b21)[blackdot,scale=0.7] at (-0.75,-0.65){};
     \draw[string,out=270,in=0,looseness=0.8] (f.connect se) to (b11.center);      \node(i11) at (-0.9,0.95){};      
     \node(i21) at (-1.4,0.95){};
     \draw[string,out=270,in=0,looseness=0.8] (f.connect sw) to (b21.center);      \draw[string,out=180,in=270] (b11.center) to (i11);      
     \draw[string,out=180,in=270] (b21.center) to (i21);
\end{pic}
\end{equation}
\end{definition}
We require $\tinymult[reddot]$ to be $\tinydot$\hspace{-6pt} -real, which gives us the following equations:
\begin{align}
\begin{aligned}\begin{tikzpicture}[scale=0.75]
          \node (0a) at (-1,1) {};
          \node (0b) at (-0.5,1) {};
          \node (0c) at (0,1) {};
          \node[reddot] (7) at (-0.5,0) {};       
          \node (10a) at (-0.5,-1) {};       
          \draw[string, out=270, in =180] (0a.center) to (7.center);
          \draw[string, out=0, in=270] (7.center) to (0c.center);
          \draw[string] (7.center) to (10a.center);
          \end{tikzpicture}\end{aligned}
          \quad=\quad
\begin{aligned}\begin{tikzpicture}[scale=0.75]
          \node (0a') at (-2,1) {};
          \node[blackdot] (a) at (-0.75,-1) {};
          \node (0a) at (-2,0.5) {};
          \node (0c) at (-1.5,0.5) {}; 
          \node[blackdot] (b) at (-1.5,-1) {};
          \node[blackdot] (c) at (0,0.5) {};
          \node (0b) at (-0.5,0.5) {};
          \node (0c') at (-1.5,1) {};
          \node[reddot] (7) at (-0.5,0) {};       
          \node (10a) at (0.5,-1.5) {};       
          \draw[string, out=270, in =180] (0c'.center) to (a.center);
          \draw[string, out=0, in =0] (7.center) to (a.center);
          \draw[string, out=180, in=0] (7.center) to (b.center);
          \draw[string, out=180, in=270] (b.center) to (0a'.center);               \draw[string,out=90 ,in=180] (7.center) to (c.center);
          \draw[string,out=0,in=90] (c.center) to (10a); 
   \end{tikzpicture}\end{aligned}
    \qquad      &  \qquad \qquad
    \begin{aligned}\begin{tikzpicture}
    \node(r)[reddot] at (0,0){};
    \node(b)[blackdot] at (0.5,-0.5){};
    \draw[string,in=left,out=270] (r.center) to (b.center);    
    \draw[string,in=right,out=270] (1,0.75) to (b.center);
    \end{tikzpicture}\end{aligned}
    \quad = \quad
    \begin{aligned}\begin{tikzpicture}
    \node(b)[reddot] at (1,-0.5){};    
    \draw[string] (1,0.75) to (b.center){};
    \end{tikzpicture}\end{aligned}          
\end{align}
\paragraph{Inverses.}
The following condition is equivalent to the binary operator having inverses, which gives us a Hopf algebra~\cite{rosskev}.
\begin{definition}[Strong complementarity]
Given two $\dagger$-commutative Frobenius algebras $\tinymult$ and $\tinymult[reddot]$, they are \textit{strongly complementary} if the following composite linear maps are both unitary: 
\begin{align}
    \begin{aligned}\begin{tikzpicture}[yscale=0.5,xscale=0.6]
          \node (0) at (0,0) {};
          \node (0a) at (0,1) {};
          \node[blackdot] (1) at (0.5,2) {};
          \node[reddot] (2) at (1.5,1) {};
          \node (3) at (1.5,0) {};
          \node (4) at (2,3) {};
          \node (4a) at (2,2) {};
          \node (5) at (0.5,3) {};
          \draw[string] (0) to (0a.center);
          \draw[string,out=90,in=180] (0a.center) to (1.center);
          \draw[string,out=0,in=180] (1.center) to (2.center);
          \draw[string,out=0,in=270] (2.center) to (4a.center);
          \draw[string] (4a.center) to (4);
          \draw[string] (2.center) to (3);
          \draw[string] (1.center) to (5);
      \end{tikzpicture}\end{aligned}
      \qquad & \qquad
    \begin{aligned}\begin{tikzpicture}[yscale=-0.5,xscale=0.6]
          \node (0) at (0,0) {};
          \node (0a) at (0,1) {};
          \node[reddot] (1) at (0.5,2) {};
          \node[blackdot] (2) at (1.5,1) {};
          \node (3) at (1.5,0) {};
          \node (4) at (2,3) {};
          \node (4a) at (2,2) {};
          \node (5) at (0.5,3) {};
          \draw[string] (0) to (0a.center);
          \draw[string,out=90,in=180] (0a.center) to (1.center);
          \draw[string,out=0,in=180] (1.center) to (2.center);
          \draw[string,out=0,in=270] (2.center) to (4a.center);
          \draw[string] (4a.center) to (4);
          \draw[string] (2.center) to (3);
          \draw[string] (1.center) to (5);
      \end{tikzpicture}\end{aligned}
      \end{align}
\end{definition}
The following theorem is useful in understanding how the character theory of an Abelian group can be derived using tensor diagrams. 
\begin{theorem}[\cite{zeng}, Theorem 9]\label{thm:zeng} A pair of $\dagger$-commutative Frobenius algebras are strongly complementarity if and only if the corresponding  bases are mutually unbiased. \end{theorem} 
\paragraph{Character group.}We shall now see that the basis states copyable by  $\tinycomult[reddot]$  form the character group, which has binary operator $\tinymult$. 
Let $\chi$ be the change of basis linear map that maps the red basis to the black basis up to a normalization factor, as follows:
\begin{equation}\label{eq:addhad}
\begin{aligned}\begin{tikzpicture}
   \node(r)[reddot]{};
   \node(b) at(0,0.5){};
   \draw[string] (r.center) to (b.center);
   \draw[string] (0,1) to (b.center);
   \draw[string,out=90,in=left] (-0.5,-1) to (r.center);
   \draw[string,out=90,in=right] (0.5,-1) to (r.center);
\end{tikzpicture}
\end{aligned} 
\quad = \quad \frac{1}{d}
\begin{aligned}\begin{tikzpicture}
   \node(r)[blackdot]{};
      \node(b)[morphism,wedge,hflip,scale=0.5] at(0,0.5){\large$\chi$};
   \node(b1)[morphism,wedge,scale=0.5] at(-0.4,-0.5){\large$\chi$};
   \node(b2)[morphism,wedge,scale=0.5] at(0.4,-0.5){\large$\chi$};
   \draw[string] (r.center) to (b.south);
   \draw[string] (b.north) to (0,1);
   \draw[string,out=90,in=270] (-0.4,-1) to (b1.south);
   \draw[string,out=90,in=270] (0.4,-1) to (b2.south);
   \draw[string,out=left,in=90] (r.center) to (b1.north);
   \draw[string,out=right,in=90] (r.center) to (b2.north);
\end{tikzpicture}
\end{aligned} 
\end{equation} 
Considering the discussion below Definition~\ref{def:had} and Theorem~\ref{thm:zeng} we can see that $\chi$ is a Hadamard . This Hadamard is the Fourier transform of the group and its rows, which are the copyable states of the red basis, are the irreducible characters.  Apart from the axioms for a Hadamard (see Definition~\ref{def:had}), it can easily be shown that the following equations hold which gives us the expected character theory.  \begin{align}\label{eq:chi}
\begin{aligned}\begin{tikzpicture}
   \node(r)[reddot]{};
   \node(b)[morphism,wedge,scale=0.5] at(0,0.5){\large$\chi$};
   \draw[string] (r.center) to (b.south);
   \draw[string] (0,1) to (b.north);
   \draw[string,out=90,in=left] (-0.5,-1) to (r.center);
   \draw[string,out=90,in=right] (0.5,-1) to (r.center);
\end{tikzpicture}
\end{aligned} 
\quad = \quad
\begin{aligned}\begin{tikzpicture}
   \node(r)[blackdot]{};
   \node(b1)[morphism,wedge,scale=0.5] at(-0.4,-0.5){\large$\chi$};
   \node(b2)[morphism,wedge,scale=0.5] at(0.4,-0.5){\large$\chi$};
   \draw[string] (r.center) to (0,1);
   \draw[string,out=90,in=270] (-0.4,-1) to (b1.south);
   \draw[string,out=90,in=270] (0.4,-1) to (b2.south);
   \draw[string,out=left,in=90] (r.center) to (b1.north);
   \draw[string,out=right,in=90] (r.center) to (b2.north);
\end{tikzpicture}
\end{aligned} 
\qquad & \qquad
\begin{pic}
\node(1)[] at (0,2){};
\node(2)[morphism,wedge,scale=0.5] at (0,1){\large$\chi$};
\node(3)[reddot] at (0,0){};
\draw (1.center) to (2.north);
\draw (2.south) to (3.center);
\end{pic}\quad = \quad
\begin{pic}
\node(1)[] at (0,2){};
\node(2) at (0,1){};
\node(3)[blackdot] at (0,0){};
\draw (1.center) to (2.center);
\draw (2.center) to (3.center);
\end{pic}
\qquad & \qquad
\begin{pic}[yscale=-1]
\node(1)[reddot] at (0,2){};
\node(2)[morphism,wedge,scale=0.5,hflip] at (0,1){\large$\chi$};
\node(3) at (0,0){};
\draw (1.center) to (2.south);
\draw (2.north) to (3.center);
\end{pic}\quad = \quad
\begin{pic}[yscale=1]
\node(1)[] at (0,2){};
\node(2) at (0,1){};
\node(3)[blackdot] at (0,0){};
\draw (1.center) to (2.center);
\draw (2.center) to (3.center);
\end{pic}
\end{align}
Let $\chi_i(x):=\bra i \chi \ket x$ so $\chi_i (x)$ is the $i$th irreducible character applied to an element of the group $x$. The left hand equation  is then equivalent to; for all \mbox{$i,x,y \in \range{0}, \chi_i(x+y)=\chi_i(x)\chi_i(y)$} which is the expected property of a character.
For a  more detailed discussion of the above please refer to the 2015 paper by Gogioso and Zeng~\cite{stefwill}. 
We summarize the results of this subsection in the following theorem:
\begin{theorem}~\cite{stefwill}
Given a $d$-dimensional Hilbert space with a $\dagger$-qSCFA $\tinymult[reddot]$ and a $\dfrob$ $\tinymult$ the following are equivalent:
\begin{itemize}
\item The copyable states of $\tinymult$ form an abelian complex group algebra under the linearly extended binary operator $\tinymult[reddot]$;
\item The algebras $\tinymult$ and $\tinymult[reddot]$ form a strongly complementary bialgebra and $\tinymult[reddot]$ is $\tinydot$\hspace{-6pt} -real.
\end{itemize}
\end{theorem}
\subsection{Finite fields}
We now define a finite field.
\begin{definition}[Finite field]\label{def:algff}~\cite{lidl}
A finite set $A$ together with closed binary operators $\tinydot[reddot]$ and $\tinydot[yellowdot]$ is a \textit{finite field} if:
\begin{itemize}
\item \textbf{Addition}: The operator $\tinydot[reddot]$ is an abelian group on the set $A$ with unit $0\in A$;
\item \textbf{Multiplication}: The operator $\tinydot[yellowdot]$ is an abelian group on the subset $A':=A\backslash \{0\}$;
\item \textbf{Distributivity}: For all $a,b,c \in A$, $a\tinydot[yellowdot](b\tinydot[reddot]c)=(a\tinydot[reddot]b)\tinydot[yellowdot](a\tinydot[reddot]c)$.
\end{itemize}
\end{definition}
Finite fields only exist in prime power dimensions~\cite{lidl} so we take $d=p^n$  for some prime $p$ and $n \in \mathbb N$, and as usual take the black wires to represent the  Hilbert space $\mathcal H\cong \C^d$. 
\paragraph{Addition.}In formulating  a diagrammatic notation for finite fields as algebraic structures defined over Hilbert spaces we start with an abelian group algebra representing addition. We therefore take $\tinymult$ and $\tinymult[reddot]$ to be a pair of  strongly complementary \mbox{$\dagger$-commutative} Frobenius bialgebras with black special, red quasi-special and red $\tinydot$\hspace{-6pt} -real . As seen in the last subsection $\tinycomult[reddot]$ copies the additive characters which form the columns of a Fourier Hadamard matrix on $\mathcal H$ which we will again call $\chi$ with the formal definition given by equation~\eqref{eq:addhad}. 
\paragraph{Multiplication.}The multiplication of a finite field also forms an abelian group on the non-zero elements. We introduce another Hilbert space, $\mathcal H' \cong \C^{d-1}$ which we represent as green wires, and a $\dagger$SCFA $\tinymultgr$. We also introduce linear maps to relate the green and black Hilbert spaces:
\begin{align}\label{eq:projj}
p:= 
\begin{aligned}\begin{tikzpicture}[scale=0.65]
           \node(3)[whitedot,projdot] at (0,0){};
           \node (1) at (0,-1) {}; 
           \node (2) at (0,1) {}; 
           \draw[green][string,line width=1pt] (2.center) to (3.center);
           \draw[string] (3.center) to (1.center);           
\end{tikzpicture}
\end{aligned} 
\qquad & \qquad \qquad \iota:= 
\begin{aligned}\begin{tikzpicture}[scale=0.65]
           \node(3)[whitedot,projdot] at (0,0){};
           \node (1) at (0,-1) {}; 
           \node (2) at (0,1) {}; 
           \draw[string] (2.center) to (3.center);
           \draw[green][string,line width=1pt] (3.center) to (1.center);           
\end{tikzpicture}
\end{aligned} 
\end{align}
We require the following relationships between $p$, $\iota$, $\tinymultgr$, $\tinymult$ and $\tinyunit[reddot]$:
 \begin{align}\label{eq:projprop}
 \begin{aligned}\begin{tikzpicture}[scale=0.65]
           \node(4) at (0,-2){}; 
           \node(3)[whitedot,projdot] at (0,0){};
           \node (1)[whitedot,projdot] at (0,-1) {}; 
           \node (2) at (0,1) {}; 
           \draw[green][string,line width=1pt] (2.center) to (3.center);
           \draw[green][string, line width=1pt] (4.center) to (1.center); 
           \draw[string] (3.center) to (1.center);           
\end{tikzpicture}
\end{aligned} 
 = 
\begin{aligned}\begin{tikzpicture}[scale=0.65]
           \node(4) at (0,-2){};           
           \node (2) at (0,1) {}; 
           \draw[green][string, line width=1pt] (2.center) to (4.center);
               \end{tikzpicture}
\end{aligned} 
\qquad & \quad
\begin{aligned}\begin{tikzpicture}[scale=0.65]
           \node(3)[whitedot,projdot] at (0,0){};
           \node[blackdot] (1) at (0,-1.5) {}; 
           \node (2) at (0,1.5) {}; 
           \draw[green][string, line width=1pt] (2.center) to (3.center);
           \draw[string] (3.center) to (1.center);           
\end{tikzpicture}
\end{aligned} 
 = 
\begin{aligned}\begin{tikzpicture}[scale=0.65]
           \node(3) at (0,0){};
           \node[greendot] (1) at (0,-1.5) {}; 
           \node (2) at (0,1.5) {}; 
           \draw[green][string, line width=1pt] (2.center) to (3.center);
           \draw[green][string, line width=1pt] (3.center) to (1.center);           
\end{tikzpicture}
\end{aligned} 
\qquad & \quad
\begin{aligned}\begin{tikzpicture}[scale=0.65]
\node(b)[blackdot] at (0,0){};
\node(1)[whitedot,projdot] at (0,0.5){};
\node(2)[whitedot,projdot] at (-0.5,-0.5){};
\node(3)[whitedot,projdot] at (0.5,-0.5){};
\node(i1) at (-0.5,-1.5){};
\node(i2) at (0.5,-1.5){};
\node(o) at (0,1.5){};
\draw[string,out=90, in=left] (2.center) to (b.center);
\draw[string,out=90, in=right] (3.center) to (b.center);
\draw[string,out=270, in=90] (1.center) to (b.center);
\draw[green][string, line width=1pt,out=90,in=270] (1.center) to (o.center);
\draw[green][string, line width=1pt,in=90,out=270] (2.center) to (i1.center);
\draw[green][string, line width=1pt,in=90,out=270] (3.center) to (i2.center);
\end{tikzpicture}\end{aligned}
= 
\begin{aligned}\begin{tikzpicture}[scale=0.65]
\node(b)[greendot] at (0,0){};
\node(1) at (0,0.5){};
\node(2) at (-0.5,-0.5){};
\node(3) at (0.5,-0.5){};
\node(i1) at (-0.5,-1.5){};
\node(i2) at (0.5,-1.5){};
\node(o) at (0,1.5){};
\draw[green][string, line width=1pt,out=90, in=left] (2.center) to (b.center);
\draw[green][string, line width=1pt,out=90, in=right] (3.center) to (b.center);
\draw[green][string, line width=1pt,out=270, in=90] (1.center) to (b.center);
\draw[green][string, line width=1pt,out=90,in=270] (1.center) to (o.center);
\draw[green][string, line width=1pt,in=90,out=270] (2.center) to (i1.center);
\draw[green][string, line width=1pt,in=90,out=270] (3.center) to (i2.center);
\end{tikzpicture}\end{aligned}
\qquad & \qquad
\begin{aligned}\begin{tikzpicture}[scale=0.65]
\node(b)[greendot] at (0,0){};
\node(1)[whitedot,projdot] at (0,0.5){};
\node(2)[whitedot,projdot] at (-0.5,-0.5){};
\node(3)[whitedot,projdot] at (0.5,-0.5){};
\node(i1) at (-0.5,-1.5){};
\node(i2) at (0.5,-1.5){};
\node(o) at (0,1.5){};
\draw[green][string, line width=1pt,out=90, in=left] (2.center) to (b.center);
\draw[green][string, line width=1pt,out=90, in=right] (3.center) to (b.center);
\draw[green][string, line width=1pt,out=270, in=90] (1.center) to (b.center);
\draw[string,out=90,in=270] (1.center) to (o.center);
\draw[string,in=90,out=270] (2.center) to (i1.center);
\draw[string,in=90,out=270] (3.center) to (i2.center);
\end{tikzpicture}\end{aligned}
 = 
\begin{aligned}\begin{tikzpicture}[scale=0.65]
\node(b)[blackdot] at (0,0){};
\node(1) at (0,0.5){};
\node(2) at (-0.5,-0.5){};
\node(3) at (0.5,-0.5){};
\node(i1) at (-0.5,-1.5){};
\node(i2) at (0.5,-1.5){};
\node(o) at (0,1.5){};
\draw[string,out=90, in=left] (2.center) to (b.center);
\draw[string,out=90, in=right] (3.center) to (b.center);
\draw[string,out=270, in=90] (1.center) to (b.center);
\draw[string,out=90,in=270] (1.center) to (o.center);
\draw[string,in=90,out=270] (2.center) to (i1.center);
\draw[string,in=90,out=270] (3.center) to (i2.center);
\end{tikzpicture}\end{aligned}
 - 
\begin{aligned}\begin{tikzpicture}[scale=0.65]
\node(b)[reddot] at (0,0.25){};
\node(1) at (0,0.5){};
\node[reddot](2) at (-0.5,-0.25){};
\node[reddot](3) at (0.5,-0.25){};
\node(i1) at (-0.5,-1.5){};
\node(i2) at (0.5,-1.5){};
\node(o) at (0,1.5){};
\draw[string,out=270, in=90] (1.center) to (b.center);
\draw[string,out=90,in=270] (1.center) to (o.center);
\draw[string,in=90,out=270] (2.center) to (i1.center);
\draw[string,in=90,out=270] (3.center) to (i2.center);
\end{tikzpicture}\end{aligned}
\end{align}     
We will assume that the black basis has been ordered such that $\tinystate[black,state,label={[label distance=-0.08cm]330:\tiny{$0$}}]:=\tinyunit[reddot]$. This makes $p$ and $\iota$ an isomorphism between $\mathcal H'$ and the $d-1$-dimensional subspace of $\mathcal H$ spanned by the non-zero black states. This isomorphism takes the green basis states to the non-zero black states. The Hilbert space $\mathcal H'$ is the analogue of the set $A'$ in Definition~\ref{def:algff}. 

The following lemma shows that the linear map given by $\iota \circ p$ is  equal to the projector defined by equation~\eqref{eq:proj}, we will make use of this projector.  
\begin{lemma}\label{lem:a}
The following equation holds.
\begin{align}\label{eq:proj2}
 \begin{aligned}\begin{tikzpicture}[scale=0.65]
           \node(4) at (0,-2){}; 
           \node(3)[whitedot,projdot] at (0,0){};
           \node (1)[whitedot,projdot] at (0,-1) {}; 
           \node (2) at (0,1) {}; 
           \draw[string] (2.center) to (3.center);
           \draw[string] (4.center) to (1.center); 
           \draw[green][string,line width=1pt] (3.center) to (1.center);           
\end{tikzpicture}
\end{aligned} 
\quad = \quad
\begin{aligned}\begin{tikzpicture}[scale=0.65]
           \node(4) at (0,-2){};           
           \node (2) at (0,1) {}; 
           \draw[string] (2.center) to (4.center);
               \end{tikzpicture}
\end{aligned} 
\quad - \quad
 \begin{aligned}\begin{tikzpicture}[scale=0.65]
           \node(4) at (0,-2){}; 
           \node[reddot](3) at (0,0){};
           \node[reddot] (1) at (0,-1) {}; 
           \node (2) at (0,1) {}; 
           \draw[string] (2.center) to (3.center);
           \draw[string] (4.center) to (1.center);
\end{tikzpicture}
\end{aligned} 
\end{align}
\end{lemma}
\begin{proof}
By the $4$th equation of~\eqref{eq:projprop}:
\begin{align*}
\begin{aligned}\begin{tikzpicture}[scale=0.65]
\node(b)[greendot] at (0,0){};
\node(1)[whitedot,projdot] at (0,0.5){};
\node(2)[whitedot,projdot] at (-0.5,-0.5){};
\node(3)[whitedot,projdot] at (0.5,-0.5){};
\node(i1) at (-0.5,-1.5){};
\node(i2) at (0.5,-1.5){};
\node(o) at (0,1.5){};
\draw[green][string,line width=1pt,out=90, in=left] (2.center) to (b.center);
\draw[green][string,line width=1pt,out=90, in=right] (3.center) to (b.center);
\draw[green][string,line width=1pt,out=270, in=90] (1.center) to (b.center);
\draw[string,out=90,in=270] (1.center) to (o.center);
\draw[string,in=90,out=270] (2.center) to (i1.center);
\draw[string,in=90,out=270] (3.center) to (i2.center);
\end{tikzpicture}\end{aligned}
\quad =& \quad
\begin{aligned}\begin{tikzpicture}[scale=0.65]
\node(b)[blackdot] at (0,0){};
\node(1) at (0,0.5){};
\node(2) at (-0.5,-0.5){};
\node(3) at (0.5,-0.5){};
\node(i1) at (-0.5,-1.5){};
\node(i2) at (0.5,-1.5){};
\node(o) at (0,1.5){};
\draw[string,out=90, in=left] (2.center) to (b.center);
\draw[string,out=90, in=right] (3.center) to (b.center);
\draw[string,out=270, in=90] (1.center) to (b.center);
\draw[string,out=90,in=270] (1.center) to (o.center);
\draw[string,in=90,out=270] (2.center) to (i1.center);
\draw[string,in=90,out=270] (3.center) to (i2.center);
\end{tikzpicture}\end{aligned}
\quad - \quad
\begin{aligned}\begin{tikzpicture}[scale=0.65]
\node(b)[reddot] at (0,0.25){};
\node(1) at (0,0.5){};
\node[reddot](2) at (-0.5,-0.25){};
\node[reddot](3) at (0.5,-0.25){};
\node(i1) at (-0.5,-1.5){};
\node(i2) at (0.5,-1.5){};
\node(o) at (0,1.5){};
\draw[string,out=270, in=90] (1.center) to (b.center);
\draw[string,out=90,in=270] (1.center) to (o.center);
\draw[string,in=90,out=270] (2.center) to (i1.center);
\draw[string,in=90,out=270] (3.center) to (i2.center);
\end{tikzpicture}\end{aligned}
 \quad \Rightarrow
\begin{aligned}\begin{tikzpicture}[scale=0.65]
\node(b)[greendot] at (0,0){};
\node(1)[whitedot,projdot] at (0,0.5){};
\node(2)[whitedot,projdot] at (-0.5,-0.5){};
\node(3)[whitedot,projdot] at (0.5,-0.5){};
\node(i1) at (-0.5,-1.5){};
\node[blackdot](i2) at (0.5,-1.5){};
\node(o) at (0,1.5){};
\draw[green][string,line width=1pt,out=90, in=left] (2.center) to (b.center);
\draw[green][string,line width=1pt,out=90, in=right] (3.center) to (b.center);
\draw[green][string,line width=1pt,out=270, in=90] (1.center) to (b.center);
\draw[string,out=90,in=270] (1.center) to (o.center);
\draw[string,in=90,out=270] (2.center) to (i1.center);
\draw[string,in=90,out=270] (3.center) to (i2.center);
\end{tikzpicture}\end{aligned}
\quad = \quad
\begin{aligned}\begin{tikzpicture}[scale=0.65]
\node(b)[blackdot] at (0,0){};
\node(1) at (0,0.5){};
\node(2) at (-0.5,-0.5){};
\node(3) at (0.5,-0.5){};
\node(i1) at (-0.5,-1.5){};
\node[blackdot](i2) at (0.5,-1.5){};
\node(o) at (0,1.5){};
\draw[string,out=90, in=left] (2.center) to (b.center);
\draw[string,out=90, in=right] (3.center) to (b.center);
\draw[string,out=270, in=90] (1.center) to (b.center);
\draw[string,out=90,in=270] (1.center) to (o.center);
\draw[string,in=90,out=270] (2.center) to (i1.center);
\draw[string,in=90,out=270] (3.center) to (i2.center);
\end{tikzpicture}\end{aligned}
\quad - \quad
\begin{aligned}\begin{tikzpicture}[scale=0.65]
\node(b)[reddot] at (0,0.25){};
\node(1) at (0,0.5){};
\node[reddot](2) at (-0.5,-0.25){};
\node[reddot](3) at (0.5,-0.25){};
\node(i1) at (-0.5,-1.5){};
\node[blackdot](i2) at (0.5,-1.5){};
\node(o) at (0,1.5){};
\draw[string,out=270, in=90] (1.center) to (b.center);
\draw[string,out=90,in=270] (1.center) to (o.center);
\draw[string,in=90,out=270] (2.center) to (i1.center);
\draw[string,in=90,out=270] (3.center) to (i2.center);
\end{tikzpicture}\end{aligned}
\\
\super{\eqref{eq:un}\eqref{eq:bialg}\eqref{eq:projprop}}\Rightarrow
\begin{aligned}\begin{tikzpicture}[scale=0.65]
\node(b)[greendot] at (0,0){};
\node(1)[whitedot,projdot] at (0,0.5){};
\node(2)[whitedot,projdot] at (-0.5,-0.5){};
\node(3) at (0.5,-0.5){};
\node(i1) at (-0.5,-1.5){};
\node[greendot](i2) at (0.5,-1.5){};
\node(o) at (0,1.5){};
\draw[green][string,line width=1pt,out=90, in=left] (2.center) to (b.center);
\draw[green][string,line width=1pt,out=90, in=right] (3.center) to (b.center);
\draw[green][string,line width=1pt,out=270, in=90] (1.center) to (b.center);
\draw[string,out=90,in=270] (1.center) to (o.center);
\draw[string,in=90,out=270] (2.center) to (i1.center);
\draw[green][string,line width=1pt,in=90,out=270] (3.center) to (i2.center);
\end{tikzpicture}\end{aligned}
\quad \super{\eqref{eq:un}}=& \quad
\begin{pic}[scale=0.65]
           \node(4) at (0,-2){}; 
           \node(3)[whitedot,projdot] at (0,0){};
           \node (1)[whitedot,projdot] at (0,-1) {}; 
           \node (2) at (0,1) {}; 
           \draw[string] (2.center) to (3.center);
           \draw[string] (4.center) to (1.center); 
           \draw[green][string,line width=1pt] (3.center) to (1.center);           \end{pic}\quad = 
\begin{aligned}\begin{tikzpicture}[scale=0.65]
\node(b) at (0,0){};
\node(1) at (0,0.5){};
\node(2) at (-0.5,-0.5){};
\node(3) at (0.5,-0.5){};
\node(i1) at (0,-1.5){};
\node(o) at (0,1.5){};
\draw[string,out=270, in=90] (1.center) to (b.center);
\draw[string,out=90,in=270] (1.center) to (o.center);
\draw[string] (b.center) to (i1.center);
\end{tikzpicture}\end{aligned}
\quad - \quad
\begin{aligned}\begin{tikzpicture}[scale=0.65]
\node(b)[reddot] at (0,0.25){};
\node(1) at (0,0.5){};
\node[reddot](2) at (0,-0.25){};
\node(i1) at (0,-1.5){};
\node(o) at (0,1.5){};
\draw[string,out=270, in=90] (1.center) to (b.center);
\draw[string,out=90,in=270] (1.center) to (o.center);
\draw[string,in=90,out=270] (2.center) to (i1.center);
\end{tikzpicture}\end{aligned}
\end{align*} 
\end{proof}
In light of this lemma we will again use $*$ to denote this projector.
\begin{equation*}
\begin{aligned}
\begin{tikzpicture}
\node(1) at (0,0.75){};
\node(3) at (0,-0.75){};
\node(b1)[whitedot,scale=1,inner sep=0.25pt] at (0,0){$\projs$};
\draw (b1.center) to (1.center);
\draw (b1.center) to (3.center);
\end{tikzpicture}\end{aligned}
\quad := \begin{aligned}\begin{tikzpicture}[scale=0.5]
           \node(4) at (0,-2){}; 
           \node(3)[whitedot,projdot] at (0,0){};
           \node (1)[whitedot,projdot] at (0,-1) {}; 
           \node (2) at (0,1) {}; 
           \draw[string] (2.center) to (3.center);
           \draw[string] (4.center) to (1.center); 
           \draw[green][string,line width=1pt] (3.center) to (1.center);           
\end{tikzpicture}
\end{aligned} 
\end{equation*}
We now introduce the multiplication acting on the subspace $\mathcal H'$ which we represent as $\tinymultgr[yellowdot]$. We require that $\tinymultgr[yellowdot]$ is a $\dagger$-qSCFA, and that $\tinymultgr[yellowdot]$ and $\tinymultgr$ are a strongly complementary bialgebra. Thus $\tinymultgr[yellowdot]$ is an abelian group on the green basis states. This also tells us that the yellow unit is a green basis state and thus isomorphic to a black basis state not equal to the red unit. We corrupt notation slightly to represent this state as $\tinyunit[yellowdot]$. We denote the multiplicative character Fourier Hadamard matrix as $\psi$ formally defined as follows:
\begin{equation}\label{eq:psi}
\begin{aligned}\begin{tikzpicture}
   \node(r)[yellowdot]{};
   \node(b) at(0,0.5){};
   \draw[string,green, line width=1pt] (r.center) to (b.center);
   \draw[string,green, line width=1pt] (0,1) to (b.center);
   \draw[string,green, line width=1pt,out=90,in=left] (-0.5,-1) to (r.center);
   \draw[string,green, line width=1pt,out=90,in=right] (0.5,-1) to (r.center);
\end{tikzpicture}
\end{aligned} 
\quad = \quad \frac{1}{d-1}
\begin{aligned}\begin{tikzpicture}
   \node(r)[blackdot]{};
      \node(b)[morphism,wedge,hflip,scale=0.5] at(0,0.5){\large$\psi$};
   \node(b1)[morphism,wedge,scale=0.5] at(-0.4,-0.5){\large$\psi$};
   \node(b2)[morphism,wedge,scale=0.5] at(0.4,-0.5){\large$\psi$};
   \draw[string,green, line width=1pt] (r.center) to (b.south);
   \draw[string,green, line width=1pt] (b.north) to (0,1);
   \draw[string,green, line width=1pt,out=90,in=270] (-0.4,-1) to (b1.south);
   \draw[string,green, line width=1pt,out=90,in=270] (0.4,-1) to (b2.south);
   \draw[string,green, line width=1pt,out=left,in=90] (r.center) to (b1.north);
   \draw[string,green, line width=1pt,out=right,in=90] (r.center) to (b2.north);
\end{tikzpicture}
\end{aligned} 
\end{equation}  
We now  introduce the multiplication on the whole Hilbert space $\mathcal H$, which we denote $\tinymult[yellowdot]$. We define $\tinymult[yellowdot]$ and $\tinyunit[yellowdot]$ as follows:
\begin{equation}\label{eq:yellowdef}
\begin{aligned}\begin{tikzpicture}[scale=0.65]
\node(b)[yellowdot] at (0,0){};
\node(1) at (0,0.5){};
\node(2) at (-0.5,-0.5){};
\node(3) at (0.5,-0.5){};
\node(i1) at (-0.5,-1.5){};
\node(i2) at (0.5,-1.5){};
\node(o) at (0,1.5){};
\draw[string,out=90, in=left] (2.center) to (b.center);
\draw[string,out=90, in=right] (3.center) to (b.center);
\draw[string,out=270, in=90] (1.center) to (b.center);
\draw[string,out=90,in=270] (1.center) to (o.center);
\draw[string,in=90,out=270] (2.center) to (i1.center);
\draw[string,in=90,out=270] (3.center) to (i2.center);
\end{tikzpicture}\end{aligned}
:= 
\begin{aligned}\begin{tikzpicture}[scale=0.65]
\node(b)[yellowdot] at (0,0){};
\node(1)[whitedot,projdot] at (0,0.5){};
\node(2)[whitedot,projdot] at (-0.5,-0.5){};
\node(3)[whitedot,projdot] at (0.5,-0.5){};
\node(i1) at (-0.5,-1.5){};
\node(i2) at (0.5,-1.5){};
\node(o) at (0,1.5){};
\draw[green][string, line width=1pt,out=90, in=left] (2.center) to (b.center);
\draw[green][string, line width=1pt,out=90, in=right] (3.center) to (b.center);
\draw[green][string, line width=1pt,out=270, in=90] (1.center) to (b.center);
\draw[string,out=90,in=270] (1.center) to (o.center);
\draw[string,in=90,out=270] (2.center) to (i1.center);
\draw[string,in=90,out=270] (3.center) to (i2.center);
\end{tikzpicture}\end{aligned}
\quad +\quad
\begin{aligned}\begin{tikzpicture}[scale=0.65]
\node(b)[reddot] at (0,0.25){};
\node(1) at (0,0.5){};
\node[reddot](2) at (-0.5,-0.25){};
\node[greendot](3) at (0.5,-0.25){};
\node(i1) at (-0.5,-1.5){};
\node(i2) at (0.5,-1.5){};
\node(o) at (0,1.5){};
\node(p)[whitedot,projdot] at (0.5,-1){};
\draw[string,out=270, in=90] (1.center) to (b.center);
\draw[string,out=90,in=270] (1.center) to (o.center);
\draw[string,in=90,out=270] (2.center) to (i1.center);
\draw[green][string, line width=1pt,in=90,out=270] (3.center) to (p.center);
\draw[string,in=90,out=270] (0.5,-1) to (i2.center);
\end{tikzpicture}\end{aligned}
\quad +\quad
\begin{aligned}\begin{tikzpicture}[scale=0.65]
\node(b)[reddot] at (0,0.25){};
\node(1) at (0,0.5){};
\node[greendot](2) at (-0.5,-0.25){};
\node[reddot](3) at (0.5,-0.25){};
\node(i1) at (-0.5,-1.5){};
\node(i2) at (0.5,-1.5){};
\node(o) at (0,1.5){};
\node(p)[whitedot,projdot] at (-0.5,-1){};
\draw[string,out=270, in=90] (1.center) to (b.center);
\draw[string,out=90,in=270] (1.center) to (o.center);
\draw[green][string, line width=1pt,in=90,out=270] (2.center) to (p.center);
\draw[string,in=90,out=270] (-0.5,-1) to (i1.center);
\draw[string,in=90,out=270] (3.center) to (i2.center);
\end{tikzpicture}\end{aligned}
\quad +\quad
\begin{aligned}\begin{tikzpicture}[scale=0.65]
\node(b)[reddot] at (0,0.25){};
\node(1) at (0,0.5){};
\node[reddot](2) at (-0.5,-0.25){};
\node[reddot](3) at (0.5,-0.25){};
\node(i1) at (-0.5,-1.5){};
\node(i2) at (0.5,-1.5){};
\node(o) at (0,1.5){};
\draw[string,out=270, in=90] (1.center) to (b.center);
\draw[string,out=90,in=270] (1.center) to (o.center);
\draw[string,in=90,out=270] (2.center) to (i1.center);
\draw[string,in=90,out=270] (3.center) to (i2.center);
\end{tikzpicture}\end{aligned}
\end{equation}
\begin{equation}\label{eq:yellowdef}
\begin{aligned}\begin{tikzpicture}[scale=0.65]
\node(b)[yellowdot] at (0,0){};
\node(1) at (0,0.5){};
\node(o) at (0,1.5){};
\draw[string,out=270, in=90] (1.center) to (b.center);
\draw[string,out=90,in=270] (1.center) to (o.center);\end{tikzpicture}\end{aligned}
:= 
\begin{aligned}\begin{tikzpicture}[scale=0.65]
\node(b)[greendot] at (0,0){};
\node(1)[whitedot,projdot] at (0,0.5){};
\node(o) at (0,1.5){};
\draw[green][string, line width=1pt,out=270, in=90] (1.center) to (b.center);
\draw[string,out=90,in=270] (1.center) to (o.center);
\end{tikzpicture}\end{aligned}\end{equation}
This ensures that $\tinymult[yellowdot]$ agrees with $\tinymultgr[yellowdot]$ on the subspace isomorphic to $\mathcal H'$.  The linear map $\tinymult[yellowdot]$ is associative, commutative and  unital with unit $\tinyunit[yellowdot]$,  as can easily be proven from the axioms and the definition of $\tinymult[yellowdot]$. We also require that $\tinymult[yellowdot]$ and $\tinymult$ form a bialgebra (this implies the requirement already made that $\tinymultgr[yellowdot]$ and $\tinymultgr$ form a bialgebra). Although $\tinymult[yellowdot]$ and $\tinymult$ are not strongly complementary following condition can easily be derived from the definitions:
 \begin{equation}\label{diag:last}
 \begin{pic}[yscale=0.5,xscale=2/3]
\node[blackdot] (1) at (0,0) {};
\node[blackdot] (2) at (0,2) {};
\node[yellowdot] (b) at (0.625,0.5) {};
\node[yellowdot] (a) at (0.625,1.5) {};
\node (a') at (1.25,-1) {};
\node (b') at (1.25,3) {};
\node (ga)[whitedot,projdot] at (0,-0.5) {};
\node (g1)[whitedot,projdot] at (0,2.5) {};
\draw[string,out=180,in=180] (1.center) to (2.center);
\draw[string,out=0,in=180] (1.center) to (b.center);
\draw[string,out=0,in=180] (2.center) to (a.center);
\draw[string] (a.center) to (b.center);
\draw[string] (0,-1) to (1.center);
\draw[string] (0,3) to (2.center);
\draw[green][string, line width=1pt] (0,-1) to (ga.center);
\draw[green][string, line width=1pt] (0,3) to (g1.center);
\draw[string,out=90,in=0] (a'.center) to (b.center);
\draw[string,out=270,in=0] (b'.center) to (a.center);
\end{pic}
\quad=\quad
\begin{pic}[yscale=0.5,xscale=2/3]
\draw[green][string, line width=1pt] (0,-1) to (0,3);
\node (a) at (0.75,-1) {};
\node (b) at (0.75,3) {};
\draw[string] (a.center) to (b.center);
\end{pic}
 \end{equation}

\paragraph{Distributivity.}Finally we relate $\tinymult[reddot]$ and $\tinymult[yellowdot]$ as follows.
\begin{definition}[Left distributivity]
Let $\tinymult[reddot]$ and $\tinymult[yellowdot]$ each form a bialgebra with $\tinymult$ Yellow \textit{left distributes over} red if the following equation holds:
\begin{equation}\label{2}
\begin{aligned}\begin{tikzpicture}[scale=1]
          \node (0a) at (-1,0.25) {};
          \node (0b) at (-0.5,0.25) {};
          \node (0c) at (0.5,0.25) {};

          \node[reddot] (1) at (0,1) {};
          \node[yellowdot] (2) at (-0.5,1.5) {};

          \node (5a) at (0-0.5,2) {};

          \draw[string, out=90, in =180] (0a.center) to (2.center);
          \draw[string, out=0, in=90] (2.center) to (1.center);
          \draw[string, out=0, in=90] (1.center) to (0c.center);
          \draw[string, out=180, in=90] (1.center) to (0b.center);

          \draw[string] (2.center) to (5a.center);
         
          \end{tikzpicture}\end{aligned} 
\quad = \quad
\begin{aligned}\begin{tikzpicture}[yscale=0.5,xscale=0.65]
           \node[blackdot] (2) at (0,0.9) {};
           \node (3) at (1.75,0.25) {};
           \node (6) at (0.75,0.25) {};
           \node[yellowdot] (4) at (0,2.1) {};
           \node[yellowdot] (5) at (1,2.1) {};
           \node[reddot](7) at (0.5,2.7) {};
           \draw[string] (0,0.25) to (2.center);
           \draw[string](7.center) to (0.5,3.5);

           \draw[string, in=180, out=180, looseness=1.2] (2.center) to (4.center);
           \draw[string, in=0, out=90, looseness=1] (3.center) to (5.center);
           \draw[string, in=180, out=right] (2.center) to (5.center);
           \draw[string, in=0, out=90] (6.center) to (4.center);
           \draw[string, out=0, in=90] (7.center) to (5.center);
           \draw[string, out=180, in=90] (7.center) to (4.center);
    \end{tikzpicture}\end{aligned}
\ignore{    \quad , \quad
\begin{aligned}\begin{tikzpicture}[xscale=-1]
          \node (0a) at (-1,0.25) {};
          \node (0b) at (-0.5,0.25) {};
          \node (0c) at (0.5,0.25) {};

          \node[reddot] (1) at (0,1) {};
          \node[yellowdot] (2) at (-0.5,1.5) {};

          \node (5a) at (0-0.5,2) {};

          \draw[string, out=90, in =180] (0a.center) to (2.center);
          \draw[string, out=0, in=90] (2.center) to (1.center);
          \draw[string, out=0, in=90] (1.center) to (0c.center);
          \draw[string, out=180, in=90] (1.center) to (0b.center);

          \draw[string] (2.center) to (5a.center);
         
          \end{tikzpicture}\end{aligned} 
\quad = \quad
\begin{aligned}\begin{tikzpicture}[yscale=0.5,xscale=-0.65]
           \node[blackdot] (2) at (0,0.9) {};
           \node (3) at (1.75,0.25) {};
           \node (6) at (0.75,0.25) {};
           \node[yellowdot] (4) at (0,2.1) {};
           \node[yellowdot] (5) at (1,2.1) {};
           \node[reddot] (7) at (0.5,2.7) {};
           \draw[string] (0,0.25) to (2.center);
           \draw[string](7.center) to (0.5,3.5);

           \draw[string, in=180, out=180, looseness=1.2] (2.center) to (4.center);
           \draw[string, in=0, out=90, looseness=1] (3.center) to (5.center);
           \draw[string, in=180, out=right] (2.center) to (5.center);
           \draw[string, in=0, out=90] (6.center) to (4.center);
           \draw[string, out=0, in=90] (7.center) to (5.center);
           \draw[string, out=180, in=90] (7.center) to (4.center);
    \end{tikzpicture}\end{aligned}
    }\end{equation} 
\end{definition} 
\textit{Right distributivity} is defined by reflecting both sides of equation~\eqref{2} in a vertical axis. We now show that right distributivity follows from left distributivity and commutativity.
\begin{lemma}
If $\tinymult[blackdot]$,$\tinymult[reddot]$ and $\tinymult[yellowdot]$ are commutative and yellow left distributes over red, then yellow right distributes over red; so the following equation holds: 
\begin{equation}\label{eq:rdist}
\begin{aligned}\begin{tikzpicture}[xscale=-1]
          \node (0a) at (-1,0.25) {};
          \node (0b) at (-0.5,0.25) {};
          \node (0c) at (0.5,0.25) {};

          \node[reddot] (1) at (0,1) {};
          \node[yellowdot] (2) at (-0.5,1.5) {};

          \node (5a) at (0-0.5,2) {};

          \draw[string, out=90, in =180] (0a.center) to (2.center);
          \draw[string, out=0, in=90] (2.center) to (1.center);
          \draw[string, out=0, in=90] (1.center) to (0c.center);
          \draw[string, out=180, in=90] (1.center) to (0b.center);

          \draw[string] (2.center) to (5a.center);
         
          \end{tikzpicture}\end{aligned} 
\quad = \quad
\begin{aligned}\begin{tikzpicture}[yscale=0.5,xscale=-0.65]
           \node[blackdot] (2) at (0,0.9) {};
           \node (3) at (1.75,0.25) {};
           \node (6) at (0.75,0.25) {};
           \node[yellowdot] (4) at (0,2.1) {};
           \node[yellowdot] (5) at (1,2.1) {};
           \node[reddot] (7) at (0.5,2.7) {};
           \draw[string] (0,0.25) to (2.center);
           \draw[string](7.center) to (0.5,3.5);

           \draw[string, in=180, out=180, looseness=1.2] (2.center) to (4.center);
           \draw[string, in=0, out=90, looseness=1] (3.center) to (5.center);
           \draw[string, in=180, out=right] (2.center) to (5.center);
           \draw[string, in=0, out=90] (6.center) to (4.center);
           \draw[string, out=0, in=90] (7.center) to (5.center);
           \draw[string, out=180, in=90] (7.center) to (4.center);
    \end{tikzpicture}\end{aligned}
\end{equation} \end{lemma}
\begin{proof}
\begin{align*}
\begin{aligned}\begin{tikzpicture}[xscale=-1]
          \node (0a) at (-1.2,-0.25) {};
          \node (0b) at (-0.5,-0.25) {};
          \node (0c) at (0.5,-0.25) {};
          \node[reddot] (1) at (0,1) {};
          \node[yellowdot] (2) at (-0.5,1.5) {};
          \node (5a) at (-0.5,2.5) {};
          \draw[string, out=90, in =180] (0a.center) to (2.center);
          \draw[string, out=0, in=90] (2.center) to (1.center);
          \draw[string, out=0, in=90] (1.center) to (0c.center);
          \draw[string, out=180, in=90] (1.center) to (0b.center);
          \draw[string] (2.center) to (5a.center);
          \end{tikzpicture}\end{aligned} 
 \super{\eqref{eq:comm}}= 
\begin{aligned}\begin{tikzpicture}[xscale=-1]
          \node (0a) at (-1.2,-0.25) {};
          \node (0b) at (-0.5,-0.25) {};
          \node (0c) at (0.5,-0.25) {};
          \node[reddot] (1) at (0,1) {};
          \node (1l) at (0.25,0.75) {};
          \node (1r) at (-0.25,0.75) {};
          \node[yellowdot] (2) at (-0.5,1.5) {};
          \node (2l) at (-0.25,1.25) {};
          \node (2r) at (-0.75,1.325) {};
          \node (5a) at (-0.5,2.5) {};
          \draw[string, out=90, in =270] (0a.center) to (2l.center);
          \draw[string, out=90, in =0] (2l.center) to (2.center);
          \draw[string, out=180, in=90] (2.center) to (2r.center);
          \draw[string, out=270, in=90] (2r.center) to (1.center);
          \draw[string, out=270, in=90] (1r.center) to (0c.center);
          \draw[string, out=180, in=90] (1.center) to (1r.center);
          \draw[string, out=0, in=90] (1.center) to (1l.center);
          \draw[string, out=270, in=90] (1l.center) to (0b.center);
          \draw[string] (2.center) to (5a.center);
          \end{tikzpicture}\end{aligned} = \begin{aligned}\begin{tikzpicture}[yscale=0.5]
          \node (0a) at (-1,0.25) {};
          \node (0b) at (-0.5,0.25) {};
          \node (0c) at (0.5,0.25) {};         
          \node[reddot] (1) at (0,1) {};
          \node[yellowdot] (2) at (-0.5,1.5) {};
          \node (5a) at (-0.5,3) {};
          \draw[string, out=90, in =180] (0a.center) to (2.center);
          \draw[string, out=0, in=90] (2.center) to (1.center);
          \draw[string, out=0, in=90] (1.center) to (0c.center);
          \draw[string, out=180, in=90] (1.center) to (0b.center);
          \draw[string] (2.center) to (5a.center);
          \node (0a7) at (-1,-2.5) {};
          \node (0b7) at (-0.5,-2.5) {};
          \node (0c7) at (0.5,-2.5) {};
          \draw[string, out=90, in=270] (0a7.center) to (0c.center);
          \draw[string, out=90, in=270] (0b7.center) to (0b.center);
          \draw[string, out=90, in=270] (0c7.center) to (0a.center);
          \end{tikzpicture}\end{aligned}  \super{\eqref{2}}= 
\begin{aligned}\begin{tikzpicture}[yscale=0.25,xscale=0.65]
           \node[blackdot] (2) at (0,0.25) {};
           \node (3) at (1.75,0.25) {};
           \node (6) at (0.75,0.25) {};
           \node[yellowdot] (4) at (0,1.75) {};
           \node[yellowdot] (5) at (1,1.75) {};
           \node[reddot](7) at (0.5,2.7) {};
           \node(out) at (0.5,5.5) {};
           \draw[string] (0,0.25) to (2.center);
           \draw[string](7.center) to (out.center);
           \draw[string, in=180, out=180, looseness=1.2] (2.center) to (4.center);
           \draw[string, in=0, out=90, looseness=1] (3.center) to (5.center);
           \draw[string, in=180, out=right] (2.center) to (5.center);
           \draw[string, in=0, out=90] (6.center) to (4.center);
           \draw[string, out=0, in=90] (7.center) to (5.center);
           \draw[string, out=180, in=90] (7.center) to (4.center);
          \node (0a7) at (0,-5.5) {};
          \node (0b7) at (0.75,-5.5) {};
          \node (0c7) at (2.5,-5.5) {};
          \draw[string, out=90, in=270,looseness=0.75] (0a7.center) to (3.center);
          \draw[string, out=90, in=270] (0b7.center) to (6.center);
          \draw[string, out=90, in=270,looseness=0.75] (0c7.center) to (0,0.25);
    \end{tikzpicture}\end{aligned} \super{\eqref{eq:comm}}= 
\begin{aligned}\begin{tikzpicture}[yscale=0.25,xscale=0.65]
           \node[blackdot] (2) at (0,-2.25) {};
           \node(2l) at (-0.35,-1.75) {};
           \node(2r) at (0.35,-1.5) {};
           \node (3) at (1.35,-0.25) {};
           \node (6) at (0.75,-0.75) {};
           \node[yellowdot] (4) at (0,1.5) {};
           \node (4u) at (0,1.9) {};
           \node (4l) at (-0.35,1) {};
           \node (4r) at (0.35,1) {};
           \node[yellowdot] (5) at (1,1.5) {};
           \node (5u) at (1,1.9) {};
           \node(5l) at (0.65,1) {};
           \node(5r) at (1.35,1) {};
           \node[reddot](7) at (0.5,4) {};
           \node(7l) at (0.15,3.5) {};
           \node(7r) at (0.85,3.5) {};
           \node(m) at (-0.35,-0.5) {};
           \node(m') at (0.625,-0.5) {};
           \draw[string] (0,-2) to (2.center);
           \draw[string](7.center) to (0.5,5.5);
           \draw[string, in=270, out=90, looseness=1.2] (2r.center) to (m.center);
           \draw[string, in=270, out=90, looseness=1.2] (m.center) to (4r.center);
           \draw[string, in=0, out=270, looseness=1.2] (2r.center) to (2.center);
           \draw[string, in=90, out=0, looseness=1.2] (4.center) to (4r.center);
           \draw[string, in=270, out=90, looseness=1] (3.center) to (5l.center);
           \draw[string, in=90, out=180, looseness=1] (5.center) to (5l.center);
           \draw[string, in=270, out=90] (2l.center) to (5r.center);
           \draw[string, in=180, out=270] (2l.center) to (2.center);
           \draw[string, in=90, out=0] (5.center) to (5r.center);
           \draw[string, in=270, out=90] (6.center) to (4l.center);
           \draw[string, in=90, out=180] (4.center) to (4l.center);
           \draw[string, out=270, in=90] (7l.center) to (5u.center);
           \draw[string, out=270, in=90] (5u.center) to (5.center);
           \draw[string, out=180, in=90] (7.center) to (7l.center);
           \draw[string, out=270, in=90] (7r.center) to (4u.center);
           \draw[string, out=270, in=90] (4u.center) to (4.center);
           \draw[string, out=0, in=90] (7.center) to (7r.center);
          \node (0a7) at (-0.25,-5.5) {};
          \node (0b7) at (0.75,-5.5) {};
          \node (0c7) at (2.25,-5.5) {};
          \draw[string, out=90, in=270,looseness=0.75] (0a7.center) to (3.center);
          \draw[string, out=90, in=270] (0b7.center) to (6.center);
          \draw[string, out=90, in=270,looseness=0.75] (0c7.center) to (0,-2);
    \end{tikzpicture}\end{aligned} = \begin{aligned}\begin{tikzpicture}[yscale=0.5,xscale=-0.65]
           \node[blackdot] (2) at (0,0.9) {};
           \node (3) at (2,-0.75) {};
           \node (6) at (1,-0.75) {};
           \node[yellowdot] (4) at (0,2.1) {};
           \node[yellowdot] (5) at (1,2.1) {};
           \node[reddot] (7) at (0.5,2.7) {};
           \draw[string] (0,-0.75) to (2.center);
           \draw[string](7.center) to (0.5,4.5);
           \draw[string, in=180, out=180, looseness=1.2] (2.center) to (4.center);
           \draw[string, in=0, out=90, looseness=1] (3.center) to (5.center);
           \draw[string, in=180, out=right] (2.center) to (5.center);
           \draw[string, in=0, out=90] (6.center) to (4.center);
           \draw[string, out=0, in=90] (7.center) to (5.center);
           \draw[string, out=180, in=90] (7.center) to (4.center);
    \end{tikzpicture}\end{aligned}
\end{align*}
\end{proof}
\paragraph{Additive characters.}We also require the following interaction between the yellow unit and $\chi$:
\begin{equation}\label{1}
\begin{pic}
\node(1)[yellowdot] at (0,2){};
\node(2)[morphism,wedge,scale=0.5] at (0,1.5){\large$\chi$};
\node(3)[yellowdot] at (0,1){};
\node(b)[blackdot] at (0.5,0.75){};
\node(i) at (-0.5,0){};
\node(o) at (1,2.25){};
\draw (1.center) to (2.north);
\draw (2.south) to (3.center);
\draw[string,in=left,out=90] (i.center) to (3.center);
\draw[string,in=left,out=right] (3.center) to (b.center);
\draw[string,in=270,out=right] (b.center) to (o.center);
\end{pic}
\quad = \quad
\begin{pic}
\node(1)[] at (0,2.25){};
\node(2)[morphism,wedge,scale=0.5] at (0,1.125){\large$\chi$};
\node(3) at (0,0){};
\draw (1.center) to (2.north);
\draw (2.south) to (3.center);
\end{pic}
\end{equation} 
This corresponds to the the following algebraic equation which can be recovered by composing by computational basis states: $\chi_a(b)=\chi_1(a\tinydot[yellowdot] b)$.
\begin{definition}[Complex finite field ]
Given a $d$-dimensional Hilbert space represented by black wires and a $d-1$-dimensional  Hilbert space represented by green wires,  \textit{a complex finite field } is a pair of $\dfrob$s $\tinymult$ and $\tinymultgr$,  a pair of $\dagger$-qSCFAs $\tinymult[reddot]$ and $\tinymultgr[yellowdot]$ as well as $\tinymult[yellowdot]$ as defined by equation~\eqref{eq:yellowdef}, linear maps $\chi$ and $\psi$ defined by equations~\eqref{eq:addhad} and~\eqref{eq:psi}, linear maps $p$ and $\iota$ defined by equation~\eqref{eq:projj} and obeying equations~\eqref{eq:proj} such that equations~\eqref{2} and~\eqref{1} hold. We denote a complex finite field $(\tinymult[reddot]$, $\tinymult$ , $\tinymultgr[yellowdot]$, $\tinymult[yellowdot]$, $\tinymultgr$,  $\chi$, $\psi$).
\end{definition}
 We summarize the results of this subsection in the following theorem.
\begin{theorem}
Given a compex finite field  $(\tinymult[reddot]$, $\tinymult$ , $\tinymultgr[yellowdot]$, $\tinymult[yellowdot]$, $\tinymultgr$,  $\chi$, $\psi$), $\tinymult[reddot]$ and $\tinymult[yellowdot]$ are the linear extension of the addition and multiplication respectively of a finite field with the underlying set of elements given by the states copyable by $\tinymult$. $\chi$ and  $\psi$ are the complex Fourier Hadamards for the additive and multiplicative group respectively.
\end{theorem}
\begin{proof}
The binary operator $\tinymult[reddot]$ forms an abelian group on the states copyable by $\tinymult$, which is the first axiom of Definition~\ref{def:algff}.
. On the subspace of $\mathcal H$ isomorphic to $\mathcal H'$ which is spanned by the non-zero black states $\tinymult[yellowdot]$ agrees with $\tinymultgr[yellowdot]$ and thus forms an abelian group, thus fulfilling the second axiom of Definition~\ref{def:algff}. Equation~\eqref{2} is precisely the linear extension of distributivity, the third axiom of Definition~\ref{def:algff}. The properties of $\chi$ and $\psi$ were proven by Gogioso and Zeng~\cite{stefwill}.
\end{proof}
\subsection{A construction of $d+1$ MUBs}\label{section:construction}
We now give an application of the complex finite fields developed in the previous subsection to the problem of constructing maximal families of MUBs. First we present two lemmas which will be necessary to proving the main result of this section.
 \begin{lemma}
Given a complex finite field  $(\tinymult[reddot]$, $\tinymult$ , $\tinymultgr[yellowdot]$, $\tinymult[yellowdot]$, $\tinymultgr$,  $\chi$, $\psi$), the following equation holds:
\begin{equation}\label{eqs:lem1}
\begin{pic}
   \node(r1)[reddot]at (0,0.5){};
   \node(r2)[reddot]at (-0.75,0){};
   \node(y1)[yellowdot]at (0.75,0){};
   \node(y2)[yellowdot]at (0.5,-0.5){};
   \node(b)[blackdot]at (1,-0.75){};
   \node(i1) at (-1.25,-1){};
   \node(i2) at (0,-1){};
   \node(i3) at (0.25,-1){};
   \node(i4) at (1,-1){};
   \node(o) at (0,1){};
   \draw[string] (o.center) to (r1.center);
   \draw[string,in=90,out=180] (r1.center) to (r2.center);
   \draw[string,in=90,out=0] (r1.center) to (y1.center);
   \draw[string,in=90,out=180] (r2.center) to (i1.center);
   \draw[string,in=90,out=0] (r2.center) to (y2.center);
   \draw[string,in=90,out=180] (y1.center) to (i2.center);
   \draw[string,in=0,out=0] (y1.center) to (b.center);
   \draw[string,in=90,out=180] (y2.center) to (i3.center);
   \draw[string,in=180,out=0] (y2.center) to (b.center);
   \draw[string] (b.center) to (i4.center);
\end{pic}\quad = \quad
\begin{pic}
   \node(r1)[reddot]at (0,0.5){};
   \node(r2)[reddot]at (-0.75,0){};
   \node(y1)[yellowdot]at (1,0){};
   \node(y2)[yellowdot]at (0.25,-0.25){};
   \node(b)[blackdot]at (1,-0.75){};
   \node(i1) at (-1.25,-1){};
   \node(i2) at (0,-1){};
   \node(i3) at (0.25,-1){};
   \node(i4) at (1,-1){};
   \node(o) at (0,1){};
   \draw[string] (o.center) to (r1.center);
   \draw[string,in=90,out=180] (r1.center) to (r2.center);
   \draw[string,in=90,out=0] (r1.center) to (y1.center);
   \draw[string,in=90,out=180] (r2.center) to (i1.center);
   \draw[string,in=90,out=0] (r2.center) to (y2.center);
   \draw[string,in=90,out=180,looseness=0.8] (y1.center) to (i3.center);
   \draw[string,in=0,out=0] (y1.center) to (b.center);
   \draw[string,in=90,out=180] (y2.center) to (i2.center);
   \draw[string,in=180,out=0] (y2.center) to (b.center);
   \draw[string] (b.center) to (i4.center);
\end{pic}
\end{equation}
\end{lemma}
\begin{proof}
\begin{equation*}\text{LHS}\quad\super{Red\eqref{eq:as}}=
\begin{pic}
   \node(r1)[reddot]at (0,1){};
   \node(r2)[reddot]at (0.5,0.5){};
   \node(y1)[yellowdot]at (0.75,0){};
   \node(y2)[yellowdot]at (0.5,-0.5){};
   \node(b)[blackdot]at (1,-0.75){};
   \node(i1) at (-1.25,-1){};
   \node(i2) at (0,-1){};
   \node(i3) at (0.25,-1){};
   \node(i4) at (1,-1){};
   \node(o) at (0,1.5){};
   \draw[string] (o.center) to (r1.center);
   \draw[string,in=90,out=0] (r1.center) to (r2.center);
   \draw[string,in=90,out=180] (r1.center) to (i1.center);
   \draw[string,in=90,out=0] (r2.center) to (y1.center);
   \draw[string,in=90,out=180] (r2.center) to (y2.center);
   \draw[string,in=90,out=180] (y1.center) to (i2.center);
   \draw[string,in=0,out=0] (y1.center) to (b.center);
   \draw[string,in=90,out=180] (y2.center) to (i3.center);
   \draw[string,in=180,out=0] (y2.center) to (b.center);
   \draw[string] (b.center) to (i4.center);
\end{pic}\quad \super{Red\eqref{eq:comm}}= \quad
\begin{pic}
   \node(r1)[reddot]at (0,1){};
   \node(r2)[reddot]at (0.5,0.5){};
   \node(y1)[yellowdot]at (0.75,-0.25){};
   \node(y2)[yellowdot]at (0.5,-0.5){};
   \node(b)[blackdot]at (1,-0.75){};
   \node(i1) at (-1.25,-1){};
   \node(i2) at (0,-1){};
   \node(i3) at (0.25,-1){};
   \node(i4) at (1,-1){};
   \node(o) at (0,1.5){};
   \draw[string] (o.center) to (r1.center);
   \draw[string,in=90,out=0] (r1.center) to (r2.center);
   \draw[string,in=90,out=180] (r1.center) to (i1.center);
   \draw[string,in=90,out=180,looseness=1.8] (r2.center) to (y1.center);
   \draw[string,in=90,out=0,looseness=1.3] (r2.center) to (y2.center);
   \draw[string,in=90,out=180] (y1.center) to (i2.center);
   \draw[string,in=0,out=0] (y1.center) to (b.center);
   \draw[string,in=90,out=180] (y2.center) to (i3.center);
   \draw[string,in=180,out=0] (y2.center) to (b.center);
   \draw[string] (b.center) to (i4.center);
\end{pic}\quad \super{Red\eqref{eq:as}}= \quad
\begin{pic}
   \node(r1)[reddot]at (0,1){};
   \node(r2)[reddot]at (-0.75,0.5){};
   \node(y1)[yellowdot]at (0.75,-0.25){};
   \node(y2)[yellowdot]at (0.5,-0.5){};
   \node(b)[blackdot]at (1,-0.75){};
   \node(i1) at (-1.25,-1){};
   \node(i2) at (0,-1){};
   \node(i3) at (0.25,-1){};
   \node(i4) at (1,-1){};
   \node(o) at (0,1.5){};
   \draw[string] (o.center) to (r1.center);
   \draw[string,in=90,out=180] (r1.center) to (r2.center);
   \draw[string,in=90,out=0] (r1.center) to (y2.center);
   \draw[string,in=90,out=180] (r2.center) to (i1.center);
   \draw[string,in=90,out=0] (r2.center) to (y1.center);
   \draw[string,in=90,out=180] (y1.center) to (i2.center);
   \draw[string,in=0,out=0] (y1.center) to (b.center);
   \draw[string,in=90,out=180] (y2.center) to (i3.center);
   \draw[string,in=180,out=0] (y2.center) to (b.center);
   \draw[string] (b.center) to (i4.center);
\end{pic}\quad \super{Black\eqref{eq:comm}}=\quad\text{RHS}
\end{equation*}
\end{proof}
\begin{lemma}
Given a complex finite field  $(\tinymult[reddot]$, $\tinymult$ , $\tinymultgr[yellowdot]$, $\tinymult[yellowdot]$, $\tinymultgr$,  $\chi$, $\psi$), the following equation holds:
\begin{equation}\label{eqs:lem2}\begin{pic}
   \node(r1)[reddot] at (-0.75,1){};
   \node(y1)[yellowdot] at (-1.5,0.5){};
   \node(y2)[yellowdot] at (0,0.5){};
   \node(r2)[reddot] at (0.75,0){};
   \node(y3)[yellowdot] at (1.5,-0.5){};
   \node(b1)[blackdot] at (-1.5,-1){};
   \node(b2)[blackdot] at (0.75,-1){};
   \node(i1) at (-1.5,-1.5){};
   \node(i2) at (-0.75,-1.5){};
   \node(i3) at (0.75,-1.5){};
   \node(i4) at (2,-1.5){};
   \node(o) at (-0.75,1.5){};
   \draw[string] (o.center) to (r1.center);
   \draw[string,in=90,out=180] (r1.center) to (y1.center);
   \draw[string,in=90,out=0] (r1.center) to (y2.center);
   \draw[string,in=180,out=180] (y1.center) to (b1.center);
   \draw[string,in=180,out=0] (y1.center) to (b2.center);
   \draw[string,in=90,out=180] (y2.center) to (i2.center);
   \draw[string,in=90,out=0] (y2.center) to (r2.center);
   \draw[string,in=0,out=180] (r2.center) to (b1.center);
   \draw[string,in=90,out=0] (r2.center) to (y3.center);
   \draw[string,in=0,out=180] (y3.center) to (b2.center);
   \draw[string,in=90,out=0] (y3.center) to (i4.center);
   \draw[string] (b1.center) to (i1.center);
   \draw[string] (b2.center) to (i3.center);
\end{pic}\quad = \quad
\begin{pic}
   \node(r1)[reddot] at (0,1){};
   \node(y1)[yellowdot] at (-0.75,0){};
   \node(y2)[yellowdot] at (0.75,0.5){};
   \node(r2)[reddot] at (0,0){};
   \node(y3)[yellowdot] at (0.75,-0.5){};
   \node(b1)[blackdot] at (-0.75,-1){};
   \node(b2)[blackdot] at (0,-1){};
   \node(i1) at (-0.75,-1.5){};
   \node(i2) at (0,-1.5){};
   \node(i3) at (1.5,-1.5){};
   \node(i4) at (2.75,-1.5){};
   \node(o) at (0,1.5){};
   \draw[string] (o.center) to (r1.center);
   \draw[string,out=180,in=90] (r1.center) to (y1.center);
   \draw[string,out=0,in=90] (r1.center) to (y2.center);
   \draw[string,out=180,in=180] (y1.center) to (b1.center);
   \draw[string,out=0,in=180] (y1.center) to (b2.center);
   \draw[string,out=180,in=90] (y2.center) to (r2.center);
   \draw[string,out=0,in=90] (y2.center) to (i3.center);
   \draw[string,out=180,in=0] (r2.center) to (b1.center);
   \draw[string,out=0,in=90] (r2.center) to (y3.center);
   \draw[string,out=180,in=0] (y3.center) to (b2.center);
   \draw[string,out=0,in=90] (y3.center) to (i4.center);
   \draw[string] (b1.center) to (i1.center);
   \draw[string] (b2.center) to (i2.center);
\end{pic}
\end{equation}
\end{lemma}
\begin{proof}
\begin{align*}\text{LHS}\quad \super{\eqref{2}}=&
\begin{pic}
   \node(r1)[reddot] at (-0.75,1){};
   \node(y1)[yellowdot] at (-1.5,0.5){};
   \node(r2)[reddot] at (0,0.5){};
   \node(y2)[yellowdot] at (0.75,0){};
   \node(y3)[yellowdot] at (1.5,-0.5){};
   \node(b1)[blackdot] at (-1.5,-1){};
   \node(b2)[blackdot] at (0.75,-1){};
   \node(y4)[yellowdot] at (-0.75,-0.25){};
   \node(b3)[blackdot] at (-0.75,-1){};
   \node(i1) at (-1.5,-1.5){};
   \node(i2) at (-0.75,-1.5){};
   \node(i3) at (0.75,-1.5){};
   \node(i4) at (2,-1.5){};
   \node(o) at (-0.75,1.5){};
   \draw[string] (o.center) to (r1.center);
   \draw[string,in=90,out=180] (r1.center) to (y1.center);
   \draw[string,in=90,out=0] (r1.center) to (r2.center);
   \draw[string,in=180,out=180] (y1.center) to (b1.center);
   \draw[string,in=180,out=0] (y1.center) to (b2.center);
   \draw[string,in=90,out=180] (r2.center) to (y4.center);
   \draw[string,in=180,out=180] (y4.center) to (b3.center);
   \draw[string,in=90,out=0] (r2.center) to (y2.center);
   \draw[string,in=0,out=0] (y4.center) to (b1.center);
   \draw[string,in=0,out=180] (y2.center) to (b3.center);
   \draw[string,in=90,out=0] (y2.center) to (y3.center);
   \draw[string,in=0,out=180] (y3.center) to (b2.center);
   \draw[string,in=90,out=0] (y3.center) to (i4.center);
   \draw[string] (b1.center) to (i1.center);
   \draw[string] (b2.center) to (i3.center);
   \draw[string] (b3.center) to (i2.center);
\end{pic}\quad \super{Yellow\eqref{eq:sm}}= \quad
\begin{pic}
   \node(r1)[reddot] at (-0.75,1){};
   \node(y1)[yellowdot] at (-1.5,0.5){};
   \node(r2)[reddot] at (0,0.5){};
   \node(y2)[yellowdot] at (0.75,0){};
   \node(y3)[yellowdot] at (0.5,-0.5){};
   \node(b1)[blackdot] at (-1.5,-1){};
   \node(b2)[blackdot] at (0.75,-1){};
   \node(y4)[yellowdot] at (-1.125,-0.25){};
   \node(b3)[blackdot] at (-0.75,-1){};
   \node(i1) at (-1.5,-1.5){};
   \node(i2) at (-0.75,-1.5){};
   \node(i3) at (0.75,-1.5){};
   \node(i4) at (2,-1.5){};
   \node(o) at (-0.75,1.5){};
   \draw[string] (o.center) to (r1.center);
   \draw[string,in=90,out=180] (r1.center) to (y1.center);
   \draw[string,in=90,out=0] (r1.center) to (r2.center);
   \draw[string,in=180,out=180] (y1.center) to (b1.center);
   \draw[string,in=180,out=0] (y1.center) to (b2.center);
   \draw[string,in=90,out=180] (r2.center) to (y4.center);
   \draw[string,in=180,out=0] (y4.center) to (b3.center);
   \draw[string,in=90,out=0] (r2.center) to (y2.center);
   \draw[string,in=0,out=180] (y4.center) to (b1.center);
   \draw[string,in=90,out=180] (y2.center) to (y3.center);
   \draw[string,in=0,out=0] (y2.center) to (b2.center);
   \draw[string,in=90,out=0] (y3.center) to (i4.center);
   \draw[string,in=0,out=180] (y3.center) to (b3.center);
   \draw[string] (b1.center) to (i1.center);
   \draw[string] (b2.center) to (i3.center);
   \draw[string] (b3.center) to (i2.center);
\end{pic}\\ 
\quad \super{Red\eqref{eq:as}}=& \quad\begin{pic}
   \node(r1)[reddot] at (-0.75,1){};
   \node(y1)[yellowdot] at (-1.75,0.25){};
   \node(r2)[reddot] at (-1.25,0.5){};
   \node(y2)[yellowdot] at (0.75,0){};
   \node(y3)[yellowdot] at (0.5,-0.5){};
   \node(b1)[blackdot] at (-1.5,-1){};
   \node(b2)[blackdot] at (0.75,-1){};
   \node(y4)[yellowdot] at (-1.125,-0.25){};
   \node(b3)[blackdot] at (-0.75,-1){};
   \node(i1) at (-1.5,-1.5){};
   \node(i2) at (-0.75,-1.5){};
   \node(i3) at (0.75,-1.5){};
   \node(i4) at (2,-1.5){};
   \node(o) at (-0.75,1.5){};
   \draw[string] (o.center) to (r1.center);
   \draw[string,in=90,out=180] (r1.center) to (r2.center);
   \draw[string,in=90,out=0] (r1.center) to (y2.center);
   \draw[string,in=180,out=180] (y1.center) to (b1.center);
   \draw[string,in=180,out=0] (y1.center) to (b2.center);
   \draw[string,in=90,out=0] (r2.center) to (y4.center);
   \draw[string,in=180,out=0] (y4.center) to (b3.center);
   \draw[string,in=90,out=180] (r2.center) to (y1.center);
   \draw[string,in=0,out=180] (y4.center) to (b1.center);
   \draw[string,in=90,out=180] (y2.center) to (y3.center);
   \draw[string,in=0,out=0] (y2.center) to (b2.center);
   \draw[string,in=90,out=0] (y3.center) to (i4.center);
   \draw[string,in=0,out=180] (y3.center) to (b3.center);
   \draw[string] (b1.center) to (i1.center);
   \draw[string] (b2.center) to (i3.center);
   \draw[string] (b3.center) to (i2.center);
\end{pic} 
\quad \super{Red\eqref{eq:comm}}= \quad
\begin{pic}
   \node(r1)[reddot] at (-0.75,1){};
   \node(y1)[yellowdot] at (-1.75,0){};
   \node(r2)[reddot] at (-1.25,0.75){};
   \node(y2)[yellowdot] at (0.75,0){};
   \node(y3)[yellowdot] at (0.5,-0.5){};
   \node(b1)[blackdot] at (-1.5,-1){};
   \node(b2)[blackdot] at (0.75,-1){};
   \node(y4)[yellowdot] at (-1.125,-0.5){};
   \node(b3)[blackdot] at (-0.75,-1){};
   \node(i1) at (-1.5,-1.5){};
   \node(i2) at (-0.75,-1.5){};
   \node(i3) at (0.75,-1.5){};
   \node(i4) at (2,-1.5){};
   \node(o) at (-0.75,1.5){};
   \node(r2mr) at (-1,0.5){};
   \draw[string] (o.center) to (r1.center);
   \draw[string,in=90,out=180] (r1.center) to (r2.center);
   \draw[string,in=90,out=0] (r1.center) to (y2.center);
   \draw[string,in=180,out=180] (y1.center) to (b1.center);
   \draw[string,in=180,out=0] (y1.center) to (b2.center);
   \draw[string,in=90,out=180] (r2.center) to (y4.center);
   \draw[string,in=180,out=0] (y4.center) to (b3.center);
   \draw[string,in=90,out=0] (r2.center) to (r2mr.center);
   \draw[string,in=90,out=270] (r2mr.center) to (y1.center);
   \draw[string,in=0,out=180] (y4.center) to (b1.center);
   \draw[string,in=90,out=180] (y2.center) to (y3.center);
   \draw[string,in=0,out=0] (y2.center) to (b2.center);
   \draw[string,in=90,out=0] (y3.center) to (i4.center);
   \draw[string,in=0,out=180] (y3.center) to (b3.center);
   \draw[string] (b1.center) to (i1.center);
   \draw[string] (b2.center) to (i3.center);
   \draw[string] (b3.center) to (i2.center);
\end{pic}\\
\quad \super{Red\eqref{eq:as}}=& \quad
\begin{pic}
   \node(r1)[reddot] at (-0.75,1){};
   \node(y1)[yellowdot] at (-1,0){};
   \node(r2)[reddot] at (0,0.75){};
   \node(y2)[yellowdot] at (0.75,0){};
   \node(y3)[yellowdot] at (0.5,-0.5){};
   \node(b1)[blackdot] at (-1.5,-1){};
   \node(b1ml) at (-1.75,-0.75){};
   \node(b1mr) at (-1.25,-0.75){};
   \node(b2)[blackdot] at (0.75,-1){};
   \node(y4)[yellowdot] at (-1.75,0){};
   \node(y4ml) at (-2,-0.25){};
   \node(b3)[blackdot] at (-0.75,-1){};
   \node(i1) at (-1.5,-1.5){};
   \node(i2) at (-0.75,-1.5){};
   \node(i3) at (0.75,-1.5){};
   \node(i4) at (2,-1.5){};
   \node(o) at (-0.75,1.5){};
   \node(r2mr) at (-1,0.5){};
   \draw[string] (o.center) to (r1.center);
   \draw[string,in=90,out=0] (r1.center) to (r2.center);
   \draw[string,in=90,out=180] (r1.center) to (y4.center);
   \draw[string,in=90,out=180] (y1.center) to (b1ml.center);
   \draw[string,in=180,out=270] (b1ml.center) to (b1.center);
   \draw[string,in=180,out=0] (y1.center) to (b2.center);
   \draw[string,in=90,out=180] (r2.center) to (y1.center);
   \draw[string,in=180,out=0] (y4.center) to (b3.center);
   \draw[string,in=90,out=0] (r2.center) to (y2.center);
   \draw[string,in=90,out=180] (y4.center) to (y4ml.center);
   \draw[string,in=90,out=270] (y4ml.center) to (b1mr.center);
   \draw[string,in=0,out=270] (b1mr.center) to (b1.center);
   \draw[string,in=90,out=180] (y2.center) to (y3.center);
   \draw[string,in=0,out=0] (y2.center) to (b2.center);
   \draw[string,in=90,out=0] (y3.center) to (i4.center);
   \draw[string,in=0,out=180] (y3.center) to (b3.center);
   \draw[string] (b1.center) to (i1.center);
   \draw[string] (b2.center) to (i3.center);
   \draw[string] (b3.center) to (i2.center);
\end{pic}\quad \super{\eqref{eq:rdist}}=\quad\text{RHS}
\end{align*}
\end{proof}
We construct a $\UEBM$ as follows:
\begin{theorem}\label{thm:const}
Given a complex finite field  $(\tinymult[reddot]$, $\tinymult$ , $\tinymultgr[yellowdot]$, $\tinymult[yellowdot]$, $\tinymultgr$,  $\chi$, $\psi$) the following is a $\UEBM$:
\begin{equation}\label{eq:mainconst}
\begin{aligned}
\begin{tikzpicture}
\node (f) [morphism,wedge, connect s length=1.95cm width=0.5mm, connect se length=1.95cm, width=1cm, connect n length=1.95cm,connect sw length=1.95cm] at (0,0) {$U_{FF}$};
\node  at (f.connect s) {};
\node  at (f.connect se) {};
\node  at (f.connect n) {};
\end{tikzpicture}
\end{aligned}
:=\quad
\begin{aligned}
\begin{tikzpicture}
\node (chi) [morphism,wedge,scale=0.75] at (0,-0.1) {\large$\chi$};
\node(b1)[blackdot]  at (-0.25,0.5) {};
\node(add)[reddot] at (0.25,1){};
\node(x)[yellowdot] at (1.25,-0.5){};
\node(b2)[blackdot] at (0.5,-1){};
\node(*) at (0.5,-1.5){};
\node(m) at (1.25,-1){};
\node(i1) at (-1,-2.5){};
\node(i2) at (0.5,-2.5){};
\node(i3) at (1.25,-2.5){};
\node(i2a) at (0.5,-1.75){};
\node(i3a) at (1.25,-1.75){};
\node(o) at (0.25,2){};
\node(g1) at (1.25,-2){};
\node(g2)[whitedot,inner sep=0.1pt] at (1.25,-2){$\projs$};
\draw[green][string, line width=1pt](g1.center) to (g2.center);
\draw[string,in=left,out=90,looseness=0.6] (i1.center) to (b1.center);
\draw[string,in=90,out=right] (b1.center) to (chi.north);
\draw[string,in=left,out=90] (b1.center) to (add.center);
\draw[string,in=90,out=right] (add.center) to (x.center);
\draw[string,in=270,out=90] (add.center) to (o.center);
\draw[string,in=270,out=90,looseness=1] (i2.center) to (i2a.center);
\draw[string,in=270,out=90,looseness=1] (i2a.center) to (m.center);
\draw[string,in=right,out=90,looseness=1] (m.center) to (x.center);
\draw[string,in=270,out=90] (i3.center) to (g1.center);
\draw[string,in=270,out=90] (g2.center) to (i3a.center);
\draw[string,in=270,out=90] (i3a.center) to (b2.center);
\draw[string,in=270,out=left] (b2.center) to (chi.south);
\draw[string,in=left,out=right] (b2.center) to (x.center);
\end{tikzpicture}
\end{aligned}
\quad + \quad
\begin{aligned}
\begin{tikzpicture}
\node(b2)[reddot] at (0.6,0.2){};
\node(m) at (1,-0.25){};
\node(s)[reddot] at (1.75,-2){};
\draw[string,out=90,in=right,looseness=0.5] (i3.center) to (b2.center);
\node(i1) at (0,-2.5){};
\node(i2) at (1.2,-2.5){};
\node(i3) at (1.75,-2.5){};
\node(o) at (0.6,2){};
\draw[string,out=90,in=left,looseness=0.5] (i1.center) to (b2.center);
\draw[string] (i3.center) to (s.center);
\draw[string] (o.center) to (b2.center);
\end{tikzpicture}
\end{aligned}
\end{equation}
\end{theorem}
\begin{proof}
We first prove that $U_{FF}$ is a UEB. We do this by showing that $U_{FF}$ is equivalent to a \textit{shift and multiply basis}. First we rearrange equation~\eqref{eq:mainconst}.
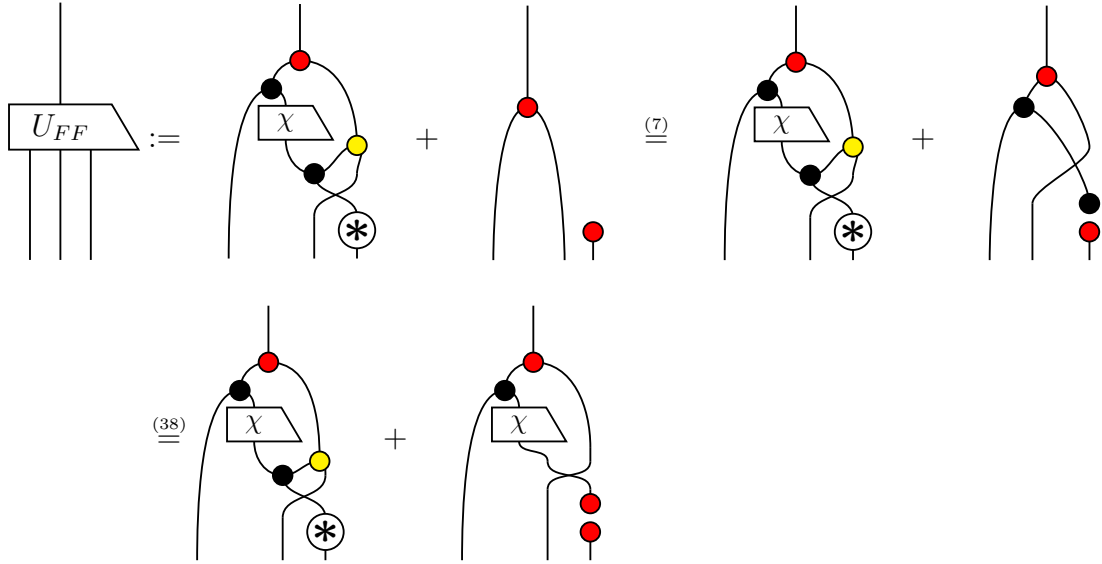
\begin{align*}
\begin{aligned}
\begin{tikzpicture}[scale=0.75]
\node (f) [morphism,wedge, connect s length=1.4625cm width=0.5mm, connect se length=1.4625cm, width=1cm, connect n length=1.35cm,connect sw length=1.4625cm] at (0,0) {$U_{FF}$};
\node  at (f.connect s) {};
\node  at (f.connect se) {};
\node  at (f.connect n) {};
\end{tikzpicture}
\end{aligned}
:=&\quad
\begin{aligned}
\begin{tikzpicture}[scale=0.75]
\node (chi) [morphism,wedge,scale=0.75] at (0,-0.1) {\large$\chi$};
\node(b1)[blackdot]  at (-0.25,0.5) {};
\node(add)[reddot] at (0.25,1){};
\node(x)[yellowdot] at (1.25,-0.5){};
\node(b2)[blackdot] at (0.5,-1){};
\node(*) at (0.5,-1.5){};
\node(m) at (1.25,-1){};
\node(i1) at (-1,-2.5){};
\node(i2) at (0.5,-2.5){};
\node(i3) at (1.25,-2.5){};
\node(i2a) at (0.5,-1.75){};
\node(i3a) at (1.25,-1.75){};
\node(o) at (0.25,2){};
\node(g1)[whitedot,projdot] at (1.25,-2.25){};
\node(g2)[whitedot,inner sep=0.1pt] at (1.25,-2){$\projs$};
\draw[green][string, line width=1pt](g1.center) to (g2.center);
\draw[string,in=left,out=90,looseness=0.6] (i1.center) to (b1.center);
\draw[string,in=90,out=right] (b1.center) to (chi.north);
\draw[string,in=left,out=90] (b1.center) to (add.center);
\draw[string,in=90,out=right] (add.center) to (x.center);
\draw[string,in=270,out=90] (add.center) to (o.center);
\draw[string,in=270,out=90,looseness=1] (i2.center) to (i2a.center);
\draw[string,in=270,out=90,looseness=1] (i2a.center) to (m.center);
\draw[string,in=right,out=90,looseness=1] (m.center) to (x.center);
\draw[string,in=270,out=90] (i3.center) to (g1.center);
\draw[string,in=270,out=90] (g2.center) to (i3a.center);
\draw[string,in=270,out=90] (i3a.center) to (b2.center);
\draw[string,in=270,out=left] (b2.center) to (chi.south);
\draw[string,in=left,out=right] (b2.center) to (x.center);
\end{tikzpicture}
\end{aligned}
\quad + \quad
\begin{aligned}
\begin{tikzpicture}[scale=0.75]
\node(b2)[reddot] at (0.6,0.2){};
\node(m) at (1,-0.25){};
\node(s)[reddot] at (1.75,-2){};
\draw[string,out=90,in=right,looseness=0.5] (i3.center) to (b2.center);
\node(i1) at (0,-2.5){};
\node(i2) at (1.2,-2.5){};
\node(i3) at (1.75,-2.5){};
\node(o) at (0.6,2){};
\draw[string,out=90,in=left,looseness=0.5] (i1.center) to (b2.center);
\draw[string] (i3.center) to (s.center);
\draw[string] (o.center) to (b2.center);
\end{tikzpicture}
\end{aligned}
\quad \super{\eqref{eq:un}}=\quad
\begin{aligned}
\begin{tikzpicture}[scale=0.75]
\node (chi) [morphism,wedge,scale=0.75] at (0,-0.1) {\large$\chi$};
\node(b1)[blackdot]  at (-0.25,0.5) {};
\node(add)[reddot] at (0.25,1){};
\node(x)[yellowdot] at (1.25,-0.5){};
\node(b2)[blackdot] at (0.5,-1){};
\node(*) at (0.5,-1.5){};
\node(m) at (1.25,-1){};
\node(i1) at (-1,-2.5){};
\node(i2) at (0.5,-2.5){};
\node(i3) at (1.25,-2.5){};
\node(i2a) at (0.5,-1.75){};
\node(i3a) at (1.25,-1.75){};
\node(o) at (0.25,2){};
\node(g1)[whitedot,projdot] at (1.25,-2.25){};
\node(g2)[whitedot,inner sep=0.1pt] at (1.25,-2){$\projs$};
\draw[green][string, line width=1pt](g1.center) to (g2.center);
\draw[string,in=left,out=90,looseness=0.6] (i1.center) to (b1.center);
\draw[string,in=90,out=right] (b1.center) to (chi.north);
\draw[string,in=left,out=90] (b1.center) to (add.center);
\draw[string,in=90,out=right] (add.center) to (x.center);
\draw[string,in=270,out=90] (add.center) to (o.center);
\draw[string,in=270,out=90,looseness=1] (i2.center) to (i2a.center);
\draw[string,in=270,out=90,looseness=1] (i2a.center) to (m.center);
\draw[string,in=right,out=90,looseness=1] (m.center) to (x.center);
\draw[string,in=270,out=90] (i3.center) to (g1.center);
\draw[string,in=270,out=90] (g2.center) to (i3a.center);
\draw[string,in=270,out=90] (i3a.center) to (b2.center);
\draw[string,in=270,out=left] (b2.center) to (chi.south);
\draw[string,in=left,out=right] (b2.center) to (x.center);
\end{tikzpicture}
\end{aligned}
\quad + \quad
\begin{aligned}
\begin{tikzpicture}[scale=0.75]
\node(b2)[reddot] at (1,0.75){};
\node(b3)[blackdot] at (0.6,0.2){};
\node(m) at (1,-0.25){};
\node(s)[reddot] at (1.75,-2){};
\node(i1) at (0,-2.5){};
\node(i2) at (0.75,-2.5){};
\node(i3) at (1.75,-2.5){};
\node(o) at (1,2){};
\node(k) at (1.75,-0.5){};
\node(r)[blackdot] at (1.75,-1.5){};
\node(1) at (0.75,-1.5){};
\draw[string,out=90,in=left,looseness=0.5] (i1.center) to (b3.center);
\draw[string,out=90,in=left,looseness=0.5] (b3.center) to (b2.center);
\draw[string] (i3.center) to (s.center);
\draw[string] (o.center) to (b2.center);
\draw[string,out=90,in=270,looseness=0.5] (i2.center) to (1.center);
\draw[string,out=90,in=270,looseness=0.5] (1.center) to (k.center);
\draw[string,out=right,in=90,looseness=0.5] (b2.center) to (k.center);
\draw[string,out=right,in=90,looseness=0.5] (b3.center) to (r.center);
\end{tikzpicture}
\end{aligned}
\\
\super{\eqref{eq:chi}}=&\begin{aligned}
\begin{tikzpicture}[scale=0.75]
\node (chi) [morphism,wedge,scale=0.75] at (0,-0.1) {\large$\chi$};
\node(b1)[blackdot]  at (-0.25,0.5) {};
\node(add)[reddot] at (0.25,1){};
\node(x)[yellowdot] at (1.15,-0.75){};
\node(b2)[blackdot] at (0.5,-1){};
\node(*) at (0.5,-1.5){};
\node(m) at (1.25,-1){};
\node(i1) at (-1,-2.5){};
\node(i2) at (0.5,-2.5){};
\node(i3) at (1.25,-2.5){};
\node(i2a) at (0.5,-1.75){};
\node(i3a) at (1.25,-1.75){};
\node(o) at (0.25,2){};
\node(g1)[whitedot,projdot] at (1.25,-2.25){};
\node(g2)[whitedot,inner sep=0.1pt] at (1.25,-2){$\projs$};
\draw[green][string, line width=1pt](g1.center) to (g2.center);
\draw[string,in=left,out=90,looseness=0.6] (i1.center) to (b1.center);
\draw[string,in=90,out=right] (b1.center) to (chi.north);
\draw[string,in=left,out=90] (b1.center) to (add.center);
\draw[string,in=90,out=right] (add.center) to (x.center);
\draw[string,in=270,out=90] (add.center) to (o.center);
\draw[string,in=270,out=90,looseness=1] (i2.center) to (i2a.center);
\draw[string,in=270,out=90,looseness=1] (i2a.center) to (m.center);
\draw[string,in=right,out=90,looseness=1] (m.center) to (x.center);
\draw[string,in=270,out=90] (i3.center) to (g1.center);
\draw[string,in=270,out=90] (g2.center) to (i3a.center);
\draw[string,in=270,out=90] (i3a.center) to (b2.center);
\draw[string,in=270,out=left] (b2.center) to (chi.south);
\draw[string,in=left,out=right] (b2.center) to (x.center);
\end{tikzpicture}
\end{aligned}
\quad + \quad
\begin{aligned}
\begin{tikzpicture}[scale=0.75]
\node (chi) [morphism,wedge,scale=0.75] at (0,-0.1) {\large$\chi$};
\node(b1)[blackdot]  at (-0.25,0.5) {};
\node(add)[reddot] at (0.25,1){};
\node(x) at (1.25,-0.5){};
\node(b2) at (0.5,-0.75){};
\node(*) at (0.5,-1.5){};
\node(m) at (1.25,-0.75){};
\node(i1) at (-1,-2.5){};
\node(i2) at (0.5,-2.5){};
\node(i3) at (1.25,-2.5){};
\node(i2a) at (0.5,-1.25){};
\node(i3a) at (1.25,-1.25){};
\node(o) at (0.25,2){};
\node[reddot](g1) at (1.25,-2){};
\node[reddot](g2) at (1.25,-1.5){};
\draw[string,in=left,out=90,looseness=0.6] (i1.center) to (b1.center);
\draw[string,in=90,out=right] (b1.center) to (chi.north);
\draw[string,in=left,out=90] (b1.center) to (add.center);
\draw[string,in=90,out=right] (add.center) to (x.center);
\draw[string,in=270,out=90] (add.center) to (o.center);
\draw[string,in=270,out=90,looseness=1] (i2.center) to (i2a.center);
\draw[string,in=270,out=90,looseness=1] (i2a.center) to (m.center);
\draw[string,in=270,out=90,looseness=1] (m.center) to (x.center);
\draw[string,in=270,out=90] (i3.center) to (g1.center);
\draw[string,in=270,out=90] (g2.center) to (i3a.center);
\draw[string,in=270,out=90] (i3a.center) to (b2.center);
\draw[string,in=270,out=90] (b2.center) to (chi.south);
\end{tikzpicture}
\end{aligned}
\end{align*}
We now prove that the following linear map $P$, as defined below, is a permutation.
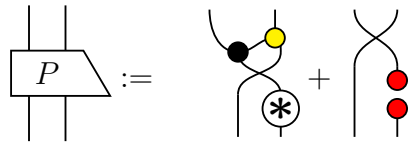
\begin{align*}
\begin{pic}[scale=0.75]
\node(o1) at (-0.32,2.4){};
\node(m) at (0.32,2.4){};
\node(i2) at (-0.32,0){};
\node(i3) at (0.32,0){};
\node(p)[morphism,wedge] at (0,1.2){$P$};
\draw[string] (i2.center) to (p.south west);
\draw[string] (i3.center) to (p.south east);
\draw[string] (m.center) to (p.north east);
\draw[string] (o1.center) to (p.north west);
\end{pic}:= \quad
\begin{aligned}
\begin{tikzpicture}[scale=0.75]
\node(b)[blackdot] at (0.5,-0.75){};
\node(m) at (1.25,-0.75){};
\node(i1) at (0.5,-2.25){};
\node(i2) at (1.25,-2.25){};
\node(i1a) at (0.5,-1.5){};
\node(i2a) at (1.25,-1.5){};
\node(g1)[whitedot,projdot] at (1.25,-2){};
\node(g2)[whitedot,inner sep=0.1pt] at (1.25,-1.75){$\projs$};
\node(o1) at (0,0){};
\node(o2) at (1.15,0){};
\node(y)[yellowdot] at (1.15,-0.5){};
\draw[green][string, line width=1pt](g1.center) to (g2.center);
\draw[string,in=270,out=90,looseness=1] (i1.center) to (i1a.center);
\draw[string,in=270,out=90,looseness=1] (i1a.center) to (m.center);
\draw[string,in=right,out=90,looseness=1] (m.center) to (y.center);
\draw[string,in=270,out=90] (i2.center) to (g1.center);
\draw[string,in=270,out=90] (g2.center) to (i2a.center);
\draw[string,in=270,out=90] (i2a.center) to (b.center);
\draw[string,in=180,out=270] (o1.center) to (b.center);
\draw[string,in=left,out=right] (b.center) to (y.center);
\draw[string](y.center) to(o2.center);
\end{tikzpicture}
\end{aligned}
+
\begin{pic}[scale=0.75]
\node(b) at (0.5,0){};
\node(m) at (1.25,0){};
\node(i1) at (0.5,-2.25){};
\node(i2) at (1.25,-2.25){};
\node(i1a) at (0.5,-1){};
\node(i2a) at (1.25,-1){};
\node[reddot](g1) at (1.25,-1.75){};
\node[reddot](g2) at (1.25,-1.25){};
\draw[string,in=270,out=90,looseness=1] (i1.center) to (i1a.center);
\draw[string,in=270,out=90,looseness=1] (i1a.center) to (m.center);
\draw[string,in=270,out=90] (i2.center) to (g1.center);
\draw[string,in=270,out=90] (g2.center) to (i2a.center);
\draw[string,in=270,out=90] (i2a.center) to (b.center);
\end{pic}
\end{align*}
First we show that equation~\eqref{eq:puni} holds for $P$.
\begin{align*}\begin{pic}[scale=0.75]
\node(i2) at (-0.32,-2.25){};
\node(i3) at (0.32,-2.25){};
\node(p)[morphism,wedge] at (0,-0.6){$P$};
\draw[string] (i2.center) to (p.south west);
\draw[string] (i3.center) to (p.south east);
\node(i2a) at (-0.32,2.25){};
\node(i3a) at (0.32,2.25){};
\node(pa)[morphism,wedge,hflip] at (0,0.6){$P$};
\draw[string] (i2a.center) to (pa.north west);
\draw[string] (i3a.center) to (pa.north east);
\draw[string] (p.north east) to (pa.south east);
\draw[string] (p.north west) to (pa.south west);
\end{pic}
:=& \quad \begin{aligned}
\begin{tikzpicture}[scale=0.75]
\node(b)[blackdot] at (0.5,-0.75){};
\node(m) at (1.25,-0.75){};
\node(i1) at (0.5,-2.25){};
\node(i2) at (1.25,-2.25){};
\node(i1a) at (0.5,-1.5){};
\node(i2a) at (1.25,-1.5){};
\node(g1)[whitedot,projdot] at (1.25,-2){};
\node(g2)[whitedot,inner sep=0.1pt] at (1.25,-1.75){$\projs$};
\node(o2) at (1.15,0){};
\node(y)[yellowdot] at (1.15,-0.5){};
\draw[green][string, line width=1pt](g1.center) to (g2.center);
\draw[string,in=270,out=90,looseness=1] (i1.center) to (i1a.center);
\draw[string,in=270,out=90,looseness=1] (i1a.center) to (m.center);
\draw[string,in=right,out=90,looseness=1] (m.center) to (y.center);
\draw[string,in=270,out=90] (i2.center) to (g1.center);
\draw[string,in=270,out=90] (g2.center) to (i2a.center);
\draw[string,in=270,out=90] (i2a.center) to (b.center);\draw[string,in=left,out=right] (b.center) to (y.center);
\draw[string](y.center) to(o2.center);
\node(b1)[blackdot] at (0.5,0.75){};
\node(m1) at (1.25,0.75){};
\node(i11) at (0.5,2.25){};
\node(i21) at (1.25,2.25){};
\node(i1a1) at (0.5,1.5){};
\node(i2a1) at (1.25,1.5){};
\node(g11)[whitedot,projdot] at (1.25,2){};
\node(g21)[whitedot,inner sep=0.1pt] at (1.25,1.75){$\projs$};
\node(o21) at (1.15,0){};
\node(y1)[yellowdot] at (1.15,0.5){};
\draw[green][string, line width=1pt](g11.center) to (g21.center);
\draw[string,in=90,out=270,looseness=1] (i11.center) to (i1a1.center);
\draw[string,in=90,out=270,looseness=1] (i1a1.center) to (m1.center);
\draw[string,in=right,out=270,looseness=1] (m1.center) to (y1.center);
\draw[string,in=90,out=270] (i21.center) to (g11.center);
\draw[string,in=90,out=270] (g21.center) to (i2a1.center);
\draw[string,in=90,out=270] (i2a1.center) to (b1.center);
\draw[string,in=180,out=180] (b.center) to (b1.center);
\draw[string,in=left,out=right] (b1.center) to (y1.center);
\draw[string](y1.center) to(o21.center);
\end{tikzpicture}
\end{aligned}
+
\begin{pic}[scale=0.75]
\node(b) at (0.5,0){};
\node(m) at (1.25,0){};
\node(i1) at (0.5,-2.25){};
\node(i2) at (1.25,-2.25){};
\node(i1a) at (0.5,-1){};
\node(i2a) at (1.25,-1){};
\node[reddot](g1) at (1.25,-1.75){};
\node[reddot](g2) at (1.25,-1.25){};
\draw[string,in=270,out=90,looseness=1] (i1.center) to (i1a.center);
\draw[string,in=270,out=90,looseness=1] (i1a.center) to (m.center);
\draw[string,in=270,out=90] (i2.center) to (g1.center);
\draw[string,in=270,out=90] (g2.center) to (i2a.center);
\draw[string,in=270,out=90] (i2a.center) to (b.center);
\node(b1) at (0.5,0){};
\node(m1) at (1.25,0){};
\node(i11) at (0.5,2.25){};
\node(i21) at (1.25,2.25){};
\node(i1a1) at (0.5,1){};
\node(i2a1) at (1.25,1){};
\node[reddot](g11) at (1.25,1.75){};
\node[reddot](g21) at (1.25,1.25){};
\draw[string,in=90,out=270,looseness=1] (i11.center) to (i1a1.center);
\draw[string,in=90,out=270,looseness=1] (i1a1.center) to (m1.center);
\draw[string,in=90,out=270] (i21.center) to (g11.center);
\draw[string,in=90,out=270] (g21.center) to (i2a1.center);
\draw[string,in=90,out=270] (i2a1.center) to (b1.center);
\end{pic}\quad \super{\eqref{diag:last}}= \quad
\begin{pic}[scale=0.75]
\node(i1) at (-0.32,-2.25){};
\node(i2) at (0.32,-2.25){};
\node(o1) at (-0.32,2.25){};
\node(o2) at (0.32,2.25){};
\node(m1)[whitedot,inner sep=0.25pt] at (0.32,0){$\projs$};
\node(m2) at (0.32,0){};
\draw[string](i1.center) to (o1.center);
\draw[string](i2.center) to (m2.center);
\draw[string](m1.center) to (o2.center);
\draw[green][string, line width=1pt](m1.center) to (m2.center);
\end{pic}\quad + \quad
\begin{pic}[scale=0.75]
\node(i1) at (-0.32,-2.25){};
\node(i2) at (0.32,-2.25){};
\node(o1) at (-0.32,2.25){};
\node(o2) at (0.32,2.25){};
\node[reddot](m1) at (0.32,0.5){};
\node[reddot](m2) at (0.32,-0.5){};
\draw[string](i1.center) to (o1.center);
\draw[string](i2.center) to (m2.center);
\draw[string](m1.center) to (o2.center);
\end{pic}
\quad \super{\eqref{eq:proj2}}= \quad
\begin{pic}[scale=0.75]
\node(i1) at (-0.32,-2.25){};
\node(i2) at (0.32,-2.25){};
\node(o1) at (-0.32,2.25){};
\node(o2) at (0.32,2.25){};
\draw[string](i1.center) to (o1.center);
\draw[string](i2.center) to (o2.center);
\end{pic}
\end{align*}
Now we show that equation~\eqref{eq:pfunc} holds for $P$.
\begin{align*}\begin{pic}[scale=0.75]
\node(o2) at (-0.36,3){};
\node(o3) at (0.36,3){};
\node(o4) at (1,3){};
\node(i2) at (-0.32,0){};
\node(i3) at (0.32,0){};
\node(o1) at (-1,3){};
\node(p)[morphism,wedge] at (0,1.2){$P$};
\node(b1)[blackdot] at (-0.32,2){};
\node(b2)[blackdot] at (0.32,2){};
\draw[string] (i2.center) to (p.south west);
\draw[string] (i3.center) to (p.south east);
\draw[string] (b2.center) to (p.north east);
\draw[string] (b1.center) to (p.north west);
\draw[string,out=left,in=270] (b1.center) to (o1.center);
\draw[string,out=right,in=270] (b1.center) to (o3.center);
\draw[string,out=left,in=270] (b2.center) to (o2.center);
\draw[string,out=right,in=270] (b2.center) to (o4.center);
\end{pic}:=& \quad 
\begin{pic}[scale=0.75]
\node(b)[blackdot] at (0.5,-0.75){};
\node(m) at (1.25,-0.75){};
\node(i1) at (0.5,-2.25){};
\node(i2) at (1.25,-2.25){};
\node(i1a) at (0.5,-1.5){};
\node(i2a) at (1.25,-1.5){};
\node(g1)[whitedot,projdot] at (1.25,-2){};
\node(g2)[whitedot,inner sep=0.1pt] at (1.25,-1.75){$\projs$};
\node[blackdot](o1) at (0,0){};
\node[blackdot](o2) at (1.15,0){};
\node(o1a) at (-0.5,0.75){};
\node(o1b) at (1,0.75){};
\node(o2a) at (0.65,0.75){};
\node(o2b) at (1.65,0.75){};
\node(y)[yellowdot] at (1.15,-0.5){};
\draw[string,out=left,in=270] (o1.center) to (o1a.center);
\draw[string,out=right,in=270] (o1.center) to (o1b.center);
\draw[string,out=left,in=270] (o2.center) to (o2a.center);
\draw[string,out=right,in=270] (o2.center) to (o2b.center);
\draw[green][string, line width=1pt](g1.center) to (g2.center);
\draw[string,in=270,out=90,looseness=1] (i1.center) to (i1a.center);
\draw[string,in=270,out=90,looseness=1] (i1a.center) to (m.center);
\draw[string,in=right,out=90,looseness=1] (m.center) to (y.center);
\draw[string,in=270,out=90] (i2.center) to (g1.center);
\draw[string,in=270,out=90] (g2.center) to (i2a.center);
\draw[string,in=270,out=90] (i2a.center) to (b.center);
\draw[string,in=180,out=270] (o1.center) to (b.center);
\draw[string,in=left,out=right] (b.center) to (y.center);
\draw[string](y.center) to(o2.center);
\end{pic}+
\begin{pic}[scale=0.75]
\node[blackdot](b) at (0.5,0){};
\node[blackdot](m) at (1.25,0){};
\node(o1a) at (0,0.75){};
\node(o1b) at (1,0.75){};
\draw[string,out=left,in=270] (b.center) to (o1a.center);
\draw[string,out=right,in=270] (b.center) to (o1b.center);
\draw[string,out=left,in=270] (m.center) to (o2a.center);
\draw[string,out=right,in=270] (m.center) to (o2b.center);
\node(o2a) at (0.75,0.75){};
\node(o2b) at (1.75,0.75){};
\node(i1) at (0.5,-2.25){};
\node(i2) at (1.25,-2.25){};
\node(i1a) at (0.5,-1){};
\node(i2a) at (1.25,-1){};
\node[reddot](g1) at (1.25,-1.75){};
\node[reddot](g2) at (1.25,-1.25){};
\draw[string,in=270,out=90,looseness=1] (i1.center) to (i1a.center);
\draw[string,in=270,out=90,looseness=1] (i1a.center) to (m.center);
\draw[string,in=270,out=90] (i2.center) to (g1.center);
\draw[string,in=270,out=90] (g2.center) to (i2a.center);
\draw[string,in=270,out=90] (i2a.center) to (b.center);
\end{pic}\quad \super{\eqref{eq:bialg}}= \quad
\begin{pic}[scale=3/8]
\node(o1) at (-2,3){};
\node(o2) at (-0.5,3){};
\node(o3) at (0.5,3){};
\node(o4) at (2,3){};
\node(i1) at (-0.5,-3){};
\node(i2) at (1.5,-3){};
\node(m1) at (-0.5,-0.5){};
\node(m2) at (1.5,-0.5){};
\node(m3) at (-0.5,-2){};
\node(g1)[whitedot,projdot] at (1.5,-2.5){};
\node(g2)[whitedot,inner sep=0.1pt] at (1.5,-2){$\projs$};
\node(b1)[blackdot] at (-1.25,2){};
\node(b2)[blackdot] at (-0.5,0){};
\node(b3)[blackdot] at (0.5,1){};
\node(b4)[blackdot] at (2,1){};
\node(y1)[yellowdot] at (0.5,2){};
\node(y2)[yellowdot] at (2,2){};
\draw[string,in=180,out=270] (o1.center) to (b1.center);
\draw[string,in=right,out=270] (o2.center) to (b1.center);
\draw[string] (o3.center) to (y1.center);
\draw[string] (o4.center) to (y2.center);
\draw[string,in=left,out=left] (y1.center) to (b3.center);
\draw[string,in=left,out=right] (y1.center) to (b4.center);
\draw[string,in=right,out=right] (y2.center) to (b4.center);
\draw[string,out=left,in=right] (y2.center) to (b3.center);
\draw[string,in=left,out=270] (b1.center) to (b2.center);
\draw[string,in=right,out=270] (b3.center) to (b2.center);
\draw[string,in=90,out=270] (b4.center) to (m2.center);
\draw[string,in=90,out=270] (b2.center) to (m1.center);
\draw[string,in=90,out=270] (m1.center) to (g2.center);
\draw[string,in=90,out=270] (m2.center) to (m3.center);
\draw[string,in=90,out=270] (m3.center) to (i1.center);
\draw[green][string, line width=1pt,in=90,out=270] (g2.center) to (g1.center);
\draw[string,in=90,out=270] (g1.center) to (i2.center);
\end{pic}+
\begin{pic}[scale=0.75]
\node[reddot](b) at (0,0){};
\node[reddot](r) at (1,0){};
\node[blackdot](m) at (0.5,-1){};
\node(o1a) at (0,0.75){};
\node(o1b) at (1,0.75){};
\draw[string,out=90,in=270] (b.center) to (o1a.center);
\draw[string] (r.center) to (o1b.center);
\draw[string,out=left,in=270] (m.center) to (o2a.center);
\draw[string,out=right,in=270] (m.center) to (o2b.center);
\node(o2a) at (0.75,0.75){};
\node(o2b) at (1.75,0.75){};
\node(i1) at (0.5,-2.25){};
\node(i2) at (1.25,-2.25){};
\node(i1a) at (0.5,-1){};
\node(i2a) at (1.25,-1){};
\node[reddot](g1) at (1.25,-1.75){};
\draw[string,in=270,out=90,looseness=1] (i1.center) to (i1a.center);
\draw[string,in=270,out=90,looseness=1] (i1a.center) to (m.center);
\draw[string,in=270,out=90] (i2.center) to (g1.center);
\end{pic}
\\
\super{\eqref{eq:sm}}=& \quad
\begin{pic}[scale=3/8]
\node(o1) at (-2,3){};
\node(o2) at (-0.5,3){};
\node(o3) at (0.5,3){};
\node(o4) at (2.75,3){};
\node(i1) at (-0.5,-3){};
\node(i2) at (1.5,-3){};
\node(g1)[whitedot,projdot] at (1.5,-2.5){};
\node(g2)[whitedot,inner sep=0.1pt] at (1.5,-2){$\projs$};
\node(b1)[blackdot] at (-1.25,1){};
\node(b2)[blackdot] at (1.5,-0.8){};
\node(b3)[blackdot] at (2,1){};
\node(b4)[blackdot] at (-0.5,-0.9){};
\node(y1)[yellowdot] at (-0.5,2){};
\node(y2)[yellowdot] at (2.75,2){};
\draw[string,in=180,out=270] (o1.center) to (b1.center);
\draw[string] (o2.center) to (y1.center);
\draw[string,out=270,in=left] (o3.center) to (b3.center);
\draw[string] (o4.center) to (y2.center);
\draw[string,in=left,out=right] (y1.center) to (b4.center);
\draw[string] (i1.center) to (b4.center);
\draw[string,in=right,out=right] (y2.center) to (b4.center);
\draw[string,out=left,in=right] (y2.center) to (b3.center);
\draw[string,in=left,out=270] (b1.center) to (b2.center);
\draw[string,in=right,out=270] (b3.center) to (b2.center);
\draw[green][string, line width=1pt,in=90,out=270] (g2.center) to (g1.center);
\draw[string,in=90,out=270] (g1.center) to (i2.center);
\draw[string] (g2.center) to (b2.center);
\draw[string,in=right,out=left] (y1.center) to (b1.center);
\end{pic}+
\begin{pic}[scale=0.75]
\node[reddot](b) at (0,-0.25){};
\node[reddot](r) at (1,0){};
\node[blackdot](m) at (0,-1.75){};
\node(o1a) at (0,0.75){};
\node(o1b) at (1,0.75){};
\draw[string,out=90,in=270] (b.center) to (o1a.center);
\draw[string] (r.center) to (o1b.center);
\draw[string,out=left,in=270] (m.center) to (o2a.center);
\draw[string,out=right,in=270] (m.center) to (o2b.center);
\node(o2a) at (0.75,0.75){};
\node(o2b) at (1.75,0.75){};
\node(i1) at (0,-2.25){};
\node(i2) at (1.25,-2.25){};
\node(i1a) at (0.5,-1){};
\node(i2a) at (1.25,-1){};
\node[blackdot](g1) at (1.25,-1.75){};
\node[reddot](r1) at (0.25,-1){};
\node[reddot](r2) at (1,-0.5){};
\draw[string] (i1.center) to (m.center);
\draw[string,in=270,out=90] (i2.center) to (g1.center);
\draw[string,out=left,in=270] (g1.center) to (r1.center);
\draw[string,out=right,in=270] (g1.center) to (r2.center);
\end{pic}\quad = \quad
\begin{pic}[scale=0.75]
\node(o1) at (-0.32,2.4){};
\node(m) at (0.32,2.4){};
\node(i2) at (-0.32,0){};
\node(i3) at (0.32,0){};
\node(p)[morphism,wedge] at (0,1.2){$P$};
\draw[string] (m.center) to (p.north east);
\draw[string] (o1.center) to (p.north west);
\node(o1a) at (1.68,2.4){};
\node(ma) at (2.32,2.4){};
\node(i2a) at (1.68,0){};
\node(i3a) at (2.32,0){};
\node(pa)[morphism,wedge] at (2,1.2){$P$};
\node(b1)[blackdot] at (0.68,-0.25){};
\node(b2)[blackdot] at (1.32,-0.25){};
\draw[string] (ma.center) to (pa.north east);
\draw[string] (o1a.center) to (pa.north west);
\draw[string,out=left,in=270] (b1.center) to (p.south west);
\draw[string,out=left,in=270,looseness=0.6] (b2.center) to (p.south east);
\draw[string,out=right,in=270,looseness=0.6] (b1.center) to (pa.south west);
\draw[string,out=right,in=270] (b2.center) to (pa.south east);
\node(a1) at (0.68,-0.75){};
\node(a2) at (1.32,-0.75){};
\draw[string] (a1.center) to (b1.center);
\draw[string] (a2.center) to (b2.center);
\end{pic}\end{align*}
So $P$ is a permutation and so $U_{FF}$ is equal to $V_{P(i,j)}$, where $V_{ij}$ is given by the following:
\begin{equation}
V_{ij}:=
 \begin{aligned}\begin{pic}
          \node (in) at (0,0) {};         
          \node (H)[morphism,wedge,scale=0.75] at(1,1.5) {\large$\chi$};
          \node (b1)[blackdot,scale=1.2] at (0.5,2) {};
          \node (b2)[reddot,scale=1.2] at (1,3) {};          
          \node (i)[state,black,scale=0.5,label={[label distance=0.08cm]330:i}] at(1,1) {};
          \node (0b) at (0.5,2) {};
          \node (j)[state,black,scale=0.5,label={[label distance=0.08cm]330:j}] at(1.5,2.5) {};
          \node (out) at (1,4) {}; 
          
          \draw[string,out=90,in=180,looseness=0.75] (in.center) to (b1.center);           \draw[string,out=90,in=0] (H.north) to (b1.center);
          \draw[string,out=90,in=180] (b1.center) to (b2.center);
          \draw[string,out=90,in=0] (j.center) to (b2.center);
          \draw[string] (b2.center) to (out.center);
          \draw[string] (i) to (H.south);
\end{pic}\end{aligned}
\end{equation}
Since $\tinymult[reddot]$ is a finite abelian group it is a finite quasigroup and thus a Latin square. $\chi$ is a Hadamard and so $V$ is a shift and multiply basis, and therefore a UEB ~\cite{werner2001all,mustothesis,mypaper1}. $V$ and $U_{FF}$ are equivalent by equation~\eqref{eq:uebeq}, and so $U_{FF}$ is a UEB. 
\paragraph{Commuting property.}\ignore{We now show the commuting property for $U_{FF}$. For $i,x \in \range{0}$ and $a,b\in \range{1}$ consider the following:
\begin{align*}[U_{FF}]_{xa}\circ[U_{FF}]_{xb}\ket i=&\chi_i[a]\chi_{(i \tinydot[reddot](a\tinydot[yellowdot]x))}[b]\ket{(i \tinydot[reddot](a\tinydot[yellowdot]x))\tinydot[reddot](b\tinydot[yellowdot]x)}\\ \super{\eqref{1}}=&\chi_1[i \tinydot[yellowdot]a]\chi_1[(i \tinydot[reddot](a\tinydot[yellowdot]x))\tinydot[yellowdot]b]\ket{(i \tinydot[reddot](a\tinydot[yellowdot]x))\tinydot[reddot](b\tinydot[yellowdot]x)}\\
\super{\eqref{eq:chi}}=&\chi_1[(i \tinydot[yellowdot]a))\tinydot[reddot]((i \tinydot[reddot](a\tinydot[yellowdot]x))\tinydot[yellowdot]b)]\ket{(i \tinydot[reddot](a\tinydot[yellowdot]x))\tinydot[reddot](b\tinydot[yellowdot]x)}
\\\super{\eqref{eq:as}\eqref{eq:comm}\eqref{2}}=&\chi_1[(i \tinydot[yellowdot]b))\tinydot[reddot]((i \tinydot[reddot](b\tinydot[yellowdot]x))\tinydot[yellowdot]a)]\ket{(i \tinydot[reddot](b\tinydot[yellowdot]x))\tinydot[reddot](a\tinydot[yellowdot]x)}
\\\super{\eqref{eq:chi}}=&\chi_1[i \tinydot[yellowdot]b]\chi_1[(i \tinydot[reddot](b\tinydot[yellowdot]x))\tinydot[yellowdot]a]\ket{(i \tinydot[reddot](b\tinydot[yellowdot]x))\tinydot[reddot](a\tinydot[yellowdot]x)}
\\\super{\eqref{1}}=&\chi_i[b]\chi_{(i \tinydot[reddot](b\tinydot[yellowdot]x))}[a]\ket{(i \tinydot[reddot](b\tinydot[yellowdot]x))\tinydot[reddot](a\tinydot[yellowdot]x)}
\\=&[U_{FF}]_{xb}\circ[U_{FF}]_{xa}\ket i
\end{align*}
This is true for all $i$ and so the classes $C_x$ with $x \in \range{0}$ commute as required. For $C_*$ we have the following:
\begin{align*}
[U_{FF}]_{a0}\circ[U_{FF}]_{b0}\ket i=&\ket{((i\tinydot[reddot]a)\tinydot[reddot]b}
\\\super{\eqref{eq:as}}=&\ket{((i\tinydot[reddot]b)\tinydot[reddot]a}
\\=&[U_{FF}]_{b0}\circ[U_{FF}]_{a0}\ket i
\end{align*}
Again this holds for all $i$ and we have the required result.}

%%%%%%%%%%%%%%%%%%%%%%%%%%%%%%%%%

We now prove the following:
\begin{align}\label{eq:uffzero}
\begin{aligned}
\begin{tikzpicture}[scale=0.75]
\node (f) [morphism,wedge,scale=0.75, connect n ] at (0,1) {$U_{FF}$};
\node(i2)  at (-0.4125,-1.7) {};
\node(i2a)  at (0.4125,-1.7) {};
\node (e) [morphism,wedge, scale=0.75,connect sw length = 1.65cm,connect s length = 1.65cm] at (-0.24,-0.25) {$U_{FF}$};
\node(m1) at (0.75,-0.25) {};
\node(m2) at (1,-0.25) {};
\node[state,black,scale=0.25,label={[label distance=-0.125cm]315:\tiny{0}}](b1) at (0.75,-1.4){};
\node[state,black,scale=0.25,hflip,label={[label distance=-0.125cm]45:\tiny{0}}](b2) at (0.75,-2){};
\node[state,black,scale=0.25,label={[label distance=-0.125cm]315:\tiny{0}}](b3) at (1.25,-1.4){};
\node[state,black,scale=0.25,hflip,label={[label distance=-0.125cm]45:\tiny{0}}](b4) at (1.25,-2){};
\draw[string] (e.north) to (f.south west);
\draw[string,in=270,out=90,looseness=1.25] (i2a.center) to (m1.center);
\draw[string,in=270,out=90,looseness=1.25] (m1.center) to (f.south);
\draw[string,in=270,out=90] (b1.center) to (e.south east);
\draw[string,out=90,in=270] (b3.center) to (m2.center);
\draw[string,out=90,in=270] (m2.center) to (f.south east);
\draw[string] (i2a.center) to +(0,-0.5);
\draw[string] (b2.center) to (0.75,-2.2);
\draw[string] (b4.center) to (1.25,-2.2);
\end{tikzpicture}
\end{aligned}
\quad = \quad
\begin{aligned}
\begin{tikzpicture}[scale=0.75]
\node (f) [morphism,wedge,scale=0.75, connect n ] at (0,1) {$U_{FF}$};
\node(i2)  at (-0.4125,-1.7) {};
\node(i2a)  at (0.4125,-1.7) {};
\node (e) [morphism,wedge, scale=0.75,connect sw length = 1.65cm] at (-0.24,-0.25){$U_{FF}$};
\node(m1) at (0.75,-0.25){} ;
\node(m2) at (1,-0.25) {};
\node[state,black,scale=0.25,label={[label distance=-0.125cm]315:\tiny{0}}](b1) at (0.75,-1.4){};
\node[state,black,scale=0.25,hflip,label={[label distance=-0.125cm]45:\tiny{0}}](b2) at (0.75,-2){};
\node[state,black,scale=0.25,label={[label distance=-0.125cm]315:\tiny{0}}](b3) at (1.25,-1.4){};
\node[state,black,scale=0.25,hflip,label={[label distance=-0.125cm]45:\tiny{0}}](b4) at (1.25,-2){};
\draw[string] (i2.center) to (e.south);
\draw[string] (e.north) to (f.south west);
\draw[string,in=270,out=90,looseness=1.25] (i2a.center) to (m1.center);
\draw[string,in=270,out=90,looseness=1.25] (m1.center) to (f.south);
\draw[string,in=270,out=90] (b1.center) to (e.south east);
\draw[string,out=90,in=270] (b3.center) to (m2.center);
\draw[string,out=90,in=270] (m2.center) to (f.south east);
\draw[string,out=270,in=90] (i2.center) to +(0.825,-0.5);
\draw[string,out=270,in=90] (i2a.center) to +(-0.825,-0.5);
\draw[string] (b2.center) to (0.75,-2.2);
\draw[string] (b4.center) to (1.25,-2.2);
\end{tikzpicture}
\end{aligned}
\end{align}

\begin{align*}
\begin{aligned}
\begin{tikzpicture}[scale=0.75]
\node (f) [morphism,wedge,scale=0.75, connect n ] at (0,1) {$U_{FF}$};
\node(i2)  at (-0.4125,-1.7) {};
\node(i2a)  at (0.4125,-1.7) {};
\node (e) [morphism,wedge, scale=0.75,connect sw length = 1.65cm,connect s length = 1.65cm] at (-0.24,-0.25) {$U_{FF}$};
\node(m1) at (0.75,-0.25) {};
\node(m2) at (1,-0.25) {};
\node[state,black,scale=0.25,label={[label distance=-0.125cm]315:\tiny{0}}](b1) at (0.75,-1.4){};
\node[state,black,scale=0.25,hflip,label={[label distance=-0.125cm]45:\tiny{0}}](b2) at (0.75,-2){};
\node[state,black,scale=0.25,label={[label distance=-0.125cm]315:\tiny{0}}](b3) at (1.25,-1.4){};
\node[state,black,scale=0.25,hflip,label={[label distance=-0.125cm]45:\tiny{0}}](b4) at (1.25,-2){};
\draw[string] (e.north) to (f.south west);
\draw[string,in=270,out=90,looseness=1.25] (i2a.center) to (m1.center);
\draw[string,in=270,out=90,looseness=1.25] (m1.center) to (f.south);
\draw[string,in=270,out=90] (b1.center) to (e.south east);
\draw[string,out=90,in=270] (b3.center) to (m2.center);
\draw[string,out=90,in=270] (m2.center) to (f.south east);
\draw[string] (i2a.center) to +(0,-0.5);
\draw[string] (b2.center) to (0.75,-2.2);
\draw[string] (b4.center) to (1.25,-2.2);
\end{tikzpicture}
\end{aligned}
:=& \quad 
\begin{pic}[scale=1/3]
\node (chi) [morphism,wedge,scale=0.5] at (-0.5,-0.1) {\large$\chi$};
\node(b1)[blackdot,proofdiagram]  at (-1,0.75) {};
\node(add)[reddot,proofdiagram] at (-0.5,1.25){};
\node(x)[yellowdot,proofdiagram] at (1.75,-0.5){};
\node(b2)[blackdot,proofdiagram] at (0.5,-1){};
\node(*) at (1,-0.5){};
\node(m) at (2,-1){};
\node(i1) at (-2,-4){};
\node(i2) at (1,-1.5){};
\node(i3)[reddot,proofdiagram] at (-0.75,-3){};
\node(i2a) at (1,-0.75){};
\node(i3a) at (1.75,-0.75){};
\node(o) at (0.75,2.5){};
\node(g1)[whitedot,projdot] at (-0.5,-2){};
\node(g2)[whitedot,scale=1,inner sep=0.25pt] at (-0.5,-2){$\boldsymbol{*}$};
\draw[green][string, line width=1pt,in=270,out=90](g1.center) to (g2.center);
\draw[string,in=left,out=90,looseness=0.6] (i1.center) to (b1.center);
\draw[string,in=90,out=right] (b1.center) to (chi.north);
\draw[string,in=left,out=90] (b1.center) to (add.center);
\draw[string,in=90,out=right] (add.center) to (x.center);
\draw[string,in=270,out=90,looseness=1] (i2.center) to (m.center);
\draw[string,in=right,out=90,looseness=1] (m.center) to (x.center);
\draw[string,in=270,out=90] (i3.center) to (g1.center);
\draw[string,in=270,out=90] (g2.center) to (b2.center);
\draw[string,in=270,out=left] (b2.center) to (chi.south);
\draw[string,in=left,out=right] (b2.center) to (x.center);
\node (chi7) [morphism,wedge,scale=0.5] at (2,3.1) {\large$\chi$};
\node(b17)[blackdot,proofdiagram]  at (1.5,4) {};
\node(add7)[reddot,proofdiagram] at (2,4.5){};
\node(x7)[yellowdot,proofdiagram] at (4.25,2.75){};
\node(b27)[blackdot,proofdiagram] at (3,2.25){};
\node(*7) at (2.25,1.5){};
\node(m7) at (4.5,2.25){};
\node(i17) at (0.75,1){};
\node(i27) at (2.5,-2){};
\node[reddot,proofdiagram](i37) at (3,-0.25){};
\node(i2a7) at (2.25,1.25){};
\node(i3a7) at (3,1.75){};
\node(o7) at (2,5){};
\node(g17)[whitedot,projdot] at (3,0.75){};
\node(g27)[whitedot,scale=1,inner sep=0.25pt] at (3,0.8){$\boldsymbol{*}$};
\draw[green][string,line width=1pt](g17.center) to (g27.center);
\draw[string,in=left,out=90,looseness=0.6] (add.center) to (b17.center);
\draw[string,in=90,out=right] (b17.center) to (chi7.north);
\draw[string,in=left,out=90] (b17.center) to (add7.center);
\draw[string,in=90,out=right] (add7.center) to (x7.center);
\draw[string,in=270,out=90] (add7.center) to (o7.center);
\draw[string,in=270,out=90,looseness=1] (i27.center) to (i2a7.center);
\draw[string,in=270,out=90,looseness=1] (i2a7.center) to (m7.center);
\draw[string,in=right,out=90,looseness=1] (m7.center) to (x7.center);
\draw[string,in=270,out=90] (i37.center) to (g17.center);
\draw[string,in=270,out=90] (g27.center) to (i3a7.center);
\draw[string,in=270,out=90] (i3a7.center) to (b27.center);
\draw[string,in=270,out=left] (b27.center) to (chi7.south);
\draw[string,in=left,out=right] (b27.center) to (x7.center);
\node(in1) at (-0.5,-4){};
\draw[string,in=90,out=270] (i2.center) to (in1.center);
\node(in2) at (1,-4){};
\draw[string,in=90,out=270,looseness=1] (i27.center) to (in2.center);
\node[reddot,proofdiagram](b2) at (2,-3.5){};
\node[reddot,proofdiagram](b4) at (3,-3.5){};
\node(in3) at (2,-4){};
\node(in4) at (3,-4){};
\draw[string] (b2.center) to (in3.center);
\draw[string] (b4.center) to (in4.center);
\end{pic}+
\begin{pic}[scale=1/3]
\node (chi) [morphism,wedge,scale=0.5] at (-0.5,-0.15) {\large$\chi$};
\node(b1)[blackdot,proofdiagram]  at (-1,0.75) {};
\node(add)[reddot,proofdiagram] at (-0.5,1.25){};
\node(x)[yellowdot,proofdiagram] at (1.75,-0.5){};
\node(b2)[blackdot,proofdiagram] at (0.5,-1){};
\node(*) at (1,-0.5){};
\node(m) at (2,-1){};
\node(i1) at (-2,-4){};
\node(i2) at (1,-1.5){};
\node(i3)[reddot,proofdiagram] at (-0.75,-3){};
\node(i2a) at (1,-0.75){};
\node(i3a) at (1.75,-0.75){};
\node(o) at (0.75,2.5){};
\node(g1)[whitedot,projdot] at (-0.5,-2){};
\node(g2)[whitedot,scale=1,inner sep=0.25pt] at (-0.5,-2){$\boldsymbol{*}$};
\draw[green][string, line width=1pt,out=90,in=270](g1.center) to (g2.center);
\draw[string,in=left,out=90,looseness=0.6] (i1.center) to (b1.center);
\draw[string,in=90,out=right] (b1.center) to (chi.north);
\draw[string,in=left,out=90] (b1.center) to (add.center);
\draw[string,in=90,out=right] (add.center) to (x.center);
\draw[string,in=270,out=90,looseness=1] (i2.center) to (m.center);
\draw[string,in=right,out=90,looseness=1] (m.center) to (x.center);
\draw[string,in=270,out=90] (i3.center) to (g1.center);
\draw[string,in=270,out=90] (g2.center) to (b2.center);
\draw[string,in=270,out=left] (b2.center) to (chi.south);
\draw[string,in=left,out=right] (b2.center) to (x.center);
\node (chi7)  at (1.75,3.25) {};
\node(b17)  at (1.5,4) {};
\node(add7)[reddot,proofdiagram] at (1.5,4.5){};
\node(x7) at (3,3){};
\node(b27) at (2.25,2.5){};
\node(*7) at (2.25,2){};
\node(m7) at (3,2.5){};
\node(i17) at (0.75,1){};
\node(i27) at (2.5,-2){};
\node[reddot,proofdiagram](i37) at (3,0.5){};
\node(i2a7) at (2.25,1.75){};
\node(i3a7) at (3,1.75){};
\node(o7) at (1.5,5){};
\node(g17) at (3,0.75){};
\node(g27)[reddot,proofdiagram] at (3,1.5){};
\draw[string](g17.center) to (g27.center);
\draw[string,in=left,out=90,looseness=0.6] (add.center) to (add7.center);
\draw[string,in=90,out=right] (add7.center) to (*7.center);
\draw[string,in=270,out=90] (add7.center) to (o7.center);
\draw[string,in=270,out=90,looseness=1] (i27.center) to (*7.center);
\draw[string,in=270,out=90] (i37.center) to (g17.center);
\node(in1) at (-0.5,-4){};
\draw[string,in=90,out=270] (i2.center) to (in1.center);
\node(in2) at (1,-4){};
\draw[string,in=90,out=270,looseness=1] (i27.center) to (in2.center);
\node[reddot,proofdiagram](b2) at (2,-3.5){};
\node[reddot,proofdiagram](b4) at (3,-3.5){};
\node(in3) at (2,-4){};
\node(in4) at (3,-4){};
\draw[string] (b2.center) to (in3.center);
\draw[string] (b4.center) to (in4.center);
\end{pic}+
\begin{pic}[scale=1/3]
\node (chi)  at (0.5,0.75) {};
\node(b1)  at (0.25,1.5) {};
\node(add)[reddot,proofdiagram] at (0.75,2){};
\node(x) at (1.75,0.5){};
\node(b2) at (1,0){};
\node(*) at (1,-0.5){};
\node(m) at (1.75,0){};
\node(i1) at (-1,-4){};
\node(i2) at (1,-1.5){};
\node(i3)[reddot,proofdiagram] at (1.75,-2){};
\node(i2a) at (1,-0.75){};
\node(i3a) at (1.75,-0.75){};
\node(o) at (0.75,2.5){};
\node(g1) at (1.75,-1.75){};
\node(g2)[reddot,proofdiagram] at (1.75,-1){};
\draw[string](g1.center) to (g2.center);
\draw[string,in=left,out=90,looseness=0.6] (i1.center) to (add.center);
\draw[string,in=90,out=right] (add.center) to (x.center);
\draw[string,in=270,out=90,looseness=1] (i2.center) to (i2a.center);
\draw[string,in=270,out=90,looseness=1] (i2a.center) to (m.center);
\draw[string,in=270,out=90,looseness=1] (m.center) to (x.center);
\draw[string,in=270,out=90] (i3.center) to (g1.center);
\node (chi7) [morphism,wedge,scale=0.5] at (2,3.1) {\large$\chi$};
\node(b17)[blackdot,proofdiagram]  at (1.5,4) {};
\node(add7)[reddot,proofdiagram] at (2,4.5){};
\node(x7)[yellowdot,proofdiagram] at (4.25,2.75){};
\node(b27)[blackdot,proofdiagram] at (3,2.25){};
\node(*7) at (2.25,2){};
\node(m7) at (4.5,2.25){};
\node(i17) at (0.75,1){};
\node(i27) at (2.5,-2){};
\node[reddot,proofdiagram](i37) at (3,-0.25){};
\node(i2a7) at (2.25,1.25){};
\node(i3a7) at (3,1.75){};
\node(o7) at (2,5){};
\node(g17)[whitedot,projdot] at (3,0.75){};
\node(g27)[whitedot,scale=1,inner sep=0.25pt] at (3,0.8){$\boldsymbol{*}$};
\draw[green][string, line width=1pt](g17.center) to (g27.center);
\draw[string,in=left,out=90,looseness=0.6] (add.center) to (b17.center);
\draw[string,in=90,out=right] (b17.center) to (chi7.north);
\draw[string,in=left,out=90] (b17.center) to (add7.center);
\draw[string,in=90,out=right] (add7.center) to (x7.center);
\draw[string,in=270,out=90] (add7.center) to (o7.center);
\draw[string,in=270,out=90,looseness=1] (i27.center) to (i2a7.center);
\draw[string,in=270,out=90,looseness=1] (i2a7.center) to (m7.center);
\draw[string,in=right,out=90,looseness=1] (m7.center) to (x7.center);
\draw[string,in=270,out=90] (i37.center) to (g17.center);
\draw[string,in=270,out=90] (g27.center) to (i3a7.center);
\draw[string,in=270,out=90] (i3a7.center) to (b27.center);
\draw[string,in=270,out=left] (b27.center) to (chi7.south);
\draw[string,in=left,out=right] (b27.center) to (x7.center);
\node(in1) at (-0.5,-4){};
\draw[string,in=90,out=270] (i2.center) to (in1.center);
\node(in2) at (1,-4){};
\draw[string,in=90,out=270,looseness=1] (i27.center) to (in2.center);
\node[reddot,proofdiagram](b2) at (2,-3.5){};
\node[reddot,proofdiagram](b4) at (3,-3.5){};
\node(in3) at (2,-4){};
\node(in4) at (3,-4){};
\draw[string] (b2.center) to (in3.center);
\draw[string] (b4.center) to (in4.center);
\end{pic}+
\begin{pic}[scale=1/3]
\node (chi)  at (0.5,0.75) {};
\node(b1)  at (0.25,1.5) {};
\node(add)[reddot,proofdiagram] at (0.75,2){};
\node(x) at (1.75,0.5){};
\node(b2) at (1,0){};
\node(*) at (1,-0.5){};
\node(m) at (1.75,0){};
\node(i1) at (-1,-4){};
\node(i2) at (1,-1.5){};
\node(i3)[reddot,proofdiagram] at (1.75,-2){};
\node(i2a) at (1,-0.75){};
\node(i3a) at (1.75,-0.75){};
\node(o) at (0.75,2.5){};
\node(g1) at (1.75,-1.75){};
\node(g2)[reddot,proofdiagram] at (1.75,-1){};
\draw[string](g1.center) to (g2.center);
\draw[string,in=left,out=90,looseness=0.6] (i1.center) to (add.center);
\draw[string,in=90,out=right] (add.center) to (x.center);
\draw[string,in=270,out=90,looseness=1] (i2.center) to (i2a.center);
\draw[string,in=270,out=90,looseness=1] (i2a.center) to (m.center);
\draw[string,in=270,out=90,looseness=1] (m.center) to (x.center);
\draw[string,in=270,out=90] (i3.center) to (g1.center);
\node (chi7)  at (1.75,3.25) {};
\node(b17)  at (1.5,4) {};
\node(add7)[reddot,proofdiagram] at (1.5,4.5){};
\node(x7) at (3,3){};
\node(b27) at (2.25,2.5){};
\node(*7) at (2.25,2){};
\node(m7) at (3,2.5){};
\node(i17) at (0.75,1){};
\node(i27) at (2.5,-2){};
\node[reddot,proofdiagram](i37) at (3,0.25){};
\node(i2a7) at (2.25,1.75){};
\node(i3a7) at (3,1.75){};
\node(o7) at (1.5,5){};
\node(g17) at (3,0.75){};
\node(g27)[reddot,proofdiagram] at (3,1.5){};
\draw[string](g17.center) to (g27.center);
\draw[string,in=left,out=90,looseness=0.6] (add.center) to (add7.center);
\draw[string,in=90,out=right] (add7.center) to (*7.center);
\draw[string,in=270,out=90] (add7.center) to (o7.center);
\draw[string,in=270,out=90,looseness=1] (i27.center) to (*7.center);
\draw[string,in=270,out=90] (i37.center) to (g17.center);
\node(in1) at (-0.5,-4){};
\draw[string,in=90,out=270] (i2.center) to (in1.center);
\node(in2) at (1,-4){};
\draw[string,in=90,out=270,looseness=1] (i27.center) to (in2.center);
\node[reddot,proofdiagram](b2) at (2,-3.5){};
\node[reddot,proofdiagram](b4) at (3,-3.5){};
\node(in3) at (2,-4){};
\node(in4) at (3,-4){};
\draw[string] (b2.center) to (in3.center);
\draw[string] (b4.center) to (in4.center);
\end{pic}
\\ \super{\eqref{eq:proj2}}=& \quad 0 \quad + \quad 0 \quad + \quad 0 \quad +
\begin{pic}
   \node(i1) at (-1,0){};
   \node(i2) at (-0.5,0){};
   \node(i3) at (0.25,0){};
   \node(i4) at (0.5,0){};
   \node(i5) at (1,0){};
   \node(o) at (-0.325,3){};
   \node[reddot,proofdiagram](r1) at (-0.75,1){};
   \node[reddot,proofdiagram](r2) at (-0.325,2){};
   \node[reddot,proofdiagram](r3) at (0.5,0.25){};
   \node[reddot,proofdiagram](r4) at (1,0.25){};
   \draw[string,out=90,in=180,looseness=0.8] (i1.center) to (r1.center);
   \draw[string,out=90,in=180,looseness=0.8] (r1.center) to (r2.center);
   \draw[string] (r2.center) to (o.center);
   \draw[string,out=90,in=0,looseness=0.8] (i2.center) to (r1.center);
   \draw[string,out=90,in=0,looseness=0.8] (i3.center) to (r2.center);
   \draw[string] (r3.center) to (i4.center);
   \draw[string] (r4.center) to (i5.center);
\end{pic}\\
\\ \super{Red \eqref{eq:sm}}=& \quad 0 \quad + \quad 0 \quad + \quad 0 \quad +
\begin{pic}
   \node(i1) at (-1,0){};
   \node(i2) at (-0.5,0){};
   \node(i3) at (0.25,0){};
   \node(i4) at (0.5,0){};
   \node(i5) at (1,0){};
   \node(o) at (-0.325,3){};
   \node[reddot,proofdiagram](r1) at (-0.75,1){};
   \node[reddot,proofdiagram](r2) at (-0.325,2){};
   \node[reddot,proofdiagram](r3) at (0.5,0.25){};
   \node[reddot,proofdiagram](r4) at (1,0.25){};
   \draw[string,out=90,in=180,looseness=0.8] (i1.center) to (r1.center);
   \draw[string,out=90,in=180,looseness=0.8] (r1.center) to (r2.center);
   \draw[string] (r2.center) to (o.center);
   \draw[string,out=90,in=0,looseness=0.8] (i3.center) to (r1.center);
   \draw[string,out=90,in=0,looseness=0.8] (i2.center) to (r2.center);
   \draw[string] (r3.center) to (i4.center);
   \draw[string] (r4.center) to (i5.center);
\end{pic}\\ \super{\eqref{eq:proj2}}=& \quad 
\begin{pic}[scale=1/3]
\node (chi) [morphism,wedge,scale=0.5] at (-0.5,0.1) {\large$\chi$};
\node(b1)[blackdot,proofdiagram]  at (-1,1) {};
\node(add)[reddot,proofdiagram] at (-0.5,1.5){};
\node(x)[yellowdot,proofdiagram] at (1.75,-0.25){};
\node(b2)[blackdot,proofdiagram] at (0.5,-0.75){};
\node(*) at (1,-0.5){};
\node(m) at (2,-0.75){};
\node(i1) at (-2,-4){};
\node(i2) at (1,-1.5){};
\node(i3)[reddot,proofdiagram] at (-0.75,-3){};
\node(i2a) at (1,-0.75){};
\node(i3a) at (1.75,-0.75){};
\node(o) at (0.75,2.5){};
\node(g1)[whitedot,projdot] at (-0.5,-2){};
\node(g2)[whitedot,scale=1,inner sep=0.25pt] at (-0.25,-1.75){$\boldsymbol{*}$};
\draw[green][string, line width=1pt](g1.center) to (g2.center);
\draw[string,in=left,out=90,looseness=0.6] (i1.center) to (b1.center);
\draw[string,in=90,out=right] (b1.center) to (chi.north);
\draw[string,in=left,out=90] (b1.center) to (add.center);
\draw[string,in=90,out=right] (add.center) to (x.center);
\draw[string,in=270,out=90,looseness=1] (i2.center) to (m.center);
\draw[string,in=right,out=90,looseness=1] (m.center) to (x.center);
\draw[string,in=270,out=90] (i3.center) to (g1.center);
\draw[string,in=270,out=90] (g2.center) to (b2.center);
\draw[string,in=270,out=left] (b2.center) to (chi.south);
\draw[string,in=left,out=right] (b2.center) to (x.center);
\node (chi7) [morphism,wedge,scale=0.5] at (2,3.1) {\large$\chi$};
\node(b17)[blackdot,proofdiagram]  at (1.5,4) {};
\node(add7)[reddot,proofdiagram] at (2,4.5){};
\node(x7)[yellowdot,proofdiagram] at (4.25,2.75){};
\node(b27)[blackdot,proofdiagram] at (3,2.25){};
\node(*7) at (2.25,2){};
\node(m7) at (4.5,2.25){};
\node(i17) at (0.75,1){};
\node(i27) at (2.5,-2){};
\node[reddot,proofdiagram](i37) at (3,-0.25){};
\node(i2a7) at (2.25,1.25){};
\node(i3a7) at (3,1.75){};
\node(o7) at (2,5){};
\node(g17)[whitedot,projdot] at (3,0.75){};
\node(g27)[whitedot,scale=1,inner sep=0.25pt] at (3,0.8){$\boldsymbol{*}$};
\draw[green][string, line width=1pt](g17.center) to (g27.center);
\draw[string,in=left,out=90,looseness=0.6] (add.center) to (b17.center);
\draw[string,in=90,out=right] (b17.center) to (chi7.north);
\draw[string,in=left,out=90] (b17.center) to (add7.center);
\draw[string,in=90,out=right] (add7.center) to (x7.center);
\draw[string,in=270,out=90] (add7.center) to (o7.center);
\draw[string,in=270,out=90,looseness=1] (i27.center) to (i2a7.center);
\draw[string,in=270,out=90,looseness=1] (i2a7.center) to (m7.center);
\draw[string,in=right,out=90,looseness=1] (m7.center) to (x7.center);
\draw[string,in=270,out=90] (i37.center) to (g17.center);
\draw[string,in=270,out=90] (g27.center) to (i3a7.center);
\draw[string,in=270,out=90] (i3a7.center) to (b27.center);
\draw[string,in=270,out=left] (b27.center) to (chi7.south);
\draw[string,in=left,out=right] (b27.center) to (x7.center);
\node(in2) at (-0.5,-4){};
\node(in1) at (1,-4){};
\draw[string,in=90,out=270] (i2.center) to (in1.center);
\draw[string,in=90,out=270,looseness=1] (i27.center) to (in2.center);
\node[reddot,proofdiagram](b2) at (2,-3.5){};
\node[reddot,proofdiagram](b4) at (3,-3.5){};
\node(in3) at (2,-4){};
\node(in4) at (3,-4){};
\draw[string] (b2.center) to (in3.center);
\draw[string] (b4.center) to (in4.center);
\end{pic}+
\begin{pic}[scale=1/3]
\node (chi) [morphism,wedge,scale=0.5] at (-0.5,0.1) {\large$\chi$};
\node(b1)[blackdot,proofdiagram]  at (-1,1) {};
\node(add)[reddot,proofdiagram] at (-0.5,1.5){};
\node(x)[yellowdot,proofdiagram] at (1.75,-0.25){};
\node(b2)[blackdot,proofdiagram] at (0.5,-0.75){};
\node(*) at (1,-0.5){};
\node(m) at (2,-0.75){};
\node(i1) at (-2,-4){};
\node(i2) at (1,-1.5){};
\node(i3)[reddot,proofdiagram] at (-0.75,-3){};
\node(i2a) at (1,-0.75){};
\node(i3a) at (1.75,-0.75){};
\node(o) at (0.75,2.5){};
\node(g1)[whitedot,projdot] at (-0.5,-2){};
\node(g2)[whitedot,scale=1,inner sep=0.25pt] at (-0.5,-2){$\boldsymbol{*}$};
\draw[green][string,line width=1pt](g1.center) to (g2.center);
\draw[string,in=left,out=90,looseness=0.6] (i1.center) to (b1.center);
\draw[string,in=90,out=right] (b1.center) to (chi.north);
\draw[string,in=left,out=90] (b1.center) to (add.center);
\draw[string,in=90,out=right] (add.center) to (x.center);
\draw[string,in=270,out=90,looseness=1] (i2.center) to (m.center);
\draw[string,in=right,out=90,looseness=1] (m.center) to (x.center);
\draw[string,in=270,out=90] (i3.center) to (g1.center);
\draw[string,in=270,out=90] (g2.center) to (b2.center);
\draw[string,in=270,out=left] (b2.center) to (chi.south);
\draw[string,in=left,out=right] (b2.center) to (x.center);
\node (chi7)  at (1.75,3.25) {};
\node(b17)  at (1.5,4) {};
\node(add7)[reddot,proofdiagram] at (1.5,4.5){};
\node(x7) at (3,3){};
\node(b27) at (2.25,2.5){};
\node(*7) at (2.25,2){};
\node(m7) at (3,2.5){};
\node(i17) at (0.75,1){};
\node(i27) at (2.5,-2){};
\node[reddot,proofdiagram](i37) at (3,0.5){};
\node(i2a7) at (2.25,1.75){};
\node(i3a7) at (3,1.75){};
\node(o7) at (1.5,5){};
\node(g17) at (3,0.75){};
\node(g27)[reddot,proofdiagram] at (3,1.5){};
\draw[string](g17.center) to (g27.center);
\draw[string,in=left,out=90,looseness=0.6] (add.center) to (add7.center);
\draw[string,in=90,out=right] (add7.center) to (*7.center);
\draw[string,in=270,out=90] (add7.center) to (o7.center);
\draw[string,in=270,out=90,looseness=1] (i27.center) to (*7.center);
\draw[string,in=270,out=90] (i37.center) to (g17.center);
\node(in2) at (-0.5,-4){};
\node(in1) at (1,-4){};
\draw[string,in=90,out=270] (i2.center) to (in1.center);
\draw[string,in=90,out=270,looseness=1] (i27.center) to (in2.center);
\node[reddot,proofdiagram](b2) at (2,-3.5){};
\node[reddot,proofdiagram](b4) at (3,-3.5){};
\node(in3) at (2,-4){};
\node(in4) at (3,-4){};
\draw[string] (b2.center) to (in3.center);
\draw[string] (b4.center) to (in4.center);
\end{pic}+
\begin{pic}[scale=1/3]
\node (chi)  at (-0.5,-0.15) {};
\node(b1)  at (0.25,1.5) {};
\node(add)[reddot,proofdiagram] at (0.75,2){};
\node(x) at (1.75,0.5){};
\node(b2) at (1,0){};
\node(*) at (1,-0.5){};
\node(m) at (1.75,0){};
\node(i1) at (-1,-4){};
\node(i2) at (1,-1.5){};
\node(i3)[reddot,proofdiagram] at (1.75,-2){};
\node(i2a) at (1,-0.75){};
\node(i3a) at (1.75,-0.75){};
\node(o) at (0.75,2.5){};
\node(g1) at (1.75,-1.75){};
\node(g2)[reddot,proofdiagram] at (1.75,-1){};
\draw[string](g1.center) to (g2.center);
\draw[string,in=left,out=90,looseness=0.6] (i1.center) to (add.center);
\draw[string,in=90,out=right] (add.center) to (x.center);
\draw[string,in=270,out=90,looseness=1] (i2.center) to (i2a.center);
\draw[string,in=270,out=90,looseness=1] (i2a.center) to (m.center);
\draw[string,in=270,out=90,looseness=1] (m.center) to (x.center);
\draw[string,in=270,out=90] (i3.center) to (g1.center);
\node (chi7) [morphism,wedge,scale=0.5] at (2,3.1) {\large$\chi$};
\node(b17)[blackdot,proofdiagram]  at (1.5,4) {};
\node(add7)[reddot,proofdiagram] at (2,4.5){};
\node(x7)[yellowdot,proofdiagram] at (4.25,2.75){};
\node(b27)[blackdot,proofdiagram] at (3,2.25){};
\node(*7) at (2.25,2){};
\node(m7) at (4.5,2.25){};
\node(i17) at (0.75,1){};
\node(i27) at (2.5,-2){};
\node[reddot,proofdiagram](i37) at (3,-0.25){};
\node(i2a7) at (2.25,1.25){};
\node(i3a7) at (3,1.75){};
\node(o7) at (2,5){};
\node(g17)[whitedot,projdot] at (3,0.75){};
\node(g27)[whitedot,scale=1,inner sep=0.25pt] at (3,0.8){$\boldsymbol{*}$};
\draw[green][string,line width=1pt](g17.center) to (g27.center);
\draw[string,in=left,out=90,looseness=0.6] (add.center) to (b17.center);
\draw[string,in=90,out=right] (b17.center) to (chi7.north);
\draw[string,in=left,out=90] (b17.center) to (add7.center);
\draw[string,in=90,out=right] (add7.center) to (x7.center);
\draw[string,in=270,out=90] (add7.center) to (o7.center);
\draw[string,in=270,out=90,looseness=1] (i27.center) to (i2a7.center);
\draw[string,in=270,out=90,looseness=1] (i2a7.center) to (m7.center);
\draw[string,in=right,out=90,looseness=1] (m7.center) to (x7.center);
\draw[string,in=270,out=90] (i37.center) to (g17.center);
\draw[string,in=270,out=90] (g27.center) to (i3a7.center);
\draw[string,in=270,out=90] (i3a7.center) to (b27.center);
\draw[string,in=270,out=left] (b27.center) to (chi7.south);
\draw[string,in=left,out=right] (b27.center) to (x7.center);
\node(in2) at (-0.5,-4){};
\node(in1) at (1,-4){};
\draw[string,in=90,out=270] (i2.center) to (in1.center);
\draw[string,in=90,out=270,looseness=1] (i27.center) to (in2.center);
\node[reddot,proofdiagram](b2) at (2,-3.5){};
\node[reddot,proofdiagram](b4) at (3,-3.5){};
\node(in3) at (2,-4){};
\node(in4) at (3,-4){};
\draw[string] (b2.center) to (in3.center);
\draw[string] (b4.center) to (in4.center);
\end{pic}+
\begin{pic}[scale=1/3]
\node (chi)  at (0.5,0.75) {};
\node(b1)  at (0.25,1.5) {};
\node(add)[reddot,proofdiagram] at (0.75,2){};
\node(x) at (1.75,0.5){};
\node(b2) at (1,0){};
\node(*) at (1,-0.5){};
\node(m) at (1.75,0){};
\node(i1) at (-1,-4){};
\node(i2) at (1,-1.5){};
\node(i3)[reddot,proofdiagram] at (1.75,-2){};
\node(i2a) at (1,-0.75){};
\node(i3a) at (1.75,-0.75){};
\node(o) at (0.75,2.5){};
\node(g1) at (1.75,-1.75){};
\node(g2)[reddot,proofdiagram] at (1.75,-1){};
\draw[string](g1.center) to (g2.center);
\draw[string,in=left,out=90,looseness=0.6] (i1.center) to (add.center);
\draw[string,in=90,out=right] (add.center) to (x.center);
\draw[string,in=270,out=90,looseness=1] (i2.center) to (i2a.center);
\draw[string,in=270,out=90,looseness=1] (i2a.center) to (m.center);
\draw[string,in=270,out=90,looseness=1] (m.center) to (x.center);
\draw[string,in=270,out=90] (i3.center) to (g1.center);
\node (chi7)  at (1.75,3.25) {};
\node(b17)  at (1.5,4) {};
\node(add7)[reddot,proofdiagram] at (1.5,4.5){};
\node(x7) at (3,3){};
\node(b27) at (2.25,2.5){};
\node(*7) at (2.25,2){};
\node(m7) at (3,2.5){};
\node(i17) at (0.75,1){};
\node(i27) at (2.5,-2){};
\node[reddot,proofdiagram](i37) at (3,0.5){};
\node(i2a7) at (2.25,1.75){};
\node(i3a7) at (3,1.75){};
\node(o7) at (1.5,5){};
\node(g17) at (3,0.75){};
\node(g27)[reddot,proofdiagram] at (3,1.5){};
\draw[string](g17.center) to (g27.center);
\draw[string,in=left,out=90,looseness=0.6] (add.center) to (add7.center);
\draw[string,in=90,out=right] (add7.center) to (*7.center);
\draw[string,in=270,out=90] (add7.center) to (o7.center);
\draw[string,in=270,out=90,looseness=1] (i27.center) to (*7.center);
\draw[string,in=270,out=90] (i37.center) to (g17.center);
\node(in2) at (-0.5,-4){};
\node(in1) at (1,-4){};
\draw[string,in=90,out=270] (i2.center) to (in1.center);
\draw[string,in=90,out=270,looseness=1] (i27.center) to (in2.center);
\node[reddot,proofdiagram](b2) at (2,-3.5){};
\node[reddot,proofdiagram](b4) at (3,-3.5){};
\node(in3) at (2,-4){};
\node(in4) at (3,-4){};
\draw[string] (b2.center) to (in3.center);
\draw[string] (b4.center) to (in4.center);
\end{pic}\quad = \quad
\begin{aligned}
\begin{tikzpicture}[scale=0.75]
\node (f) [morphism,wedge,scale=0.75, connect n ] at (0,1) {$U_{FF}$};
\node(i2)  at (-0.4125,-1.7) {};
\node(i2a)  at (0.4125,-1.7) {};
\node (e) [morphism,wedge, scale=0.75,connect sw length = 1.65cm] at (-0.24,-0.25){$U_{FF}$};
\node(m1) at (0.75,-0.25){} ;
\node(m2) at (1,-0.25) {};
\node[state,black,scale=0.25,label={[label distance=-0.125cm]315:\tiny{0}}](b1) at (0.75,-1.4){};
\node[state,black,scale=0.25,hflip,label={[label distance=-0.125cm]45:\tiny{0}}](b2) at (0.75,-2){};
\node[state,black,scale=0.25,label={[label distance=-0.125cm]315:\tiny{0}}](b3) at (1.25,-1.4){};
\node[state,black,scale=0.25,hflip,label={[label distance=-0.125cm]45:\tiny{0}}](b4) at (1.25,-2){};
\draw[string] (i2.center) to (e.south);
\draw[string] (e.north) to (f.south west);
\draw[string,in=270,out=90,looseness=1.25] (i2a.center) to (m1.center);
\draw[string,in=270,out=90,looseness=1.25] (m1.center) to (f.south);
\draw[string,in=270,out=90] (b1.center) to (e.south east);
\draw[string,out=90,in=270] (b3.center) to (m2.center);
\draw[string,out=90,in=270] (m2.center) to (f.south east);
\draw[string,out=270,in=90] (i2.center) to +(0.825,-0.5);
\draw[string,out=270,in=90] (i2a.center) to +(-0.825,-0.5);
\draw[string] (b2.center) to (0.75,-2.2);
\draw[string] (b4.center) to (1.25,-2.2);
\end{tikzpicture}
\end{aligned}
\end{align*}
We now show that:
\begin{equation}\label{eq:uffnzero}
\begin{aligned}
\begin{tikzpicture}[scale=0.75]
\node (f) [morphism,wedge,scale=0.75, connect n ] at (0,1) {$U_{FF}$};
\node(i2)  at (-0.4125,-2.2) {};
\node(i2a)  at (0.4125,-2.2) {};
\node[whitedot,scale=1,inner sep=0.25pt](i3)  at (0.75,-1.7) {$\boldsymbol{*}$};
\node[whitedot,scale=1,inner sep=0.25pt](i4)  at (1.25,-1.7) {$\boldsymbol{*}$};
\node (e) [morphism,wedge, scale=0.75,connect sw length = 1.65cm] at (-0.24,-0.25) {$U_{FF}$};
\node(b)[blackdot] at (0,-1.1){};
\node(c)[blackdot] at (0,-1.7){};
\node(m1) at (0.75,-0.25) {};
\node(m2) at (1,-0.25) {};
\draw[string,out=90,in=left] (i2.center) to (c.center);
\draw[string,out=90,in=right] (i2a.center) to (c.center);
\draw[string] (b.center) to (c.center);
\draw[string] (e.north) to (f.south west);
\draw[string,in=270,out=left] (b.center) to (e.south);
\draw[string,in=270,out=right,looseness=1.25] (b.center) to (m1.center);
\draw[string,in=270,out=90,looseness=1.25] (m1.center) to (f.south);
\draw[string,in=270,out=90] (i3.center) to (e.south east);
\draw[string,out=90,in=270] (i4.center) to (m2.center);
\draw[string,out=90,in=270] (m2.center) to (f.south east);
\draw[string] (i3.center) to +(0,-0.5);
\draw[string] (i4.center) to +(0,-0.5);
\end{tikzpicture}
\end{aligned}
\quad = \quad
\begin{aligned}
\begin{tikzpicture}[scale=0.75]
\node (f) [morphism,wedge,scale=0.75, connect n ] at (0,1) {$U_{FF}$};
\node(i2)  at (-0.4125,-2.2) {};
\node(i2a)  at (0.4125,-2.2) {};
\node[whitedot,scale=1,inner sep=0.25pt](i3)  at (0.75,-1.7) {$\boldsymbol{*}$};
\node[whitedot,scale=1,inner sep=0.25pt](i4)  at (1.25,-1.7) {$\boldsymbol{*}$};
\node (e) [morphism,wedge,scale=0.75,connect sw length = 1.65cm] at (-0.24,-0.25){$U_{FF}$};
\node(b)[blackdot,scale=0.95] at (0,-1.3){};
\node(c)[blackdot,scale=0.95] at (0,-1.7){};
\node(m1) at (0.75,-0.25) {};
\node(m2) at (1,-0.25) {};
\draw[string,out=90,in=left] (i2.center) to (c.center);
\draw[string,out=90,in=right] (i2a.center) to (c.center);
\draw[string] (b.center) to (c.center);
\draw[string] (e.north) to (f.south west);
\draw[string,in=270,out=left] (b.center) to (e.south);
\draw[string,in=270,out=right,looseness=1.25] (b.center) to (m1.center);
\draw[string,in=270,out=90,looseness=1.25] (m1.center) to (f.south);
\draw[string,in=270,out=90] (i3.center) to (e.south east);
\draw[string,out=90,in=270] (i4.center) to (m2.center);
\draw[string,out=90,in=270] (m2.center) to (f.south east);
\draw[string,out=270,in=90] (i3.center) to +(0.5,-0.5);
\draw[string,out=270,in=90] (i4.center) to +(-0.5,-0.5);
\end{tikzpicture}
\end{aligned}
\end{equation}
\begin{align*}
\begin{aligned}
\begin{tikzpicture}[scale=11/12]
\node (f) [morphism,wedge,scale=0.75, connect n ] at (0,1) {$U_{FF}$};
\node(i2)  at (-0.4125,-2.2) {};
\node(i2a)  at (0.4125,-2.2) {};
\node[whitedot,scale=1,inner sep=0.25pt](i3)  at (0.75,-1.7) {$\boldsymbol{*}$};
\node[whitedot,scale=1,inner sep=0.25pt](i4)  at (1.25,-1.7) {$\boldsymbol{*}$};
\node (e) [morphism,wedge, scale=0.75,connect sw length = 1cm] at (-0.24,-0.25) {$U_{FF}$};
\node(b)[blackdot] at (0,-1.1){};
\node(c)[blackdot] at (0,-1.7){};
\node(m1) at (0.75,-0.25) {};
\node(m2) at (1,-0.25) {};
\draw[string,out=90,in=left] (i2.center) to (c.center);
\draw[string,out=90,in=right] (i2a.center) to (c.center);
\draw[string] (b.center) to (c.center);
\draw[string] (e.north) to (f.south west);
\draw[string,in=270,out=left] (b.center) to (e.south);
\draw[string,in=270,out=right,looseness=1.25] (b.center) to (m1.center);
\draw[string,in=270,out=90,looseness=1.25] (m1.center) to (f.south);
\draw[string,in=270,out=90] (i3.center) to (e.south east);
\draw[string,out=90,in=270] (i4.center) to (m2.center);
\draw[string,out=90,in=270] (m2.center) to (f.south east);
\draw[string] (i3.center) to +(0,-0.5);
\draw[string] (i4.center) to +(0,-0.5);
\node(i)  at (-0.6,-2.2) {};
\draw[string,in=270,out=90] (i.center) to (e.connect sw);
\end{tikzpicture}
\end{aligned}
:=& \quad 
\begin{pic}[scale=1/3]
\node (chi) [morphism,wedge,scale=0.5] at (0.4,1.2) {\large$\chi$};
\node(b1)[blackdot,proofdiagram]  at (-0.1,2.1) {};
\node(add)[reddot,proofdiagram] at (0.4,2.6){};
\node(x)[yellowdot,proofdiagram] at (1.85,0.4){};
\node(b2)[blackdot,proofdiagram] at (0.9,0){};
\node(*) at (1,-0.5){};
\node(m) at (2.15,0.15){};
\node(i1) at (-2,-5){};
\node(i2) at (-0.5,-2){};
\node(i3)[whitedot,projdot] at (2,-3.5){};
\node(i2a) at (1,-0.75){};
\node(i3a) at (1.75,-0.75){};
\node(o) at (0.75,2.5){};
\node(g1)[whitedot,projdot] at (1.75,-1.35){};
\node(g2)[whitedot,inner sep=0.25pt] at (1.75,-1.3){$\boldsymbol{*}$};
\draw[green][string,line width=1pt](g1.center) to (g2.center);
\draw[string,in=left,out=90,looseness=0.6] (i1.center) to (b1.center);
\draw[string,in=90,out=right] (b1.center) to (chi.north);
\draw[string,in=left,out=90] (b1.center) to (add.center);
\draw[string,in=90,out=right] (add.center) to (x.center);
\draw[string,in=270,out=90,looseness=1] (i2.center) to (i2a.center);
\draw[string,in=270,out=90,looseness=1] (i2a.center) to (m.center);
\draw[string,in=right,out=90,looseness=1] (m.center) to (x.center);
\draw[string,in=270,out=90] (i3.center) to (g1.center);
\draw[string,in=270,out=90] (g2.center) to (i3a.center);
\draw[string,in=270,out=90] (i3a.center) to (b2.center);
\draw[string,in=270,out=left] (b2.center) to (chi.south);
\draw[string,in=left,out=right] (b2.center) to (x.center);
\node (chi7) [morphism,wedge,scale=0.5] at (1.4,3.6) {\large$\chi$};
\node(b17)[blackdot,proofdiagram]  at (0.95,4.6) {};
\node(add7)[reddot,proofdiagram] at (2,5.1){};
\node(x7)[yellowdot,proofdiagram] at (3.1,3){};
\node(b27)[blackdot,proofdiagram] at (2.25,2.5){};
\node(*7) at (2.25,2){};
\node(m7) at (3.4,2.75){};
\node(i17) at (0.75,1){};
\node(i27) at (2.5,-1){};
\node(i37)[whitedot,projdot] at (3,-3.5){};
\node(i2a7) at (2.25,1.5){};
\node(i3a7) at (3,1.75){};
\node(o7) at (2,6){};
\node(g17)[whitedot,projdot] at (3,0.75){};
\node(g27)[whitedot,scale=1,inner sep=0.25pt] at (3,0.8){$\boldsymbol{*}$};
\draw[green][string,line width=1pt](g17.center) to (g27.center);
\draw[string,in=left,out=90,looseness=0.6] (add.center) to (b17.center);
\draw[string,in=90,out=right] (b17.center) to (chi7.north);
\draw[string,in=left,out=90] (b17.center) to (add7.center);
\draw[string,in=90,out=right] (add7.center) to (x7.center);
\draw[string,in=270,out=90] (add7.center) to (o7.center);
\draw[string,in=270,out=90,looseness=1] (i27.center) to (i2a7.center);
\draw[string,in=270,out=90,looseness=1] (i2a7.center) to (m7.center);
\draw[string,in=right,out=90,looseness=1] (m7.center) to (x7.center);
\draw[string,in=270,out=90] (i37.center) to (g17.center);
\draw[string,in=270,out=90] (g27.center) to (i3a7.center);
\draw[string,in=270,out=90] (i3a7.center) to (b27.center);
\draw[string,in=270,out=left] (b27.center) to (chi7.south);
\draw[string,in=left,out=right] (b27.center) to (x7.center);
\node(in1) at (-0.5,-5){};
\node(in2) at (1,-5){};
\node(b2)[whitedot,inner sep=0.25pt] at (2,-3.5){$\boldsymbol{*}$};
\node(b4)[whitedot,inner sep=0.25pt] at (3,-3.5){$\boldsymbol{*}$};
\node(in3) at (2,-5){};
\node(in4) at (3,-5){};
\draw[string] (b2.center) to (in3.center);
\draw[string] (b4.center) to (in4.center);
\draw[green][string,line width=1pt] (b4.center) to (i37.center);
\draw[green][string,line width=1pt] (b2.center) to (i3.center);
\node(t1)[blackdot,proofdiagram] at (0.25,-3.5){};
\node(t2)[blackdot,proofdiagram] at (0.25,-2.75){};
\draw[string] (t1.center) to (t2.center);
\draw[string,out=180,in=90] (t1.center) to (in1.center);
\draw[string,out=0,in=90] (t1.center) to (in2.center);
\draw[string,in=180,out=270] (i2.center) to (t2.center);
\draw[string,in=0,out=270,looseness=1] (i27.center) to (t2.center);
\end{pic}+
\begin{pic}[scale=1/3]
\node (chi) [morphism,wedge,scale=0.5] at (0.4,1.2) {\large$\chi$};
\node(b1)[blackdot,proofdiagram]  at (-0.1,2.1) {};
\node(add)[reddot,proofdiagram] at (0.4,2.6){};
\node(x)[yellowdot,proofdiagram] at (1.85,0.4){};
\node(b2)[blackdot,proofdiagram] at (0.9,0){};
\node(*) at (1,-0.5){};
\node(m) at (2.1,0.15){};
\node(i1) at (-2,-5){};
\node(i2) at (-0.5,-2){};
\node(i3) at (2,-3.5){};
\node(i2a) at (1,-0.75){};
\node(i3a) at (1.75,-0.75){};
\node(o) at (0.75,2.5){};
\node(g1)[whitedot,inner sep=0.25pt] at  (1.75,-1.325){$\boldsymbol{*}$};
\draw[string,in=left,out=90,looseness=0.6] (i1.center) to (b1.center);
\draw[string,in=90,out=right] (b1.center) to (chi.north);
\draw[string,in=left,out=90] (b1.center) to (add.center);
\draw[string,in=90,out=right] (add.center) to (x.center);
\draw[string,in=270,out=90,looseness=1] (i2.center) to (i2a.center);
\draw[string,in=270,out=90,looseness=1] (i2a.center) to (m.center);
\draw[string,in=right,out=90,looseness=1] (m.center) to (x.center);
\draw[string,in=270,out=90] (i3.center) to (g1.center);
\draw[string,in=270,out=90] (g2.center) to (i3a.center);
\draw[string,in=270,out=90] (i3a.center) to (b2.center);
\draw[string,in=270,out=left] (b2.center) to (chi.south);
\draw[string,in=left,out=right] (b2.center) to (x.center);
\node (chi7)  at (1.75,3.25) {};
\node(b17)  at (1.5,4) {};
\node(add7)[reddot,proofdiagram] at (1.5,4.5){};
\node(x7) at (3,3){};
\node(b27) at (2.25,2.5){};
\node(*7) at (2.25,2){};
\node(m7) at (3,2.5){};
\node(i17) at (0.75,1){};
\node(i27) at (2.5,-1){};
\node(i2a7) at (2.25,1.75){};
\node(i3a7) at (3,1.75){};
\node(o7) at (1.5,6){};
\node(g17) at (3,0.75){};
\draw[string,in=left,out=90,looseness=0.6] (add.center) to (add7.center);
\draw[string,in=90,out=right] (add7.center) to (*7.center);
\draw[string,in=270,out=90] (add7.center) to (o7.center);
\draw[string,in=270,out=90,looseness=1] (i27.center) to (*7.center);
\node(in1) at (-0.5,-5){};
\node(in2) at (1,-5){};
\draw[string,in=180,out=270] (i2.center) to (t2.center);
\draw[string,in=0,out=270,looseness=1] (i27.center) to (t2.center);
\node(b2)[whitedot,inner sep=0.25pt] at (2,-3.5){$\boldsymbol{*}$};
\node[reddot,proofdiagram](b4) at (3,-2){};
\node(u1)[whitedot,inner sep=0.25pt] at (3,-3.5){$\boldsymbol{*}$};
\node(u2) at (3,-3.5){};
\draw[string,line width=1pt] (u1.center) to (u2.center);
\draw[string,line width=1pt] (b2.center) to (i3.center);
\node(in3) at (2,-5){};
\node(in4) at (3,-5){};
\draw[string] (b2.center) to (in3.center);
\draw[string] (b4.center) to (in4.center);
\node(t1)[blackdot,proofdiagram] at (0.25,-3.5){};
\node(t2)[blackdot,proofdiagram] at (0.25,-2.75){};
\draw[string] (t1.center) to (t2.center);
\draw[string,out=180,in=90] (t1.center) to (in1.center);
\draw[string,out=0,in=90] (t1.center) to (in2.center);
\end{pic}+
\begin{pic}[scale=1/3]
\node (chi)  at (0.5,0.75) {};
\node(b1)  at (0.25,1.5) {};
\node(add)[reddot,proofdiagram] at (0.4,2){};
\node(x) at (1.75,0.5){};
\node(b2) at (1,0){};
\node(*) at (1,-0.5){};
\node(m) at (1.75,0){};
\node(i1) at (-2,-5){};
\node(i2) at (-0.5,-2){};
\node(i2a) at (1,-0.75){};
\node(i3a) at (1.75,-0.75){};
\node(o) at (0.75,2.5){};
\draw[string,in=left,out=90,looseness=0.6] (i1.center) to (add.center);
\draw[string,in=90,out=right] (add.center) to (x.center);
\draw[string,in=270,out=90,looseness=1] (i2.center) to (i2a.center);
\draw[string,in=270,out=90,looseness=1] (i2a.center) to (m.center);
\draw[string,in=270,out=90,looseness=1] (m.center) to (x.center);
\node (chi7) [morphism,wedge,scale=0.5] at (1.4,3.6) {\large$\chi$};
\node(b17)[blackdot,proofdiagram]  at (1,4.6) {};
\node(add7)[reddot,proofdiagram] at (2,5.1){};
\node(x7)[yellowdot,proofdiagram] at (3.1,3){};
\node(b27)[blackdot,proofdiagram] at (2.25,2.5){};
\node(*7) at (2.25,2){};
\node(m7) at (3.4,2.75){};
\node(i17) at (0.75,1){};
\node(i27) at (2.5,0){};
\node(i37) at (3,-2.75){};
\node(i2a7) at (2.25,1.75){};
\node(i3a7) at (3,1.75){};
\node(o7) at (2,6){};
\node(g17) at (3,0.75){};
\node(g27)[whitedot,scale=1,inner sep=0.25pt] at (3,1){$\boldsymbol{*}$};
\draw[green][string,line width=1pt](g17.center) to (g27.center);
\draw[string,in=left,out=90,looseness=0.6] (add.center) to (b17.center);
\draw[string,in=90,out=right] (b17.center) to (chi7.north);
\draw[string,in=left,out=90] (b17.center) to (add7.center);
\draw[string,in=90,out=right] (add7.center) to (x7.center);
\draw[string,in=270,out=90] (add7.center) to (o7.center);
\draw[string,in=270,out=90,looseness=1] (i27.center) to (i2a7.center);
\draw[string,in=270,out=90,looseness=1] (i2a7.center) to (m7.center);
\draw[string,in=right,out=90,looseness=1] (m7.center) to (x7.center);
\draw[string,in=270,out=90] (i37.center) to (g17.center);
\draw[string,in=270,out=90] (g27.center) to (i3a7.center);
\draw[string,in=270,out=90] (i3a7.center) to (b27.center);
\draw[string,in=270,out=left] (b27.center) to (chi7.south);
\draw[string,in=left,out=right] (b27.center) to (x7.center);
\node(in1) at (-0.5,-5){};
\node(in2) at (1,-5){};
\draw[string,in=180,out=270] (i2.center) to (t2.center);
\draw[string,in=0,out=270,looseness=0.7] (i27.center) to (t2.center);
\node[reddot,proofdiagram](b2) at (2,-2){};
\node(u1)[whitedot,inner sep=0.25pt] at (2,-3.5){$\boldsymbol{*}$};
\node(u2)[whitedot,inner sep=0.25pt] at (2,-3.5){$\boldsymbol{*}$};
\draw[string] (b2.center) to (i3.center);
\node(b4)[whitedot,inner sep=0.25pt] at (3,-3.5){$\boldsymbol{*}$};
\draw[string] (b4.center) to (i37.center);
\node(in3) at (2,-5){};
\node(in4) at (3,-5){};
\draw[string] (b2.center) to (in3.center);
\draw[string] (b4.center) to (in4.center);
\node(t1)[blackdot,proofdiagram] at (0.25,-3.5){};
\node(t2)[blackdot,proofdiagram] at (0.25,-2.75){};
\draw[string] (t1.center) to (t2.center);
\draw[string,out=180,in=90] (t1.center) to (in1.center);
\draw[string,out=0,in=90] (t1.center) to (in2.center);
\draw[green][string,line width=1pt] (u1.center) to (u2.center);
\end{pic}+
\begin{pic}[scale=1/3]
\node (chi)  at (0.5,0.75) {};
\node(b1)  at (0.25,1.5) {};
\node(add)[reddot,proofdiagram] at (0.75,2){};
\node(x) at (1.75,0.5){};
\node(b2) at (1,0){};
\node(*) at (1,-0.5){};
\node(m) at (1.75,0){};
\node(i1) at (-2,-5){};
\node(i2) at (-0.5,-2){};
\node(i2a) at (1,-0.75){};
\node(i3a) at (1.75,-0.75){};
\node(o) at (0.75,2.5){};
\draw[string,in=left,out=90,looseness=0.6] (i1.center) to (add.center);
\draw[string,in=90,out=right] (add.center) to (x.center);
\draw[string,in=270,out=90,looseness=1] (i2.center) to (i2a.center);
\draw[string,in=270,out=90,looseness=1] (i2a.center) to (m.center);
\draw[string,in=270,out=90,looseness=1] (m.center) to (x.center);
\node (chi7)  at (1.75,3.25) {};
\node(b17)  at (1.5,4) {};
\node(add7)[reddot,proofdiagram] at (1.5,4.5){};
\node(x7) at (3,3){};
\node(b27) at (2.25,2.5){};
\node(*7) at (2.25,2){};
\node(m7) at (3,2.5){};
\node(i17) at (0.75,1){};
\node(i27) at (2.325,0){};
\node(i2a7) at (2.25,1.75){};
\node(i3a7) at (3,1.75){};
\node(o7) at (1.5,6){};
\draw[string,in=left,out=90,looseness=0.6] (add.center) to (add7.center);
\draw[string,in=90,out=right] (add7.center) to (*7.center);
\draw[string,in=270,out=90] (add7.center) to (o7.center);
\draw[string,in=270,out=90,looseness=1] (i27.center) to (*7.center);
\node(in1) at (-0.5,-5){};
\node(in2) at (1,-5){};
\draw[string,in=180,out=270] (i2.center) to (t2.center);
\draw[string,in=0,out=270,looseness=0.7] (i27.center) to (t2.center);
\node[reddot,proofdiagram](b2) at (2,-2){};
\node[reddot,proofdiagram](b4) at (3,-2){};
\node(u1)[whitedot,inner sep=0.25pt] at (2,-3.5){$\boldsymbol{*}$};
\node(p1)[whitedot,inner sep=0.25pt] at (3,-3.5){$\boldsymbol{*}$};
\node(u2) at (2,-3.5){};
\node(p2) at (3,-3.5){};
\node(in3) at (2,-5){};
\node(in4) at (3,-5){};
\draw[string] (u1.center) to (in3.center);
\draw[string] (b2.center) to (u2.center);
\draw[string] (b4.center) to (in4.center);
\draw[string,green] (u2.center) to (u1.center);
\draw[string,green] (p2.center) to (p1.center);
\node(t1)[blackdot,proofdiagram] at (0.25,-3.5){};
\node(t2)[blackdot,proofdiagram] at (0.25,-2.75){};
\draw[string] (t1.center) to (t2.center);
\draw[string,out=180,in=90] (t1.center) to (in1.center);
\draw[string,out=0,in=90] (t1.center) to (in2.center);
\end{pic}\\  \super{\eqref{eq:proj2}\eqref{eq:projprop}}=& \quad
\begin{pic}[scale=1/3]
\node (chi) [morphism,wedge,scale=0.5] at (0.4,1.2) {\large$\chi$};
\node(b1)[blackdot,proofdiagram]  at (-0.1,2.1) {};
\node(add)[reddot,proofdiagram] at (0.4,2.6){};
\node(x)[yellowdot,proofdiagram] at (1.85,0.4){};
\node(b2)[blackdot,proofdiagram] at (0.9,0){};
\node(*) at (1,-0.5){};
\node(m) at (2.15,0.15){};
\node(i1) at (-2,-5){};
\node(i2) at (-0.5,-2){};
\node(i3) at (2,-2.75){};
\node(i2a) at (1,-0.75){};
\node(i3a) at (1.75,-0.75){};
\node(o) at (0.75,2.5){};
\node(g1) at (1.75,-1.75){};
\node(g2) at (1.75,-1){};
\draw[string](g1.center) to (g2.center);
\draw[string,in=left,out=90,looseness=0.6] (i1.center) to (b1.center);
\draw[string,in=90,out=right] (b1.center) to (chi.north);
\draw[string,in=left,out=90] (b1.center) to (add.center);
\draw[string,in=90,out=right] (add.center) to (x.center);
\draw[string,in=270,out=90,looseness=1] (i2.center) to (i2a.center);
\draw[string,in=270,out=90,looseness=1] (i2a.center) to (m.center);
\draw[string,in=right,out=90,looseness=1] (m.center) to (x.center);
\draw[string,in=270,out=90] (i3.center) to (g1.center);
\draw[string,in=270,out=90] (g2.center) to (i3a.center);
\draw[string,in=270,out=90] (i3a.center) to (b2.center);
\draw[string,in=270,out=left] (b2.center) to (chi.south);
\draw[string,in=left,out=right] (b2.center) to (x.center);
\node (chi7) [morphism,wedge,scale=0.5] at (1.4,3.6) {\large$\chi$};
\node(b17)[blackdot,proofdiagram]  at (1,4.6) {};
\node(add7)[reddot,proofdiagram] at (2,5.1){};
\node(x7)[yellowdot,proofdiagram] at (3.1,3){};
\node(b27)[blackdot,proofdiagram] at (2.25,2.5){};
\node(*7) at (2.25,2){};
\node(m7) at (3.4,2.75){};
\node(i17) at (0.75,1){};
\node(i27) at (2.5,-1){};
\node(i37) at (3,-2.75){};
\node(i2a7) at (2.25,1.5){};
\node(i3a7) at (3,1.75){};
\node(o7) at (2,6){};
\node(g17) at (3,0.75){};
\node(g27) at (3,1.5){};
\draw[string](g17.center) to (g27.center);
\draw[string,in=left,out=90,looseness=0.6] (add.center) to (b17.center);
\draw[string,in=90,out=right] (b17.center) to (chi7.north);
\draw[string,in=left,out=90] (b17.center) to (add7.center);
\draw[string,in=90,out=right] (add7.center) to (x7.center);
\draw[string,in=270,out=90] (add7.center) to (o7.center);
\draw[string,in=270,out=90,looseness=1] (i27.center) to (i2a7.center);
\draw[string,in=270,out=90,looseness=1] (i2a7.center) to (m7.center);
\draw[string,in=right,out=90,looseness=1] (m7.center) to (x7.center);
\draw[string,in=270,out=90] (i37.center) to (g17.center);
\draw[string,in=270,out=90] (g27.center) to (i3a7.center);
\draw[string,in=270,out=90] (i3a7.center) to (b27.center);
\draw[string,in=270,out=left] (b27.center) to (chi7.south);
\draw[string,in=left,out=right] (b27.center) to (x7.center);
\node(in1) at (-0.5,-5){};
\node(in2) at (1,-5){};
\draw[string,in=180,out=270] (i2.center) to (t2.center);
\draw[string,in=0,out=270,looseness=1] (i27.center) to (t2.center);
\node(b2)[whitedot,inner sep=0.25pt] at (2,-3.5){$\boldsymbol{*}$};
\node(b4)[whitedot,inner sep=0.25pt] at (3,-3.5){$\boldsymbol{*}$};
\node(in3) at (2,-5){};
\node(in4) at (3,-5){};
\draw[string] (b2.center) to (in3.center);
\draw[string] (b4.center) to (in4.center);
\draw[string] (b4.center) to (i37.center);
\draw[string] (b2.center) to (i3.center);
\node(t1)[blackdot,proofdiagram] at (0.25,-3.5){};
\node(t2)[blackdot,proofdiagram] at (0.25,-2.75){};
\draw[string] (t1.center) to (t2.center);
\draw[string,out=180,in=90] (t1.center) to (in1.center);
\draw[string,out=0,in=90] (t1.center) to (in2.center);
\end{pic}\quad + \quad 0 \quad + \quad 0 \quad + \quad 0 \\  \super{\eqref{1}}=& \quad 
\begin{pic}[scale=1/3]
\node(chi) [morphism,wedge,scale=0.5,connect n length=0.35cm,connect s length=0.35cm] at (-1.25,-0.1) {\large$\chi$};
\node(b1)[blackdot,proofdiagram]  at (-1.5,1.5) {};
\node(add)[reddot,proofdiagram] at (-1,2){};
\node(x)[yellowdot,proofdiagram] at (1.5,-0.25){};
\node(b2)[blackdot,proofdiagram] at (0,-2){};
\node(y)[yellowdot,proofdiagram] at (chi.connect s){};
\node(y2)[yellowdot,proofdiagram] at (chi.connect n){};
\node(k)[blackdot,proofdiagram] at (-0.4,-1.35){};
\node(i1) at (-4,-4){};
\node(i3)[whitedot,projdot] at (2,-2.75){};
\node(mm) at (-1.9,-1.35){};
\node(t1)[blackdot,proofdiagram] at (0.25,-3.5){};
\node(t2)[blackdot,proofdiagram] at (2.25,-1){};
\node (chi7) [morphism,wedge,scale=0.5,connect n length=0.35cm,connect s length=0.35cm] at (1.75,3.9) {\large$\chi$};
\node(b17)[blackdot,proofdiagram]  at (1.5,5.5) {};
\node(add7)[reddot,proofdiagram] at (2,6){};
\node(x7)[yellowdot,proofdiagram] at (4.5,3.75){};
\node(b27)[blackdot,proofdiagram] at (3,2){};
\node(y7)[yellowdot,proofdiagram] at (chi7.connect s){};
\node(y27)[yellowdot,proofdiagram] at (chi7.connect n){};
\node(k7)[blackdot,proofdiagram] at (2.6,2.65){};
\node(m7) at (5,2.5){};
\node(i17) at (0.75,1){};
\node(i37)[whitedot,projdot] at (3,-2.75){};
\node(o7) at (2,7){};
\node(mm7) at (1.1,2.65){};
\node(in1) at (-0.5,-4){};
\node(in2) at (1,-4){};
\node(b2k)[whitedot,inner sep=0.25pt] at (2,-3.1){$\boldsymbol{*}$};
\node(b4k)[whitedot,inner sep=0.25pt] at (3,-3.1){$\boldsymbol{*}$};
\node(in3) at (2,-4){};
\node(in4) at (3,-4){};
\draw[string,out=180,in=0] (k.center) to (y.center);
\draw[string,in=left,out=90,looseness=0.6] (i1.center) to (b1.center);
\draw[string,in=0,out=right,looseness=1.5] (b1.center) to (k.center);
\draw[string,in=left,out=90] (b1.center) to (add.center);
\draw[string,in=90,out=right] (add.center) to (x.center);
\draw[string,in=270,out=90] (i3.center) to (b2.center);
\draw[string,in=270,out=left] (b2.center) to (mm.center);
\draw[string,in=180,out=90] (mm.center) to (y.center);
\draw[string,in=left,out=right] (b2.center) to (x.center);
\draw[string,out=180,in=0] (k7.center) to (y7.center);
\draw[string,in=0,out=right,looseness=1.5] (b17.center) to (k7.center);
\draw[string,in=270,out=left] (b27.center) to (mm7.center);
\draw[string,in=180,out=90] (mm7.center) to (y7.center);
\draw[string,in=left,out=90,looseness=0.6] (add.center) to (b17.center);
\draw[string,in=left,out=90] (b17.center) to (add7.center);
\draw[string,in=90,out=right] (add7.center) to (x7.center);
\draw[string,in=270,out=90] (add7.center) to (o7.center);
\draw[string,in=right,out=90,looseness=1] (m7.center) to (x7.center);
\draw[string,in=270,out=90] (i37.center) to (b27.center);
\draw[string,in=left,out=right] (b27.center) to (x7.center);
\draw[string,in=180,out=0] (x.center) to (t2.center);
\draw[string,in=0,out=270,looseness=1] (m7.center) to (t2.center);
\draw[string] (b2k.center) to (in3.center);
\draw[string] (b4k.center) to (in4.center);
\draw[green][string,line width=1pt] (b4k.center) to (i37.center);
\draw[green][string,line width=1pt] (b2k.center) to (i3.center);
\draw[string,out=90,in=270] (t1.center) to (t2.center);
\draw[string,out=180,in=90] (t1.center) to (in1.center);
\draw[string,out=0,in=90] (t1.center) to (in2.center);
\end{pic}\qquad  \super{Black, yellow\eqref{eq:sm}}= \quad 
\begin{pic}[scale=1/3]
\node(chi) [morphism,wedge,scale=0.5,connect n length=0.35cm,connect s length=0.35cm] at (-4,5.5) {\large$\chi$};
\node(b1)[blackdot,proofdiagram]  at (-4,-2) {};
\node(add)[reddot,proofdiagram] at (1,2.5){};
\node(x)[yellowdot,proofdiagram] at (2.5,1.75){};
\node(b2)[blackdot,proofdiagram] at (2,-2){};
\node(y)[yellowdot,proofdiagram] at (chi.connect s){};
\node(y2)[yellowdot,proofdiagram] at (chi.connect n){};
\node(i1) at (-4,-4){};
\node(i3)[whitedot,projdot] at (2,-2.75){};
\node(mm) at (-3,2){};
\node(t1)[blackdot,proofdiagram] at (0.25,-3.5){};
\node(t2)[blackdot,proofdiagram] at (4.5,1){};
\node (chi7) [morphism,wedge,scale=0.5,connect n length=0.35cm,connect s length=0.35cm] at (-1.8,5.5) {\large$\chi$};
\node(b17)[blackdot,proofdiagram]  at (1,3.5) {};
\node(add7)[reddot,proofdiagram] at (2,6){};
\node(x7)[yellowdot,proofdiagram] at (4.5,3){};
\node(b27)[blackdot,proofdiagram] at (3,-1.5){};
\node(y7)[yellowdot,proofdiagram] at (chi7.connect s){};
\node(y27)[yellowdot,proofdiagram] at (chi7.connect n){};
\node(m7) at (5,2.5){};
\node(i17) at (0.75,1){};
\node(i37)[whitedot,projdot] at (3,-2.75){};
\node(o7) at (2,7){};
\node(mm7) at (-2.5,4){};
\node(in1) at (-0.5,-4){};
\node(in2) at (1,-4){};
\node(b2k)[whitedot,inner sep=0.25pt] at (2,-3.1){$\boldsymbol{*}$};
\node(b4k)[whitedot,inner sep=0.25pt] at (3,-3.1){$\boldsymbol{*}$};
\node(in3) at (2,-4){};
\node(in4) at (3,-4){};
\draw[string,in=270,out=90,looseness=0.6] (i1.center) to (b1.center);
\draw[string,in=180,out=180,looseness=1] (b1.center) to (y.center);
\draw[string,in=left,out=0] (b1.center) to (add.center);
\draw[string,in=90,out=right] (add.center) to (x.center);
\draw[string,in=270,out=90] (i3.center) to (b2.center);
\draw[string,in=270,out=left] (b2.center) to (mm.center);
\draw[string,in=0,out=90] (mm.center) to (y.center);
\draw[string,in=left,out=right] (b2.center) to (x.center);
\draw[string,out=180,in=0] (b17.center) to (y7.center);
\draw[string,in=270,out=left] (b27.center) to (mm7.center);
\draw[string,in=180,out=90] (mm7.center) to (y7.center);
\draw[string] (add.center) to (b17.center);
\draw[string,in=left,out=0] (b17.center) to (add7.center);
\draw[string,in=90,out=right] (add7.center) to (x7.center);
\draw[string,in=270,out=90] (add7.center) to (o7.center);
\draw[string,in=right,out=0,looseness=1] (t2.center) to (x7.center);
\draw[string,in=270,out=90] (i37.center) to (b27.center);
\draw[string,in=left,out=right] (b27.center) to (x7.center);
\draw[string,in=180,out=0] (x.center) to (t2.center);
\draw[string] (b2k.center) to (in3.center);
\draw[string] (b4k.center) to (in4.center);
\draw[green][string,line width=1pt] (b4k.center) to (i37.center);
\draw[green][string,line width=1pt] (b2k.center) to (i3.center);
\draw[string,out=90,in=270] (t1.center) to (t2.center);
\draw[string,out=180,in=90] (t1.center) to (in1.center);
\draw[string,out=0,in=90] (t1.center) to (in2.center);
\end{pic}\quad \\  \super{Red, black\eqref{eq:bialg}}=& \quad 
\begin{pic}[scale=1/3]
\node(chi) [morphism,wedge,scale=0.5,connect n length=0.35cm,connect s length=0.35cm] at (-4,5.5) {\large$\chi$};
\node(b1)[blackdot,proofdiagram]  at (-4,-2) {};
\node(add1)[blackdot,proofdiagram] at (0.5,2.5){};
\node(add2)[blackdot,proofdiagram] at (1.5,2.5){};
\node(x)[yellowdot,proofdiagram] at (1.5,1.75){};
\node(b2)[blackdot,proofdiagram] at (2,-2){};
\node(y)[yellowdot,proofdiagram] at (chi.connect s){};
\node(y2)[yellowdot,proofdiagram] at (chi.connect n){};
\node(i1) at (-4,-4){};
\node(i3)[whitedot,projdot] at (2,-2.75){};
\node(mm) at (-3,2){};
\node(t1)[blackdot,proofdiagram] at (0.25,-3.5){};
\node(t2)[blackdot,proofdiagram] at (4.5,1){};
\node (chi7) [morphism,wedge,scale=0.5,connect n length=0.35cm,connect s length=0.35cm] at (-1.8,5.5) {\large$\chi$};
\node(b171)[reddot,proofdiagram]  at (0.5,3.5) {};
\node(b172)[reddot,proofdiagram]  at (1.5,3.5) {};
\node(add7)[reddot,proofdiagram] at (2,6){};
\node(x7)[yellowdot,proofdiagram] at (4.5,3){};
\node(b27)[blackdot,proofdiagram] at (3,-1.5){};
\node(y7)[yellowdot,proofdiagram] at (chi7.connect s){};
\node(y27)[yellowdot,proofdiagram] at (chi7.connect n){};
\node(m7) at (5,2.5){};
\node(i17) at (0.75,1){};
\node(i37)[whitedot,projdot] at (3,-2.75){};
\node(o7) at (2,7){};
\node(mm7) at (-2.5,4){};
\node(in1) at (-0.5,-4){};
\node(in2) at (1,-4){};
\node(b2k)[whitedot,inner sep=0.25pt] at (2,-3.1){$\boldsymbol{*}$};
\node(b4k)[whitedot,inner sep=0.25pt] at (3,-3.1){$\boldsymbol{*}$};
\node(in3) at (2,-4){};
\node(in4) at (3,-4){};
\draw[string,in=270,out=90,looseness=0.6] (i1.center) to (b1.center);
\draw[string,in=180,out=180,looseness=1] (b1.center) to (y.center);
\draw[string,in=270,out=0] (b1.center) to (add1.center);
\draw[string,in=90,out=270] (add2.center) to (x.center);
\draw[string,in=270,out=90] (i3.center) to (b2.center);
\draw[string,in=270,out=left] (b2.center) to (mm.center);
\draw[string,in=0,out=90] (mm.center) to (y.center);
\draw[string,in=left,out=right,looseness=1.3] (b2.center) to (x.center);
\draw[string,out=90,in=0] (b171.center) to (y7.center);
\draw[string,in=270,out=left] (b27.center) to (mm7.center);
\draw[string,in=180,out=90] (mm7.center) to (y7.center);
\draw[string] (add1.center) to (b172.center);
\draw[string] (add2.center) to (b171.center);
\draw[string,in=180,out=180] (add1.center) to (b171.center);
\draw[string,in=0,out=0] (add2.center) to (b172.center);
\draw[string,in=left,out=90] (b172.center) to (add7.center);
\draw[string,in=90,out=right] (add7.center) to (x7.center);
\draw[string,in=270,out=90] (add7.center) to (o7.center);
\draw[string,in=right,out=0,looseness=1] (t2.center) to (x7.center);
\draw[string,in=270,out=90] (i37.center) to (b27.center);
\draw[string,in=left,out=right] (b27.center) to (x7.center);
\draw[string,in=180,out=0] (x.center) to (t2.center);
\draw[string] (b2k.center) to (in3.center);
\draw[string] (b4k.center) to (in4.center);
\draw[green][string,line width=1pt] (b4k.center) to (i37.center);
\draw[green][string,line width=1pt] (b2k.center) to (i3.center);
\draw[string,out=90,in=270] (t1.center) to (t2.center);
\draw[string,out=180,in=90] (t1.center) to (in1.center);
\draw[string,out=0,in=90] (t1.center) to (in2.center);
\end{pic}\qquad  \super{Yellow, black\eqref{eq:bialg}}= \quad 
\begin{pic}[scale=1/3]
\node(chi) [morphism,wedge,scale=0.5,connect n length=0.35cm,connect s length=0.35cm] at (-4,5.5) {\large$\chi$};
\node(b1)[blackdot,proofdiagram]  at (-4,-2) {};
\node(add1)[blackdot,proofdiagram] at (0.5,2.5){};
\node(add21)[yellowdot,proofdiagram] at (1.5,2){};
\node(add22)[yellowdot,proofdiagram] at (2.5,2){};
\node(x1)[blackdot,proofdiagram] at (1.5,1.25){};
\node(x2)[blackdot,proofdiagram] at (2.5,1.25){};
\node(b2)[blackdot,proofdiagram] at (2,-2){};
\node(y)[yellowdot,proofdiagram] at (chi.connect s){};
\node(y2)[yellowdot,proofdiagram] at (chi.connect n){};
\node(i1) at (-4,-4){};
\node(i3)[whitedot,projdot] at (2,-2.75){};
\node(mm) at (-3,2){};
\node(t1)[blackdot,proofdiagram] at (0.25,-3.5){};
\node(t2)[blackdot,proofdiagram] at (4.25,0.5){};
\node (chi7) [morphism,wedge,scale=0.5,connect n length=0.35cm,connect s length=0.35cm] at (-1.8,5.5) {\large$\chi$};
\node(b171)[reddot,proofdiagram]  at (0.5,3.5) {};
\node(b172)[reddot,proofdiagram]  at (1.5,3.5) {};
\node(add7)[reddot,proofdiagram] at (2,6){};
\node(x7)[yellowdot,proofdiagram] at (4.5,3){};
\node(b27)[blackdot,proofdiagram] at (3,-1.5){};
\node(y7)[yellowdot,proofdiagram] at (chi7.connect s){};
\node(y27)[yellowdot,proofdiagram] at (chi7.connect n){};
\node(m7) at (5,2.5){};
\node(i17) at (0.75,1){};
\node(i37)[whitedot,projdot] at (3,-2.75){};
\node(o7) at (2,7){};
\node(mm7) at (-2.5,4){};
\node(in1) at (-0.5,-4){};
\node(in2) at (1,-4){};
\node(b2k)[whitedot,inner sep=0.25pt] at (2,-3.1){$\boldsymbol{*}$};
\node(b4k)[whitedot,inner sep=0.25pt] at (3,-3.1){$\boldsymbol{*}$};
\node(in3) at (2,-4){};
\node(in4) at (3,-4){};
\draw[string,in=270,out=90,looseness=0.6] (i1.center) to (b1.center);
\draw[string,in=180,out=180,looseness=1] (b1.center) to (y.center);
\draw[string,in=270,out=0] (b1.center) to (add1.center);
\draw[string] (add21.center) to (x2.center);
\draw[string] (add22.center) to (x1.center);
\draw[string,in=0,out=0] (add22.center) to (x2.center);
\draw[string,in=180,out=180] (add21.center) to (x1.center);
\draw[string,in=270,out=90] (i3.center) to (b2.center);
\draw[string,in=270,out=left] (b2.center) to (mm.center);
\draw[string,in=0,out=90] (mm.center) to (y.center);
\draw[string,in=270,out=right,looseness=1.3] (b2.center) to (x1.center);
\draw[string,out=90,in=0] (b171.center) to (y7.center);
\draw[string,in=270,out=left] (b27.center) to (mm7.center);
\draw[string,in=180,out=90] (mm7.center) to (y7.center);
\draw[string] (add1.center) to (b172.center);
\draw[string,out=90,in=0] (add21.center) to (b171.center);
\draw[string,in=180,out=180] (add1.center) to (b171.center);
\draw[string,in=0,out=90] (add22.center) to (b172.center);
\draw[string,in=left,out=90] (b172.center) to (add7.center);
\draw[string,in=90,out=right] (add7.center) to (x7.center);
\draw[string,in=270,out=90] (add7.center) to (o7.center);
\draw[string,in=right,out=0,looseness=1] (t2.center) to (x7.center);
\draw[string,in=270,out=90] (i37.center) to (b27.center);
\draw[string,in=left,out=right] (b27.center) to (x7.center);
\draw[string,in=180,out=270] (x2.center) to (t2.center);
\draw[string] (b2k.center) to (in3.center);
\draw[string] (b4k.center) to (in4.center);
\draw[green][string,line width=1pt] (b4k.center) to (i37.center);
\draw[green][string,line width=1pt] (b2k.center) to (i3.center);
\draw[string,out=90,in=270] (t1.center) to (t2.center);
\draw[string,out=180,in=90] (t1.center) to (in1.center);
\draw[string,out=0,in=90] (t1.center) to (in2.center);
\end{pic}\quad \\  \super{\eqref{eq:sm}}=& 
\begin{pic}[scale=1/3]
\node(chi) [morphism,wedge,scale=0.5,connect n length=0.35cm,connect s length=1.6cm] at (-4.5,6) {\large$\chi$};
\node(b1)[blackdot,proofdiagram]  at (-4,-2) {};
\node(add1)[blackdot,proofdiagram] at (-4.5,2.5){};
\node(add21)[yellowdot,proofdiagram] at (-2,3.25){};
\node(add22)[yellowdot,proofdiagram] at (1.5,4){};
\node(x1)[blackdot,proofdiagram] at (-3,2.5){};
\node(x2)[blackdot,proofdiagram] at (3,3){};
\node(b2)[blackdot,proofdiagram] at (2,-2){};
\node(y)[yellowdot,proofdiagram] at (chi.connect s){};
\node(y2)[yellowdot,proofdiagram] at (chi.connect n){};
\node(i1) at (-4,-4){};
\node(i3)[whitedot,projdot] at (2,-2.75){};
\node(mm) at (-3,2){};
\node(t1)[blackdot,proofdiagram] at (0.25,-3.5){};
\node(t2)[blackdot,proofdiagram] at (0.25,-2.5){};
\node (chi7) [morphism,wedge,scale=0.5,connect n length=0.35cm,connect s length=0.35cm] at (-2,6) {\large$\chi$};
\node(b171)[reddot,proofdiagram]  at (-3,4) {};
\node(b172)[reddot,proofdiagram]  at (1,5) {};
\node(add7)[reddot,proofdiagram] at (2,6){};
\node(x7)[yellowdot,proofdiagram] at (3,5){};
\node(b27)[blackdot,proofdiagram] at (3,-1.5){};
\node(y7)[yellowdot,proofdiagram] at (chi7.connect s){};
\node(y27)[yellowdot,proofdiagram] at (chi7.connect n){};
\node(m7) at (5,2.5){};
\node(i17) at (0.75,1){};
\node(i37)[whitedot,projdot] at (3,-2.75){};
\node(o7) at (2,7){};
\node(mm7) at (0,2){};
\node(in1) at (-0.5,-4){};
\node(in2) at (1,-4){};
\node(b2k)[whitedot,inner sep=0.25pt] at (2,-3.1){$\boldsymbol{*}$};
\node(b4k)[whitedot,inner sep=0.25pt] at (3,-3.1){$\boldsymbol{*}$};
\node(in3) at (2,-4){};
\node(in4) at (3,-4){};
\draw[string,in=270,out=90] (i1.center) to (b1.center);
\draw[string,in=180,out=180,looseness=1] (add1.center) to (y.center);
\draw[string,in=270,out=180,looseness=0.7] (b1.center) to (add1.center);
\draw[string,out=0,in=180] (add21.center) to (t2.center);
\draw[string,out=180,in=0] (add22.center) to (b2.center);
\draw[string,in=180,out=0] (add22.center) to (x2.center);
\draw[string,in=0,out=180] (add21.center) to (x1.center);
\draw[string,in=270,out=90] (i3.center) to (b2.center);
\draw[string,in=0,out=left] (x1.center) to (y.center);
\draw[string,in=270,out=90] (mm.center) to (x1.center);
\draw[string,in=270,out=180,looseness=1.1] (b2.center) to (mm.center);
\draw[string,out=90,in=180] (b171.center) to (y7.center);
\draw[string,in=270,out=left] (b27.center) to (mm7.center);
\draw[string,in=0,out=90] (mm7.center) to (y7.center);
\draw[string,out=0,in=180] (b1.center) to (b172.center);
\draw[string,out=90,in=0] (add21.center) to (b171.center);
\draw[string,in=180,out=0] (add1.center) to (b171.center);
\draw[string,in=0,out=90] (add22.center) to (b172.center);
\draw[string,in=left,out=90] (b172.center) to (add7.center);
\draw[string,in=90,out=right] (add7.center) to (x7.center);
\draw[string,in=270,out=90] (add7.center) to (o7.center);
\draw[string,in=right,out=0,looseness=1] (x2.center) to (x7.center);
\draw[string,in=270,out=90] (i37.center) to (b27.center);
\draw[string,in=left,out=right] (b27.center) to (x7.center);
\draw[string,in=0,out=270,looseness=0.7] (x2.center) to (t2.center);
\draw[string] (b2k.center) to (in3.center);
\draw[string] (b4k.center) to (in4.center);
\draw[green][string,line width=1pt] (b4k.center) to (i37.center);
\draw[green][string,line width=1pt] (b2k.center) to (i3.center);
\draw[string,out=90,in=270] (t1.center) to (t2.center);
\draw[string,out=180,in=90] (t1.center) to (in1.center);
\draw[string,out=0,in=90] (t1.center) to (in2.center);
\end{pic}\quad \super{Yellow, black\eqref{eq:bialg}}= \qquad 
\begin{pic}[scale=1/3]
\node(chi) [morphism,wedge,scale=0.5,connect s length=1.6cm] at (-4.4,4.2) {\large$\chi$};
\node(b1)[blackdot,proofdiagram]  at (-4,-2) {};
\node(add1)[blackdot,proofdiagram] at (-4.5,0.5){};
\node(add21)[yellowdot,proofdiagram] at (-2,1.25){};
\node(add22)[yellowdot,proofdiagram] at (1.5,4){};
\node(x1)[blackdot,proofdiagram] at (-3,0.5){};
\node(x2)[blackdot,proofdiagram] at (3,3){};
\node(b2)[blackdot,proofdiagram] at (2,-2){};
\node(y)[yellowdot,proofdiagram] at (chi.connect s){};
\node(i1) at (-4,-4){};
\node(i3)[whitedot,projdot] at (2,-2.75){};
\node(mm) at (-3,0){};
\node(t1)[blackdot,proofdiagram] at (0.25,-3.5){};
\node(t2)[blackdot,proofdiagram] at (0.25,-2.5){};
\node (chi7) [morphism,wedge,scale=0.5,connect s length=0.35cm] at (-2.1,4.2) {\large$\chi$};
\node(b171)[reddot,proofdiagram]  at (-3,2) {};
\node(b172)[reddot,proofdiagram]  at (1,5) {};
\node(add7)[reddot,proofdiagram] at (2,6){};
\node(x7)[yellowdot,proofdiagram] at (3,5){};
\node(b27)[blackdot,proofdiagram] at (3,-1.5){};
\node(y7)[yellowdot,proofdiagram] at (chi7.connect s){};
\node(m7) at (5,2.5){};
\node(i17) at (0.75,1){};
\node(i37)[whitedot,projdot] at (3,-2.75){};
\node(o7) at (2,7){};
\node(mm7) at (0,0){};
\node(in1) at (-0.5,-4){};
\node(in2) at (1,-4){};
\node(b2k)[whitedot,inner sep=0.25pt] at (2,-3.1){$\boldsymbol{*}$};
\node(b4k)[whitedot,inner sep=0.25pt] at (3,-3.1){$\boldsymbol{*}$};
\node(in3) at (2,-4){};
\node(in4) at (3,-4){};
\node(q)[blackdot,proofdiagram] at (-3.25,5){};
\node(qy)[yellowdot,proofdiagram] at (-3.25,5.75){};
\draw[string,in=180,out=90] (chi.north) to (q.center);
\draw[string,in=0,out=90] (chi7.north) to (q.center);
\draw[string] (q.center) to (qy.center);
\draw[string,in=270,out=90] (i1.center) to (b1.center);
\draw[string,in=180,out=180,looseness=1] (add1.center) to (y.center);
\draw[string,in=270,out=180,looseness=0.7] (b1.center) to (add1.center);
\draw[string,out=0,in=180] (add21.center) to (t2.center);
\draw[string,out=180,in=0] (add22.center) to (b2.center);
\draw[string,in=180,out=0] (add22.center) to (x2.center);
\draw[string,in=0,out=180] (add21.center) to (x1.center);
\draw[string,in=270,out=90] (i3.center) to (b2.center);
\draw[string,in=0,out=left] (x1.center) to (y.center);
\draw[string,in=270,out=90] (mm.center) to (x1.center);
\draw[string,in=270,out=180,looseness=1.1] (b2.center) to (mm.center);
\draw[string,out=90,in=180] (b171.center) to (y7.center);
\draw[string,in=270,out=left] (b27.center) to (mm7.center);
\draw[string,in=0,out=90] (mm7.center) to (y7.center);
\draw[string,out=0,in=180] (b1.center) to (b172.center);
\draw[string,out=90,in=0] (add21.center) to (b171.center);
\draw[string,in=180,out=0] (add1.center) to (b171.center);
\draw[string,in=0,out=90] (add22.center) to (b172.center);
\draw[string,in=left,out=90] (b172.center) to (add7.center);
\draw[string,in=90,out=right] (add7.center) to (x7.center);
\draw[string,in=270,out=90] (add7.center) to (o7.center);
\draw[string,in=right,out=0,looseness=1] (x2.center) to (x7.center);
\draw[string,in=270,out=90] (i37.center) to (b27.center);
\draw[string,in=left,out=right] (b27.center) to (x7.center);
\draw[string,in=0,out=270,looseness=0.7] (x2.center) to (t2.center);
\draw[string] (b2k.center) to (in3.center);
\draw[string] (b4k.center) to (in4.center);
\draw[green][string,line width=1pt] (b4k.center) to (i37.center);
\draw[green][string,line width=1pt] (b2k.center) to (i3.center);
\draw[string,out=90,in=270] (t1.center) to (t2.center);
\draw[string,out=180,in=90] (t1.center) to (in1.center);
\draw[string,out=0,in=90] (t1.center) to (in2.center);
\end{pic}\\ \super{\eqref{eq:chi}}=& \quad 
\begin{pic}[scale=1/3]
\node(chi) [morphism,wedge,scale=0.5,connect n length=0.35cm] at (-5,5.75) {\large$\chi$};
\node(b1)[blackdot,proofdiagram]  at (-4,-2) {};
\node(add1)[blackdot,proofdiagram] at (-6.5,1){};
\node(add21)[yellowdot,proofdiagram] at (-3.5,2){};
\node(add22)[yellowdot,proofdiagram] at (2.5,4.5){};
\node(x1)[blackdot,proofdiagram] at (-5,1){};
\node(x2)[blackdot,proofdiagram] at (4,3.5){};
\node(b2)[blackdot,proofdiagram] at (2,-2){};
\node(y)[yellowdot,proofdiagram] at (-6.5,3){};
\node(i1) at (-4,-4){};
\node(i3)[whitedot,projdot] at (2,-2.75){};
\node(mm) at (-5,1){};
\node(t1)[blackdot,proofdiagram] at (0.25,-3.5){};
\node(t2)[blackdot,proofdiagram] at (0.25,-2.5){};
\node(b171)[reddot,proofdiagram]  at (-5,3) {};
\node(b172)[reddot,proofdiagram]  at (1.25,5.5) {};
\node(add7)[reddot,proofdiagram] at (2,6){};
\node(x7)[yellowdot,proofdiagram] at (4,5){};
\node(b27)[blackdot,proofdiagram] at (3,-1.5){};
\node(y7)[yellowdot,proofdiagram] at (-3.5,4){};
\node(m7) at (5,2.5){};
\node(i17) at (0.75,1){};
\node(i37)[whitedot,projdot] at (3,-2.75){};
\node(o7) at (2,7){};
\node(mm7) at (0,0){};
\node(in1) at (-0.5,-4){};
\node(in2) at (1,-4){};
\node(b2k)[whitedot,inner sep=0.25pt] at (2,-3.1){$\boldsymbol{*}$};
\node(b4k)[whitedot,inner sep=0.25pt] at (3,-3.1){$\boldsymbol{*}$};
\node(in3) at (2,-4){};
\node(in4) at (3,-4){};
\node(q)[reddot,proofdiagram] at (-5,4.75){};
\node(qy)[yellowdot,proofdiagram] at (chi.connect n){};
\draw[string,in=180,out=90] (y.center) to (q.center);
\draw[string,in=0,out=90] (y7.center) to (q.center);
\draw[string] (q.center) to (chi.south);
\draw[string,in=270,out=90] (i1.center) to (b1.center);
\draw[string,in=180,out=180,looseness=1] (add1.center) to (y.center);
\draw[string,in=270,out=180,looseness=0.7] (b1.center) to (add1.center);
\draw[string,out=0,in=180] (add21.center) to (t2.center);
\draw[string,out=180,in=0,looseness=0.7] (add22.center) to (b2.center);
\draw[string,in=180,out=0] (add22.center) to (x2.center);
\draw[string,in=0,out=180] (add21.center) to (x1.center);
\draw[string,in=270,out=90] (i3.center) to (b2.center);
\draw[string,in=0,out=left] (x1.center) to (y.center);
\draw[string,in=270,out=90] (mm.center) to (x1.center);
\draw[string,in=270,out=180,looseness=1.1] (b2.center) to (mm.center);
\draw[string,out=90,in=180] (b171.center) to (y7.center);
\draw[string,in=270,out=left] (b27.center) to (mm7.center);
\draw[string,in=0,out=90] (mm7.center) to (y7.center);
\draw[string,out=0,in=180] (b1.center) to (b172.center);
\draw[string,out=90,in=0] (add21.center) to (b171.center);
\draw[string,in=180,out=0] (add1.center) to (b171.center);
\draw[string,in=0,out=90] (add22.center) to (b172.center);
\draw[string,in=left,out=90] (b172.center) to (add7.center);
\draw[string,in=90,out=right] (add7.center) to (x7.center);
\draw[string,in=270,out=90] (add7.center) to (o7.center);
\draw[string,in=right,out=0,looseness=1] (x2.center) to (x7.center);
\draw[string,in=270,out=90] (i37.center) to (b27.center);
\draw[string,in=left,out=right] (b27.center) to (x7.center);
\draw[string,in=0,out=270,looseness=0.7] (x2.center) to (t2.center);
\draw[string] (b2k.center) to (in3.center);
\draw[string] (b4k.center) to (in4.center);
\draw[green][string,line width=1pt] (b4k.center) to (i37.center);
\draw[green][string,line width=1pt] (b2k.center) to (i3.center);
\draw[string,out=90,in=270] (t1.center) to (t2.center);
\draw[string,out=180,in=90] (t1.center) to (in1.center);
\draw[string,out=0,in=90] (t1.center) to (in2.center);
\end{pic}\quad \super{\eqref{eqs:lem1}\eqref{eqs:lem2}}= \quad 
\begin{pic}[scale=1/3]
\node(chi) [morphism,wedge,scale=0.5,connect n length=0.35cm] at (-5,5.75) {\large$\chi$};
\node(b1)[blackdot,proofdiagram]  at (-4,-2) {};
\node(add1)[blackdot,proofdiagram] at (-6.5,1){};
\node(add21)[yellowdot,proofdiagram] at (-0.5,2){};
\node(add22)[yellowdot,proofdiagram] at (3,4){};
\node(x1)[blackdot,proofdiagram] at (-3.5,1){};
\node(x2)[blackdot,proofdiagram] at (4,3.5){};
\node(b2)[blackdot,proofdiagram] at (2,-2){};
\node(y)[yellowdot,proofdiagram] at (-6.5,4){};
\node(i1) at (-4,-4){};
\node(i3)[whitedot,projdot] at (2,-2.75){};
\node(mm7) at (-5,1){};
\node(t1)[blackdot,proofdiagram] at (0.25,-3.5){};
\node(t2)[blackdot,proofdiagram] at (0.25,-2.5){};
\node(b171)[reddot,proofdiagram]  at (-2,3) {};
\node(b172)[reddot,proofdiagram]  at (1.25,5.5) {};
\node(add7)[reddot,proofdiagram] at (2,6){};
\node(x7)[yellowdot,proofdiagram] at (3.5,5){};
\node(b27)[blackdot,proofdiagram] at (3,-1.5){};
\node(y7)[yellowdot,proofdiagram] at (-3.5,4){};
\node(m7) at (5,2.5){};
\node(i17) at (0.75,1){};
\node(i37)[whitedot,projdot] at (3,-2.75){};
\node(o7) at (2,7){};
\node(mm) at (-3.5,0.5){};
\node(in1) at (-0.5,-4){};
\node(in2) at (1,-4){};
\node(b2k)[whitedot,inner sep=0.25pt] at (2,-3.1){$\boldsymbol{*}$};
\node(b4k)[whitedot,inner sep=0.25pt] at (3,-3.1){$\boldsymbol{*}$};
\node(in3) at (2,-4){};
\node(in4) at (3,-4){};
\node(q)[reddot,proofdiagram] at (-5,4.75){};
\node(qy)[yellowdot,proofdiagram] at (chi.connect n){};
\node(mm1) at (0.5,2){};
\node(mm2) at (2,3){};
\draw[string,in=180,out=90] (y.center) to (q.center);
\draw[string,in=0,out=90] (y7.center) to (q.center);
\draw[string] (q.center) to (chi.south);
\draw[string,in=270,out=90] (i1.center) to (b1.center);
\draw[string,in=180,out=180,looseness=1] (add1.center) to (y.center);
\draw[string,in=270,out=180,looseness=0.7] (b1.center) to (add1.center);
\draw[string,out=0,in=180] (add21.center) to (t2.center);
\draw[string,out=180,in=0,looseness=0.7] (add22.center) to (b27.center);
\draw[string,in=180,out=0] (add22.center) to (x2.center);
\draw[string,in=0,out=180] (add21.center) to (x1.center);
\draw[string,in=270,out=90] (i3.center) to (b2.center);
\draw[string,in=0,out=left] (x1.center) to (y.center);
\draw[string,in=270,out=90] (mm.center) to (x1.center);
\draw[string,in=270,out=180,looseness=0.8] (b27.center) to (mm.center);
\draw[string,out=90,in=0] (b171.center) to (y7.center);
\draw[string,in=270,out=left] (b2.center) to (mm7.center);
\draw[string,in=180,out=90] (mm7.center) to (y7.center);
\draw[string,out=0,in=270] (b1.center) to (mm1.center);
\draw[string,out=90,in=180] (mm1.center) to (b172.center);
\draw[string,out=90,in=0] (add21.center) to (b171.center);
\draw[string,in=180,out=0] (add1.center) to (b171.center);
\draw[string,in=0,out=90] (add22.center) to (b172.center);
\draw[string,in=left,out=90] (b172.center) to (add7.center);
\draw[string,in=90,out=right] (add7.center) to (x7.center);
\draw[string,in=270,out=90] (add7.center) to (o7.center);
\draw[string,in=right,out=0,looseness=1] (x2.center) to (x7.center);
\draw[string,in=270,out=90] (i37.center) to (b27.center);
\draw[string,in=270,out=right,looseness=0.7] (b2.center) to (mm2.center);
\draw[string,in=180,out=90] (mm2.center) to (x7.center);
\draw[string,in=0,out=270,looseness=0.7] (x2.center) to (t2.center);
\draw[string] (b2k.center) to (in3.center);
\draw[string] (b4k.center) to (in4.center);
\draw[green][string,line width=1pt] (b4k.center) to (i37.center);
\draw[green][string,line width=1pt] (b2k.center) to (i3.center);
\draw[string,out=90,in=270] (t1.center) to (t2.center);
\draw[string,out=180,in=90] (t1.center) to (in1.center);
\draw[string,out=0,in=90] (t1.center) to (in2.center);
\end{pic}\\ \super{\eqref{eq:chi}}=& \quad 
\begin{pic}[scale=1/3]
\node(chi) [morphism,wedge,scale=0.5] at (-6.5,5) {\large$\chi$};
\node(chi7) [morphism,wedge,scale=0.5] at (-3.5,5) {\large$\chi$};
\node(b1)[blackdot,proofdiagram]  at (-4,-2) {};
\node(add1)[blackdot,proofdiagram] at (-6.5,1){};
\node(add21)[yellowdot,proofdiagram] at (-0.5,2){};
\node(add22)[yellowdot,proofdiagram] at (3,4){};
\node(x1)[blackdot,proofdiagram] at (-3.5,1){};
\node(x2)[blackdot,proofdiagram] at (4,3.5){};
\node(b2)[blackdot,proofdiagram] at (2,-2){};
\node(y)[yellowdot,proofdiagram] at (-6.5,4){};
\node(i1) at (-4,-4){};
\node(i3)[whitedot,projdot] at (2,-2.75){};
\node(mm7) at (-5,1){};
\node(t1)[blackdot,proofdiagram] at (0.25,-3.5){};
\node(t2)[blackdot,proofdiagram] at (0.25,-2.5){};
\node(b171)[reddot,proofdiagram]  at (-2,3) {};
\node(b172)[reddot,proofdiagram]  at (1.25,5.5) {};
\node(add7)[reddot,proofdiagram] at (2,6){};
\node(x7)[yellowdot,proofdiagram] at (3.5,5){};
\node(b27)[blackdot,proofdiagram] at (3,-1.5){};
\node(y7)[yellowdot,proofdiagram] at (-3.5,4){};
\node(m7) at (5,2.5){};
\node(i17) at (0.75,1){};
\node(i37)[whitedot,projdot] at (3,-2.75){};
\node(o7) at (2,7){};
\node(mm) at (-3.5,0.5){};
\node(in1) at (-0.5,-4){};
\node(in2) at (1,-4){};
\node(b2k)[whitedot,inner sep=0.25pt] at (2,-3.1){$\boldsymbol{*}$};
\node(b4k)[whitedot,inner sep=0.25pt] at (3,-3.1){$\boldsymbol{*}$};
\node(in3) at (2,-4){};
\node(in4) at (3,-4){};
\node(q)[blackdot,proofdiagram] at (-5,6){};
\node(qy)[yellowdot,proofdiagram] at (-5,6.75){};
%\node(qy)[yellowdot,scale=0.5] at (chi.connect n){};
%\node(qqy)[yellowdot,scale=0.5] at (chi7.connect n){};
\node(mm1) at (0.5,2){};
\node(mm2) at (2,3){};
\draw[string,in=180,out=90] (chi.north) to (q.center);
\draw[string,in=0,out=90] (chi7.north) to (q.center);
\draw[string] (chi.south) to (y.center);
\draw[string] (chi7.south) to (y7.center);
\draw[string] (q.center) to (qy.center);
\draw[string,in=270,out=90] (i1.center) to (b1.center);
\draw[string,in=180,out=180,looseness=1] (add1.center) to (y.center);
\draw[string,in=270,out=180,looseness=0.7] (b1.center) to (add1.center);
\draw[string,out=0,in=180] (add21.center) to (t2.center);
\draw[string,out=180,in=0,looseness=0.7] (add22.center) to (b27.center);
\draw[string,in=180,out=0] (add22.center) to (x2.center);
\draw[string,in=0,out=180] (add21.center) to (x1.center);
\draw[string,in=270,out=90] (i3.center) to (b2.center);
\draw[string,in=0,out=left] (x1.center) to (y.center);
\draw[string,in=270,out=90] (mm.center) to (x1.center);
\draw[string,in=270,out=180,looseness=0.8] (b27.center) to (mm.center);
\draw[string,out=90,in=0] (b171.center) to (y7.center);
\draw[string,in=270,out=left] (b2.center) to (mm7.center);
\draw[string,in=180,out=90] (mm7.center) to (y7.center);
\draw[string,out=0,in=270] (b1.center) to (mm1.center);
\draw[string,out=90,in=180] (mm1.center) to (b172.center);
\draw[string,out=90,in=0] (add21.center) to (b171.center);
\draw[string,in=180,out=0] (add1.center) to (b171.center);
\draw[string,in=0,out=90] (add22.center) to (b172.center);
\draw[string,in=left,out=90] (b172.center) to (add7.center);
\draw[string,in=90,out=right] (add7.center) to (x7.center);
\draw[string,in=270,out=90] (add7.center) to (o7.center);
\draw[string,in=right,out=0,looseness=1] (x2.center) to (x7.center);
\draw[string,in=270,out=90] (i37.center) to (b27.center);
\draw[string,in=270,out=right,looseness=0.7] (b2.center) to (mm2.center);
\draw[string,in=180,out=90] (mm2.center) to (x7.center);
\draw[string,in=0,out=270,looseness=0.7] (x2.center) to (t2.center);
\draw[string] (b2k.center) to (in3.center);
\draw[string] (b4k.center) to (in4.center);
\draw[green][string,line width=1pt] (b4k.center) to (i37.center);
\draw[green][string,line width=1pt] (b2k.center) to (i3.center);
\draw[string,out=90,in=270] (t1.center) to (t2.center);
\draw[string,out=180,in=90] (t1.center) to (in1.center);
\draw[string,out=0,in=90] (t1.center) to (in2.center);
\end{pic}\quad \super{Yellow, black\eqref{eq:bialg}}=\qquad
\begin{pic}[scale=1/3]
\node(chi) [morphism,wedge,scale=0.5,connect n length=0.35cm] at (-6.5,5) {\large$\chi$};
\node(chi7) [morphism,wedge,scale=0.5,connect n length=0.35cm] at (-3.5,5) {\large$\chi$};
\node(b1)[blackdot,proofdiagram]  at (-4,-2) {};
\node(add1)[blackdot,proofdiagram] at (-6.5,1){};
\node(add21)[yellowdot,proofdiagram] at (-0.5,2){};
\node(add22)[yellowdot,proofdiagram] at (3,4){};
\node(x1)[blackdot,proofdiagram] at (-3.5,1){};
\node(x2)[blackdot,proofdiagram] at (4,3.5){};
\node(b2)[blackdot,proofdiagram] at (2,-2){};
\node(y)[yellowdot,proofdiagram] at (-6.5,4){};
\node(i1) at (-4,-4){};
\node(i3)[whitedot,projdot] at (2,-2.75){};
\node(mm7) at (-5,1){};
\node(t1)[blackdot,proofdiagram] at (0.25,-3.5){};
\node(t2)[blackdot,proofdiagram] at (0.25,-2.5){};
\node(b171)[reddot,proofdiagram]  at (-2,3) {};
\node(b172)[reddot,proofdiagram]  at (1.25,5.5) {};
\node(add7)[reddot,proofdiagram] at (2,6){};
\node(x7)[yellowdot,proofdiagram] at (3.5,5){};
\node(b27)[blackdot,proofdiagram] at (3,-1.5){};
\node(y7)[yellowdot,proofdiagram] at (-3.5,4){};
\node(m7) at (5,2.5){};
\node(i17) at (0.75,1){};
\node(i37)[whitedot,projdot] at (3,-2.75){};
\node(o7) at (2,7){};
\node(mm) at (-3.5,0.5){};
\node(in1) at (-0.5,-4){};
\node(in2) at (1,-4){};
\node(b2k)[whitedot,inner sep=0.25pt] at (2,-3.1){$\boldsymbol{*}$};
\node(b4k)[whitedot,inner sep=0.25pt] at (3,-3.1){$\boldsymbol{*}$};
\node(in3) at (2,-4){};
\node(in4) at (3,-4){};
\node(qy)[yellowdot,proofdiagram] at (chi.connect n){};
\node(qqy)[yellowdot,proofdiagram] at (chi7.connect n){};
\node(mm1) at (0.5,2){};
\node(mm2) at (2,3){};
\draw[string] (chi.south) to (y.center);
\draw[string] (chi7.south) to (y7.center);
\draw[string,in=270,out=90] (i1.center) to (b1.center);
\draw[string,in=180,out=180,looseness=1] (add1.center) to (y.center);
\draw[string,in=270,out=180,looseness=0.7] (b1.center) to (add1.center);
\draw[string,out=0,in=180] (add21.center) to (t2.center);
\draw[string,out=180,in=0,looseness=0.7] (add22.center) to (b27.center);
\draw[string,in=180,out=0] (add22.center) to (x2.center);
\draw[string,in=0,out=180] (add21.center) to (x1.center);
\draw[string,in=270,out=90] (i3.center) to (b2.center);
\draw[string,in=0,out=left] (x1.center) to (y.center);
\draw[string,in=270,out=90] (mm.center) to (x1.center);
\draw[string,in=270,out=180,looseness=0.8] (b27.center) to (mm.center);
\draw[string,out=90,in=0] (b171.center) to (y7.center);
\draw[string,in=270,out=left] (b2.center) to (mm7.center);
\draw[string,in=180,out=90] (mm7.center) to (y7.center);
\draw[string,out=0,in=270] (b1.center) to (mm1.center);
\draw[string,out=90,in=180] (mm1.center) to (b172.center);
\draw[string,out=90,in=0] (add21.center) to (b171.center);
\draw[string,in=180,out=0] (add1.center) to (b171.center);
\draw[string,in=0,out=90] (add22.center) to (b172.center);
\draw[string,in=left,out=90] (b172.center) to (add7.center);
\draw[string,in=90,out=right] (add7.center) to (x7.center);
\draw[string,in=270,out=90] (add7.center) to (o7.center);
\draw[string,in=right,out=0,looseness=1] (x2.center) to (x7.center);
\draw[string,in=270,out=90] (i37.center) to (b27.center);
\draw[string,in=270,out=right,looseness=0.7] (b2.center) to (mm2.center);
\draw[string,in=180,out=90] (mm2.center) to (x7.center);
\draw[string,in=0,out=270,looseness=0.7] (x2.center) to (t2.center);
\draw[string] (b2k.center) to (in3.center);
\draw[string] (b4k.center) to (in4.center);
\draw[green][string,line width=1pt] (b4k.center) to (i37.center);
\draw[green][string,line width=1pt] (b2k.center) to (i3.center);
\draw[string,out=90,in=270] (t1.center) to (t2.center);
\draw[string,out=180,in=90] (t1.center) to (in1.center);
\draw[string,out=0,in=90] (t1.center) to (in2.center);
\end{pic}\\ \super{\eqref{eq:sm}}=&\quad
\begin{pic}[scale=1/3]
\node(chi) [morphism,wedge,scale=0.5,connect n length=0.35cm,connect s length=0.35cm] at (-4,5.5) {\large$\chi$};
\node(b1)[blackdot,proofdiagram]  at (-4,-1) {};
\node(add1)[blackdot,proofdiagram] at (0.5,2.5){};
\node(add21)[yellowdot,proofdiagram] at (1.5,2){};
\node(add22)[yellowdot,proofdiagram] at (2.5,2){};
\node(x1)[blackdot,proofdiagram] at (1.5,1.25){};
\node(x2)[blackdot,proofdiagram] at (2.5,1.25){};
\node(b2)[blackdot,proofdiagram] at (2,-2){};
\node(y)[yellowdot,proofdiagram] at (chi.connect s){};
\node(y2)[yellowdot,proofdiagram] at (chi.connect n){};
\node(i1) at (-4,-4){};
\node(i3)[whitedot,projdot] at (2,-2.75){};
\node(mm) at (-3,2){};
\node(t1)[blackdot,proofdiagram] at (0.25,-3.5){};
\node(t2)[blackdot,proofdiagram] at (4.25,0.5){};
\node (chi7) [morphism,wedge,scale=0.5,connect n length=0.35cm,connect s length=0.35cm] at (-1.8,5.5) {\large$\chi$};
\node(b171)[reddot,proofdiagram]  at (0.5,3.5) {};
\node(b172)[reddot,proofdiagram]  at (1.5,3.5) {};
\node(add7)[reddot,proofdiagram] at (2,6){};
\node(x7)[yellowdot,proofdiagram] at (4.5,3){};
\node(b27)[blackdot,proofdiagram] at (3.75,-1.25){};
\node(y7)[yellowdot,proofdiagram] at (chi7.connect s){};
\node(y27)[yellowdot,proofdiagram] at (chi7.connect n){};
\node(m7) at (5,2.5){};
\node(i17) at (0.75,1){};
\node(i37)[whitedot,projdot] at (3,-2.75){};
\node(o7) at (2,7){};
\node(mm7) at (-2.5,4){};
\node(in1) at (-0.5,-4){};
\node(in2) at (1,-4){};
\node(b2k)[whitedot,inner sep=0.25pt] at (2,-3.1){$\boldsymbol{*}$};
\node(b4k)[whitedot,inner sep=0.25pt] at (3,-3.1){$\boldsymbol{*}$};
\node(in3) at (2,-4){};
\node(in4) at (3,-4){};
\draw[string,in=270,out=90,looseness=0.6] (i1.center) to (b1.center);
\draw[string,in=180,out=180,looseness=1] (b1.center) to (y.center);
\draw[string,in=270,out=0,looseness=0.6] (b1.center) to (add1.center);
\draw[string] (add21.center) to (x2.center);
\draw[string] (add22.center) to (x1.center);
\draw[string,in=0,out=0] (add22.center) to (x2.center);
\draw[string,in=180,out=180] (add21.center) to (x1.center);
\draw[string,in=270,out=90] (i3.center) to (b2.center);
\draw[string,in=270,out=left] (b27.center) to (mm.center);
\draw[string,in=0,out=90] (mm.center) to (y.center);
\draw[string,in=180,out=right,looseness=1.3] (b2.center) to (x7.center);
\draw[string,out=90,in=0] (b171.center) to (y7.center);
\draw[string,in=270,out=left] (b2.center) to (mm7.center);
\draw[string,in=180,out=90] (mm7.center) to (y7.center);
\draw[string] (add1.center) to (b172.center);
\draw[string,out=90,in=0] (add21.center) to (b171.center);
\draw[string,in=180,out=180] (add1.center) to (b171.center);
\draw[string,in=0,out=90] (add22.center) to (b172.center);
\draw[string,in=left,out=90] (b172.center) to (add7.center);
\draw[string,in=90,out=right] (add7.center) to (x7.center);
\draw[string,in=270,out=90] (add7.center) to (o7.center);
\draw[string,in=right,out=0,looseness=1] (t2.center) to (x7.center);
\draw[string,in=270,out=90] (i37.center) to (b27.center);
\draw[string,in=270,out=right,looseness=1] (b27.center) to (x1.center);
\draw[string,in=180,out=270] (x2.center) to (t2.center);
\draw[string] (b2k.center) to (in3.center);
\draw[string] (b4k.center) to (in4.center);
\draw[green][string,line width=1pt] (b4k.center) to (i37.center);
\draw[green][string,line width=1pt] (b2k.center) to (i3.center);
\draw[string,out=90,in=270] (t1.center) to (t2.center);
\draw[string,out=180,in=90] (t1.center) to (in1.center);
\draw[string,out=0,in=90] (t1.center) to (in2.center);
\end{pic}\qquad \super{Yellow, black\eqref{eq:bialg}}= \qquad 
\begin{pic}[scale=1/3]
\node(chi) [morphism,wedge,scale=0.5,connect n length=0.35cm,connect s length=0.35cm] at (-4,5.5) {\large$\chi$};
\node(b1)[blackdot,proofdiagram]  at (-4,-1) {};
\node(add1)[blackdot,proofdiagram] at (0.5,2.5){};
\node(add2)[blackdot,proofdiagram] at (1.5,2.5){};
\node(x)[yellowdot,proofdiagram] at (1.5,1.75){};
\node(b2)[blackdot,proofdiagram] at (2,-2){};
\node(y)[yellowdot,proofdiagram] at (chi.connect s){};
\node(y2)[yellowdot,proofdiagram] at (chi.connect n){};
\node(i1) at (-4,-4){};
\node(i3)[whitedot,projdot] at (2,-2.75){};
\node(mm) at (-3,2){};
\node(t1)[blackdot,proofdiagram] at (0.25,-3.5){};
\node(t2)[blackdot,proofdiagram] at (4.5,1){};
\node (chi7) [morphism,wedge,scale=0.5,connect n length=0.35cm,connect s length=0.35cm] at (-1.8,5.5) {\large$\chi$};
\node(b171)[reddot,proofdiagram]  at (0.5,3.5) {};
\node(b172)[reddot,proofdiagram]  at (1.5,3.5) {};
\node(add7)[reddot,proofdiagram] at (2,6){};
\node(x7)[yellowdot,proofdiagram] at (4.5,3){};
\node(b27)[blackdot,proofdiagram] at (3.75,-1){};
\node(y7)[yellowdot,proofdiagram] at (chi7.connect s){};
\node(y27)[yellowdot,proofdiagram] at (chi7.connect n){};
\node(m7) at (5,2.5){};
\node(i17) at (0.75,1){};
\node(i37)[whitedot,projdot] at (3,-2.75){};
\node(o7) at (2,7){};
\node(mm7) at (-2.5,4){};
\node(in1) at (-0.5,-4){};
\node(in2) at (1,-4){};
\node(b2k)[whitedot,inner sep=0.25pt] at (2,-3.1){$\boldsymbol{*}$};
\node(b4k)[whitedot,inner sep=0.25pt] at (3,-3.1){$\boldsymbol{*}$};
\node(in3) at (2,-4){};
\node(in4) at (3,-4){};
\draw[string,in=270,out=90,looseness=0.6] (i1.center) to (b1.center);
\draw[string,in=180,out=180,looseness=1] (b1.center) to (y.center);
\draw[string,in=270,out=0,looseness=0.7] (b1.center) to (add1.center);
\draw[string,in=90,out=270] (add2.center) to (x.center);
\draw[string,in=270,out=90] (i3.center) to (b2.center);
\draw[string,in=270,out=left] (b27.center) to (mm.center);
\draw[string,in=0,out=90] (mm.center) to (y.center);
\draw[string,in=left,out=right,looseness=1.3] (b2.center) to (x7.center);
\draw[string,out=90,in=0] (b171.center) to (y7.center);
\draw[string,in=270,out=left] (b2.center) to (mm7.center);
\draw[string,in=180,out=90] (mm7.center) to (y7.center);
\draw[string] (add1.center) to (b172.center);
\draw[string] (add2.center) to (b171.center);
\draw[string,in=180,out=180] (add1.center) to (b171.center);
\draw[string,in=0,out=0] (add2.center) to (b172.center);
\draw[string,in=left,out=90] (b172.center) to (add7.center);
\draw[string,in=90,out=right] (add7.center) to (x7.center);
\draw[string,in=270,out=90] (add7.center) to (o7.center);
\draw[string,in=right,out=0,looseness=1] (t2.center) to (x7.center);
\draw[string,in=270,out=90] (i37.center) to (b27.center);
\draw[string,in=left,out=right] (b27.center) to (x.center);
\draw[string,in=180,out=0] (x.center) to (t2.center);
\draw[string] (b2k.center) to (in3.center);
\draw[string] (b4k.center) to (in4.center);
\draw[green][string,line width=1pt] (b4k.center) to (i37.center);
\draw[green][string,line width=1pt] (b2k.center) to (i3.center);
\draw[string,out=90,in=270] (t1.center) to (t2.center);
\draw[string,out=180,in=90] (t1.center) to (in1.center);
\draw[string,out=0,in=90] (t1.center) to (in2.center);
\end{pic}\\ \super{Red, black\eqref{eq:bialg}}=&\qquad\begin{pic}[scale=1/3]
\node(chi) [morphism,wedge,scale=0.5,connect n length=0.35cm,connect s length=0.35cm] at (-4,5.5) {\large$\chi$};
\node(b1)[blackdot,proofdiagram]  at (-4,-2) {};
\node(add)[reddot,proofdiagram] at (1,2.5){};
\node(x)[yellowdot,proofdiagram] at (2.5,1.75){};
\node(b27)[blackdot,proofdiagram] at (2,-2){};
\node(y)[yellowdot,proofdiagram] at (chi.connect s){};
\node(y2)[yellowdot,proofdiagram] at (chi.connect n){};
\node(i1) at (-4,-4){};
\node(i3)[whitedot,projdot] at (2,-2.75){};
\node(mm) at (-3,2){};
\node(t1)[blackdot,proofdiagram] at (0.25,-3.5){};
\node(t2)[blackdot,proofdiagram] at (4.5,1){};
\node (chi7) [morphism,wedge,scale=0.5,connect n length=0.35cm,connect s length=0.35cm] at (-1.8,5.5) {\large$\chi$};
\node(b17)[blackdot,proofdiagram]  at (1,3.5) {};
\node(add7)[reddot,proofdiagram] at (2,6){};
\node(x7)[yellowdot,proofdiagram] at (4.5,3){};
\node(b2)[blackdot,proofdiagram] at (3.75,-1.25){};
\node(y7)[yellowdot,proofdiagram] at (chi7.connect s){};
\node(y27)[yellowdot,proofdiagram] at (chi7.connect n){};
\node(m7) at (5,2.5){};
\node(i17) at (0.75,1){};
\node(i37) at (3,-2.75){};
\node(o7) at (2,7){};
\node(mm7) at (-2.5,4){};
\node(in1) at (-0.5,-4){};
\node(in2) at (1,-4){};
\node(b2k)[whitedot,inner sep=0.25pt] at (2,-3.1){$\boldsymbol{*}$};
\node(b4k)[whitedot,inner sep=0.25pt] at (3,-3.1){$\boldsymbol{*}$};
\node(in3) at (2,-4){};
\node(in4) at (3,-4){};
\draw[string,in=270,out=90,looseness=0.6] (i1.center) to (b1.center);
\draw[string,in=180,out=180,looseness=1] (b1.center) to (y.center);
\draw[string,in=left,out=0] (b1.center) to (add.center);
\draw[string,in=90,out=right] (add.center) to (x.center);
\draw[string,in=270,out=90] (i3.center) to (b27.center);
\draw[string,in=270,out=left] (b2.center) to (mm.center);
\draw[string,in=0,out=90] (mm.center) to (y.center);
\draw[string,in=left,out=right] (b2.center) to (x.center);
\draw[string,out=180,in=0] (b17.center) to (y7.center);
\draw[string,in=270,out=left] (b27.center) to (mm7.center);
\draw[string,in=180,out=90] (mm7.center) to (y7.center);
\draw[string] (add.center) to (b17.center);
\draw[string,in=left,out=0] (b17.center) to (add7.center);
\draw[string,in=90,out=right] (add7.center) to (x7.center);
\draw[string,in=270,out=90] (add7.center) to (o7.center);
\draw[string,in=right,out=0,looseness=1] (t2.center) to (x7.center);
\draw[string,in=270,out=90] (i37.center) to (b2.center);
\draw[string,in=left,out=right] (b27.center) to (x7.center);
\draw[string,in=180,out=0] (x.center) to (t2.center);
\draw[string] (b2k.center) to (in3.center);
\draw[string] (b4k.center) to (in4.center);
\draw[green][string,line width=1pt] (b4k.center) to (i37.center);
\draw[green][string,line width=1pt] (b2k.center) to (i3.center);
\draw[string,out=90,in=270] (t1.center) to (t2.center);
\draw[string,out=180,in=90] (t1.center) to (in1.center);
\draw[string,out=0,in=90] (t1.center) to (in2.center);
\end{pic}\qquad \super{Black, yellow\eqref{eq:sm}}= \quad  
\begin{pic}[scale=1/3]
\node(chi) [morphism,wedge,scale=0.5,connect n length=0.35cm,connect s length=0.35cm] at (-1.25,-0.1) {\large$\chi$};
\node(b1)[blackdot,proofdiagram]  at (-1.5,1.5) {};
\node(add)[reddot,proofdiagram] at (-1,2){};
\node(x)[yellowdot,proofdiagram] at (1.5,-0.25){};
\node(b2)[blackdot,proofdiagram] at (0,-2){};
\node(y)[yellowdot,proofdiagram] at (chi.connect s){};
\node(y2)[yellowdot,proofdiagram] at (chi.connect n){};
\node(k)[blackdot,proofdiagram] at (-0.4,-1.35){};
\node(i1) at (-4,-4){};
\node(i3)[whitedot,inner sep=0.25pt] at (1.8,-2.9){$\boldsymbol{*}$};
\node(mm) at (-1.9,-1.35){};
\node(t1)[blackdot,proofdiagram] at (0.25,-3.5){};
\node(t2)[blackdot,proofdiagram] at (2.25,-1){};
\node (chi7) [morphism,wedge,scale=0.5,connect n length=0.35cm,connect s length=0.35cm] at (1.75,3.9) {\large$\chi$};
\node(b17)[blackdot,proofdiagram]  at (1.5,5.5) {};
\node(add7)[reddot,proofdiagram] at (2,6){};
\node(x7)[yellowdot,proofdiagram] at (4.5,3.75){};
\node(b27)[blackdot,proofdiagram] at (3,2){};
\node(y7)[yellowdot,proofdiagram] at (chi7.connect s){};
\node(y27)[yellowdot,proofdiagram] at (chi7.connect n){};
\node(k7)[blackdot,proofdiagram] at (2.6,2.65){};
\node(m7) at (5,2.5){};
\node(i17) at (0.75,1){};
\node(i37)[whitedot,inner sep=0.25pt] at (3,-2.9){$\boldsymbol{*}$};
\node(o7) at (2,7){};
\node(mm7) at (1.1,2.65){};
\node(in1) at (-0.5,-4){};
\node(in2) at (1,-4){};
\node(b2k) at (2,-3.1){};
\node(b4k) at (3,-3.1){};
\node(in3) at (2,-4){};
\node(in4) at (3,-4){};
\draw[string,out=180,in=0] (k.center) to (y.center);
\draw[string,in=left,out=90,looseness=0.6] (i1.center) to (b1.center);
\draw[string,in=0,out=right,looseness=1.5] (b1.center) to (k.center);
\draw[string,in=left,out=90] (b1.center) to (add.center);
\draw[string,in=90,out=right] (add.center) to (x.center);
\draw[string,in=270,out=90] (i3.center) to (b2.center);
\draw[string,in=270,out=left] (b2.center) to (mm.center);
\draw[string,in=180,out=90] (mm.center) to (y.center);
\draw[string,in=left,out=right] (b2.center) to (x.center);
\draw[string,out=180,in=0] (k7.center) to (y7.center);
\draw[string,in=0,out=right,looseness=1.5] (b17.center) to (k7.center);
\draw[string,in=270,out=left] (b27.center) to (mm7.center);
\draw[string,in=180,out=90] (mm7.center) to (y7.center);
\draw[string,in=left,out=90,looseness=0.6] (add.center) to (b17.center);
\draw[string,in=left,out=90] (b17.center) to (add7.center);
\draw[string,in=90,out=right] (add7.center) to (x7.center);
\draw[string,in=270,out=90] (add7.center) to (o7.center);
\draw[string,in=right,out=90,looseness=1] (m7.center) to (x7.center);
\draw[string,in=270,out=90] (i37.center) to (b27.center);
\draw[string,in=left,out=right] (b27.center) to (x7.center);
\draw[string,in=180,out=0] (x.center) to (t2.center);
\draw[string,in=0,out=270,looseness=1] (m7.center) to (t2.center);
\draw[string,out=270,in=90] (b2k.center) to (in4.center);
\draw[string,out=270,in=90] (b4k.center) to (in3.center);
\draw[green][string,line width=1pt] (b4k.center) to (i37.center);
\draw[green][string,line width=1pt] (b2k.center) to (i3.center);
\draw[string,out=90,in=270] (t1.center) to (t2.center);
\draw[string,out=180,in=90] (t1.center) to (in1.center);
\draw[string,out=0,in=90] (t1.center) to (in2.center);
\end{pic}\\ 
 \super{\eqref{1}}=& \quad
\begin{pic}[scale=1/3]
\node (chi) [morphism,wedge,scale=0.5] at (0.4,1.2) {\large$\chi$};
\node(b1)[blackdot,proofdiagram]  at (-0.1,2.1) {};
\node(add)[reddot,proofdiagram] at (0.4,2.6){};
\node(x)[yellowdot,proofdiagram] at (1.85,0.4){};
\node(b2)[blackdot,proofdiagram] at (0.9,0){};
\node(*) at (1,-0.5){};
\node(m) at (2.15,0.15){};
\node(i1) at (-2,-5){};
\node(i2) at (-0.5,-2){};
\node(i3) at (2,-2.75){};
\node(i2a) at (1,-0.75){};
\node(i3a) at (1.75,-0.75){};
\node(o) at (0.75,2.5){};
\node(g1) at (1.75,-1.75){};
\node(g2) at (1.75,-1){};
\draw[string](g1.center) to (g2.center);
\draw[string,in=left,out=90,looseness=0.6] (i1.center) to (b1.center);
\draw[string,in=90,out=right] (b1.center) to (chi.north);
\draw[string,in=left,out=90] (b1.center) to (add.center);
\draw[string,in=90,out=right] (add.center) to (x.center);
\draw[string,in=270,out=90,looseness=1] (i2.center) to (i2a.center);
\draw[string,in=270,out=90,looseness=1] (i2a.center) to (m.center);
\draw[string,in=right,out=90,looseness=1] (m.center) to (x.center);
\draw[string,in=270,out=90] (i3.center) to (g1.center);
\draw[string,in=270,out=90] (g2.center) to (i3a.center);
\draw[string,in=270,out=90] (i3a.center) to (b2.center);
\draw[string,in=270,out=left] (b2.center) to (chi.south);
\draw[string,in=left,out=right] (b2.center) to (x.center);
\node (chi7) [morphism,wedge,scale=0.5] at (1.4,3.6) {\large$\chi$};
\node(b17)[blackdot,proofdiagram]  at (1,4.6) {};
\node(add7)[reddot,proofdiagram] at (2,5.1){};
\node(x7)[yellowdot,proofdiagram] at (3.1,3){};
\node(b27)[blackdot,proofdiagram] at (2.25,2.5){};
\node(*7) at (2.25,2){};
\node(m7) at (3.4,2.75){};
\node(i17) at (0.75,1){};
\node(i27) at (2.5,-1){};
\node(i37) at (3,-2.75){};
\node(i2a7) at (2.25,1.5){};
\node(i3a7) at (3,1.75){};
\node(o7) at (2,6){};
\node(g17) at (3,0.75){};
\node(g27) at (3,1.5){};
\draw[string](g17.center) to (g27.center);
\draw[string,in=left,out=90,looseness=0.6] (add.center) to (b17.center);
\draw[string,in=90,out=right] (b17.center) to (chi7.north);
\draw[string,in=left,out=90] (b17.center) to (add7.center);
\draw[string,in=90,out=right] (add7.center) to (x7.center);
\draw[string,in=270,out=90] (add7.center) to (o7.center);
\draw[string,in=270,out=90,looseness=1] (i27.center) to (i2a7.center);
\draw[string,in=270,out=90,looseness=1] (i2a7.center) to (m7.center);
\draw[string,in=right,out=90,looseness=1] (m7.center) to (x7.center);
\draw[string,in=270,out=90] (i37.center) to (g17.center);
\draw[string,in=270,out=90] (g27.center) to (i3a7.center);
\draw[string,in=270,out=90] (i3a7.center) to (b27.center);
\draw[string,in=270,out=left] (b27.center) to (chi7.south);
\draw[string,in=left,out=right] (b27.center) to (x7.center);
\node(in1) at (-0.5,-5){};
\node(in2) at (1,-5){};
\node(b2)[whitedot,inner sep=0.25pt] at (2,-3.1){$\boldsymbol{*}$};
\node(b4)[whitedot,inner sep=0.25pt] at (3,-3.1){$\boldsymbol{*}$};
\node(in3) at (2,-5){};
\node(in4) at (3,-5){};
\draw[string,out=270,in=90] (b2.center) to (in4.center);
\draw[string,out=270,in=90] (b4.center) to (in3.center);
\draw[string] (b4.center) to (i37.center);
\draw[string] (b2.center) to (i3.center);
\node(t1)[blackdot,proofdiagram] at (0.25,-3.5){};
\node(t2)[blackdot,proofdiagram] at (0.25,-2.75){};
\draw[string] (t1.center) to (t2.center);
\draw[string,out=180,in=90] (t1.center) to (in1.center);
\draw[string,out=0,in=90] (t1.center) to (in2.center);
\draw[string,in=180,out=270] (i2.center) to (t2.center);
\draw[string,in=0,out=270,looseness=1] (i27.center) to (t2.center);
\end{pic}\quad + \quad 0 \quad + \quad 0 \quad + \quad 0 
\\\super{\eqref{eq:proj2}\eqref{eq:projprop}}=& \quad
\begin{pic}[scale=1/3]
\node (chi) [morphism,wedge,scale=0.5] at (0.4,1.2) {\large$\chi$};
\node(b1)[blackdot,proofdiagram]  at (-0.1,2.1) {};
\node(add)[reddot,proofdiagram] at (0.4,2.6){};
\node(x)[yellowdot,proofdiagram] at (1.85,0.4){};
\node(b2)[blackdot,proofdiagram] at (0.9,0){};
\node(*) at (1,-0.5){};
\node(m) at (2.15,0.15){};
\node(i1) at (-2,-5){};
\node(i2) at (-0.5,-2){};
\node(i3) at (2,-2.75){};
\node(i2a) at (1,-0.75){};
\node(i3a) at (1.75,-0.75){};
\node(o) at (0.75,2.5){};
\node(g1)[whitedot,inner sep=0.25pt] at (1.75,-1.4){$\boldsymbol{*}$};
\node(g2) at (1.75,-1){};
\draw[string](g1.center) to (g2.center);
\draw[string,in=left,out=90,looseness=0.6] (i1.center) to (b1.center);
\draw[string,in=90,out=right] (b1.center) to (chi.north);
\draw[string,in=left,out=90] (b1.center) to (add.center);
\draw[string,in=90,out=right] (add.center) to (x.center);
\draw[string,in=270,out=90,looseness=1] (i2.center) to (i2a.center);
\draw[string,in=270,out=90,looseness=1] (i2a.center) to (m.center);
\draw[string,in=right,out=90,looseness=1] (m.center) to (x.center);
\draw[string,in=270,out=90] (i3.center) to (g1.center);
\draw[string,in=270,out=90] (g2.center) to (i3a.center);
\draw[string,in=270,out=90] (i3a.center) to (b2.center);
\draw[string,in=270,out=left] (b2.center) to (chi.south);
\draw[string,in=left,out=right] (b2.center) to (x.center);
\node (chi7) [morphism,wedge,scale=0.5] at (1.4,3.6) {\large$\chi$};
\node(b17)[blackdot,proofdiagram]  at (1,4.6) {};
\node(add7)[reddot,proofdiagram] at (2,5.1){};
\node(x7)[yellowdot,proofdiagram] at (3.1,3){};
\node(b27)[blackdot,proofdiagram] at (2.25,2.5){};
\node(*7) at (2.25,2){};
\node(m7) at (3.4,2.75){};
\node(i17) at (0.75,1){};
\node(i27) at (2.5,-1){};
\node(i37) at (3,-2.75){};
\node(i2a7) at (2.25,1.5){};
\node(i3a7) at (3,1.75){};
\node(o7) at (2,6){};
\node(g17) at (3,1.3){};
\node(g27)[whitedot,inner sep=0.25pt] at (3,1.3){$\boldsymbol{*}$};
\draw[string](g17.center) to (g27.center);
\draw[string,in=left,out=90,looseness=0.6] (add.center) to (b17.center);
\draw[string,in=90,out=right] (b17.center) to (chi7.north);
\draw[string,in=left,out=90] (b17.center) to (add7.center);
\draw[string,in=90,out=right] (add7.center) to (x7.center);
\draw[string,in=270,out=90] (add7.center) to (o7.center);
\draw[string,in=270,out=90,looseness=1] (i27.center) to (i2a7.center);
\draw[string,in=270,out=90,looseness=1] (i2a7.center) to (m7.center);
\draw[string,in=right,out=90,looseness=1] (m7.center) to (x7.center);
\draw[string,in=270,out=90] (i37.center) to (g17.center);
\draw[string,in=270,out=90] (g27.center) to (i3a7.center);
\draw[string,in=270,out=90] (i3a7.center) to (b27.center);
\draw[string,in=270,out=left] (b27.center) to (chi7.south);
\draw[string,in=left,out=right] (b27.center) to (x7.center);
\node(in1) at (-0.5,-5){};
\node(in2) at (1,-5){};
\node(b2)[whitedot,inner sep=0.25pt] at (2,-3.1){$\boldsymbol{*}$};
\node(b4)[whitedot,inner sep=0.25pt] at (3,-3.1){$\boldsymbol{*}$};
\node(in3) at (2,-5){};
\node(in4) at (3,-5){};
\draw[string,out=270,in=90] (b2.center) to (in4.center);
\draw[string,in=90,out=270] (b4.center) to (in3.center);
\draw[green][string,line width=1pt] (b4.center) to (i37.center);
\draw[green][string,line width=1pt] (b2.center) to (i3.center);
\node(t1)[blackdot,proofdiagram] at (0.25,-3.5){};
\node(t2)[blackdot,proofdiagram] at (0.25,-2.75){};
\draw[string] (t1.center) to (t2.center);
\draw[string,out=180,in=90] (t1.center) to (in1.center);
\draw[string,out=0,in=90] (t1.center) to (in2.center);
\draw[string,in=180,out=270] (i2.center) to (t2.center);
\draw[string,in=0,out=270,looseness=1] (i27.center) to (t2.center);
\end{pic}+
\begin{pic}[scale=1/3]
\node (chi) [morphism,wedge,scale=0.5] at (0.4,1.2) {\large$\chi$};
\node(b1)[blackdot,proofdiagram]  at (-0.1,2.1) {};
\node(add)[reddot,proofdiagram] at (0.4,2.6){};
\node(x)[yellowdot,proofdiagram] at (1.85,0.4){};
\node(b2)[blackdot,proofdiagram] at (0.9,0){};
\node(*) at (1,-0.5){};
\node(m) at (2.15,0.15){};
\node(i1) at (-2,-5){};
\node(i2) at (-0.5,-2){};
\node(i3) at (2,-2.75){};
\node(i2a) at (1,-0.75){};
\node(i3a) at (1.75,-0.75){};
\node(o) at (0.75,2.5){};
\node(g1)[whitedot,inner sep=0.25pt] at (1.75,-1.4){$\boldsymbol{*}$};
\node(g2) at (1.75,-1){};
\draw[string](g1.center) to (g2.center);
\draw[string,in=left,out=90,looseness=0.6] (i1.center) to (b1.center);
\draw[string,in=90,out=right] (b1.center) to (chi.north);
\draw[string,in=left,out=90] (b1.center) to (add.center);
\draw[string,in=90,out=right] (add.center) to (x.center);
\draw[string,in=270,out=90,looseness=1] (i2.center) to (i2a.center);
\draw[string,in=270,out=90,looseness=1] (i2a.center) to (m.center);
\draw[string,in=right,out=90,looseness=1] (m.center) to (x.center);
\draw[string,in=270,out=90] (i3.center) to (g1.center);
\draw[string,in=270,out=90] (g2.center) to (i3a.center);
\draw[string,in=270,out=90] (i3a.center) to (b2.center);
\draw[string,in=270,out=left] (b2.center) to (chi.south);
\draw[string,in=left,out=right] (b2.center) to (x.center);
\node (chi7)  at (1.75,3.25) {};
\node(b17)  at (1.5,4) {};
\node(add7)[reddot,proofdiagram] at (1.5,4.5){};
\node(x7) at (3,3){};
\node(b27) at (2.25,2.5){};
\node(*7) at (2.25,2){};
\node(m7) at (3,2.5){};
\node(i17) at (0.75,1){};
\node(i27) at (2.5,-1){};
\node(i2a7) at (2.25,1.75){};
\node(i3a7) at (3,1.75){};
\node(o7) at (1.5,6){};
\node(g17) at (3,0.75){};
\draw[string,in=left,out=90,looseness=0.6] (add.center) to (add7.center);
\draw[string,in=90,out=right] (add7.center) to (*7.center);
\draw[string,in=270,out=90] (add7.center) to (o7.center);
\draw[string,in=270,out=90,looseness=1] (i27.center) to (*7.center);
\node(in1) at (-0.5,-5){};
\node(in2) at (1,-5){};
\draw[string,in=180,out=270] (i2.center) to (t2.center);
\draw[string,in=0,out=270,looseness=1] (i27.center) to (t2.center);
\node(b2)[whitedot,inner sep=0.25pt] at (2,-3.1){$\boldsymbol{*}$};
\node[reddot,proofdiagram](b4) at (3,-2){};
\node(u1)[whitedot,inner sep=0.25pt] at (3,-3.1){$\boldsymbol{*}$};
\node(u2) at (3,-2.75){};
\draw[green][string,line width=1pt] (u1.center) to (u2.center);
\draw[green][string,line width=1pt] (b2.center) to (i3.center);
\node(in3) at (2,-5){};
\node(in4) at (3,-5){};
\draw[string,out=270,in=90] (b2.center) to (in4.center);
\draw[string] (b4.center) to (p2.center);
\draw[string,out=90,in=270] (in3.center) to (p1.center);
\node(t1)[blackdot,proofdiagram] at (0.25,-3.5){};
\node(t2)[blackdot,proofdiagram] at (0.25,-2.75){};
\draw[string] (t1.center) to (t2.center);
\draw[string,out=180,in=90] (t1.center) to (in1.center);
\draw[string,out=0,in=90] (t1.center) to (in2.center);
\draw[green][string,line width=1pt] (u1.center) to (u2.center);
\end{pic}+
\begin{pic}[scale=1/3]
\node (chi)  at (0.5,0.75) {};
\node(b1)  at (0.25,1.5) {};
\node(add)[reddot,proofdiagram] at (0.4,2){};
\node(x) at (1.75,0.5){};
\node(b2) at (1,0){};
\node(*) at (1,-0.5){};
\node(m) at (1.75,0){};
\node(i1) at (-2,-5){};
\node(i2) at (-0.5,-2){};
\node(i2a) at (1,-0.75){};
\node(i3a) at (1.75,-0.75){};
\node(o) at (0.75,2.5){};
\draw[string,in=left,out=90,looseness=0.6] (i1.center) to (add.center);
\draw[string,in=90,out=right] (add.center) to (x.center);
\draw[string,in=270,out=90,looseness=1] (i2.center) to (i2a.center);
\draw[string,in=270,out=90,looseness=1] (i2a.center) to (m.center);
\draw[string,in=270,out=90,looseness=1] (m.center) to (x.center);
\node (chi7) [morphism,wedge,scale=0.5] at (1.4,3.6) {\large$\chi$};
\node(b17)[blackdot,proofdiagram]  at (1,4.6) {};
\node(add7)[reddot,proofdiagram] at (2,5.1){};
\node(x7)[yellowdot,proofdiagram] at (3.1,3){};
\node(b27)[blackdot,proofdiagram] at (2.25,2.5){};
\node(*7) at (2.25,2){};
\node(m7) at (3.4,2.75){};
\node(i17) at (0.75,1){};
\node(i27) at (2.5,0){};
\node(i37) at (3,-2.75){};
\node(i2a7) at (2.25,1.5){};
\node(i3a7) at (3,1.75){};
\node(o7) at (2,6){};
\node(g17) at (3,1.3){};
\node(g27)[whitedot,inner sep=0.25pt] at (3,1.3){$\boldsymbol{*}$};
\draw[string](g17.center) to (g27.center);
\draw[string,in=left,out=90,looseness=0.6] (add.center) to (b17.center);
\draw[string,in=90,out=right] (b17.center) to (chi7.north);
\draw[string,in=left,out=90] (b17.center) to (add7.center);
\draw[string,in=90,out=right] (add7.center) to (x7.center);
\draw[string,in=270,out=90] (add7.center) to (o7.center);
\draw[string,in=270,out=90,looseness=1] (i27.center) to (i2a7.center);
\draw[string,in=270,out=90,looseness=1] (i2a7.center) to (m7.center);
\draw[string,in=right,out=90,looseness=1] (m7.center) to (x7.center);
\draw[string,in=270,out=90] (i37.center) to (g17.center);
\draw[string,in=270,out=90] (g27.center) to (i3a7.center);
\draw[string,in=270,out=90] (i3a7.center) to (b27.center);
\draw[string,in=270,out=left] (b27.center) to (chi7.south);
\draw[string,in=left,out=right] (b27.center) to (x7.center);
\node(in1) at (-0.5,-5){};
\node(in2) at (1,-5){};
\draw[string,in=180,out=270] (i2.center) to (t2.center);
\draw[string,in=0,out=270,looseness=0.7] (i27.center) to (t2.center);
\node[reddot,proofdiagram](b2) at (2,-2){};
\node(u1)[whitedot,inner sep=0.25pt] at (2,-3.1){$\boldsymbol{*}$};
\node(u2) at (2,-2.75){};
\draw[string] (b2.center) to (u1.center);
\node(b4)[whitedot,inner sep=0.25pt] at (3,-3.1){$\boldsymbol{*}$};
\draw[green][string,line width=1pt] (b4.center) to (i37.center);
\node(in3) at (2,-5){};
\node(in4) at (3,-5){};
\draw[string,out=270,in=90] (u1.center) to (in4.center);
\draw[string,in=90,out=270] (b4.center) to (in3.center);
\node(t1)[blackdot,proofdiagram] at (0.25,-3.5){};
\node(t2)[blackdot,proofdiagram] at (0.25,-2.75){};
\draw[string] (t1.center) to (t2.center);
\draw[string,out=180,in=90] (t1.center) to (in1.center);
\draw[string,out=0,in=90] (t1.center) to (in2.center);
\draw[green][string,line width=1pt] (u1.center) to (u2.center);
\end{pic}+
\begin{pic}[scale=1/3]
\node (chi)  at (0.5,0.75) {};
\node(b1)  at (0.25,1.5) {};
\node(add)[reddot,proofdiagram] at (0.75,2){};
\node(x) at (1.75,0.5){};
\node(b2) at (1,0){};
\node(*) at (1,-0.5){};
\node(m) at (1.75,0){};
\node(i1) at (-2,-5){};
\node(i2) at (-0.5,-2){};
\node(i2a) at (1,-0.75){};
\node(i3a) at (1.75,-0.75){};
\node(o) at (0.75,2.5){};
\draw[string,in=left,out=90,looseness=0.6] (i1.center) to (add.center);
\draw[string,in=90,out=right] (add.center) to (x.center);
\draw[string,in=270,out=90,looseness=1] (i2.center) to (i2a.center);
\draw[string,in=270,out=90,looseness=1] (i2a.center) to (m.center);
\draw[string,in=270,out=90,looseness=1] (m.center) to (x.center);
\node (chi7)  at (1.75,3.25) {};
\node(b17)  at (1.5,4) {};
\node(add7)[reddot,proofdiagram] at (1.5,4.5){};
\node(x7) at (3,3){};
\node(b27) at (2.25,2.5){};
\node(*7) at (2.25,2){};
\node(m7) at (3,2.5){};
\node(i17) at (0.75,1){};
\node(i27) at (2.325,0){};
\node(i2a7) at (2.25,1.75){};
\node(i3a7) at (3,1.75){};
\node(o7) at (1.5,6){};
\draw[string,in=left,out=90,looseness=0.6] (add.center) to (add7.center);
\draw[string,in=90,out=right] (add7.center) to (*7.center);
\draw[string,in=270,out=90] (add7.center) to (o7.center);
\draw[string,in=270,out=90,looseness=1] (i27.center) to (*7.center);
\node(in1) at (-0.5,-5){};
\node(in2) at (1,-5){};
\draw[string,in=180,out=270] (i2.center) to (t2.center);
\draw[string,in=0,out=270,looseness=0.7] (i27.center) to (t2.center);
\node[reddot,proofdiagram](b2) at (2,-2){};
\node[reddot,proofdiagram](b4) at (3,-2){};
\node(u1)[whitedot,inner sep=0.25pt] at (2,-3.1){$\boldsymbol{*}$};
\node(p1)[whitedot,inner sep=0.25pt] at (3,-3.1){$\boldsymbol{*}$};
\node(u2) at (2,-2.75){};
\node(p2) at (3,-2.75){};
\node(in3) at (2,-5){};
\node(in4) at (3,-5){};
\draw[string,out=270,in=90] (u1.center) to (in4.center);
\draw[string] (b2.center) to (u2.center);
\draw[string] (b4.center) to (p2.center);
\draw[string,in=270,out=90] (in3.center) to (p1.center);
\draw[string,green,line width=1pt] (u2.center) to (u1.center);
\draw[string,green,line width=1pt] (p2.center) to (p1.center);
\node(t1)[blackdot,proofdiagram] at (0.25,-3.5){};
\node(t2)[blackdot,proofdiagram] at (0.25,-2.75){};
\draw[string] (t1.center) to (t2.center);
\draw[string,out=180,in=90] (t1.center) to (in1.center);
\draw[string,out=0,in=90] (t1.center) to (in2.center);
\end{pic}\\ =&\begin{aligned}
\begin{tikzpicture}[scale=11/12]
\node (f) [morphism,wedge,scale=0.75, connect n ] at (0,1) {$U_{FF}$};
\node(i2)  at (-0.4125,-2.2) {};
\node(i2a)  at (0.4125,-2.2) {};
\node[whitedot,scale=1,inner sep=0.25pt](i3)  at (0.75,-1.7) {$\boldsymbol{*}$};
\node[whitedot,scale=1,inner sep=0.25pt](i4)  at (1.25,-1.7) {$\boldsymbol{*}$};
\node (e) [morphism,wedge,scale=0.75,connect sw length = 1.65cm] at (-0.24,-0.25){$U_{FF}$};
\node(b)[blackdot,scale=0.95] at (0,-1.3){};
\node(c)[blackdot,scale=0.95] at (0,-1.7){};
\node(m1) at (0.75,-0.25) {};
\node(m2) at (1,-0.25) {};
\draw[string,out=90,in=left] (i2.center) to (c.center);
\draw[string,out=90,in=right] (i2a.center) to (c.center);
\draw[string] (b.center) to (c.center);
\draw[string] (e.north) to (f.south west);
\draw[string,in=270,out=left] (b.center) to (e.south);
\draw[string,in=270,out=right,looseness=1.25] (b.center) to (m1.center);
\draw[string,in=270,out=90,looseness=1.25] (m1.center) to (f.south);
\draw[string,in=270,out=90] (i3.center) to (e.south east);
\draw[string,out=90,in=270] (i4.center) to (m2.center);
\draw[string,out=90,in=270] (m2.center) to (f.south east);
\draw[string,out=270,in=90] (i3.center) to +(0.5,-0.5);
\draw[string,out=270,in=90] (i4.center) to +(-0.5,-0.5);
\node(i)  at (-0.6,-2.2) {};
\draw[string,in=270,out=90] (i.center) to (e.connect sw);
\end{tikzpicture}
\end{aligned}
\end{align*}We can now combine equations~\eqref{eq:uffzero} and \eqref{eq:uffnzero} to obtain the following:

\begin{equation*}
\begin{aligned}
\begin{tikzpicture}[scale=0.75]
\node (f) [morphism,wedge,scale=0.75, connect n ] at (0,1) {$U_{FF}$};
\node(i2)  at (-0.4125,-2.2) {};
\node(i2a)  at (0.4125,-2.2) {};
\node[whitedot,scale=1,inner sep=0.25pt](i3)  at (0.75,-1.7) {$\boldsymbol{*}$};
\node[whitedot,scale=1,inner sep=0.25pt](i4)  at (1.25,-1.7) {$\boldsymbol{*}$};
\node (e) [morphism,wedge, scale=0.75,connect sw length = 1.65cm] at (-0.24,-0.25) {$U_{FF}$};
\node(b)[blackdot] at (0,-1.1){};
\node(c)[blackdot] at (0,-1.7){};
\node(m1) at (0.75,-0.25) {};
\node(m2) at (1,-0.25) {};
\draw[string,out=90,in=left] (i2.center) to (c.center);
\draw[string,out=90,in=right] (i2a.center) to (c.center);
\draw[string] (b.center) to (c.center);
\draw[string] (e.north) to (f.south west);
\draw[string,in=270,out=left] (b.center) to (e.south);
\draw[string,in=270,out=right,looseness=1.25] (b.center) to (m1.center);
\draw[string,in=270,out=90,looseness=1.25] (m1.center) to (f.south);
\draw[string,in=270,out=90] (i3.center) to (e.south east);
\draw[string,out=90,in=270] (i4.center) to (m2.center);
\draw[string,out=90,in=270] (m2.center) to (f.south east);
\draw[string] (i3.center) to +(0,-0.5);
\draw[string] (i4.center) to +(0,-0.5);
\end{tikzpicture}
\end{aligned}
\quad + \quad
\begin{aligned}
\begin{tikzpicture}[scale=0.75]
\node (f) [morphism,wedge,scale=0.75, connect n ] at (0,1) {$U_{FF}$};
\node(i2)  at (-0.4125,-1.7) {};
\node(i2a)  at (0.4125,-1.7) {};
\node (e) [morphism,wedge, scale=0.75,connect sw length = 1.65cm,connect s length = 1.65cm] at (-0.24,-0.25) {$U_{FF}$};
\node(m1) at (0.75,-0.25) {};
\node(m2) at (1,-0.25) {};
\node[state,black,scale=0.25,label={[label distance=-0.125cm]315:\tiny{0}}](b1) at (0.75,-1.4){};
\node[state,black,scale=0.25,hflip,label={[label distance=-0.125cm]45:\tiny{0}}](b2) at (0.75,-2){};
\node[state,black,scale=0.25,label={[label distance=-0.125cm]315:\tiny{0}}](b3) at (1.25,-1.4){};
\node[state,black,scale=0.25,hflip,label={[label distance=-0.125cm]45:\tiny{0}}](b4) at (1.25,-2){};
\draw[string] (e.north) to (f.south west);
\draw[string,in=270,out=90,looseness=1.25] (i2a.center) to (m1.center);
\draw[string,in=270,out=90,looseness=1.25] (m1.center) to (f.south);
\draw[string,in=270,out=90] (b1.center) to (e.south east);
\draw[string,out=90,in=270] (b3.center) to (m2.center);
\draw[string,out=90,in=270] (m2.center) to (f.south east);
\draw[string] (i2a.center) to +(0,-0.5);
\draw[string] (b2.center) to (0.75,-2.2);
\draw[string] (b4.center) to (1.25,-2.2);
\end{tikzpicture}
\end{aligned}
\quad = \quad
\begin{aligned}
\begin{tikzpicture}[scale=0.75]
\node (f) [morphism,wedge,scale=0.75, connect n ] at (0,1) {$U_{FF}$};
\node(i2)  at (-0.4125,-2.2) {};
\node(i2a)  at (0.4125,-2.2) {};
\node[whitedot,scale=1,inner sep=0.25pt](i3)  at (0.75,-1.7) {$\boldsymbol{*}$};
\node[whitedot,scale=1,inner sep=0.25pt](i4)  at (1.25,-1.7) {$\boldsymbol{*}$};
\node (e) [morphism,wedge,scale=0.75,connect sw length = 1.65cm] at (-0.24,-0.25){$U_{FF}$};
\node(b)[blackdot,scale=0.95] at (0,-1.3){};
\node(c)[blackdot,scale=0.95] at (0,-1.7){};
\node(m1) at (0.75,-0.25) {};
\node(m2) at (1,-0.25) {};
\draw[string,out=90,in=left] (i2.center) to (c.center);
\draw[string,out=90,in=right] (i2a.center) to (c.center);
\draw[string] (b.center) to (c.center);
\draw[string] (e.north) to (f.south west);
\draw[string,in=270,out=left] (b.center) to (e.south);
\draw[string,in=270,out=right,looseness=1.25] (b.center) to (m1.center);
\draw[string,in=270,out=90,looseness=1.25] (m1.center) to (f.south);
\draw[string,in=270,out=90] (i3.center) to (e.south east);
\draw[string,out=90,in=270] (i4.center) to (m2.center);
\draw[string,out=90,in=270] (m2.center) to (f.south east);
\draw[string,out=270,in=90] (i3.center) to +(0.5,-0.5);
\draw[string,out=270,in=90] (i4.center) to +(-0.5,-0.5);
\end{tikzpicture}
\end{aligned}
\quad + \quad
\begin{aligned}
\begin{tikzpicture}[scale=0.75]
\node (f) [morphism,wedge,scale=0.75, connect n ] at (0,1) {$U_{FF}$};
\node(i2)  at (-0.4125,-1.7) {};
\node(i2a)  at (0.4125,-1.7) {};
\node (e) [morphism,wedge, scale=0.75,connect sw length = 1.65cm] at (-0.24,-0.25){$U_{FF}$};
\node(m1) at (0.75,-0.25){} ;
\node(m2) at (1,-0.25) {};
\node[state,black,scale=0.25,label={[label distance=-0.125cm]315:\tiny{0}}](b1) at (0.75,-1.4){};
\node[state,black,scale=0.25,hflip,label={[label distance=-0.125cm]45:\tiny{0}}](b2) at (0.75,-2){};
\node[state,black,scale=0.25,label={[label distance=-0.125cm]315:\tiny{0}}](b3) at (1.25,-1.4){};
\node[state,black,scale=0.25,hflip,label={[label distance=-0.125cm]45:\tiny{0}}](b4) at (1.25,-2){};
\draw[string] (i2.center) to (e.south);
\draw[string] (e.north) to (f.south west);
\draw[string,in=270,out=90,looseness=1.25] (i2a.center) to (m1.center);
\draw[string,in=270,out=90,looseness=1.25] (m1.center) to (f.south);
\draw[string,in=270,out=90] (b1.center) to (e.south east);
\draw[string,out=90,in=270] (b3.center) to (m2.center);
\draw[string,out=90,in=270] (m2.center) to (f.south east);
\draw[string,out=270,in=90] (i2.center) to +(0.825,-0.5);
\draw[string,out=270,in=90] (i2a.center) to +(-0.825,-0.5);
\draw[string] (b2.center) to (0.75,-2.2);
\draw[string] (b4.center) to (1.25,-2.2);
\end{tikzpicture}
\end{aligned}
\end{equation*}This concludes the proof.
\end{proof}We now present an example of a partitioned UEB in dimension $d=4$ constructed from the finite field $\mathbb F_4$. 
\begin{example}The Fourier transform Hadamard for the additive group of $\mathbb F_4$ is given by the following matrix.
\begin{equation*}
\chi:=
\beginmatrix
1 & 1 &1 & 1 \\
1 & 1 & \minus1 & \minus1 \\
1 & \minus1 & \minus1 & 1 \\
1 & \minus1 & 1 & \minus1
\endmatrix
\end{equation*}
Let $\M_{ij}:=
\begin{aligned}  
\begin{tikzpicture}[scale=0.4]
\node (f) [morphism,wedge, connect s length=0.4cm width=0.5mm, connect se length=0.4cm, width=0.5cm, connect n length=0.8cm,connect sw length=0.8cm] at (0,0) {$U_{FF}$};
\node[state,scale=0.25,black,label={[label distance=-0.08cm]330:\tiny{$i$}}]  at (f.connect s){} ;
\node[state,scale=0.25,black,label={[label distance=-0.08cm]330:\tiny{$j$}}]  at (f.connect se) {};
\node  at (f.connect n) {};
\end{tikzpicture}
\end{aligned}
$ with $U_{FF}$ as defined in equation~\eqref{eq:mainconst}, then the partitioned UEB, with partitions $C_x, x \in \{ *,0,...,3\}$, is as follows:
\[
\scriptsize{ 
\begin{array}{llll}
\M_{00} =
\beginmatrix
1 & 0 & 0 & 0 \\
 0 & 1 & 0 & 0 \\
 0 & 0 & 1 & 0 \\
 0 & 0 & 0 & 1 \\
\endmatrix
&
\M_{01} =
\beginmatrix
1 & 0 & 0 & 0 \\
 0 & 1 & 0 & 0 \\
 0 & 0 & \minus1 & 0 \\
 0 & 0 & 0 & \minus1 \\
\endmatrix
&
\M_{02} =
\beginmatrix
1 & 0 & 0 & 0 \\
 0 & \minus1 & 0 & 0 \\
 0 & 0 & \minus1 & 0 \\
 0 & 0 & 0 & 1 \\
\endmatrix
&
\M_{03} =
\beginmatrix
1 & 0 & 0 & 0 \\
 0 & \minus1 & 0 & 0 \\
 0 & 0 & 1 & 0 \\
 0 & 0 & 0 & \minus1 \\
\endmatrix
\\
\M_{10} = 
\beginmatrix
 0 & 1 & 0 & 0 \\
 1 & 0 & 0 & 0 \\
 0 & 0 & 0 & 1 \\
 0 & 0 & 1 & 0 \\
\endmatrix
&
\M_{11} =
\beginmatrix
 0 & 1 & 0 & 0 \\
 1 & 0 & 0 & 0 \\
 0 & 0 & 0 & \minus1 \\
 0 & 0 & \minus1 & 0 \\
\endmatrix
 &
\M_{12} =
\beginmatrix
 0 & 0 & \minus1 & 0 \\
 0 & 0 & 0 & 1 \\
 1 & 0 & 0 & 1 \\
 0 & \minus1 & 0 & 0 \\
\endmatrix
&
\M_{13} =
\beginmatrix
 0 & 0 & 0 & \minus1 \\
 0 & 0 & 1 & 0 \\
 0 & \minus1 & 0 & 0 \\
 1 & 0 & 0 & 0 \\
\endmatrix
\\
\M_{20} =
\beginmatrix
 0 & 0 & 1 & 0 \\
 0 & 0 & 0 & 1 \\
 1 & 0 & 0 & 0 \\
 0 & 1 & 0 & 0 \\
\endmatrix
&
\M_{21} =
\beginmatrix
 0 & 0 & \minus1 & 0 \\
 0 & 0 & 0 & \minus1 \\
 1 & 0 & 0 & 1 \\
 0 & 1 & 0 & 0 \\
\endmatrix
&
\M_{22} =
\beginmatrix
 0 & 0 & 0 & 1 \\
 0 & 0 & \minus1 & 0 \\
 0 & \minus1 & 0 & 0 \\
 1 & 0 & 0 & 0 \\
\endmatrix
&
\M_{23} =
\beginmatrix
 0 & \minus1 & 0 & 0 \\
 1 & 0 & 0 & 0 \\
 0 & 0 & 0 & \minus1 \\
 0 & 0 & 1 & 0 \\
\endmatrix
\\
\M_{30} =
\beginmatrix
 0 & 0 & 0 & 1 \\
 0 & 0 & 1 & 0 \\
 0 & 1 & 0 & 0 \\
 1 & 0 & 0 & 0 \\
\endmatrix
&
\M_{31} =
\beginmatrix
 0 & 0 & 0 & \minus1 \\
 0 & 0 & \minus1 & 0 \\
 0 & 1 & 0 & 0 \\
 1 & 0 & 0 & 0 \\
\endmatrix
&
\M_{32} =
\beginmatrix
 0 & \minus1 & 0 & 0 \\
 1 & 0 & 0 & 0 \\
 0 & 0 & 0 & 1 \\
 0 & 0 & \minus1 & 0 \\
\endmatrix
&
\M_{33} =
\beginmatrix
 0 & 0 & 1 & 0 \\
 0 & 0 & 0 & \minus1 \\
 1 & 0 & 0 & 1 \\
 0 & \minus1 & 0 & 0 \\
\endmatrix
\end{array}
}
\]
The partitions are:
\begin{align*}
C_*:=&\{ \M_{10},\M_{20},\M_{30}\}
\\ C_0:=&\{\M_{01},\M_{02},\M_{03}\}
\\ C_1:=&\{\M_{11},\M_{12},\M_{13}\}
\\ C_2:=&\{\M_{21},\M_{22},\M_{23}\}
\\ C_3:=&\{\M_{31},\M_{32},\M_{33}\}
\end{align*}
Thus since $\M_{00}=\I_4$, we have:
\begin{equation*}
U_{FF}=\{\I_4\}\sqcup C_* \sqcup C_0 \sqcup C_1 \sqcup C_2 \sqcup C_3
\end{equation*}
It can easily be  verified that this is a partition into maximal commuting sub-families and that $U_{FF}$ is a UEB. 
\end{example} 
\section{conclusion}
We have introduced a tensor diagrammatic characterization of maximal families of MUBs, partitioned unitary error bases, Hadamards and controlled Hadamards. As an application of these tensor diagrammatic characterizations we have introduced a new construction for a $\UEBM$ from a maximal family of MUBs extending work by Bandyopadhyay~\cite{Bandyopadhyay} ,which makes clear the exact nature of the correspondence between $\UEBM$s and maximal families of MUBs. Each $\UEBM$ gives rise to a unique maximal family of MUBs. Each maximal family of MUBs gives rise to an infinite family of possibly inequivalent $\UEBM$s each $\UEBM$ corresponding to a choice of controlled Hadamard. Further work in this direction is to investigate whether the property of monomiality of UEBs, introduced by Wocjan et al~\cite{boykin}, is invariant under the choice of controlled Hadamard in our construction to ensure the property is well defined.

We have also  introduced a tensor diagrammatic characterization of finite fields as algebraic structures defined over Hilbert spaces, extending existing characterizations of abelian groups~\cite{stefwill}. As an application of this and a further application of the tensor diagrammatic characterizations of $\UEBM$s we introduced a new construction of  $\UEBM$s and thus maximal families of MUBs from a finite field. This is different from the construction due to Bandyopadhyay et al~\cite{Bandyopadhyay}, with the partition being easier to calculate. Further work is necessary to investigate whether this construction could be adapted to one requiring less structure than that of a finite field.
      
\bibliographystyle{eptcs}
\bibliography{FiniteFields}
\appendix
\section{Minor lemmas}
In this section we present a number of minor diagrammatic lemmas that are essential to the proofs of the main Theorems. We assume throughout that all wires are $d$-dimensional Hilbert spaces and the operators are as defined in Section~\ref{section:gff}.
\begin{lemma}\label{lem:zerohad}Given a controlled Hadamard $H^i$ and $\dfrob$ $\tinymult$, the following equation holds.
\begin{equation}\label{eq:zerohad}
\begin{aligned}\begin{tikzpicture}[scale=0.5]
           \node[whitedot,inner sep=1pt,scale=0.5] (h) at (0,1) {$H$};
           \node[whitedot,inner sep=0.25pt,scale=0.5] at (0,1.6){$\boldsymbol{*}$};
           \node[blackdot,proofdiagram](b1) at (0,-0.25){};
           \node[blackdot,proofdiagram](b2) at (-1,0){};
           \draw[string] (0,2) to (b1.center);
           \draw[string](-1,2) to (-1,-1);
           \draw[string] (b2.center) to (h.center);
\end{tikzpicture}\end{aligned}
= 0
\end{equation}
\end{lemma} 
\begin{proof}We take the LHS of equation~\eqref{eq:zerohad}:
\begin{equation}
\begin{aligned}\begin{tikzpicture}[scale=0.5]
           \node[whitedot,inner sep=1pt,scale=0.5] (h) at (0,1.25) {$H$};
           \node[whitedot,inner sep=0.25pt,scale=0.5] at (0,1.8){$\boldsymbol{*}$};
           \node[blackdot,proofdiagram](b1) at (0,-0.25){};
           \node[blackdot,proofdiagram](b2) at (-1,0.75){};
           \draw[string] (0,2.25) to (b1.center);
           \draw[string](-1,2.25) to (-1,-1.25);
           \draw[string] (b2.center) to (h.center);
\end{tikzpicture}\end{aligned}
\super{\eqref{eq:un}}=
\begin{aligned}\begin{tikzpicture}[scale=0.5]
           \node[whitedot,inner sep=1pt,scale=0.5] (h) at (0.1,1.25) {$H$};
           \node[whitedot,inner sep=0.25pt,scale=0.5] at (0.1,1.8){$\boldsymbol{*}$};
           \node[blackdot,proofdiagram](b1) at (0.1,0){};
           \node[blackdot,proofdiagram](b2) at (-1,0.75){};
           \node[blackdot,proofdiagram](b3) at (-1,0.25){};  
           \node[blackdot,proofdiagram](b4) at (-0.25,-0.25){};                          \draw[string] (0.1,2.25) to (b1.center);
           \draw[string] (b4.center) to (b3.center);
           \draw[string](-1,2.25) to (-1,-1.25);
           \draw[string] (b2.center) to (h.center);
\end{tikzpicture}\end{aligned}
\super{\eqref{eq:had1}}=
\begin{aligned}\begin{tikzpicture}[scale=0.5]
           \node[whitedot,inner sep=1pt,scale=0.5] (h) at (0,1.25) {$H$};
           \node[whitedot,inner sep=0.25pt,scale=0.5] at (0,1.8){$\boldsymbol{*}$};
           \node[whitedot,inner sep=1pt,scale=0.5] (h1) at (0,-0.25) {$H$};
           \node[black,state,scale=0.35,label={[label distance=0.05cm]330:\tiny{$0$}}](b1) at (0,-0.75){};
           \node[blackdot,proofdiagram](b2) at (-1,0.75){};           
           \node[blackdot,proofdiagram](b3) at (-1,0.25){};
           \draw[string] (0,2.25) to (b1.center);
           \draw[string](-1,2.25) to (-1,-1.25);
           \draw[string] (b2.center) to (h.center);
           \draw[string] (b3.center) to (h1.center);
\end{tikzpicture}\end{aligned}
\super{\eqref{eq:ch}}=d
\begin{aligned}\begin{tikzpicture}[scale=0.5]
           \node[black,state,scale=0.35,label={[label distance=0.05cm]330:\tiny{$0$}}](b1) at (0,-0.25){};
           \node[whitedot,inner sep=0.25pt,scale=0.5] at (0,1.6){$\boldsymbol{*}$};
           \draw[string] (0,2) to (b1.center);
           \draw[string](-1,2) to (-1,-1.5);
\end{tikzpicture}\end{aligned}\super{\eqref{eq:projzero}}=0
\end{equation}
\end{proof}
\begin{lemma}
The following equation holds.
\begin{align}\label{eq:proj2}
 \begin{aligned}\begin{tikzpicture}[scale=0.65]
           \node(3)[whitedot,projdot] at (0,0){};
           \node (1) at (0,-1) {}; 
           \node[reddot] (2) at (0,1) {}; 
           \draw[string] (2.center) to (3.center);
           \draw[green][string,line width=1pt] (3.center) to (1.center);           \end{tikzpicture}
\end{aligned} 
\quad = \quad
0 
\end{align}
\end{lemma}
\begin{proof}
By Lemma~\ref{lem:a}:
\begin{align*}\begin{aligned}\begin{tikzpicture}[scale=0.65]
           \node(4) at (0,-2){}; 
           \node(3)[whitedot,projdot] at (0,0){};
           \node (1)[whitedot,projdot] at (0,-1) {}; 
           \node (2) at (0,1) {}; 
           \draw[string] (2.center) to (3.center);
           \draw[string] (4.center) to (1.center); 
           \draw[green][string,line width=1pt] (3.center) to (1.center);           
\end{tikzpicture}
\end{aligned} 
\quad = \quad
\begin{aligned}\begin{tikzpicture}[scale=0.65]
           \node(4) at (0,-2){};           
           \node (2) at (0,1) {}; 
           \draw[string] (2.center) to (4.center);
               \end{tikzpicture}
\end{aligned} 
\quad - \quad
 \begin{aligned}\begin{tikzpicture}[scale=0.65]
           \node(4) at (0,-2){}; 
           \node[reddot](3) at (0,0){};
           \node[reddot] (1) at (0,-1) {}; 
           \node (2) at (0,1) {}; 
           \draw[string] (2.center) to (3.center);
           \draw[string] (4.center) to (1.center);
\end{tikzpicture}
\end{aligned}
\quad \Rightarrow 
\begin{aligned}\begin{tikzpicture}[scale=0.65]
           \node(4) at (0,-2){}; 
           \node(3)[whitedot,projdot] at (0,0.25){};
           \node (1)[whitedot,projdot] at (0,-0.75) {}; 
           \node[reddot] (2) at (0,1) {}; 
           \node[whitedot,projdot](8) at (0,-1.5) {};
           \draw[string] (2.center) to (3.center);
           \draw[string] (4.center) to (1.center); 
           \draw[green][string,line width=1pt] (3.center) to (1.center);           
           \draw[green][string,line width=1pt] (4.center) to (0,-1.5);
\end{tikzpicture}
\end{aligned} 
\quad = \quad
\begin{aligned}\begin{tikzpicture}[scale=0.65]
           \node(4) at (0,-2){};           
           \node[reddot] (2) at (0,1) {}; 
           \draw[string] (2.center) to (4.center);
                     \draw[green][string,line width=1pt] (4.center) to (0,-1.5);
               \end{tikzpicture}
\end{aligned} 
\quad - \quad
 \begin{aligned}\begin{tikzpicture}[scale=0.65]
           \node(4) at (0,-2){}; 
           \node[reddot](3) at (0,0){};
           \node[reddot] (1) at (0,-1) {}; 
           \node[reddot] (2) at (0,1) {}; 
           \node[whitedot,projdot](8) at (0,-1.5) {};
           \draw[string] (2.center) to (3.center);
           \draw[string] (4.center) to (1.center);
           \draw[green][string,line width=1pt] (4.center) to (8.center);
\end{tikzpicture}\end{aligned}\quad \super{\eqref{eq:projprop}}\Rightarrow
 \begin{aligned}\begin{tikzpicture}[scale=0.65]
           \node(4) at (0,-2){}; 
           \node(3) at (0,0.25){};
           \node (1) at (0,-0.75) {}; 
           \node[reddot] (2) at (0,1) {}; 
           \node[whitedot,projdot](8) at (0,-0.5) {};
           \draw[string] (2.center) to (3.center);
           \draw[string] (4.center) to (1.center); 
           \draw[string] (3.center) to (1.center);           
           \draw[green][string,line width=1pt] (4.center) to (0,-0.5);
\end{tikzpicture}
\end{aligned} 
\quad = \quad
\begin{aligned}\begin{tikzpicture}[scale=0.65]
           \node(4) at (0,-2){}; 
           \node(3) at (0,0.25){};
           \node (1) at (0,-0.75) {}; 
           \node[reddot] (2) at (0,1) {}; 
           \node[whitedot,projdot](8) at (0,-0.5) {};
           \draw[string] (2.center) to (3.center);
           \draw[string] (4.center) to (1.center); 
           \draw[string] (3.center) to (1.center);           
           \draw[green][string,line width=1pt] (4.center) to (0,-0.5);
               \end{tikzpicture}
\end{aligned} 
\quad - \quad
 \begin{aligned}\begin{tikzpicture}[scale=0.65]
           \node(4) at (0,-2){}; 
           \node(3) at (0,0.25){};
           \node (1) at (0,-0.75) {}; 
           \node[reddot] (2) at (0,1) {}; 
           \node[whitedot,projdot](8) at (0,-0.5) {};
           \draw[string] (2.center) to (3.center);
           \draw[string] (4.center) to (1.center); 
           \draw[string] (3.center) to (1.center);           
           \draw[green][string,line width=1pt] (4.center) to (0,-0.5);
\end{tikzpicture}\end{aligned}=0
\end{align*}

\end{proof} 
\begin{lemma}
The following equation holds.
\begin{equation}\label{eq:lem2}\begin{aligned}\begin{tikzpicture}[xscale=0.65,yscale=-0.5]
\node(b)[blackdot] at (0,0){};
\node(1)[whitedot,projdot] at (0,0.5){};
\node(2) at (-0.5,-0.5){};
\node(3) at (0.5,-0.5){};
\node(i1) at (-0.5,-1.5){};
\node(i2) at (0.5,-1.5){};
\node(o)[whitedot,projdot] at (0,1.5){};
\draw[string,out=90, in=left] (2.center) to (b.center);
\draw[string,out=90, in=right] (3.center) to (b.center);
\draw[string,out=270, in=90] (1.center) to (b.center);
\draw[green][string,line width=1pt,out=90,in=270] (1.center) to (o.center);
\draw[string,out=90,in=270] (0,2.5) to (o.center);
\draw[string,in=90,out=270] (2.center) to (i1.center);
\draw[string,in=90,out=270] (3.center) to (i2.center);
\draw[string,in=90,out=270] (-0.5,-2.5) to (i1.center);
\draw[string,in=90,out=270] (0.5,-2.5) to (i2.center);
\end{tikzpicture}\end{aligned}
\quad = \quad
\begin{aligned}\begin{tikzpicture}[xscale=0.65,yscale=-0.5]
\node(b)[blackdot] at (0,0){};
\node(1) at (0,0.5){};
\node(2)[whitedot,projdot] at (-0.5,-0.5){};
\node(3)[whitedot,projdot] at (0.5,-0.5){};
\node(i1)[whitedot,projdot] at (-0.5,-1.5){};
\node(i2)[whitedot,projdot] at (0.5,-1.5){};
\node(o) at (0,1.5){};
\draw[string,out=90, in=left] (2.center) to (b.center);
\draw[string,out=90, in=right] (3.center) to (b.center);
\draw[string,out=270, in=90] (1.center) to (b.center);
\draw[string,out=90,in=270] (1.center) to (o.center);
\draw[green][string,line width=1pt,in=90,out=270] (2.center) to (i1.center);
\draw[green][string,line width=1pt,in=90,out=270] (3.center) to (i2.center);
\draw[string,in=90,out=270] (-0.5,-2.5) to (i1.center);
\draw[string,in=90,out=270] (0.5,-2.5) to (i2.center);
\draw[string,out=90,in=270] (0,2.5) to (o.center);
\end{tikzpicture}\end{aligned}
\quad = \quad
\begin{aligned}\begin{tikzpicture}[xscale=0.65,yscale=-0.5]
\node(b)[blackdot] at (0,0){};
\node(1)[whitedot,projdot] at (0,0.5){};
\node(2)[whitedot,projdot] at (-0.5,-0.5){};
\node(3)[whitedot,projdot] at (0.5,-0.5){};
\node(i1)[whitedot,projdot] at (-0.5,-1.5){};
\node(i2)[whitedot,projdot] at (0.5,-1.5){};
\node(o)[whitedot,projdot] at (0,1.5){};
\draw[string,out=90, in=left] (2.center) to (b.center);
\draw[string,out=90, in=right] (3.center) to (b.center);
\draw[string,out=270, in=90] (1.center) to (b.center);
\draw[green][string,line width=1pt,out=90,in=270] (1.center) to (o.center);
\draw[green][string,line width=1pt,in=90,out=270] (2.center) to (i1.center);
\draw[green][string,line width=1pt,in=90,out=270] (3.center) to (i2.center);
\draw[string,in=90,out=270] (-0.5,-2.5) to (i1.center);
\draw[string,in=90,out=270] (0.5,-2.5) to (i2.center);
\draw[string,out=90,in=270] (0,2.5) to (o.center);
\end{tikzpicture}\end{aligned}\end{equation}
\end{lemma}
\begin{proof}
By Lemma~\ref{lem:a} and Definition~\ref{def:bialg}:
\begin{equation*}\label{eq:lem2}\begin{aligned}\begin{tikzpicture}[xscale=0.65,yscale=-0.5]
\node(b)[blackdot] at (0,0){};
\node(1)[whitedot,projdot] at (0,0.5){};
\node(2) at (-0.5,-0.5){};
\node(3) at (0.5,-0.5){};
\node(i1) at (-0.5,-1.5){};
\node(i2) at (0.5,-1.5){};
\node(o)[whitedot,projdot] at (0,1.5){};
\draw[string,out=90, in=left] (2.center) to (b.center);
\draw[string,out=90, in=right] (3.center) to (b.center);
\draw[string,out=270, in=90] (1.center) to (b.center);
\draw[green][string,line width=1pt,out=90,in=270] (1.center) to (o.center);
\draw[string,out=90,in=270] (0,2.5) to (o.center);
\draw[string,in=90,out=270] (2.center) to (i1.center);
\draw[string,in=90,out=270] (3.center) to (i2.center);
\draw[string,in=90,out=270] (-0.5,-2.5) to (i1.center);
\draw[string,in=90,out=270] (0.5,-2.5) to (i2.center);
\end{tikzpicture}\end{aligned}
\quad = \quad
\begin{aligned}\begin{tikzpicture}[xscale=0.65,yscale=-0.5]
\node(b)[blackdot] at (0,0){};
\node(1) at (0,0.5){};
\node(2)[whitedot,projdot] at (-0.5,-0.5){};
\node(3)[whitedot,projdot] at (0.5,-0.5){};
\node(i1)[whitedot,projdot] at (-0.5,-1.5){};
\node(i2)[whitedot,projdot] at (0.5,-1.5){};
\node(o) at (0,1.5){};
\draw[string,out=90, in=left] (2.center) to (b.center);
\draw[string,out=90, in=right] (3.center) to (b.center);
\draw[string,out=270, in=90] (1.center) to (b.center);
\draw[string,out=90,in=270] (1.center) to (o.center);
\draw[green][string,line width=1pt,in=90,out=270] (2.center) to (i1.center);
\draw[green][string,line width=1pt,in=90,out=270] (3.center) to (i2.center);
\draw[string,in=90,out=270] (-0.5,-2.5) to (i1.center);
\draw[string,in=90,out=270] (0.5,-2.5) to (i2.center);
\draw[string,out=90,in=270] (0,2.5) to (o.center);
\end{tikzpicture}\end{aligned}
\quad = \quad
\begin{aligned}\begin{tikzpicture}[xscale=0.65,yscale=-0.5]
\node(b)[blackdot] at (0,0){};
\node(1)[whitedot,projdot] at (0,0.5){};
\node(2)[whitedot,projdot] at (-0.5,-0.5){};
\node(3)[whitedot,projdot] at (0.5,-0.5){};
\node(i1)[whitedot,projdot] at (-0.5,-1.5){};
\node(i2)[whitedot,projdot] at (0.5,-1.5){};
\node(o)[whitedot,projdot] at (0,1.5){};
\draw[string,out=90, in=left] (2.center) to (b.center);
\draw[string,out=90, in=right] (3.center) to (b.center);
\draw[string,out=270, in=90] (1.center) to (b.center);
\draw[green][string,line width=1pt,out=90,in=270] (1.center) to (o.center);
\draw[green][string,line width=1pt,in=90,out=270] (2.center) to (i1.center);
\draw[green][string,line width=1pt,in=90,out=270] (3.center) to (i2.center);
\draw[string,in=90,out=270] (-0.5,-2.5) to (i1.center);
\draw[string,in=90,out=270] (0.5,-2.5) to (i2.center);
\draw[string,out=90,in=270] (0,2.5) to (o.center);
\end{tikzpicture}\end{aligned}
\quad =\quad
\begin{aligned}\begin{tikzpicture}[xscale=0.86,yscale=-0.86]
\node(b)[reddot] at (0,0.25){};
\node(1) at (0,0.5){};
\node[reddot](2) at (-0.5,-0.25){};
\node[reddot](3) at (0.5,-0.25){};
\node(i1) at (-0.5,-1.5){};
\node(i2) at (0.5,-1.5){};
\node(o) at (0,1.5){};
\draw[string,out=270, in=90] (1.center) to (b.center);
\draw[string,out=90,in=270] (1.center) to (o.center);
\draw[string,in=90,out=270] (2.center) to (i1.center);
\draw[string,in=90,out=270] (3.center) to (i2.center);
\end{tikzpicture}\end{aligned}
\end{equation*}
 \end{proof}
\end{document}